 \newcommand{\mydriver}{pdflatex} %Making a PDF directly using pdflatex.

\documentclass[12pt,\mydriver]{thesis}  %12pt is larger than 11pt

\usepackage{titlesec}
   \titleformat{\chapter}
      {\normalfont\large}{Chapter \thechapter:}{1em}{}
\usepackage{natbib}
\usepackage{epigraph}
\usepackage{graphicx}
\usepackage{cite}
\usepackage{lscape}
\usepackage{indentfirst}
\usepackage{latexsym}
\usepackage{multirow}
\usepackage{epstopdf}
\usepackage{tabls}
\usepackage{wrapfig}
\usepackage{slashbox}
\usepackage{longtable}
\usepackage{supertabular}
\usepackage{setspace}

\usepackage{amsmath}
\usepackage{amssymb}
\usepackage{amsfonts}
\usepackage{amsthm}
\usepackage{stmaryrd}
\usepackage{float}
\usepackage[caption = false]{subfig}

\usepackage[normalem]{ulem}
\newcommand{\stkout}[1]{\ifmmode\text{\sout{\ensuremath{#1}}}\else\sout{#1}\fi}
\newcommand*{\defeq}{\mathrel{\vcenter{\baselineskip0.6ex \lineskiplimit0pt
                     \hbox{\scriptsize.}\hbox{\scriptsize.}}}%
                     =}
\newcommand{\gke}{\texttt{Gkeyll}}

\newcommand{\eqr}[1]{Eq.\thinspace(#1)}
\newcommand{\pfrac}[2]{\frac{\partial #1}{\partial #2}}

\newcommand{\pfraca}[1]{\frac{\partial}{\partial #1}}

\newcommand{\mvec}[1]{\mathbf{#1}}
\newcommand{\gvec}[1]{\boldsymbol{#1}}

\newcommand{\gz}{\nabla_{\mvec{z}}}
\newcommand{\gx}{\nabla_{\mvec{x}}}
\newcommand{\gv}{\nabla_{\mvec{v}}}
\newcommand{\dz}{\thinspace d\mvec{z}}
\newcommand{\dv}{\thinspace d\mvec{v}}
\newcommand{\dx}{\thinspace d\mvec{x}}

\newtheorem{proposition}{Proposition}
\newtheorem{corollary}{Corollary}
\newtheorem{lemma}{Lemma}

\DeclareMathOperator{\sech}{sech}
\DeclareMathOperator{\sign}{sign}
\DeclareMathOperator*{\spn}{span}

\usepackage[colorlinks=true,urlcolor=black,linkcolor=blue,citecolor=blue]{hyperref}

\renewcommand{\baselinestretch}{2}
\begin{document}
\pagestyle{empty}
%Abstract Page

\hbox{\ }

\renewcommand{\baselinestretch}{1}
\small \normalsize

\begin{center}
\large{{ABSTRACT}}

\vspace{3em}

\end{center}
\hspace{-.15in}
\begin{tabular}{ll}
Title of dissertation:    & {\large A Deep Dive into the Distribution Function: }\\
&                     {\large Understanding Phase Space Dynamics} \\ 
&                     {\large with Continuum Vlasov--Maxwell Simulations} \\
\ \\
&                          {\large  James Juno} \\
&                           {\large Doctor of Philosophy, 2020} \\
\ \\
Dissertation directed by: & {\large  Professor William Dorland} \\
&               {\large  Department of Physics } \\
\end{tabular}

\vspace{3em}

\renewcommand{\baselinestretch}{2}
\large \normalsize

In collisionless and weakly collisional plasmas, the particle distribution function is a rich tapestry of the underlying physics. However, actually leveraging the particle distribution function to understand the dynamics of a weakly collisional plasma is challenging. The equation system of relevance, the Vlasov--Maxwell--Fokker--Planck (VM-FP) system of equations, is difficult to numerically integrate, and traditional methods such as the particle-in-cell method introduce counting noise into the distribution function. 

In this thesis, we present a new algorithm for the discretization of VM-FP system of equations for the study of plasmas in the kinetic regime. Using the discontinuous Galerkin (DG) finite element method for the spatial discretization and a third order strong-stability preserving Runge--Kutta for the time discretization, we obtain an accurate solution for the plasma's distribution function in space and time. 

We both prove the numerical method retains key physical properties of the VM-FP system, such as the conservation of energy and the second law of thermodynamics, and demonstrate these properties numerically. These results are contextualized in the history of the DG method. We discuss the importance of the algorithm being \emph{alias-free}, a necessary condition for deriving stable DG schemes of kinetic equations so as to retain the implicit conservation relations embedded in the particle distribution function, and the computational favorable implementation using a \emph{modal, orthonormal} basis in comparison to traditional DG methods applied in computational fluid dynamics. 

A diverse array of simulations are performed which exploit the advantages of our approach over competing numerical methods. We demonstrate how the high fidelity representation of the distribution function, combined with novel diagnostics, permits detailed analysis of the energization mechanisms in fundamental plasma processes such as collisionless shocks. Likewise, we show the undesirable effect particle noise can have on both solution quality, and ease of analysis, with a study of kinetic instabilities with both our continuum VM-FP method and a particle-in-cell method. 

Our VM-FP solver is implemented in the \texttt{Gkyell} framework\footnote{https://github.com/ammarhakim/gkyl}, a modular framework for the solution to a variety of equation systems in plasma physics and fluid dynamics.  %(must be first, required, non-numbered)
%Titlepage

\thispagestyle{empty}
\hbox{\ }
\vspace{1in}
\renewcommand{\baselinestretch}{1}
\small\normalsize
\begin{center}

\large{{A Deep Dive into the Distribution Function: Understanding Phase Space Dynamics using Continuum Vlasov--Maxwell Simulations}}\\
\ \\
\ \\
\large{by} \\
\ \\
\large{James Juno}%Your full name as it appears in University records.
\ \\
\ \\
\ \\
\ \\
\normalsize
Dissertation submitted to the Faculty of the Graduate School of the \\
University of Maryland, College Park in partial fulfillment \\
of the requirements for the degree of \\
Doctor of Philosophy \\
2020
\end{center}

\vspace{7.5em}

\noindent Advisory Committee: \\
Professor William Dorland, Chair/Advisor \\
Dr. Jason TenBarge, Co-Advisor \\
Professor James Drake \\
Professor Adil Hassam \\
Professor Jacob Bedrossian
 %(must follow Abstract, required, non-numbered)
%Copyright

\thispagestyle{empty}
\hbox{\ }

\vfill
\renewcommand{\baselinestretch}{1}
\small\normalsize

\vspace{-.65in}

\begin{center}
\large{\copyright \hbox{ }Copyright by\\
James Juno  %Type your name as it appears in University records
\\
2020}
\end{center}

\vfill

\newpage

\hbox{\ }
\newpage
 %(highly recommended, non-numbered)

%Pages from this point start at lower-case Roman number ii)
\pagestyle{plain} \pagenumbering{roman} \setcounter{page}{2}
\addcontentsline{toc}{chapter}{Dedication}
%Dedication

\renewcommand{\baselinestretch}{2}
\small\normalsize
\hbox{\ }
 
\vspace{-.65in}

\begin{center}
\large{Dedication}
\end{center} 

To the memory of my father, Jim Juno, and to my wife, Anna
 %(if present, lower-case Roman)
\addcontentsline{toc}{chapter}{Preface}
%Preface

\renewcommand{\baselinestretch}{2}
\small\normalsize
\hbox{\ }
 
\vspace{-.65in}

\begin{center}
\large{Preface} 
\end{center} 

\vspace{1ex}

This thesis was an enormous labor of love, and if you are reading it now with the intention of learning about what I, and the \gke~project, accomplished, from the bottom of my heart: thank you. 
The length of this thesis requires a preface about my goals and what I hope a reader comes away with after reading it.

At every turn, we in the \gke~project have attempted to make the code accessible and user-friendly, and I think we have broadly accomplished this goal. 
I feel blessed to have had numerous conversations with fellow graduate students, post doctoral scientists, and more senior members of our community that have found \gke~to be an excellent tool, not just in the breadth of plasma physics that can be studied, but in the ease with which they have found downloading the code, building it, and running simulations everywhere from their laptops to supercomputers.

But, there is more that can be done in making a tool accessible, especially to those just entering the field of plasma physics. 
While the equation system of interest in this thesis, the Vlasov--Maxwell--Fokker--Planck system of equations, is one of the most fundamental equation systems in all of plasma physics, it is not always the case that a budding new plasma physicist has immediate exposure to the equation system, its derivation, and the wealth of physics content within the equation system.
The few universities that offer rigorous courses in kinetic theory often break up the discussions of this equation system over the course of a full year. 
In addition, some beloved textbooks that offer clear explanations of plasma kinetic theory are out of print, such as \citet{Nicholson:1983}, and may only become harder to find with time.

I do not claim to have rigorously derived the foundations of plasma physics in this thesis.
But it is my wish to impart physical intuition about plasma kinetic theory, thinking about a many-body system like a plasma in a statistical sense, and the rich physics buried in the Vlasov--Maxwell--Fokker--Planck system of equations that ultimately made the derivation and implementation of novel numerical methods such a rewarding project. 
In this vein, I hope to proceed pedagogically through the intuition that forces us to develop kinetic theory, what kinetic theory means, and how we obtain workable equations for the physics of a plasma so that when we ultimately work to discretize the equation system and numerically integrate the discrete system to model plasma phenomena, we have a sense of what properties of the continuous system of equations we would like our discretization to respect.

This thesis is not intended as a user manual for the code, at least not if a reader's goal is to find installation instructions and assistance in building the \gke~simulation framework.
I refer an interested reader in this regard to our \texttt{GitHub}\footnote{https://github.com/ammarhakim/gkyl} and documentation website\footnote{https://gkyl.readthedocs.io/en/latest/}. 
It is the goal of this thesis to explain every aspect of our numerical method, how it works, and how we can leverage this particular algorithm to perform simulations of kinetic plasmas. 
In this way, this thesis is intended as much to be an introduction to the algorithms in the \gke~simulation framework as it is to kinetic theory, especially the difference between the mathematical formulation of an algorithm, and the translation of this algorithm to code.

I have attempted to organize this thesis in a logical fashion for an aspiring plasma physicist interested in diving into the details of plasma dynamics. 
Chapter~\ref{ch:Introduction} provides an introduction to plasma physics and kinetic theory and attempts to motivate both why we need the Vlasov--Maxwell--Fokker--Planck system of equations, and from where this equation system ultimately comes. 
Importantly, while the discussion of the Vlasov--Maxwell--Fokker--Planck system of equations may not be wholly rigorous, we will in detail work through many of the properties of the continuous system in anticipation of what properties we desire a numerical method to respect in the process of discretizing the equation system of interest. 

Chapter~\ref{ch:DGFEM} will introduce our numerical method, the discontinuous Galerkin finite element method, and attempt to build intuition for how the method works and how we can apply the method generally to partial differential equations. We will then in detail discretize the Vlasov--Maxwell--Fokker--Planck system of equations and mathematically determine the properties our discrete scheme retains from the continuous system. 
Chapter~\ref{ch:DGFEM} will form a mathematically complete description of our method, before we turn to Chapter~\ref{ch:ImplementationDGFEM}, where we will translate this mathematics into an algorithm which can be implemented in code. 
This conversion to code is equally nontrivial to the mathematical formulation of the algorithm, but it is my goal that after reading Chapter~\ref{ch:ImplementationDGFEM}, a reader may dive into \gke~with newfound understanding of how to put all the pieces together into a discrete scheme that can be used for performing numerical experiments. 

Chapter~\ref{ch:Benchmarks} will involve taxing testing of the implemented numerical method for the Vlasov--Maxwell--Fokker--Planck system of equations, and attempt to demonstrate to the reader that the scheme outlined in this thesis is on firm foundation; you may trust both that the scheme discussed in this thesis is a valid one, and that the code will work for whatever you envision doing with it. 
We will conclude in Chapter~\ref{ch:Leverage} with a number of applications of my implementation of the DG discretization of the  Vlasov--Maxwell--Fokker--Planck system of equations to demonstrate the full utility of this approach, leveraging the code to understand the details of energization processes and nonlinear plasma instabilities.

Because this thesis is intended to live beyond my graduate career, I would ask future readers that find typos or issues to contact me at my personal email: junoravin@gmail.com. 
At every stage of my career, I will attempt to keep this thesis in a state of maximum utility by updating it as necessary on the \gke~documentation website.
Readers interested in reproducing the simulations presented in this thesis can do so by running the input files available through a \texttt{GitHub} repository\footnote{https://github.com/ammarhakim/gkyl-paper-inp}.
The changesets used to produce the data are documented in the input file, and where appropriate the scripts used to produce the figures in this thesis can be found alongside the input files. 
In addition, if any readers are interested in the publications which formed the basis for this thesis, I refer them to \citet{Juno:2018}, \citet*{Hakim:2019}, \citet{HakimJuno:2020}, \citet{Juno:2020}, and \citet*{Skoutnev:2019}.

Without further ado, let us begin. I hope you ultimately find this thesis as much fun to read as I had writing it. To quote Robert Louis Stevenson, ``It is one thing to mortify curiosity, another to conquer it.''
 %(if present, lower-case Roman)
\addcontentsline{toc}{chapter}{Acknowledgements}
%Acknowledgments

\renewcommand{\baselinestretch}{2}
\small\normalsize
\hbox{\ }
 
\vspace{-.65in}

\begin{center}
\large{Acknowledgments} 
\end{center} 

\vspace{1ex}

A long thesis requires a long acknowledgements, as this thesis would not exist without the support I have received from many people over the years.

I have been blessed to have received a great deal of mentorship throughout my graduate career.
It is not possible to thank just one person for serving as an advisor to my graduate career, so let me instead thank them all. 
I am grateful to Prof. William Dorland, who took me on as his student before I even officially started at the University of Maryland and who graciously accepted my ambition for myself by allowing me to pitch him the project that this thesis became. 
Bill has been a constant source of guidance, and his own experiences in developing one of the leading computational plasma physics tools provided invaluable perspective that led to some of the most strenuous tests of \gke. 
The confidence I have in the code is borne of many discussions with him on software development and verification and validation, and I think we got some exciting science out of this project to boot. 

I would like to thank Dr. Jason TenBarge, my co-advisor, for being a constant source of inspiration and a valuable scientific confidant. 
Jason probably did not know what he was getting into having an open door policy and me as a student, but I learned so much from him throughout my graduate career. 
If I have gained a reputation for asking questions at conferences, it is only because of the sheer volume of knowledge Jason has imparted to me. 
I hope for a long career of collaborations between the two of us as I move to the next stage of my career. 
I am especially grateful to Jason and his wife Helen for their friendship through the challenges of graduate school.
I have treasured immensely the good company, good food, and board games more than I can say.

Finally, I must thank Dr. Ammar Hakim for his tutelage and guidance from the very beginning of me joining the \gke~project.
Ammar has followed and mentored my growth since I was an undergraduate summer student at the Princeton Plasma Physics Lab, and he must have known he made a strong impression when I pitched Bill on continuing to work on \gke~as a University of Maryland graduate student.
I am proud of all the work I, and the \gke~team, have accomplished. 
Serving as a coding apprentice to a computational physicist as outstanding as Ammar has been more valuable than any single course I have taken in my entire education.
I hope to only grow more as a programmer through continued collaboration with him as a member of the \gke~team.

Speaking of the \gke~team, this thesis owes a tremendous debt of gratitude to all of you. 
To the members past and present, Petr Cagas, Manaure Francisquez, Tess Bernard, Valentin Skoutnev, Noah Mandell, Eric Shi, Liang Wang, Jonathan Ng, and Rupak Mukherjee, thank you. 
I am blessed to have had such an amazing team to code alongside.
Thank you especially to Petr Cagas for all his work writing \texttt{postgkyl}, which this thesis leverages extensively.
We have accomplished so much, and I look forward to many more years of working together with you all.

Thank you to the other members of my committee, Prof. Jim Drake, Prof. Adil Hassam,  and Prof. Jacob Bedrossian.
Most especially thank you to Jim, who was generous enough to offer me a post doctoral position that I happily accepted. 
I look forward to the work we will do together.

Thank you to all of my plasma compatriots at the University of Maryland, George Wilkie, Joel Dahlin, Wrick Sengupta, Lora Price, Elizabeth Paul, Michael Martin, Gareth Roberg-Clark, Harry Arnold, Qile Zhang, Rahul Gaur, Michael Nastac, Rog\'erio Jorge, and Alessandro Geraldini. 
I would like to especially thank Elizabeth and Wrick for being such incredible and hospitable friends. 
This thesis would not exist without your support.

Every scientist needs a break from science, and I have been fortunate enough to have had the very best of friends in Sandy and Clarissa Craddock, Jonathan Vannucci, Zachary Eldredge, Katie Goff, Rodney Snyder, Molly Carpenter, Steffi Rathe, Humberto Gilmer, Carl Mitchell, Sheehan Ahmed, and Jesse Rivera. 
I remain eternally grateful for the drinks and laughs we have shared.
A special thank you to Sandy, Jon, Zach, and Rodney for your companionship through the trials of qualifying exams.

Thank you to all members of The Shed, past and present, who have come along with me on adventures at the table. 
I have never played with a better dungeon master than Zak Schooley, who rekindled my love for tabletop games.
Recording Adventures in Hyperborea with James Upton, James Wiley, and Jonathan Hill has been a singular joy for me over these past few years.
I cannot wait to see what the dice have in store for us next.

There are many others to thank.
Thank you to Dr. Gregory Hammett for assisting in supervising me with Ammar when I was but an undergraduate on the \gke~team, for continuing to support the whole \gke~project, and for being a constant source of physics insight. 
Thank you to Prof. Matt Kunz and Prof. Anatoly Spitkovsky for all of our physics discussions in my home away from home in Princeton. 
A special thank you to Matt for allowing me to sit in on his Irreversible Processes in Plasmas course so I could refresh my knowledge of plasma kinetic theory in anticipation of writing this thesis. 
Thank you to Dr. Ian Abel, who is as much a plasma compatriot as my fellow graduate students and post docs. 
I have Ian to thank for developing the intuition I have on asymptotic methods, and I am so grateful for his thoughtful responses to my long emails in the early days of my graduate career when I was still learning gyrokinetics. 
Thank you to Dr. Marc Swisdak, who graciously agreed to collaborate on work that appears in this thesis. 
I am thrilled that the post doctoral position with Jim will allow us continue collaborating more in the coming years. 
Thank you to Dr. Matt Landreman and Prof. Tom Antonsen for an uncountable number of enlightening physics discussions and for the deep, engaging questions you both asked of me during group meetings. 

Thank you to the entire Chalmer's group for being such early adopters of \gke~and seeing so much potential in us, especially T\"unde F\"ul\"op, Istv\'an Pusztai, and Andr\'eas Sundstr\"om. 
It has been an absolute pleasure to collaborate with you all. 
I recall being worried that you had caught a subtle bug for the electron Landau damping and dynamo paper, and I was ready to send you an email discouraging the use of \gke~for this project, only for Istv\'an to email within 24 hours that he had figured out a physics explanation for the code'€™s behavior.
It is no exaggeration that I am overwhelmed with excitement for our future research endeavors.

Thank you to the colleagues I have made through the Solar, Heliospheric, and INterplanetary Environment (SHINE) conference, Prof. Gregory Howes and Prof. Kristopher Klein.
You both have been so encouraging of the development of \gke~, and I am elated that your encouragement has led to multiple ongoing projects, the results of which partially appear in this thesis.
Amongst my fellow graduate students at SHINE, I am grateful to have served alongside Do\v{g}a Can Su \"Ozt\"urk, Samaiyah Farid, and Emily Lichko as SHINE Student Representatives.
I feel strongly that we made the SHINE conference a better experience for students, and I hope our initiatives are a fixture in the conference for years to come.
Special thanks to Emily for taking the time to discuss her work on magnetic pumping with me, informing a rigorous test of \gke.

Thank you to those that supported me from the very beginning.
Thank you to my family, most especially my mom, Constance Lynn, who has been an endless wellspring of support.
Thank you to the teachers who lit the spark of curiosity and nurtured the flame, Kristy Elam, Richard McGowan, and Jeff Peden in high school, and Prof. Frank Toffoletto and Prof. Anthony Chan at Rice University.

And thank you most of all to my wife, Anna Wright, who this thesis is dedicated too alongside the memory of my father. 
Anna, you are the light of my life, and words do no exist to describe the extent of your support for me, and the way you inspire me each and every day. 
This thesis has been a highly collaborative effort, and your name is right alongside all of my \gke~teammates. 
It was the greatest of privileges to have you by my side through this journey, and I look forward to what we do next together.

This work was supported by a NASA Earth and Space Science Fellowship, grant no. 80NSSC17K0428. %(if present, lower-case Roman)

\renewcommand{\baselinestretch}{1}
\small\normalsize
\tableofcontents %(required, lower-case Roman)
\newpage
\listoftables %(if present, lower-case Roman)
\newpage
\listoffigures %(if present, lower-case Roman)
\newpage
% LIST OF ABBREVIATIONS
\addcontentsline{toc}{chapter}{List of Abbreviations}
%List of Abbreviations

\renewcommand{\baselinestretch}{1}
\small\normalsize
\hbox{\ }

\vspace{-4em}

\begin{center}
\large{List of Abbreviations}
\end{center} 

\vspace{3pt}

\begin{supertabular}{ll}
VM-FP & Vlasov--Maxwell--Fokker--Planck (system of equations)\\
DG & discontinuous Galerkin (finite element method) \\
\end{supertabular}

\newpage
\setlength{\parskip}{0em}
\renewcommand{\baselinestretch}{2}
\small\normalsize

%Pages from this point start at Arabic numeral 1
\setlength{\abovedisplayskip}{6pt} 
\setlength{\belowdisplayskip}{6pt}
\setcounter{page}{1}
\pagenumbering{arabic}
%Chapter 1

\renewcommand{\thechapter}{1}
\epigraph{Some of the material in this chapter has been adapted from \citet{Juno:2018}, \citet*{Hakim:2019}, and \citet{HakimJuno:2020}.}{}
\chapter{Introduction}\label{ch:Introduction}

Plasmas are ubiquitous in nature, and the study of plasmas has application to a wide variety of problems, from the development of nuclear fusion, to understanding the dynamic interaction between the solar wind and the Earth's magnetosphere, to elucidating the mysteries of astrophysical phenomena such as binary star collisions or the accretion disks of black holes. 
Unfortunately, many plasmas of interest are only weakly collisional and far from equilibrium, making the system best described by kinetic theory. 
The use of kinetic theory significantly complicates the theoretical analysis and simulation of the plasma's dynamics due to the increased dimensionality of the corresponding equations, which are solved in a combined position and velocity phase space, along with the large collection of waves and instabilities that the kinetic system supports. 

While there are many avenues for tackling the numerical solution of the kinetic equation, popular approaches such as the particle-in-cell method have deficiencies due to the counting noise inherent to the algorithm.
This noise can significantly degrade the quality of the solution, in addition to making the ultimate analysis of simulations more challenging, especially for problems requiring high signal-to-noise ratio. 
In this thesis, we outline and demonstrate the utility of an approach that directly discretizes the kinetic equation on a phase space grid. 

This approach requires care, as we must consider both the cost, since the partial differential equation is defined in a six dimensional phase space, alongside the challenges which arise from the wealth of physics buried within the equation system of interest.
For example, important conservation relations, such as the conservation of energy, are implicit to the kinetic equation, leading to additional difficulties in ensuring a discrete scheme satisfies these properties.
But this same wealth of physics contained in the kinetic equation motivates a direct discretization of the kinetic equation.
We can leverage the uncontaminated phase space from a continuum discretization to diagnose energization processes directly in phase space and carefully ascertain the nonlinear saturation mechanisms of unstable plasmas.

Some readers may be left wondering right from the beginning why the numerical solution of a plasma system is at all challenging. 
Before diving deeper into the details of the algorithm and the verification of this approach, let us take a moment to address the paradoxically simple yet subtle question of what makes plasmas so rich in their underlying physics.
We will then define some of the terminology used in this brief introduction, most importantly \emph{kinetic theory}, and how we use kinetic theory to derive a useful equation system for modeling a plasma.
This brief overview will serve as the foundation from which we will build intuition for what we want from a numerical model of a kinetic plasma, most especially the fundamental physics properties of a plasma we would like our discretization to respect.

It is the goal of this introduction to proceed in a pedagogical fashion.
We will assume no prior plasma physics knowledge, much less knowledge about the subtleties described so far concerning particle versus continuum methods.
We will connect this holistic introduction to plasma physics to these questions regarding our choice of numerical method in the final section of this introduction, Section~\ref{sec:introObjectives}, when we outline the objectives of this thesis.

\section{What is a plasma?}

Formally, a plasma is a collection of mobile, or ``free,'' charged particles. 
Collection in this case refers to the fact that a plasma is an $N$-body system, where $N \gg 1$. By mobile, or ``free,'' we mean that the particles in a plasma are not confined by inter-particle forces and the individual particles in a plasma behave similarly to a gas, as opposed to a solid or crystalline structure, albeit with the added complication of the particles being charged. 
And in this case, the fact that the particles are charged means that the particles are subject to the Lorentz force and can interact with each other via microscopic electromagnetic fields governed by Maxwell's equations.

This definition of a plasma is somewhat restrictive. 
In this case, we limit ourselves to what are commonly referred to as weakly coupled plasmas, as the mobile component of our definition implies the kinetic energy of the particles is much, much greater than the potential energy of the particles.
Likewise, we restrict our attention to plasmas which are fully ionized.

Let us be a bit more concrete about our definition of a plasma, so that we can gain more intuition for what it means to limit ourselves to this subset of so-called weakly coupled plasmas. 
Consider a gas of some number of charged species, such as a gas of protons and electrons, where each charged species has density $n_0$. 
Since $n_0$, the density, is the number of particles in a given volume, the average distance between two charged particles is roughly $n_0^{-1/3}$. 
This rough estimate for the average distance between two particles can be used to approximate the average potential energy per particle in this sample plasma,
\begin{align}
    \Phi \sim \frac{1}{4\pi \epsilon_0} \frac{e^2}{r} \sim \frac{1}{4\pi \epsilon_0} n_0^{\frac{1}{3}} e^2,
\end{align}
where $e$ is the elementary charge, i.e., the charge carried by a proton.
Likewise, we can estimate the average kinetic energy of a particle using the equi-partition theorem,
\begin{align}
    \frac{1}{2} m_s \langle v^2 \rangle \sim \frac{3}{2} k_B T_s,
\end{align}
where $m_s$ and $T_s$ are the mass and temperature of the particles of species $s$, respectively, and $\langle \cdot \rangle$ denotes an average over all particle velocities at a given point in space. Here, $k_B$ is Boltzmann's constant.

Thus, our definition of a plasma, that the average kinetic energy of the particles is much larger than the average potential energy of the particles, implies
\begin{align}
    \frac{3}{2} k_B T_s \gg \frac{1}{4\pi \epsilon_0} n_0^{\frac{1}{3}} e^2,
\end{align}
or
\begin{align}
    6 \pi n_0^{\frac{2}{3}} \left (\frac{\epsilon_0 k_B T_s}{n_0 e^2} \right ) \gg 1.
\end{align}
This expression at first glance looks somewhat unremarkable, but upon raising both sides to the $3/2$ power, and rewriting the expression in terms of a characteristic length scale of a plasma, the Debye length,
\begin{align}\label{eq:DebyeLengthDef}
    \lambda_{D_s} = \sqrt{\frac{\epsilon_0 k_B T_s}{n_0 e^2}},
\end{align}
we obtain
\begin{align} \label{eq:PlasmaDefinition}
    (6 \pi)^{\frac{3}{2}} n_0 \lambda_D^3 \gg 1.
\end{align}
Note that in the definition of the Debye length we could have a species dependent density, but since we have assumed that both the protons and electrons have the same density we have set $n_p = n_e = n_0$.

Ignoring the constant for a moment, we may gain a bit of intuition for what we have just found. $n_0 \lambda_D^3$ is the number of particles in a cube with side lengths equal to the Debye length. 
We will gain a deeper understanding of the physical significance of this expression which follows from our definition of the a plasma in the following section.

\section{The Debye length and the Plasma Parameter}\label{sec:DebyeLengthPlasmaParameter}

Because plasmas are a large collection of charged particles, inevitably, the particles will rearrange themselves in response to each other's charges. 
Consider one particular particle with a positive charge. 
Since the particle's charge is positive, the electrons in the plasma will be attracted to the particle, while the positively charged ions will be repelled, creating a local area where the density of the electrons has increased, while the density of the positively charged ions has decreased.

Without loss of generality, let us take the positively charged ions to be protons. 
Then, if the electrons have density $n_e$, and the protons have density $n_p$, Poisson's equation tells us that the electric potential for the plasma is
\begin{align}
    \nabla^2 \phi = -\frac{\rho_c}{\epsilon_0} = \frac{e}{\epsilon_0} (n_e - n_p) - q_T \delta(\mvec{r}),
\end{align}
where we denoted the charge of the particular particle as $q_P$ and used the Dirac delta function, $\delta(\mathbf{r})$, to denote the position of the particle in space.

We need to determine how the density of electrons and protons has been modified by the presence of this particular charge.
If we assume that we have waited long enough for electrons and protons to come into thermodynamic equilibrium with the particular charge, i.e., that we wait long enough that the temperature becomes a well-defined quantity, we can use equilibrium statistical mechanics.
Without insisting that the electrons and protons have the same temperature, only that we can define temperatures, the densities of the electrons and protons are given by the Boltzmann distribution,
\begin{align}
    n_e & = n_0 \exp\left( \frac{e \phi}{k_B T_e} \right), \\
    n_p & = n_0 \exp\left( \frac{-e \phi}{k_B T_p} \right),    
\end{align}
where $n_0$ is the density of the electrons and protons far away from the particular charged particle of interest, i.e., far enough away so that electric potential from the particular charged particle of interest is zero.

But, recall what we have continually reiterated from our definition of a plasma: the average potential energy of the particles is much less than the average kinetic energy. Therefore, $e \phi \ll k_B T_s$, and the exponential function can be Taylor expanded far away from $\mvec{r} = 0$, the location of the particular charged particle, so that
\begin{align}
    \nabla^2 \phi = \frac{1}{r^2}\frac{d}{dr}\left(r^2 \frac{d\phi}{dr} \right ) = \frac{e^2 n_0}{\epsilon_0 k_B} \left (\frac{1}{T_e} + \frac{1}{T_p} \right ) \phi.
\end{align}
Using \eqr{\ref{eq:DebyeLengthDef}}, we can see that the above equation simplifies to
\begin{align}
    \frac{1}{r^2}\frac{d}{dr}\left(r^2 \frac{d\phi}{dr} \right ) = \left(\frac{1}{\lambda_{D_e}^2} + \frac{1}{\lambda_{D_p}^2} \right ) \phi.
\end{align}
If one waits longer for the protons and electrons to come in to thermodynamic equilibrium with each other so the temperatures of the two species are equal, $T_e = T_p = T$, then
\begin{align}
    \left(\frac{1}{\lambda_{D_e}^2} + \frac{1}{\lambda_{D_p}^2} \right ) = \frac{2}{\lambda_D^2}.
\end{align}
The solution to this differential equation then follows from trying functions of the form $\phi = \tilde{\phi}/r$, so that
\begin{align}
    \frac{d^2 \tilde{\phi}}{dr^2} = \frac{2}{\lambda_D^2} \tilde{\phi}.
\end{align}
The only solution which respects the boundary condition that the electric potential, $\phi$, not blow up as $r \rightarrow \infty$ is a solution of the form
\begin{align}
    \phi(r) = A \exp\left( -\frac{\sqrt{2}}{\lambda_D} r \right),
\end{align}
where $A$ is a constant of integration. This constant of integration can be found by considering the boundary condition at $r=0$, where the electric potential will be dominated by the particular charged particle of interest. We know from Gauss' law that the electric potential of an individual charged particle is simply $1/(4 \pi \epsilon_0) \thinspace q_P/r$ so that the complete solution for the electric potential of an individual charged particle in a plasma is
\begin{align}
    \phi(r) = \frac{1}{4 \pi \epsilon_0} \frac{q_P}{r} \exp\left( -\frac{\sqrt{2}}{\lambda_D} r \right).\label{eq:DebyeShielding}
\end{align}

This functional form for the potential implies that the electric potential, and thus the charge, of a particle falls off much faster than just the inverse of the distance. 
It thus follows from this solution that the charged particles in a plasma rearrange themselves to cancel the charges of their neighbors, and that the characteristic length scale on which a plasma's charged particles are screened is the Debye length.

We return now to \eqr{\ref{eq:PlasmaDefinition}} with newfound understanding of the physical significance of the Debye length.
If the number of particles in a Debye cube is very large, then it becomes a bit more apparent why these plasmas are often referred to as weakly coupled. 
When the number of particles in a Debye cube is large, no individual electrostatic interaction between particles is of dynamical importance.
Because a single particle is feeling the electrostatic potential of a large number of particles in its immediate vicinity, the individual electrostatic interactions between particles are dwarfed by the accumulation of all of the electrostatic interactions.
We need not discuss the electric field one particle exerts on another; rather, what we require is the net electric field of all of the particles in a Debye cube, so that we may obtain the aggregated response of the particles in the plasma.

In this regard, we should avoid being dismissive of the individual electrostatic interactions occurring within a Debye cube in a plasma. 
It is true that the sum is greater than the individual parts in a weakly coupled plasma where there are many particles in a Debye cube.
But in this vein, we must distinguish between \emph{individual} and \emph{collective} effects.
\eqr{\ref{eq:DebyeShielding}} shows us that the electrostatic potential of an \emph{individual} particle falls off exponentially on scales larger than the Debye length, but the \emph{collective} effects of all the individual particles within the Debye cube can be of critical importance for the plasma's dynamics. 
The collective response of the plasma, on scales above and below the Debye length, is a crucial consideration in the derivation of the resulting equations of interest in the forthcoming sections. 

\section{The challenge in modeling plasmas}\label{sec:challengeModeling}
For now, let us use this discussion as a segue into the original question which galvanized defining a plasma: ``Why is modeling a plasma hard?'' 
Regardless of whether individual particle-particle interactions are important or irrelevant in a many-body plasma, it does not change the fact that the equations of motion for particles in electromagnetic fields are well-known and easy enough to solve numerically. 
So why not model all the particle-particle interactions in the many-body system?

The answer may be obvious just from the description of a plasma as a many-body system.
Since the challenge inherent in modeling a plasma can be seen even without considering the magnetic field, for readability, we will ignore the magnetic field for now and only consider the particle-particle interactions from the plasma's self-consistent electric field. 
In this case, one could evolve a single particle under the equations of motion,
\begin{align}
    \frac{d \mvec{x}_k}{dt} & = \mvec{v}_k, \label{eq:numericalCharcteristicX} \\
    \frac{d \mvec{v}_k}{dt} & = \frac{q_k}{m_k} \sum_{i=1, i \neq k}^N \mvec{E}_i, \label{eq:numericalCharcteristicV}
\end{align}
where $k$ is the label for the particle being evolved. 
One would then proceed to solve these equations for each $k = 1,\dots, N$. The electric field in this system of equations is given by
\begin{align}
    \mvec{E}_i = \frac{1}{4 \pi \epsilon_0} \frac{q_i}{(\mvec{x}_i - \mvec{x}_k)^2} \hat{\mvec{x}}_{ik}.
\end{align}
Here, $\mvec{E}_i$ is the electric field particle $i$ exerts on particle $k$, so that in this notation, $\mvec{x}_i$ is the $i^\textrm{th}$ particle's position, $\mvec{x}_k$ is the position of the particle currently being evolved, i.e. the particle the electric field $\mvec{E}_i$ is acting on, and $\hat{\mvec{x}}_{ik}$ is the unit vector pointing from $\mvec{x}_i$ to $\mvec{x}_k$.

What is the computational cost to solve these two coupled sets of ordinary differential equations? 
For each of the $N$ particles, we require at least $N-1$ operations to compute the total electric field since each particle's electric field depends on all of the other particles. 
Even if one stores the electric field of each particle so as not to recompute any particle's contribution, the computational work only reduces to a sum of the form
\begin{align}
    \sum_{i=1}^N (N-i) = \frac{N^2}{2} - \frac{N}{2}.
\end{align}

Thus, the computational complexity of such an algorithm is $\mathcal{O}(N^2)$, meaning if we double the number of particles we are evolving numerically, we quadruple the cost to compute the solution. 
In addition, this argument implies that, at minimum, this method requires on the order of $N^2$ operations to perform a single time step. 

But, modern supercomputers are already fast and will only continue to speed up with time. 
As of the completion of this thesis, we have achieved exascale computing\footnote{With the caveat that this is only for reduced precision, i.e., the supercomputers at the writing of this thesis could achieve an exaflop, $10^{18}$ floating point operations per second, if one only required single precision.}. 
Is this enough to make this approach feasible? 

We require a concrete example. 
Let us consider the ITER (International Thermonuclear Experimental Reactor) Tokamak currently being built to demonstrate the feasibility of nuclear fusion as a power source. 
According to the website for ITER (ITER \citeyear{iter}), the vacuum vessel is $840 m^3$ in volume, and the average density of the electrons in the plasma will be $\sim 10^{20} m^{-3}$. The plasma will be quasi-neutral, so a conservative estimate of the number of particles, protons, electrons, and alpha particles, inside ITER is $\sim 10^{23}$ particles. 
A single ``shot,'' or run of the experiment, is expected to last anywhere from 100 to 1000 seconds. 
So, could we model ITER through a full experimental shot, tracking every particle in the experiment?

To answer this question, we require one final piece of information: the fastest time scale in the system, so that we know how many time-steps we would require to evolve all the particles for 100 to 1000 seconds. 
The fastest time scale in a plasma can be found by considering how particles jostle about. 
We have found a characteristic length scale, the Debye length, \eqr{\ref{eq:DebyeLengthDef}}, and it is simple enough to define a characteristic velocity from the equipartition theorem
\begin{align}
    \frac{1}{2} m v_{rms}^2 \sim \frac{N}{2} k_B T,
\end{align}
where $N$ is the number of degrees of freedom. Each degree of freedom thus has root-mean-square velocity
\begin{align}
    v_{th_s} = \sqrt{\frac{k_B T_s}{m_s}}. \label{eq:thermalVelocity}
\end{align}
This speed is called the thermal velocity, and approximates the average speed of particles with temperature $T_s$ (or energy $k_B T_s$).
Note that \eqr{\ref{eq:thermalVelocity}} is typically referred to as the thermal velocity even though it is not a vector quantity, and though we will maintain this nomenclature throughout our discussion, we will attempt to minimize confusion by emphasizing \eqr{\ref{eq:thermalVelocity}} is a speed where appropriate.
The ratio of these two quantities for the electrons,
\begin{align}
    \frac{v_{th_e}}{\lambda_{D_e}} = \sqrt{\frac{e^2 n_e}{\epsilon_0 m_e}},
\end{align}
has units of inverse seconds and defines a frequency,
\begin{align}
    \omega_{pe} = \sqrt{\frac{e^2 n_e}{\epsilon_0 m_e}}. \label{eq:plasmaFrequency}
\end{align}

Although we have not demonstrated this here, this frequency is roughly the highest frequency in the system. With the Debye length as our length scale, and this frequency, $\omega_{pe}$, called the plasma frequency, setting a time scale, we have a rough estimate of the cost of numerical integration of Eqns.\thinspace\ref{eq:numericalCharcteristicX}--\ref{eq:numericalCharcteristicV}.
For the purposes of numerical integration, we wish to avoid particles moving distances greater than the Debye length on time scales shorter than the inverse plasma frequency, as this may introduce numerical instabilities into the integration of the particle orbits.

In ITER, the plasma frequency for the electrons is $\sim 10^{11}$ -- $10^{12}$ Hz, so even if our numerical method is very robust and requires only one time-step per inverse plasma frequency, we require a large number of time steps per second.
In total, assuming $N^2$ operations per time-step and $10^{14}$ time steps to model a single run of the experiment, a computer simulation of all the particle dynamics would need to do approximately $10^{60}$ floating point operations. 
If current supercomputers, even with the trade-offs in terms of floating point precision, can only perform $10^{18}$ floating point operations per second, an exa-flop, we would still require $10^{42}$ seconds of simulation time.
For reference, the universe has only existed for just over $10^{17}$ seconds, so a simulation like this requires quite a few universe lifetimes with modern computer architecture.

It is worth taking a moment to go further and try to improve this algorithm before we give up on tracking every single particle-particle interaction. 
For example, there are algorithms which would reduce the cost of computing the electric field, and thus the algorithm, from $\mathcal{O}(N^2)$ to $\mathcal{O}(N)$, by using a multipole expansion of the electrostatic potential \citep{GreengardRokhlin:1987}.
Such algorithms are commonly employed in computational cosmology for solving for the gravitational potential of a large number of dark matter particles and simulating galactic dynamics and evolution \citep{Stadel:2001}. 
Even with a multipole expansion of the electric field, the total number of operations would reduce from $10^{60}$ to only $\sim 10^{37}$, and the total time to $10^{19}$ seconds.
Unfortunately, this is not a large enough reduction, and one would have to track a lot fewer particles for a lot less time, to say nothing of the added complexity of the magnetic field acting on individual particles.
For example, it would be quite a large simulation to run on a modern supercomputer for 4 continuous months, $\sim 10^7$ seconds, so one would have to eliminate 12 orders of magnitude in some combination of the amount of time being simulated and number of particles being evolved, again to say nothing of the assumptions which made this back-of-the-envelope calculation remotely reasonable.

So, what is one to do? 
All hope is not lost for the reason we have emphasized throughout these introduction sections. 
That is, the individual particle-particle dynamics are of minimal importance in a weakly coupled plasma, and in fact what is principally important to the plasma's dynamics is its collective response to electromagnetic fields.
In other words, a weakly coupled plasma is an ideal system for which a mean-field theory may arise, one which allows for the study of the plasma of interest in a statistical sense.
We will be careful to define both what we mean by a mean-field theory and what we mean by thinking about the particle dynamics in a statistical sense in the next section.

\section{An introduction to kinetic theory}

Up until now, we have concerned ourselves with the microscopic properties of the plasma, and, as demonstrated in the previous section, this limits our ability to model the plasma. 
We would like to still respect the fact that the plasma is made of discrete particles though, and so we turn to kinetic theory.
``Kinetic'' in this case means, ``pertaining to motion,'' and kinetic theory provides the foundation to consider the motion of all of the particles in the plasma, but without the stringent requirement to track individual particle dynamics and interactions.

Consider the density of particles of species $s$, $N_s$, in a combined position and velocity space.
This density is simply a sum of Dirac delta functions denoting the individual positions and velocities of every particle in the plasma,
\begin{align}
    N_s(\mvec{x}, \mvec{v}, t) = \sum_{i=1}^{N_0} \delta(\mvec{x} - \mvec{X}_i)\delta(\mvec{v}-\mvec{V}_i),
\end{align}
where we have used capital $\mvec{X}_i$ and $\mvec{V}_i$ to specify the individual particle positions and velocities in the $\mvec{x}$--$\mvec{v}$ phase space.
The motions of the particles in this plasma in space and time are governed by the particle characteristics\footnote{Assuming the particles are traveling at velocities much less than the speed of light, $|\mvec{v}| \ll c$, so we can ignore the Lorentz boost factors, and further that the self-force due to radiation is of minimal dynamical importance.},
\begin{align}
    \frac{d \mvec{X}_i}{dt} & = \mvec{V}_i, \label{eq:physicalCharacteristicX} \\
    \frac{d \mvec{V}_i}{dt} & = \frac{q_s}{m_s} \left [\mvec{E}^m(\mvec{X}_i, t) + \mvec{V}_i \times \mvec{B}^m(\mvec{X}_i, t) \right ], \label{eq:physicalCharacteristicV}
\end{align}
similar to Eqns.\thinspace\ref{eq:numericalCharcteristicX}--\ref{eq:numericalCharcteristicV}, but with a simplified notation for the microscopic electromagnetic fields using the superscript $m$. 
We can see that change in velocity does not couple to the acceleration in such a way as to require a third set of equations for the time derivative of the acceleration of the particles and thus the two equations Eqns.\thinspace(\ref{eq:physicalCharacteristicX}) and (\ref{eq:physicalCharacteristicV}) are closed once we specify evolution equations for the electromagnetic fields.
In this case, the evolution of the electromagnetic fields is given by Maxwell's equations,
\begin{align}
  \frac{\partial \mvec{B}^m(\mvec{x}, t)}{\partial t} + \gx\times\mvec{E}^m(\mvec{x}, t) &= 0, \label{eq:micro-dbdt} \\
  \epsilon_0\mu_0\frac{\partial \mvec{E}^m(\mvec{x}, t)}{\partial t} - \gx\times\mvec{B}^m(\mvec{x}, t) &= -\mu_0\mvec{J}^m(\mvec{x}, t), \label{eq:micro-dedt} \\
  \gx\cdot\mvec{E}^m(\mvec{x}, t) &= \frac{\varrho^m_c(\mvec{x}, t)}{\epsilon_0}, \label{eq:micro-divE} \\
  \gx\cdot\mvec{B}^m(\mvec{x}, t) &= 0, \label{eq:micro-divB}
\end{align}
where the microscopic charge density and current density are given by
\begin{align}
  \varrho^m_c(\mvec{x}, t) = \sum_s  q_s \int N_s(\mvec{x}, \mvec{v}, t)\dv,
\end{align}
and
\begin{align}
  \mvec{J}^m(\mvec{x}, t) = \sum_s q_s \int \mvec{v} N_s(\mvec{x}, \mvec{v}, t)\dv,
\end{align}
respectively.

This density of particles of species $s$, $N_s$, can neither be created nor destroyed because the number of particles cannot change in time, assuming the system is closed. 
This attribute implies $N_s$ obeys a continuity equation. For a reader unfamiliar with the concept of a conservation equation, consider a quantity $f(\mvec{r}, t)$, a function of some space $\mvec{r}$ and time, which in the process of its motion in space and time can neither be created nor destroyed. Then, this quantity $f(\mvec{r}, t)$ obeys
\begin{align}
    \int_{\Omega} \pfrac{f(\mvec{r}, t)}{t} \thinspace d\mvec{r} = 0, \label{eq:simpleConservationEq}
\end{align}
where $\Omega$ is the domain the function $f(\mvec{r}, t)$ is defined in.
But this quantity $f(\mvec{r}, t)$ can still be transported throughout the domain $\Omega$.
Let us define the flux function for the function $f(\mvec{r}, t)$ as $\mvec{G}$, where $\mvec{G}$ could be as simple as a constant, or as complex as a nonlinear function of the quantity of interest, $\mvec{G} = \mvec{G}(f)$.
Then, the flux of $f(\mvec{r}, t)$ is $\mvec{G} f(\mvec{r}, t)$.
We have argued in \eqr{\ref{eq:simpleConservationEq}} that the time derivative of the integral over the whole domain of the function $f(\mvec{r}, t)$ is zero, which means that the flux of the function $f(\mvec{r}, t)$ out the boundary of the domain must also be zero,
\begin{align}
    \oint_{\partial \Omega} f(\mvec{r}, t) \mvec{G} \cdot d\mvec{S} = 0,
\end{align}
so that we can say
\begin{align}
    \int_{\Omega} \pfrac{f(\mvec{r}, t)}{t} = -\oint_{\partial \Omega} f(\mvec{r}, t) \mvec{G} \cdot d\mvec{S}.
\end{align}
But using the divergence theorem,
\begin{align}
    \oint_{\partial \Omega} f(\mvec{r}, t) \mvec{G} \cdot d\mvec{S} = \int_{\Omega} \nabla_{\mvec{r}} \cdot [ \mvec{G} f(\mvec{r}, t)] \thinspace d\mvec{r},
\end{align}
so that we can argue
\begin{align}
    \int_{\Omega} \pfrac{f(\mvec{r}, t)}{t} + \nabla_{\mvec{r}} \cdot \left [ \mvec{G} f(\mvec{r}, t) \right ] \thinspace d\mvec{r} = 0,
\end{align}
allows us to attain an evolution equation for the function $f(\mvec{r}, t)$ using the fact that the integrand itself must also be equal to zero,
\begin{align}
    \pfrac{f(\mvec{r}, t)}{t} + \nabla_{\mvec{r}} \cdot \left [ \mvec{G} f(\mvec{r}, t) \right ] = 0. \label{eq:simpleConservationDerivative}
\end{align}
Because the time rate of change of the quantity, $f(\mvec{r}, t)$, integrated over the whole domain, is zero, i.e., $f(\mvec{r}, t)$ is not appearing or disappearing over time, $f(\mvec{r}, t)$ will inevitably obey an equation of the form \eqr{\ref{eq:simpleConservationDerivative}}.

If the density of particles of species $s$, $N_s$, obeys a similar equation, what is the flux function to advect $N_s$ in the combined position-velocity phase space? 
It is simply the characteristics defined in Eqns.\thinspace\ref{eq:physicalCharacteristicX}--\ref{eq:physicalCharacteristicV}, but importantly, with a change of variables from the individual particles' physical locations and velocities to the phase space coordinates.
This change follows from the fact that $N_s$ is a function of the phase space variables $\mvec{x}$ and $\mvec{v}$, not a function of the individual particle positions.
Thus, the conservation equation governing the evolution of the density of particles of species $s$ is
\begin{align}
    \pfrac{N_s (\mvec{x}, \mvec{v}, t)}{t} + \gx & \cdot \left [ \mvec{v} N_s(\mvec{x}, \mvec{v}, t) \right ] \notag \\
    & + \gv \cdot \left \{ \frac{q_s}{m_s} \left [ \mvec{E}^m(\mvec{x}, t) + \mvec{v} \times \mvec{B}^m(\mvec{x}, t) \right ] N_s(\mvec{x}, \mvec{v}, t) \right \} = 0. \label{eq:conservativeKlimontovich}
\end{align}
This equation is more commonly known as the Klimontovich equation, or Klimontovich's equation~\citep{Klimontovich:1967,Nicholson:1983}.
Oftentimes, \eqr{\ref{eq:conservativeKlimontovich}} is rearranged to emphasize the connection between the particle characteristics, and how $N_s$ advects in phase space,
\begin{align}
     \pfrac{N_s (\mvec{x}, \mvec{v}, t)}{t} + \mvec{v} & \cdot \gx N_s(\mvec{x}, \mvec{v}, t) \notag \\
    & + \frac{q_s}{m_s} \left [ \mvec{E}^m(\mvec{x}, t) + \mvec{v} \times \mvec{B}^m(\mvec{x}, t) \right ] \cdot \gv  N_s(\mvec{x}, \mvec{v}, t) = 0, \label{eq:characteristicKlimontovich}   
\end{align}
where we have exploited the fact that
\begin{align}
    & \gx \cdot \mvec{v} = 0, \label{eq:divV} \\
    & \gv \cdot \left [ \mvec{E}^m(\mvec{x}, t) + \mvec{v} \times \mvec{B}^m(\mvec{x}, t) \right ] = 0, \label{eq:divA}
\end{align}
in the rearrangement of \eqr{\ref{eq:conservativeKlimontovich}} to \eqr{\ref{eq:characteristicKlimontovich}}. 
\eqr{\ref{eq:divV}} likely seems intuitive, the velocity coordinate $\mvec{v}$ of course does not depend on the configuration space coordinate $\mvec{x}$, and \eqr{\ref{eq:divA}} follows from properties of the cross product, in addition to the fact that the electromagnetic fields themselves do not depend on velocity.
Just as \eqr{\ref{eq:conservativeKlimontovich}} follows from the fact that the density of particles of species $s$ cannot be created or destroyed, \eqr{\ref{eq:characteristicKlimontovich}} shows that $N_s$ is constant along characteristics, i.e.,
\begin{align}
    \frac{D N_s(\mvec{x}, \mvec{v}, t)}{Dt} = 0,
\end{align}
where $D/Dt$ is a convective derivative,
\begin{align}
    \frac{D}{Dt} & = \pfrac{}{t} + \mvec{v} \cdot \gx + \frac{q_s}{m_s} \left [ \mvec{E}^m(\mvec{x}, t) + \mvec{v} \times \mvec{B}^m(\mvec{x}, t) \right ] \cdot \gv,
\end{align}
a time derivative with respect to a moving coordinate system.

The Klimontovich equation is essentially an alternative way of expressing the motion of every particle in phase space, and it suffers from the same issues discussed in Section~\ref{sec:challengeModeling}. 
We do not want to track the motion of every particle in phase space, especially if we can prioritize collective effects over individual particle-particle interactions and microscopic fields in a weakly coupled plasma.
But how does one go from the Klimontovich equation to a more suitable representation of a weakly coupled plasma's dynamics? 
How does one obtain an equation which contains the accumulated physics of the many individual particle interactions in our many-body system?

We now leverage a mathematical technique known as an ensemble average.
An ensemble average is an average over realizations of the solution, i.e., an average of the results of different initial conditions.
Imagine, if one could, solving the Klimontovich equation many times and finding with different initial conditions the collective motion of the plasma was similar while the details of the individual particle interactions varied.
A concrete example: imagine solving the Klimontovich equation repeatedly for the plasma system considered in Section~\ref{sec:DebyeLengthPlasmaParameter}. 
While the details of the relaxation to a Debye-shielded charged particle may vary from realization to realization depending on how exactly we initialize the electrons around the particular positively charged particle, we still end up at the same place: a distribution of electrons moving around a positively charged particle, shielding its charge strongly beyond this characteristic length scale of the Debye length.

So what would this mean for the collective behavior, a Debye shielded charged particle for example, to be roughly similar between different realizations of the plasma's dynamics? 
We turn now to the language of statistics to lay a solid foundation for the next derivation.
This roughly similar collective behavior is an example of the average response of the plasma to its internal, individual particle-particle, dynamics, likely with some standard deviation or variance across different realizations. 
While every realization of the Klimontovich equation is deterministic, there is also some stochasticity between different realizations. 
We now argue that a more appropriate, and ultimately more useful, way to characterize the plasma's dynamics is by focusing on this stochasticity, so as to obtain a probabilistic description of the plasma's dynamics.

We define the \emph{particle distribution function} for species $s$ as
\begin{align}
    f_s(\mvec{x}, \mvec{v}, t) = \langle N_s(\mvec{x}, \mvec{v}, t) \rangle, 
\end{align}
where $\langle \cdot \rangle$ defines the ensemble average, the average over many (formally an infinite number) realizations of the plasma. 
The particle distribution function tells us how many particles are \emph{likely} to be found in a small volume $\Delta \mvec{x} \Delta \mvec{v}$. 
Before, the density of particles of species $s$ could only take the value of 0 or 1---it was a simply a sum of Dirac delta functions for the exact location in configuration and velocity space of each particle.
We have now shifted perspective to focusing on the probability of finding a particle at a particular location in position-velocity phase space.

To obtain an equation for the evolution of the particle distribution function we ensemble average \eqr{\ref{eq:conservativeKlimontovich}}, the Klimontovich equation,
\begin{align}
    \pfrac{f_s(\mvec{x}, \mvec{v}, t) }{t} + \gx & \cdot \left [ \mvec{v} f_s(\mvec{x}, \mvec{v}, t) \right ] + \gv \cdot \left \{ \frac{q_s}{m_s} \left [ \mvec{E}(\mvec{x}, t) + \mvec{v} \times \mvec{B}(\mvec{x}, t) \right ] f_s(\mvec{x}, \mvec{v}, t) \right \} \notag \\
    & = - \left \langle \frac{q_s}{m_s} \gv \cdot \left \{ \left [\delta \mvec{E}(\mvec{x}, t) + \mvec{v} \times \delta \mvec{B}(\mvec{x}, t) \right] \delta N_s(\mvec{x}, \mvec{v}, t) \right \}  \right \rangle, \label{eq:plasmaKinetic}
\end{align}
where
\begin{align}
    \delta N_s(\mvec{x}, \mvec{v}, t) & = N_s(\mvec{x}, \mvec{v}, t) - f_s(\mvec{x}, \mvec{v}, t), \\
    \delta \mvec{E}(\mvec{x}, t) & = \mvec{E}^m(\mvec{x}, t) - \mvec{E}(\mvec{x}, t), \\
    \delta \mvec{B}(\mvec{x}, t) & = \mvec{B}^m(\mvec{x}, t) - \mvec{B}(\mvec{x}, t),
\end{align}
and we have used the shorthand $\mvec{E} = \langle \mvec{E}^m \rangle$ and $\mvec{B} = \langle \mvec{B}^m \rangle$ for the ensemble-averaged fields.
By definition, the ensemble average of the fluctuating quantities $\langle \delta N_s \rangle = \langle \delta \mvec{E} \rangle = \langle \delta \mvec{B} \rangle = 0$.
Thus, in the process of ensemble averaging the Klimontovich equation, terms proportional to $\langle N_s \delta \mvec{E} \rangle = N_s \langle \delta \mvec{E} \rangle$ and their permutations will vanish, leaving only the term which is quadratic in the fluctuating quantities.

\eqr{\ref{eq:plasmaKinetic}} is the plasma kinetic equation.
We are close to a more useful equation, as we have replaced a deterministic equation with a probabilistic equation, which will allow us to understand the plasma's collective behavior irrespective of the details of the discrete particle dynamics.
Importantly, in the process of ensemble averaging, we now have on the left hand side of \eqr{\ref{eq:plasmaKinetic}} how the plasma responds to ensemble-averaged electromagnetic fields, i.e., effective electromagnetic fields from the collective motions of the entire plasma instead of individual particle-particle electromagnetic interactions. 
But we have retained the effects of the discrete particle interactions on the right-hand side, or at least the accumulation of many discrete particle interactions.
We need one final simplification, to complete the derivation of the equation, and equation system, which is of principal interest in this thesis.

\section{Bogoliubov's Timescale Hierarchy and the\\ Vlasov--Maxwell--Fokker--Planck System of Equations}\label{sec:introKineticTheory}

To complete the probabilistic picture of a plasma, we need to know the physics of the right hand side of the plasma kinetic equation, \eqr{\ref{eq:plasmaKinetic}}. 
We have already shown in Section~\ref{sec:DebyeLengthPlasmaParameter} that the electric field from an individual particle in the plasma falls off exponentially at length scales larger than the Debye length, so we might expect the physics of these fluctuating fields to be at scales smaller than the Debye length.
Indeed, that must be the case, as the fluctuating electromagnetic fields become vanishingly small on scales larger than the Debye length, i.e., the ``microscopic'' electromagnetic fields and ensemble-averaged electromagnetic fields are indistinguishable when one is no longer considering ``microscopic'' scales. 
This justification may seem like a tautology, that once we consider length and time scales in the plasma on which collective effects arise, we no longer have to concern ourselves with these fluctuating quantities.
In fact, it can be shown that the term on the right hand side of \eqr{\ref{eq:plasmaKinetic}} scales like $\Lambda^{-1}$, the inverse of the plasma parameter,\footnote{One can see this scaling with a thought experiment. Imagine breaking an electron into an infinite number of pieces, so that $n_e \rightarrow \infty, m_e \rightarrow 0, e \rightarrow 0$ while the charge density, charge to mass ratio, and thermal velocity $n_e e, e/m_e, v_{th_e}$ all remain constant. Note that in this thought experiment, the electron temperature $T_e \rightarrow 0$ for the thermal velocity to be constant, while the electron plasma frequency and Debye length $\omega_{pe}, \lambda_D$ are both constant through the break up of the electron. Importantly, this means the plasma parameter $\Lambda = n \lambda_D^3 \rightarrow \infty$. Now, any volume, no matter how small contains an infinite number of point particles with an infinitesimal charge. Statistical mechanics tells us that the fluctuations in the density will scale like the square root of the density, $\delta N_s \sim N^{1/2} \sim \Lambda^{1/2}$, but the electromagnetic fields, for example the electric field from Poisson's equation, scales like $\delta \mvec{E} \sim e \delta N \sim N^{-1} N^{1/2} \sim N^{-1/2}$, because the charge density is constant, meaning $e \sim N^{-1}$. Thus, the right hand side of the plasma kinetic equation, \eqr{\ref{eq:plasmaKinetic}}, is constant in this thought experiment. But on the left hand side of \eqr{\ref{eq:plasmaKinetic}}, the distribution function becomes infinite in this thought experiment, $f_e \rightarrow \infty$, so the right hand side vanishes with the scaling of the left hand side, $N \sim \Lambda$. The contribution of the fluctuating fields is thus $\Lambda^{-1}$ smaller in scaling for the evolution of the particle distribution function.} 
so it is tempting to argue the the fluctuating fields contribute negligibly to the dynamics of a weakly coupled plasma where $\Lambda \gg 1$---there are many particles in a Debye cube that their individual electromagnetic interactions cannot possibly be of consequence.

But the physics of the $\sim \Lambda$ Coulomb collisions the particles are experiencing within a Debye cube is slightly more subtle.
While each individual Coulomb collision a particle experiences is a small effect, a small deviation to its trajectory, the cumulative effect of many Coulomb collisions can significantly perturb the path of the particle.
One may have to wait an exceedingly long time for the cumulative effect of many Coulomb collisions to noticeably affect the plasma's dynamics compared to the collective motion of the plasma contained in the left hand side of \eqr{\ref{eq:plasmaKinetic}}, especially given the scaling of the right hand side compared to the left hand side of $\Lambda^{-1}$.
But, wait long enough, and small deviations will accumulate to make an impact on the dynamics of these plasma particles.

How long is long enough to wait for Coloumb collisions to be of dynamical importance? 
Bogoliubov's timescale hierarchy \citep{Nicholson:1983} tells us that a plasma's dynamical evolution consists of the following stages:
\begin{enumerate}
    \item Pair correlations are established, leading to shielded Coloumb potentials on Debye scales. These correlations are established on the time scale of the inverse electron plasma frequency, \eqr{\ref{eq:plasmaFrequency}}, and once these correlations are established, for $t \omega_{pe} \gtrsim 1$, collective behavior dominates over individual particle interactions.
    \item The plasma relaxes to local thermodynamic equilibrium. We will show in the next section, Section~\ref{sec:PropertiesKineticEquation}, that this relaxation is contained in the physics of collisions, the right hand side of \eqr{\ref{eq:plasmaKinetic}}. If we define a collision frequency $\nu$, we expect the plasma to relax to local thermodynamic equilibrium on time scales $\nu t \gtrsim 1$, a much longer time scale than the plasma frequency $\nu/\omega_{pe} \sim \Lambda^{-1}$, given the scaling of the terms in \eqr{\ref{eq:plasmaKinetic}}.\footnote{We can also argue for the difference in the time scale of collisions versus the plasma frequency by estimating the size of the mean free path, the average distance a particle travels before it experiences a significant deflection due to a binary inter-particle Coulomb collision, compared to the Debye length. Here significant deflection could mean the accumulation of many small angle Coulomb collisions, i.e., small deviations due to individual electrostatic interactions, or by one large angle collision due to a close fly-by of one plasma particle of another. The mean free path can be estimated from the collisional cross section $\sigma$,
    \begin{align}
        \lambda_{mfp} \sim \frac{1}{n \sigma} \sim \frac{T^2}{n e^4}, 
    \end{align}
    where we have estimated the collisional cross section $\sigma \sim d^2$ by balancing the potential energy at a distance $d$ with the average kinetic energy of the particle, $e^2/d \sim T$. Comparing the mean free path and the Debye length, we have,
    \begin{align}
        \frac{\lambda_{mfp}}{\lambda_D} \sim  \frac{T^2}{n e^4} \sqrt{\frac{e^2 n}{T}} \sim n \lambda_D^3, 
    \end{align}
    which is the plasma parameter $\Lambda \sim n \lambda_D^3 \gg 1$. But if the mean free path is much larger than the Debye length, than considering a thermal particle moving with velocity $v_{th}$,
    \begin{align}
        \frac{v_{th}}{v_{th}} \frac{\lambda_{mfp}}{\lambda_D} = \frac{\nu}{\omega_{pe}} \sim \Lambda.
    \end{align}
    }
    \item On time scales $\nu t \gg 1$, the plasma attempts to relax to \emph{global} thermodynamic equilibrium. The plasma's boundary conditions or sources may prevent this global relaxation from occurring, but on these time scales, we would seek alternative means of describing the plasma so as to capture its transport.
\end{enumerate}
We have engaged in a small amount of circumlocution as we attempted to not get too far ahead of ourselves in a heuristic derivation of the equation system of interest.
A detailed derivation of the collisional response of the plasma, valid for all the timescales in Bogoliubov's hierarchy, is a longer calculation.
For now, we state that because collisions are the accumulation of many small effects, that the right hand side of \eqr{\ref{eq:plasmaKinetic}} must inevitably be a Fokker-Planck operator~\citep{Landau:1936,Helander:2005},
\begin{align}
    \left \langle \frac{q_s}{m_s} \gv \cdot \left ( [\delta \mvec{E} + \mvec{v} \times \delta \mvec{B} \right] \delta N_s)  \right \rangle \sim \gv \cdot  \left [ -\left (\mvec{A} f_s \right ) + \frac{1}{2} \gv \cdot \left (\overleftrightarrow{\mvec{D}} f_s \right ) \right ],\label{eq:roughFP}
\end{align}
where we have dropped the spatial dependence temporarily for notational convenience.
%We argue more substantively in Appendix~\ref{app:smallAngleCollisions} that collisions within a plasma are dominated by small angle collisions, and thus collisions are indeed the the accumulation of many small effects.
We note that the details of the derivation of \eqr{\ref{eq:roughFP}} can be found in Chapter 3 and Appendix A of \citet{Nicholson:1983}, where the author performs the full BBGKY (Bogoliubov-Born-Green-Kirkwood-Yvon) hierarchy to derive the equation system of interest, including the Fokker--Planck equation.

Each individual Coulomb collision has a small effect on the trajectory of a particle in a plasma, so in analogy with Brownian motion in a gas, the cumulative effect of many Coulomb collisions is a diffusive process in velocity space. 
The exact expressions for the drag coefficient, $\mvec{A}$, and the diffusion tensor, $\overleftrightarrow{\mvec{D}}$, in \eqr{\ref{eq:roughFP}} require more careful treatment, and a more in depth discussion and derivation \citep{Rosenbluth:1957}.
We choose, in this thesis, a simplified form for the drag and diffusion coefficients,
\begin{align}
    \mvec{A} & = \mvec{u} - \mvec{v}, \label{eq:DoughertyDrag} \\
    \overleftrightarrow{\mvec{D}} & = \frac{2 T}{m} \overleftrightarrow{\mvec{I}}, \label{eq:DoughertyDiffusion}
\end{align}
where $\overleftrightarrow{\mvec{I}}$ is the identity tensor.
These simplified drag and diffusion coefficients are related to the velocity moments of the particle distribution function,
\begin{align}
    \mvec{u}(\mvec{x}, t) & = \frac{\int \mvec{v} f(\mvec{x}, \mvec{v}, t) \dv}{\int f(\mvec{x}, \mvec{v}, t) \dv}, \label{eq:flowDefinition} \\
    \frac{T(\mvec{x}, t)}{m} & = \frac{1}{3} \frac{\int |\mvec{v} - \mvec{u}(\mvec{x}, t)|^2 f(\mvec{x}, \mvec{v}, t) \dv}{\int f(\mvec{x}, \mvec{v}, t) \dv}, \label{eq:temperatureDefinition}
\end{align}
where the factor of $1/3$ in the second equation follows from the fact that there are three velocity dimensions. 

You probably recognize \eqr{\ref{eq:temperatureDefinition}} as the thermal velocity squared from \eqr{\ref{eq:thermalVelocity}}.
We could use the thermal velocity squared\footnote{Note that in this definition, Boltzmann's constant has been absorbed into the temperature, $k_B T_s \rightarrow T_s$, so that the units of temperature are an energy, e.g. electron-volts or Joules.} as the diffusion coefficient, $v_{th}^2 = T/m$, but we want to emphasize the connection between the drag and diffusion coefficients and the velocity moments of the particle distribution function, so we will switch notation and define the diffusion coefficient with respect to the temperature.
While these expressions for the drag and diffusion coefficients may appear unintuitive at first glance, there is a rich history for their use as a lowest order approximation to the Fokker-Planck behavior we expect the plasma to have due to Coloumb collisions.

The ``full'' Fokker--Planck operator \citep{Rosenbluth:1957} includes additional physics, most importantly that collisions should be velocity dependent and that faster particles experience fewer collisions.
However, the solution to the complete Fokker--Planck operator is more computationally demanding.
For the purposes of the algorithms and physics presented in this thesis, the most important component of modeling collisions is that collisions are a Fokker--Planck operator modeling the drag and diffusion in velocity space particles should experience from the accumulation of many small-angle collisions.
We will continue to call this operator the Fokker--Planck operator throughout this thesis, but with these drag and diffusion coefficients, one can commonly find the names Lenard-Bernstein and Dougherty attached~\citep{Lenard:1958,Dougherty:1964}.
These particular drag and diffusion coefficients will also appear particularly inspired after we explore the properties of the kinetic equation in the next section, Section~\ref{sec:PropertiesKineticEquation}.

Let us now bring all the pieces together to describe the equation system in totality. We begin with the Vlasov--Fokker--Planck equation for the evolution of the particle distribution function for each species $s$ in phase space,
\begin{align}
    \pfrac{f_s(\mvec{x}, \mvec{v}, t) }{t} +&  \gx \cdot \left [ \mvec{v} f_s(\mvec{x}, \mvec{v}, t) \right ] + \gv \cdot \left \{ \frac{q_s}{m_s} \left [ \mvec{E}(\mvec{x}, t) + \mvec{v} \times \mvec{B}(\mvec{x}, t) \right ] f_s(\mvec{x}, \mvec{v}, t) \right \} \notag \\
    & = \nu_s \gv \cdot \left \{ [\mvec{v} - \mvec{u}_s(\mvec{x}, t)] f_s(\mvec{x}, \mvec{v}, t) + \frac{T_s(\mvec{x}, t)}{m_s} \gv f_s(\mvec{x}, \mvec{v}, t) \right \},
\end{align}
where we have added the collision frequency $\nu_s$, which will allow us to accurately characterize the contribution of the collision operator to the dynamics in comparison to the collisionless evolution from the macroscopic electromagnetic fields. 
In other words, we can pick the collision frequency $\nu_s$ to be $\Lambda^{-1}$ smaller than the electron plasma frequency, $\omega_{pe}$, as it should be.
This equation is coupled to the ensemble-averaged Maxwell's equations for the evolution of the macroscopic electromagnetic fields,
\begin{align}
  \frac{\partial \mvec{B}(\mvec{x}, t)}{\partial t} + \gx\times\mvec{E}(\mvec{x}, t) &= 0, \label{eq:dbdt} \\
  \epsilon_0\mu_0\frac{\partial \mvec{E}(\mvec{x}, t)}{\partial t} - \gx\times\mvec{B}(\mvec{x}, t) &= -\mu_0\mvec{J}(\mvec{x}, t), \label{eq:dedt} \\
  \gx\cdot\mvec{E}(\mvec{x}, t) &= \frac{\varrho_c(\mvec{x}, t)}{\epsilon_0}, \label{eq:divE} \\
  \gx\cdot\mvec{B}(\mvec{x}, t) &= 0, \label{eq:divB}
\end{align}
where the current density $\mvec{J}$ and charge density $\rho_c$ are related to velocity moments of the particle distribution function,
\begin{align}
    \rho_c(\mvec{x}, t) & = \sum_s q_s \int f_s(\mvec{x}, \mvec{v}, t)\dv, \label{eq:chargeDensity} \\
    \mvec{J}(\mvec{x}, t) & = \sum_s q_s \int \mvec{v} f_s(\mvec{x}, \mvec{v}, t)\dv. \label{eq:currentDensity}
\end{align}

Having closed the equation system with the coupling between the electromagnetic fields and the particle distribution function, the \emph{Vlasov--Maxwell--Fokker--Planck} system of equations is complete. 
This equation system forms the foundation for the theory of weakly coupled plasmas and will be the principal focus for the remainder of this thesis.
The particle distribution function contains a wealth of data, and we are strongly motivated by the veritable treasure trove of information the particle distribution function holds.
We thus want to make sure however we choose to numerically integrate the Vlasov--Maxwell--Fokker--Planck system of equations, we can still leverage the particle distribution function to understand the plasma's dynamics.

While we now have an equation system we can actually use the computer to solve, having simplified the N-body dynamics of the plasma to a probabilistic equation system in a six dimensional phase space, we must still be careful in our next steps for how we discretize the Vlasov--Maxwell--Fokker-Planck system of equations.
In this next section, we will review many of the most important properties of the Vlasov--Maxwell--Fokker-Planck system of equations.
These properties of the \emph{continuous} system of equations will help us ultimately make an informed decision in both our choice, and implementation, of the numerical method for constructing the \emph{discrete} Vlasov--Maxwell--Fokker-Planck system of equations.

\section{Properties of the Vlasov--Maxwell--Fokker--Planck \\System of Equations}\label{sec:PropertiesKineticEquation}

Before we begin this discussion of the properties of the continuous Vlasov--Maxwell--Fokker--Planck (VM-FP) system of equations, we want to simplify some of our notation for readability.
Firstly, we will separate the collisionless and collisional components of the system of equations,
\begin{align}
    \pfrac{f^{collisionless}_s}{t} & = -\gx \cdot (\mvec{v} f_s) - \gv \cdot \left [ \frac{q_s}{m_s} \left ( \mvec{E}+ \mvec{v} \times \mvec{B} \right ) f_s \right ], \label{eq:collisionlessComponent} \\
    \pfrac{f^{c}_s}{t} & = \nu_s \gv \cdot \left [ (\mvec{v} - \mvec{u}_s) f_s + \frac{T_s}{m_s} \gv f_s \right ], \label{eq:collisionalComponent} \\
    \pfrac{f_s}{t} & = \pfrac{f^{collisionless}_s}{t} + \pfrac{f^{c}_s}{t},
\end{align}
and we will drop the notation for the explicit dependence on configuration space and phase space. 
Importantly, this separation is for readability of the coming discussion of the properties of the VM-FP system of equations, and is not due to any explicit need to separate the collisionless and collisional components of the plasma's evolution.
These contributions to the plasma's dynamics are on equal footing, and we could just as easily demonstrate the properties of VM-FP system of equations holistically, but we feel this makes the subsequent discussion unnecessarily dense.
Suffice to say, if, for example both the collisionless and collisional components of the VM-FP system of equations conserve the total energy in the system, we know that together the whole system conserves the energy.

%We note that the effect of cross-species collisions, for example electron-ion and and ion-electron collisions, can also be modeled by the simplified Fokker-Planck operator employed in this thesis,
%\begin{align}
%    \pfrac{f^{c}_s}{t} = \nu_s \gv \cdot & \left [ (\mvec{v} - \mvec{u}_s) f + \frac{T_s}{m_s} \gv f_s \right ] \notag \\
%    & + \sum_r \nu_{rs} \gv \cdot \left [ (\mvec{v} - \mvec{u}_{rs}) f + \frac{T_{rs}}{m_s} \gv f_s \right ],
%\end{align}
%where $\mvec{u}_{rs}$ and $T_{rs}$ are to-be-determined coefficients which model the drag and diffusion one species experiences due to collisions with other species in the plasma. 
%We show in Appendix~\ref{app:multSpeciesModelFP} that we can choose these coefficients based on properties we desire, so we will focus on the form in \eqr{\ref{eq:collisionalComponent}} in our discussions of the inherent properties of the collision operator. 

For brevity of notation, we will introduce the full phase space variable $\mvec{z} = (\mvec{x}, \mvec{v})$ so that the collisionless component of the VM-FP system of equations can be written as,
\begin{align}
    \pfrac{f^{collisionless}_s}{t} = -\gz \cdot (\gvec{\alpha}_s f_s), \label{eq:vlasovConservationEq}
\end{align}
a conservation equation in the full phase space with phase space flux,
\begin{align}
    \gvec{\alpha}_s = \left (\mvec{v}, \frac{q_s}{m_s} [\mvec{E} + \mvec{v} \times \mvec{B}] \right ).
\end{align}
We will also use the notation $K$ to define the phase space domain that the distribution function is defined on and $\Omega$ to define the configuration space domain that velocity space moments and electromagnetic fields are defined on.

We have hinted at the connection between the bulk properties of the plasma and the velocity moments of the particle distribution function. 
Both the components of the drag and diffusion coefficients, Eqns.\thinspace\ref{eq:flowDefinition}--\ref{eq:temperatureDefinition}, and the charge density and current density that couple the particle dynamics to the electromagnetic fields, Eqns.\thinspace\ref{eq:chargeDensity}--\ref{eq:currentDensity}, are defined with integrals over velocity space of the particle distribution function.
We should solidify this connection with a number of definitions which will prove critical to our discussion of the properties of the continuous VM-FP system of equations.
We will focus on the first few velocity moments,
\begin{align}
    \rho_s = m_s n_s & = m_s \int f_s \dv, \label{eq:massDefinition} \\
    \gvec{\mathcal{M}}_s = m_s n_s \mvec{u}_s & = m_s \int \mvec{v} f_s \dv, \label{eq:momentumDefinition} \\
    \mathcal{E}_s = \frac{3}{2} n_s T_s + \frac{1}{2} m_s n_s |\mvec{u}_s|^2 & = \frac{1}{2} m_s \int |\mvec{v}|^2 f_s \dv, \label{eq:energyDefinition}
\end{align}
i.e., the mass density, momentum density, and energy density of the plasma species with label $s$.

To gain intuition for why these quantities can be defined this way, recall what the particle distribution function is: the probability of finding a particle in a given volume $\Delta \mvec{x} \Delta \mvec{v}$.
Thus, if we integrate the particle distribution in velocity space, \eqr{\ref{eq:massDefinition}}, we are computing the number density of the particles (the number of particles per unit volume) at a given configuration space location, or the mass density at a given configuration space location.
For the higher velocity moments, we can make similar connections.
The velocity weighted moment, which includes a factor of the particle mass, \eqr{\ref{eq:momentumDefinition}}, tells us the amount of momentum per unit volume, the momentum density, at a particular configuration space location.
We might have also guessed this physical interpretation for \eqr{\ref{eq:momentumDefinition}} by considering what statistics tells us the first velocity moment is: the average velocity of the particles.
This same logic can be applied to \eqr{\ref{eq:energyDefinition}}; the second velocity moment, weighted by $m_s/2$, gives us the total energy density---internal, $3/2 \thinspace n_s T_s$, plus kinetic, $1/2 \thinspace m_s n_s |\mvec{u}_s|^2$---of the particles at a given configuration space location.
And, in the language of statistics, the second velocity moment is related to the spread, or standard deviation, of the particle velocities.

In both velocity space moment cases, we should be careful not to make the connection between our physical intuition and our knowledge of statistics superficial.
The average velocity is only $\mvec{u}$, not the full definition of the momentum density in \eqr{\ref{eq:momentumDefinition}}, and we have to account for this average velocity when computing the real standard deviation, i.e.,
\begin{align}
    \sigma \propto \sqrt{\int |\mvec{v} - \mvec{u}_s|^2 f_s \dv},
\end{align}
where we have used the standard notation of the variable $\sigma$ for the standard deviation.
It is not the energy density, \eqr{\ref{eq:energyDefinition}}, that is the variance of the particle distribution function.
Only the square root of the internal energy, $3/2 \thinspace n_s T_s$, will enter into the definition of the variance, because in subtracting off the average velocity we are eliminating the kinetic energy, $1/2 \thinspace m_s n_s |\mvec{u}_s|^2$, component.
If we recall our definitions for the various components of the drag and diffusion coefficients, \eqr{\ref{eq:flowDefinition}} and \eqr{\ref{eq:temperatureDefinition}}, we can make the parallels concrete, and drive home some of the intuition for our choice of simplified drag and diffusion coefficients.
The particle distribution function is centered around some velocity $\mvec{u}$, the average velocity of the particles, with some variance in velocity space quantifying the thermal spread of the particles\footnote{In the diffusion coefficient \eqr{\ref{eq:temperatureDefinition}}, $T_s/m_s$ is the standard deviation squared. The variance must have the same units as velocity and thus the actual spread in velocity space of the distribution function is $\sqrt{T_s/m_s}$, the thermal velocity, again noting that we have absorbed Boltzmann's constant into our definition of the temperature, $k_B T_s \rightarrow T_s$.}, $\sqrt{T_s/m_s}$.
Thus, we naturally have links between our physical intuition for how much of the plasma's mass, momentum, and energy is at a single physical location in configuration space and the statistical nature of the particle distribution function quantifying the probability of particles being located in a given volume $\Delta \mvec{x} \Delta \mvec{v}$.

Let us now move on to properties of the continuous VM-FP system of equations. 
Since one set of properties we wish to quantify are the conservation relations inherent to the system of equations, we will need to assume specific boundary conditions for the distribution function and electromagnetic fields.
In particular, we will assume the distribution function $f(\mvec{x},\mvec{v}\rightarrow\pm\infty, t)\rightarrow 0$ faster than the logarithmic singularity $\ln(f_s)$.
Note that in this assumption, it naturally follows that $f(\mvec{x},\mvec{v}\rightarrow\pm\infty, t)\rightarrow 0$ faster than $\mvec{v}^n$ for finite $n$.
Likewise, we will take configuration space to be either periodic or some similar self-contained boundary condition, such as a reflecting wall for $\mvec{E}, \mvec{B}$, and the distribution function at the edge of configuration space.

We wish to be rigorous at this point and prove many of these properties, but to avoid the discussion becoming overly cumbersome here in the introduction, we prove all of the forthcoming properties of the VM-FP system of equations in Appendix~\ref{app:proofsContinuous}.
Here, we will only state the properties to foreshadow the work we will do in the upcoming chapters on retaining properties of the continuous VM-FP system of equations when we discretize and numerically integrate the equation system.
We will first focus on the collisionless component, \eqr{\ref{eq:collisionlessComponent}}, of the VM-FP system of equations, often referred to as the Vlasov--Maxwell part.
The Vlasov--Maxwell system of equations has the following properties:
\begin{proposition} \label{prop:collisionlessMassConservation}
The Vlasov--Maxwell system conserves mass,
   \begin{align}
      \frac{d}{dt}\left ( m_s \int_K f_s \dz \right ) = 0.
   \end{align}
\end{proposition}
\begin{proposition}\label{prop:collisionlessL2}
The collisionless Vlasov--Maxwell system conserves the $L^2$ norm of the distribution function, i.e.,
   \begin{align}
      \frac{d}{dt} \left ( \frac{1}{2} \int_K f_s^2 \dz \right ) = 0.
   \end{align}
\end{proposition}
\begin{proposition}\label{prop:collisionlessEntropy}
The collisionless Vlasov--Maxwell system conserves the entropy density $S = -f \ln(f)$ of the system\footnote{Note that it is the physicists' convention to include a minus sign in the definition of the entropy, thus making the entropy a non-decreasing quantity and the Maxwellian the maximum entropy state. The minus sign could be dropped, as is often done in the theory of hyperbolic conservation laws, and then the entropy would be a non-increasing quantity and the Maxwellian would minimize the entropy. For a discussion of the Maxwellian velocity distribution as the entropy maximizing particle distribution function, see Proposition~\ref{prop:collisionEntropy} and Corollary~\ref{coro:HTheorem}.},
\begin{align}
    \frac{d}{dt} \left [\int_K -f_s \ln(f_s) \dz \right ] = 0.
\end{align}
\end{proposition}
\begin{proposition}\label{prop:collisionlessMomentum}
The Vlasov-Maxwell system conserves the total, particles plus fields, momentum,
   \begin{align}
      \frac{d}{dt} \left (\int_\Omega \sum_s \gvec{\mathcal{M}}_s  +  \epsilon_0 \mvec{E}\times\mvec{B} \dx \right )= 0.
   \end{align}
The first term is the total particle momentum, and the second term is the momentum carried by the electromagnetic fields.
\end{proposition}
\begin{proposition}\label{prop:collisionlessEnergyConservation}
The Vlasov-Maxwell system conserves the total, particles plus fields, energy,
   \begin{align}
      \frac{d}{dt} \left (\int_\Omega \sum_s \mathcal{E}_s+ \frac{\epsilon_0}{2} |\mvec{E}|^2 + \frac{1}{2 \mu_0}|\mvec{B}|^2 \dx \right )= 0.
   \end{align}
The first term is the total particle energy, and the second two terms are the energy contained in the electromagnetic fields.
\end{proposition}

So, the collisionless Vlasov--Maxwell system of equations conserves mass, total momentum, and total energy, and additionally the entropy of the particles is unchanged by the collisionless component of the VM-FP system of equations.
The latter property of entropy conservation in the collisionless system naturally leads us to a discussion of collisions.
We alluded to the effect collisions would have on the thermodynamics of the plasma with Bogoliubov's timescale hierarchy in Section~\ref{sec:introKineticTheory}.
We will now make the connection concrete with a discussion of the properties of the Fokker-Planck collision operator, \eqr{\ref{eq:collisionalComponent}}.
We first focus on the conservation properties of the Fokker-Planck collision operator, and then we will discuss the effect of the collision operator on the thermodynamics of the system.
As with our discussion of the collisionless Vlasov-Maxwell system of equations, the proofs for the properties of the continuous Fokker--Planck collision operator can be found in Appendix~\ref{app:proofsContinuous}.
\begin{proposition}\label{prop:collisionMassConservation}
The Fokker--Planck equation conserves mass,
   \begin{align}
      \frac{d}{dt}\left ( m_s \int_K f^c_s \dz \right ) = 0.
   \end{align}
\end{proposition}
\begin{proposition}\label{prop:collisionMomentumConservation}
The Fokker--Planck equation conserves the particle momentum,
   \begin{align}
      \frac{d}{dt} \left (\int_K m_s \mvec{v} f_s^c \dz \right )= 0.
   \end{align}
\end{proposition}
\begin{proposition}\label{prop:collisionEnergyConservation}
The Fokker--Planck equation conserves the particle energy,
   \begin{align}
      \frac{d}{dt} \left (\int_K \frac{1}{2} m_s |\mvec{v}|^2 f_s^c \dz \right )= 0.
   \end{align}
\end{proposition}
\begin{proposition}\label{prop:collisionEntropy}
The Fokker--Planck equation leads to a non-decreasing entropy density, $S = -f \ln(f)$, of the system,
   \begin{align}
      \frac{d}{dt} \left [\int_K -f_s^c \ln(f_s^c) \dz \right ] \geq 0.
   \end{align}
Thus, the Vlasov--Maxwell--Fokker--Planck system of equations satisfies the Second Law of Thermodynamics, $\Delta S \geq 0$.
\end{proposition}
\begin{corollary}\label{coro:HTheorem}
The maximum entropy solution to the Fokker--Planck collision operator is attained by the Maxwellian velocity distribution,
\begin{align}
    f_s = n_s \left (\frac{m_s}{2 \pi T_s} \right )^{\frac{3}{2}} \exp \left (-m_s \frac{|\mvec{v} - \mvec{u}_s|^2}{2 T_s} \right ).
\end{align}
Thus, the Vlasov--Maxwell--Fokker--Planck system of equations satisfies Boltzmann's H-theorem, and a plasma in local thermodynamic equilibrium is described by the Maxwellian velocity distribution.
\end{corollary}

So, the Fokker--Planck component of the VM-FP system of equations also conserves mass, momentum, and energy, so that the complete equation system possesses these properties.
And the Fokker--Planck component is a critical piece of the evolution of the thermodynamics of the plasma, governing both entropy production and providing us the form of the distribution function which maximizes the entropy and describes local thermodynamic equilibrium---see Appendix~\ref{app:proofsContinuous} for further discussions of the connection between the Maxwellian velocity distribution and local thermodynamic equilibrium.
We should reiterate that our discussion of the collision operator in the VM-FP system of equations utilizes simplified drag and diffusion coefficients \citep{Lenard:1958, Dougherty:1964}, Eqns.\thinspace (\ref{eq:flowDefinition}) and (\ref{eq:temperatureDefinition}), and that, while collisions in a plasma are well approximated by a Fokker--Planck operator, the real drag and diffusion coefficients are more complex \citep{Rosenbluth:1957}.
Nonetheless, this equation system contains all the ingredients required to characterize a weakly coupled plasma, a plasma whose collective motions dominate over individual particle-particle interactions.
This equation system is simultaneously more computationally tractable than integrating all the particle trajectories, while also still containing the properties our physical intuition tells us the plasma should have despite this perspective shift to a probabilistic picture from the deterministic picture of individual particle motions.

This discussion naturally leads us into the next section.
We have presented an equation system for modeling a myriad of plasma systems, relevant everywhere from laboratories, to the heliosphere, to astrophysical systems such as the interstellar and intracluster medium.
We want to now utilize the computer to understand the dynamics of weakly coupled plasmas.
But just because we have made the problem of simulating plasma dynamics computationally tractable, shifting our perspective from integrating every single particle's equations of motion to focusing on the collective behavior we know to be of critical importance, does not imply we have made the problem easy.
There is a rich history in tackling the numerical integration of the Vlasov--Maxwell--Fokker--Planck system of equations, and it is worth reviewing this history to motivate the novel approach derived and implemented in this thesis.

\section{A Brief History of Kinetic Numerical Methods\\ and the Objectives of This Thesis}\label{sec:introObjectives}
We restate here, in its entirety, the VM-FP, or  Vlasov--Maxwell--Fokker--Planck, system of equations,
\begin{align}
    \pfrac{f_s}{t} = -\gz \cdot (\gvec{\alpha}_s f_s) & + \nu_s \gv \cdot \left [ (\mvec{v} - \mvec{u}_s) f_s + \frac{T_s}{m_s} \gv f_s \right ], \notag \\
    \frac{\partial \mvec{B}}{\partial t} + \gx\times\mvec{E} = 0, & \quad \epsilon_0\mu_0\frac{\partial \mvec{E}}{\partial t} - \gx\times\mvec{B} = -\mu_0\mvec{J}, \notag \\
    \gx\cdot\mvec{E} = \frac{\varrho_c}{\epsilon_0}, & \quad \gx\cdot\mvec{B} = 0, \notag
\end{align}
where,
\begin{align}
    \gvec{\alpha}_s = & \left ( \mvec{v}, \frac{q_s}{m_s} [\mvec{E} + \mvec{v} \times \mvec{B}] \right ), \notag \\
    \mvec{u_s} = \frac{\int \mvec{v} f_s \dv}{\int f_s \dv}, & \quad \frac{T_s}{m_s} = \frac{1}{3} \frac{\int |\mvec{v} - \mvec{u}_s|^2 f_s \dv}{\int f_s \dv}, \notag \\
    \rho_c = \sum_s q_s \int f_s\dv, & \quad \mvec{J} = \sum_s q_s \int \mvec{v} f_s\dv, \notag 
\end{align}
define the phase space flux, flow and temperature per mass, and charge density and current density, which close the system of equations and couple the electromagnetic fields to the motion of the particles.
This equation system provides an alternative, ultimately more useful, perspective on the evolution of the plasma by shifting from a purely deterministic picture to a probabilistic picture; we track the evolution of the particle distribution function for the probability of finding particles in a phase space volume $\Delta \mvec{x} \Delta \mvec{v}$ instead of every individual particle in the plasma.

Given the discussion in Section~\ref{sec:PropertiesKineticEquation}, we would like however we ultimately discretize the VM-FP system of equations to retain some of these properties of the continuous system of equations.
But, we also want to weigh the computational feasibility of our approach.
The VM-FP system of equations involves the solution of a high dimensional, up to six dimensions plus time, partial differential equation, and this presents its own challenges numerically.

Because of the high dimensionality of the Vlasov--Fokker--Planck equation for the dynamics of the particle distribution function, the most common numerical techniques historically have been Monte Carlo methods, principally the particle-in-cell (PIC) method \citep{Dawson:1962, Langdon:1970, Dawson:1983, birdsallbook}.
This approach attempts to alleviate the computational challenge in integrating the Vlasov--Fokker--Planck equation in the six dimensional phase space by discretizing the particle distribution function as a collection of ``macroparticles,'' i.e., particles of finite size \citep[see, e.g,][and references therein]{Lapenta:2012}.
Maxwell's equations are then discretized on a grid, and the charge and current density of the ``macroparticles'' are deposited on the grid for the coupling.
By making the particles have finite size, the scheme essentially smooths over the spatial scales of the particle size, eliminating discrete particle effects.
Thus, despite the numerical method involving the integration of particle trajectories, the PIC method really is a discretization of the VM-FP system of equations.
There are additional subtleties for the Fokker--Planck component of the equation system since the collisional component of the dynamics occurs inside the macroparticle's finite size; thus, numerical, unphysical, collisions can arise \citep{Hockney:1968, Okuda:1970, Okuda:1972, Hockney:1971, Langdon:1979, Krommes:2007}, and the implementation of a physical collision operator requires modifications to the underlying particle-in-cell algorithm \citep{Lemons:2009}.

As a consequence of discretizing the particle distribution function as a collection of macroparticles, the numerical method only requires a configuration space grid---the velocity space discretization is implicit in the sampling of the particles to compute quantities such as the charge and current density.
Thus, the dimensionality of the problem is reduced from six to three, with the freedom to use as many, or as few, particles per configuration space grid cell as deemed necessary to resolve the kinetic plasma physics encompassed in the VM-FP system of equations.
This reduction in dimensionality, combined with modern algorithms for particle-sorting and sampling, allows one to construct efficient schemes for the complete particle-in-cell algorithm which perform well on the largest supercomputers in the world \citep[e.g.,][]{Fonseca:2008,Bowers:2009,Germaschewski:2016}.

Discretizing the particle distribution function as a collection of macroparticles has its disadvantages though, chief among them the particle noise that is introduced via the particle's finite size.%, see Appendix~\ref{app:noiseParticleMethods}.
This pollution of the solution of the VM-FP system of equations is a real travesty, as the particle distribution function is such a rich tapestry of the underlying physics of the weakly coupled plasma.
One can always mollify this concern by increasing the number of particles in the simulation, but the counting noise decreases like $1/\sqrt{N}$, where $N$ is the number of particles per grid cell. 

In addition to degrading the quality of the solution and potentially making the ultimate analysis more challenging, the particle noise inherent to the PIC algorithm can have more severe consequences, potentially giving incorrect or deceptive answers in situations requiring high signal to noise ratios.
For example,~\citet{Camporeale:2016} have demonstrated that a large number of particles-per-cell is required to correctly identify wave-particle resonances and compare well with linear theory. 
There are means of reducing noise in PIC methods, such as the delta-f PIC method \citep{Parker:1993, Hu:1994, Denton:1995, Belova:1997, ChengJianhua:2013, Kunz:2014b}, but noise mitigation techniques like the delta-f PIC method can break down if the distribution function deviates significantly from its initial value.
Further noise mitigation techniques, such as very high order particle shapes, e.g., particle-in-wavelets \citep{vanyenNguyen:2010, vanyenNguyen:2011} and von Mises distributions based on Kernel Density Estimation theory \citep{Wu:2018}, and time-dependent deformable shape functions for the particles \citep{Coppa:1996,Abel:2012,Hahn:2015,KatesHarbeck:2016} are active areas of research.
However, these more sophisticated particle shape functions add significant computational complexity to the algorithm.
Thus, preliminary application of some of these techniques is done in post-processing to assist in analysis \citep{Totorica:2018}, and not \emph{in situ} during a simulation, so any issues due to noise that arise during the course of a simulation are not mitigated.

We thus have strong motivation, both from a desire to eliminate noise and a desire to fully leverage the particle distribution function in our analyses, to directly discretize the VM-FP system of equations on a phase space grid.
But as we have said before, direct discretization of a six dimensional, plus time, partial differential equation, presents its own challenges.
To mitigate the cost, much of the current body of research on direct discretization of the VM-FP system of equations has focused on the hybrid approximation \citep{Valentini:2007, Valentini:2010, Greco:2012, Perrone:2013, Servidio:2014, Valentini:2016, Kempf:2012, Kempf:2013, Pokhotelov:2013, Palmroth:2018}.
In this approximation, proton species are treated with the Vlasov--Maxwell system of equations, with potentially a Fokker--Planck equation for the ion-ion collisions \citep{Pezzi:2015, Pezzi:2019}, while the electrons are taken to be a massless, isothermal background.
This approximation still requires the solution of the VM-FP system of equations on a high dimensional phase space grid, but the challenges in multi-scale modeling of a plasma, from the electron to the proton scales to the macroscopic dynamics, are alleviated.
There are exceptions in recent years \citep{Vencels:2016, Wettervik:2017, Roytershteyn:2018, Roytershteyn:2019}, but the direct discretization approach for the full VM-FP system of equations for the solution of a multi-species weakly coupled plasma, including the effects of collisions, is not common.

It is the objective of this thesis to outline, derive, and implement a novel scheme for the numerical integration of the multi-species VM-FP system of equations. 
Such a scheme should, as much as possible, respect the properties derived in Section~\ref{sec:PropertiesKineticEquation}.
But, in order for our scheme to accomplish this goal, we must be careful to respect the fact that many of these properties, most especially the conservation properties, are \emph{implicit} to the equation system being evolved.
In other words, we must, for example, encode the fact that the second velocity space moment is a conserved quantity in our evolution of the particle distribution function.
Especially for Fokker--Planck collision operators, such schemes are an active area of research \citep{Taitano:2015, Hirvijoki:2017, Hirvijoki:2018}, but the task of a robust, accurate, conservative, and cost effective numerical method for the full VM-FP system of equations is a tall task.
We have tackled this task in this thesis, and applied the resulting algorithm to a wide variety of plasma systems to solve outstanding questions about the energization mechanisms in fundamental plasma processes and the nonlinear dynamics of saturated plasma instabilities using the pristine, noise-free, distribution function granted to us by a continuum discretization of the VM-FP system of equations.

As an example of the power of this approach of direct discretization, we show in Figure~\ref{fig:proton-dist-proof-of-concept} the results of a simulation we will discuss in Chapter~\ref{ch:Leverage}.
\begin{figure}
    \centering
    \vspace{-0.8in}
    \includegraphics[width=\textwidth]{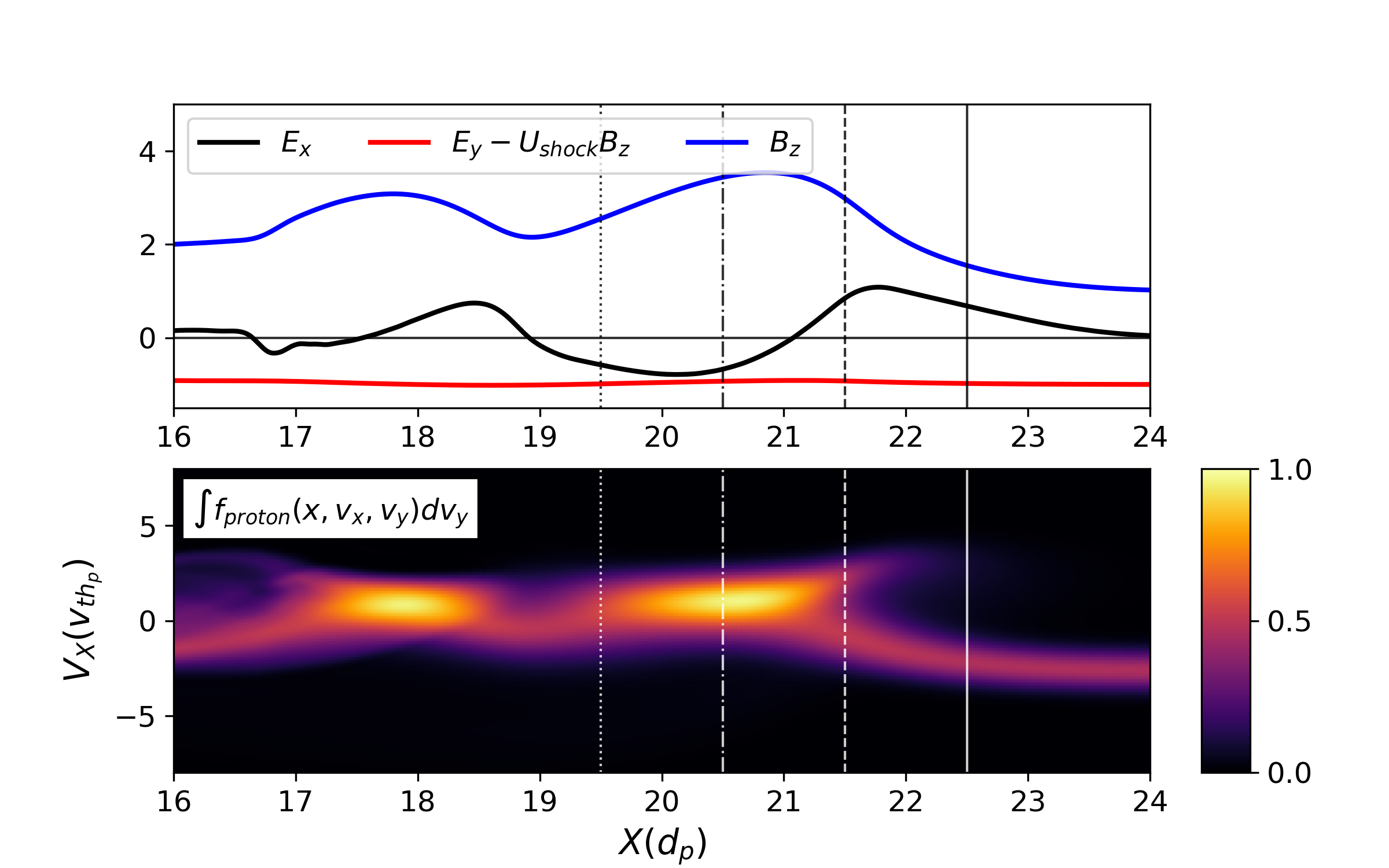}
    \includegraphics[width=0.7\textwidth]{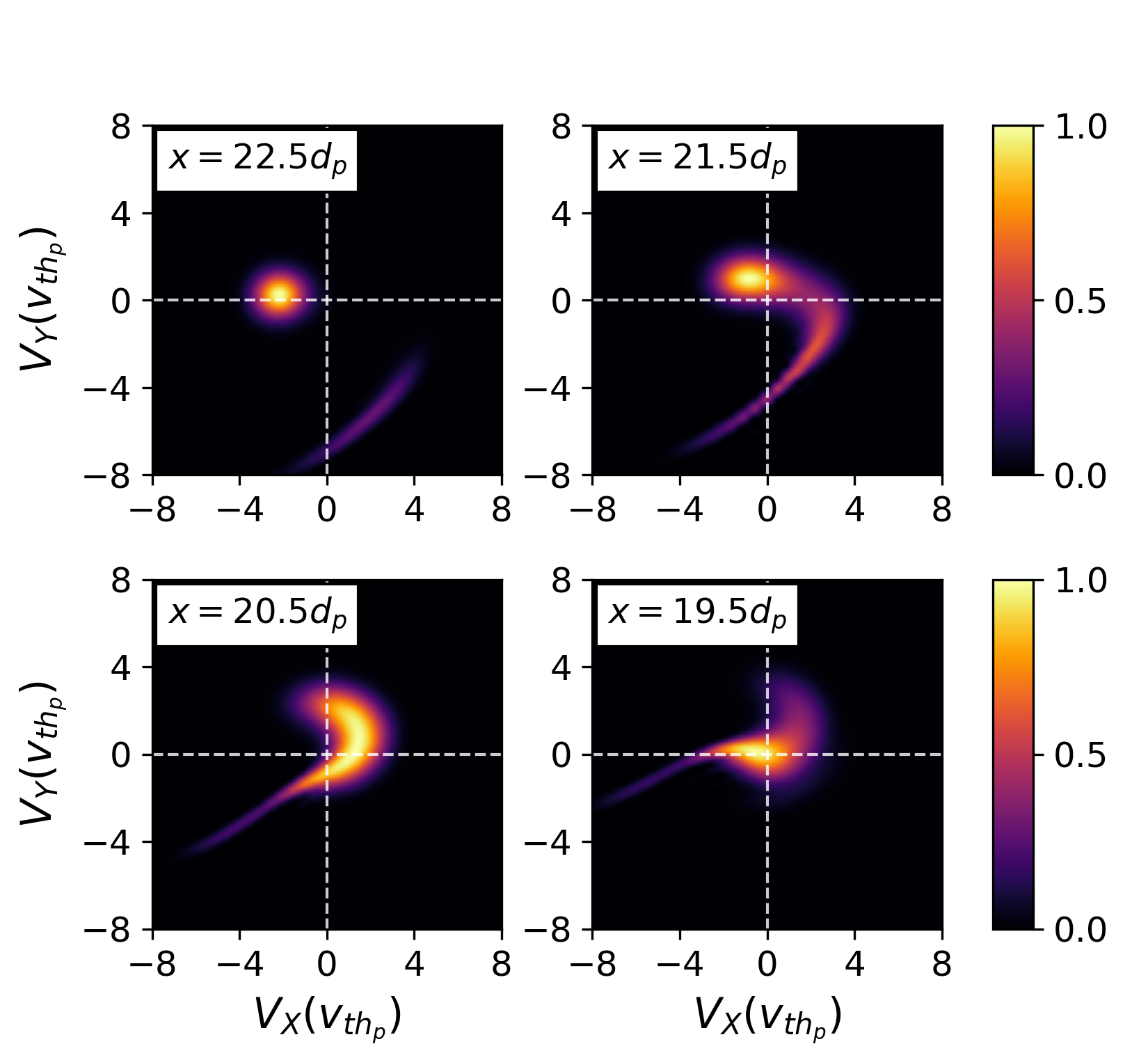}
    \vspace{-0.17in}
    \caption{The electromagnetic fields (top plot) and proton distribution function due to a collisionless shock, where the kinetic energy of an incoming supersonic flow is dissipated and converted into other forms of energy, e.g., thermal energy, on scales smaller than the particle mean-free path, such as the proton inertial length $d_p = c/\omega_{pp}$. We plot the reduced proton distribution function in $x-v_x$ (second from top plot) and slices of the proton distribution function in $v_x-v_y$ (bottom plots) at the specified lines in the $x-v_x$ plots, $x = 19.5, 20.5, 21.5,$ and $22.5\thinspace  d_p$. We will discuss this structure and the specific energization mechanisms of this collisionless shock in Chapter~\ref{ch:Leverage}, but for now we draw attention to the quality of the solution from a continuum representation of the distribution function using a phase space grid. By directly discretizing the VM-FP system of equations in phase space, we can represent fine-scale structure in velocity space which we can leverage to dive in to the details of the energization of the protons.}
    \label{fig:proton-dist-proof-of-concept}
\end{figure}
Figure~\ref{fig:proton-dist-proof-of-concept} shows the proton distribution function undergoing energization due to a collisionless shock, a shock wave which forms on scales smaller than the particle's mean-free path.
The conversion of energy in collisionless shocks, from the kinetic energy of the incoming supersonic flow to other forms of energy, e.g., thermal energy, thus occurs due to kinetic processes such as wave-particle interactions and small-scale instabilities rather than inter-particle collisions.
We will study this system in greater detail when we discuss analysis techniques for extracting data from such a pristine representation of the distribution function.
Suffice to say, the quality of the distribution function from the continuum approach discussed in this thesis is made manifest by inspection of the structure the algorithm can resolve on a phase space grid.

Having motivated our wish to directly discretize the VM-FP system of equations, and briefly demonstrated the capability to resolve detailed particle distribution function structure in kinetic plasma processes like collisionless shocks with this approach, we now discuss the organization of the rest of the thesis.
We will describe the numerical method, the discontinuous Galerkin finite element method, in Chapter~\ref{ch:DGFEM}.
Chapter~\ref{ch:DGFEM} will form a complete mathematical description of our discrete system, including what properties the discrete VM-FP system of equations retains compared to the continuous VM-FP system of equations, and the stability properties of the algorithm.
We will then move to a discussion of the implementation of the algorithm in Chapter~\ref{ch:ImplementationDGFEM}. 
This discussion will detail two of the major breakthroughs in this thesis: the requirement that the algorithm be \emph{alias-free} so it retains the properties of the discrete scheme, most especially the stability and conservation properties, and the specific choice of an \emph{orthonormal, modal} basis expansion in the discontinuous Galerkin method to optimize the computational complexity of the algorithm.

Chapter~\ref{ch:Benchmarks} will numerically demonstrate the accuracy and robustness of the implemented scheme. 
We will show via a variety of numerical tests the proven properties of the discrete scheme, and compare a number of numerical experiements to known analytic solutions.
Chapter~\ref{ch:Leverage} will be a tour-de-force showcase of the power of the implemented scheme.
With access to a high fidelity representation of the particle distribution function from our direct discretization, we will examine energization mechanisms in fundamental plasma processes directly in phase space, such as the collisionless shock shown in Figure~\ref{fig:proton-dist-proof-of-concept}, and conclude with an application comparison between the particle-in-cell method and our continuum approach that shows explicitly where particle noise can pollute the simulation of plasma kinetic systems.

The scheme is implemented within the \gke~framework. 
\gke~is a general purpose, open-source, simulation framework with support for five- \citep{Hakim:2006} and ten-moment multi-fluid \citep{Hakim:2008,Wang:2015,Ng:2015,Wang:2019}, full-f gyrokinetic \citep{Shi:2015,Shi2017thesis,Mandell:2020}, and Vlasov--Maxwell--Fokker--Planck systems \citep{Juno:2018,Hakim:2019,HakimJuno:2020}.
For the purposes of reproducibility, the source code for \gke~is available through \texttt{GitHub}\footnote{https://github.com/ammarhakim/gkyl}, and all input files for the simulations run in this thesis are available through a \texttt{GitHub} repository\footnote{https://github.com/ammarhakim/gkyl-paper-inp}, with the changesets used to produce the data documented in the input file. Additional documentation can be found through the \gke~documentation website\footnote{https://gkyl.readthedocs.io/en/latest/}.
%Chapter 2

\renewcommand{\thechapter}{2}
\epigraph{Some of the material in this chapter has been adapted from \citet{Juno:2018}, \citet*{Hakim:2019}, and \citet{HakimJuno:2020}.}{}
\chapter{The Discontinuous Galerkin Finite Element Method}\label{ch:DGFEM}

The method we will employ to discretize the Vlasov--Maxwell--Fokker--Planck system of equations is called the discontinuous Galerkin finite element method, or DG for short.
DG was first introduced to study neutron transport \citep{ReedHill:1973} and became an active area of study in numerical methods after the general formulation of the algorithm by \citet{Cockburn:1998b,Cockburn:2001}.
DG has become an enticing method for a variety of problems, from computational fluid dynamics to seismology and wave equations\citep[see, e.g.,][and references therein]{Hesthaven:2007}, because DG methods are constructed to combine advantages of both finite element methods and finite volume methods.
By combining the power of the finite element method, principally the high order accuracy and flexibility in the chosen basis expansion, with the benefits of a finite volume method, such as locality of data and the ability to construct conservative discretizations, one can design robust, physically-motivated, numerical methods for the chosen equation or equation system of interest.
In fact, DG has become a particularly active area of research in recent years for kinetic equations such as the Vlasov--Maxwell--Fokker--Planck system of equations, and its subsidiaries Vlasov--Poisson and Vlasov--Ampere \citep{Cheng:2011, Cheng:2013, Cheng:2014a, Cheng:2014b}

It is worth taking a moment to give some intuition for the construction of the DG method in a more general context before diving in to our discretization of the VM-FP system of equations.
We will define what we mean by a ``Galerkin'' method, and then apply DG to a simple hyperbolic partial differential equation.
In doing so, we will be able to connect with our knowledge of other numerical methods, and see why DG is often discussed as a hybrid finite volume-finite element method, combining the strengths of both numerical methods into a singular, powerful, means of discretizing a partial differential equation.

\section{$L^2$ Minimization of the Error}\label{sec:L2ErrorMinimization}

The two essential ingredients of a Galerkin method are the definition of some finite dimensional space of functions and a definition of errors.
The former allows us to connect the function space the continuous equation, or equation system, lives in, to a discrete representation of the solution to our equation or equation system.
The latter gives us a unique way of finding the discrete representation, as we would like to minimize the errors of our discrete representation of our solution.

Consider an interval $[-1, 1]$ and the function space of polynomials of order $p$, $\mathbb{P}^p$.
The particular space of polynomials will form a complete basis on our interval\footnote{A good example of such a complete basis would be the Legendre polynomials up to some order $n$, $P_n(x)$.}.
On this interval, we will employ the inner product,
\begin{align}
    \langle f, g \rangle_{L^2} = \int_{-1}^1 f(x) g(x) \thinspace dx, \label{eq:L2innerproduct}
\end{align}
with the following norm,
\begin{align}
    \langle f, f \rangle_{L^2} = \int_{-1}^1 f^2(x) \thinspace dx, \label{eq:L2norm}
\end{align}
the $L^2$ norm.

In general, we want to solve problems of the form
\begin{align}
    \pfrac{f(x,t)}{t} = G[f], \label{eq:generalEquationCont}
\end{align}
where $G[f]$ is some operator for $f$.
$G[f]$ may be a very general operator, such as in the VM-FP system of equations wherein we have first order terms, e.g., the collisionless advection in phase space, and second order terms, e.g., the collision operator.
In seeking an approximation of our solution $f(x,t)$, we will expand $f(x,t)$ in our basis set,
\begin{align}
    f(x,t) \approx f_h(x,t) \defeq \sum_{k=1}^N f_k(t) \phi_k(x), \label{eq:approxSolH}
\end{align}
where $\phi_k(x) \in \mathbb{P}^p$, for $k=1,\dots, N$.
Thus, the problem of interest is approximated as
\begin{align}
    \sum_{k=1}^N \frac{d f_k(t)}{dt} \phi_k(x) = G[f_h],
\end{align}
and we need to determine the time evolution of the coefficients $f_k(t)$.
Note that we have changed notation from $\partial/\partial_t$ to $d/dt$ to emphasize that the coefficients $f_k$ are only a function of time.

We defined a norm in \eqr{\ref{eq:L2norm}}, so let us minimize the error with respect to this norm,
\begin{align}
    E_{L^2} = \int_{-1}^1 \left (\sum_{k=1}^N \frac{d f_k(t)}{dt} \phi_k(x) - G[f_h] \right )^2 \thinspace dx,
\end{align}
by taking the derivative of the error with respect to each time-dependent coefficient,
\begin{align}
    \pfrac{E_{L^2}}{f'_\ell} = 2 \int_{-1}^1 \phi_\ell(x) \left (\sum_{k=1}^N \frac{d f_k(t)}{dt} \phi_k(x) - G[f_h] \right ) \thinspace dx. \label{eq:derivativeError}
\end{align}
Here, we have used the shorthand $f'_\ell = df_\ell/dt$.
To minimize the error with respect to the time derivative of the coefficients, we set \eqr{\ref{eq:derivativeError}} equal to 0,
\begin{align}
    \int_{-1}^1 \sum_k \frac{df_k(t)}{dt} \phi_k(x) \phi_\ell(x) \thinspace dx = \int_{-1}^1 G[f_h] \phi_\ell(x) \thinspace dx. \label{eq:errorMin1D}
\end{align}
To give a bit more insight into how one could then evaluate this expression to find each of the time dependent coefficients, consider what this expression reduces to if the polynomials $\phi_k(x) \in \mathbb{P}^p$ for $k=1,\dots, N$ are an orthonormal basis set such that
\begin{align}
    \int_{-1}^1 \phi_k(x) \phi_\ell(x) \thinspace dx = \delta_{k\ell},
\end{align}
where $\delta_{k\ell} = 1$ if $k=\ell$ and zero otherwise.
Then our equation for the time evolution of the coefficients would reduce to
\begin{align}
    \frac{d f_\ell}{dt} = \int_{-1}^1 G[f_h] \phi_\ell(x) \thinspace dx, \label{eq:orthonormal1D}
\end{align}
for $\ell = 1,\dots, N$, and we would then have a system of ordinary differential equations to solve for each of $df_\ell/dt$.

The discussion up to this point has been somewhat abstract, so we would like to make this concrete in two ways.
First, let us perform the $L^2$ minimization of the error on a non-polynomial function.
In doing so, we would like to show what it means to take a function in some infinite dimensional space, since it would take an infinite number of polynomials to represent this function normally, and project it to a finite dimensional subspace.

We plot in Figure~\ref{fig:L2ErrorMinimization} the projection of the function $f(x) = x^4 + \sin(5x)$ onto a number of different basis expansions.
Here, we have a further generalization of the previous discussion for the Galerkin method, where the domain of $[-1, 1]$ is further subdivided into non-overlapping cells, and the projection is done within each cell.
As we move to higher and higher polynomial order, we can see the reduction, even just visually, of the error between the exact solution and our discrete representation of the solution.
This reduction in the error with higher polynomial order is our first evidence of the connection between the discontinuous Galerkin method and finite element methods, where higher order basis sets correspond to higher accuracy.

\begin{figure}[ht]
    \centering
    \includegraphics[width=0.325\textwidth]{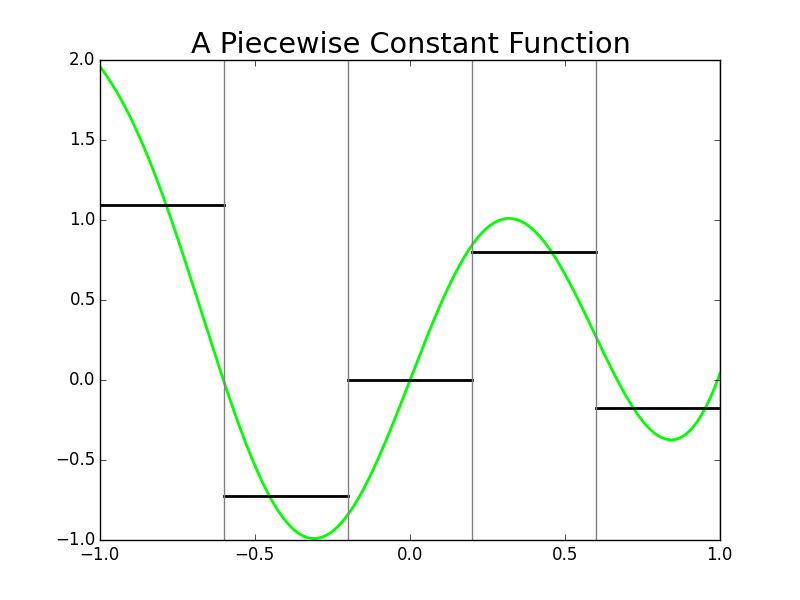}
    \includegraphics[width=0.325\textwidth]{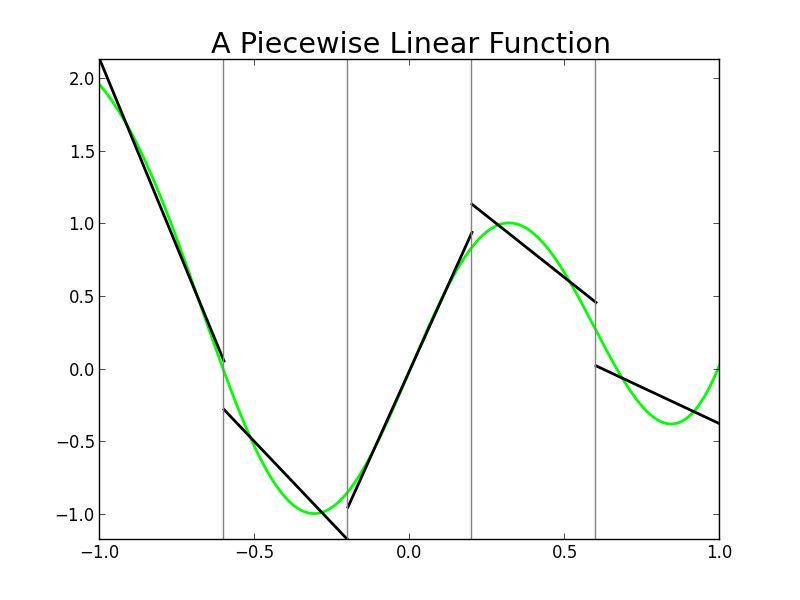}
    \includegraphics[width=0.325\textwidth]{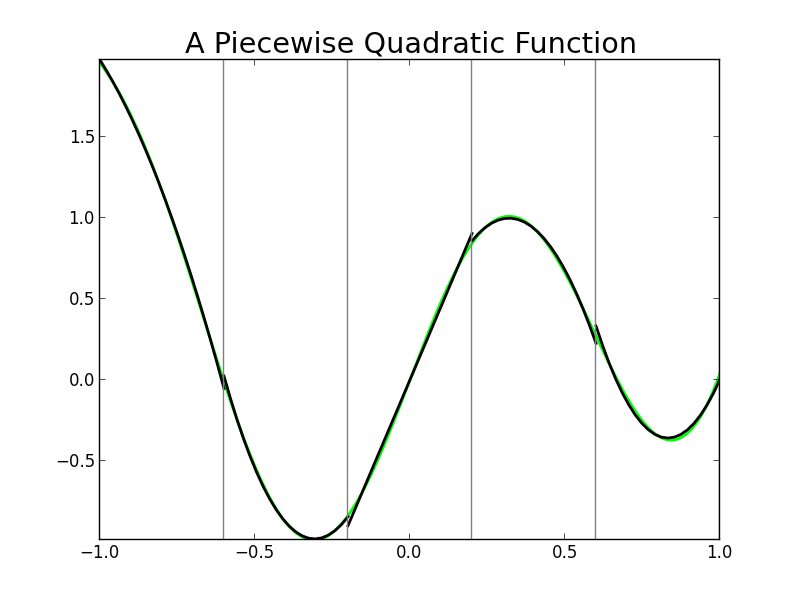}
    \caption{The projection of $f(x) = x^4 + \sin(5x)$ onto piecewise constant (left), piecewise linear (middle), and piecewise quadratice (right) functions. The domain from $[-1, 1]$ is divided into non-overlapping cells and the projection is done within each cell to minimize the $L^2$ error. We begin to see some of the connection between the discontinuous Galerkin method and finite element methods, as moving to higher polynomial order manifestly reduces the $L^2$ error between the exact solution and projected solution.}
    \label{fig:L2ErrorMinimization}
\end{figure}

The second way we will make our discussion of the Galerkin minimization of the $L^2$ error less abstract is by considering the full discretization of the constant advection equation in one dimension,
\begin{align}
    \pfrac{f(x,t)}{t} + \lambda \pfrac{f(x,t)}{x} = 0.
\end{align}
Define the domain of the advection equation as $\Omega$, which we will divide into non-overlapping cells $I_j \in \Omega_j$, for $j = 1,...,N_j$. Plugging $-\lambda \partial f/\partial x$ into \eqr{\ref{eq:errorMin1D}} for the operator $G[f_h]$, and integrating by parts we obtain
\begin{align}
    \int_{I_j} \frac{df_{h,j}}{dt} \phi_\ell \thinspace dx = -\lambda \phi_{\ell,j+1/2} \hat{F}_{j+1/2} + \lambda \phi_{\ell,j-1/2} \hat{F}_{j-1/2} + \lambda \int_{I_j} \frac{d \phi_\ell}{dx} f_{h,j} \thinspace dx, \label{eq:1DAvectionDG}
\end{align}
where the subscripts $j \pm 1/2$ define the right, $+$, and left, $-$, sides of the cell respectively, and $f_{h,j}$ is the projection of the solution in each cell $I_j$ as defined by \eqr{\ref{eq:approxSolH}}. Note that the solution in each cell requires a minimization of the error for every $\phi_\ell, \ell=1,\dots,N$, for however many basis functions in each cell one has, and further that the full solution is a direct sum over all cells $I_j \in \Omega_j$,
\begin{align}
    f_h(x,t) = \bigoplus_{j=1}^{N_j} f_{h,j} (x,t).
\end{align}
Since we have a solution in each cell $I_j$, the integration by parts gives us a means to connect the solution within each cell to its neighbors, but we need to prescribe the numerical flux function, $\hat{F}_{j\pm1/2}$.
A natural choice for the constant advection equation is known as upwind fluxes,
\begin{align}
    \hat{F}(f_h^+, f_h^-) =
    \begin{cases}
    f_h^- \textrm{ if } \lambda>0 \\
    f_h^+ \textrm{ if } \lambda<0, 
    \end{cases} \label{eq:1DUpwindFlux}
\end{align}
where the superscript plus-minus is the solution evaluated just inside, $-$, or just outside $+$, the cell interface---see Figure~\ref{fig:plusminusexp} for a visualization of this notation.

\begin{figure}[ht]
    \centering
    \includegraphics[width=0.9\textwidth]{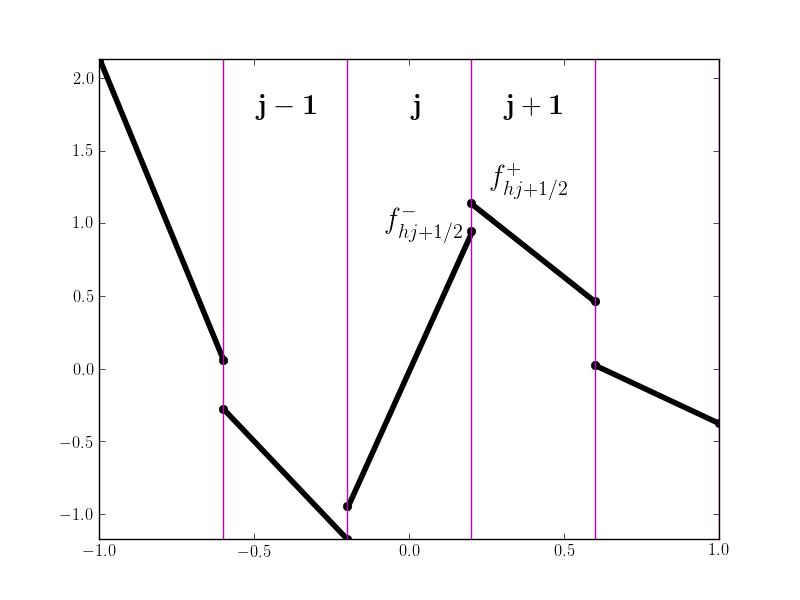}
    \caption{Annotated piecewise linear representation to make our notation more clear, most especially superscript plus-minus, where the solution is evaluated just inside, $-$, or just outside $+$, the cell interface.}
    \label{fig:plusminusexp}
\end{figure}

To make further progress, let us consider two cases. The first case is one in which our basis expansion is just the set of piecewise constant basis functions,
\begin{align}
    \phi = \{ 1 \}.
\end{align}
Substituting the piecewise constant basis function into \eqr{\ref{eq:1DAvectionDG}}, we obtain,
\begin{align}
    \frac{df_j}{dt} \Delta x = -\lambda(f_j - f_{j-1}),
\end{align}
since the derivative of a constant function is 0, and the integral of the left hand side in \eqr{\ref{eq:1DAvectionDG}} when the basis function is a constant is the volume of the cell, $\Delta x$.
We can immediately recognize this formula as a first order finite volume method, or an upwind finite difference method, if you prefer.
We can then discretize the time derivative with a forward Euler method to obtain
\begin{align}
    f^{n+1}_j = f^n_j - \frac{\lambda \Delta t}{\Delta x} (f_j - f_{j-1}),
\end{align}
and should we choose, we could combine multiple forward Euler steps into a multi-stage method, such as a Runge--Kutta method.

The second case is one in which our basis functions are a piecewise linear expansion,
\begin{align}
    \phi_{1,2} = \{ 1, 2 (x - x_j)/\Delta x \}, \label{eq:1DPolynomialDGExample}
\end{align}
where $x_j$ is the cell center value of cell $I_j$.
We can obtain update formulas for a forward Euler step for the constant and linear coefficients when employing the piecewise linear basis,
\begin{align}
    f^{n+1}_{1,j} & = f^n_{1,j} - \frac{\lambda \Delta t}{\Delta x} \left ( \hat{F}_{j+1/2} - \hat{F}_{j-1/2} \right ), \\
    f^{n+1}_{2,j} & = f^n_{2,j} - 3\frac{\lambda \Delta t}{\Delta x} \left ( \hat{F}_{j+1/2} + \hat{F}_{j-1/2} \right ) + 6 \frac{\lambda \Delta t}{\Delta x} f^n_{1,j},
\end{align}
which again, can be combined into a general multi-stage time-stepping method.
Note that the numerical flux function $\hat{F}_{j\pm1/2}$ is still given by \eqr{\ref{eq:1DUpwindFlux}}, but due to the piecewise linear representation within a cell, we will need to evaluate the numerical flux function at the corresponding cell interfaces when implementing the method.

So the switch from piecewise constant basis functions, which produced a standard first order finite volume method, to piecewise linear basis functions, led to more general update formulas.
As we might expect, the accuracy of the method has also improved as a result of switching to a higher order set of basis functions.
To see this, we plot in Figure~\ref{fig:1DAdvectionGauss} the result of advecting a Gaussian pulse on a domain $[0,1]$ with $N_j = 32$ (32 cells) and periodic boundary conditions one full period.
The size of the time-step is chosen to satisfy stability constraints for a forward Euler time-step.
We expect that after one period, the initial condition and the final solution should be identical, since the exact solution of the linear advection equation is simply $f_0(x - \lambda t, t)$, where $f_0$ is the initial condition at $t=0$.
However, the first order finite volume method has significant numerical diffusion, leading to a less accurate representation of the solution than the piecewise linear basis function solution.

\begin{figure}[ht]
    \centering
    \includegraphics[width=0.49\textwidth]{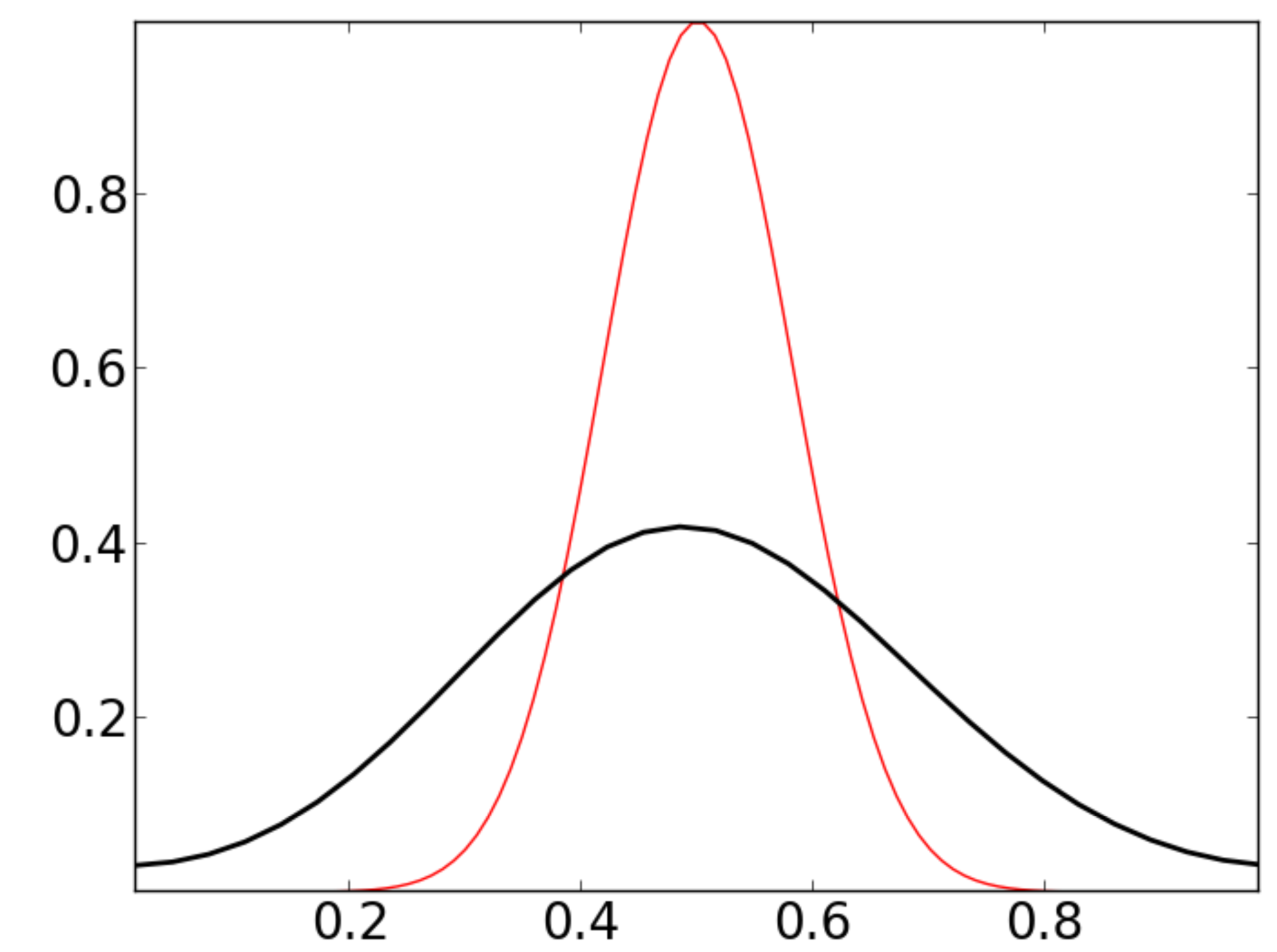}
    \includegraphics[width=0.49\textwidth]{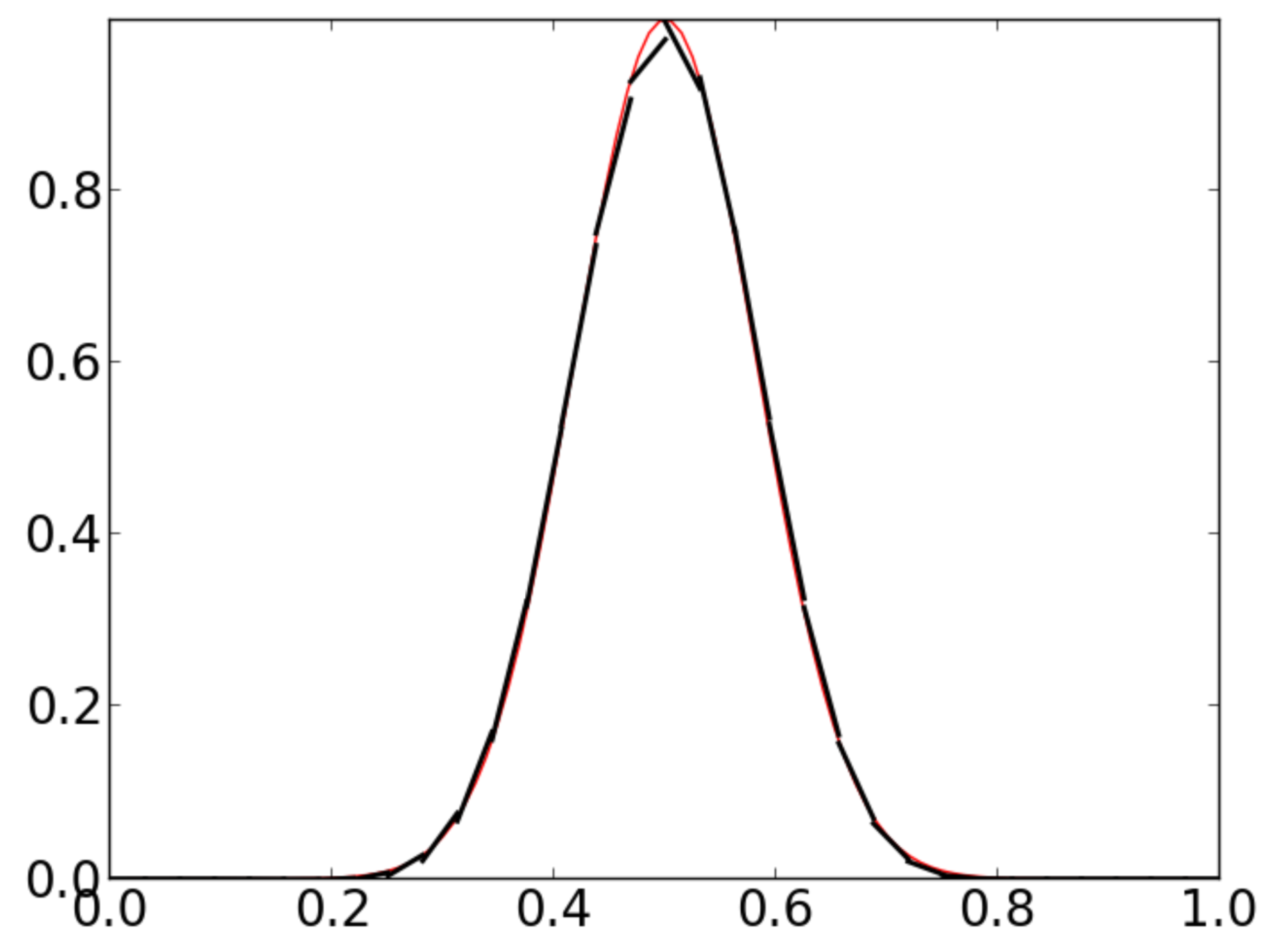}
    \caption{Comparison of advection of a Gaussian pulse one period with a piecewise constant (left) and piecewise linear (right) basis function expansion and upwind fluxes. While the piecewise constant solution suffers from numerical diffusion which leads to poor agreement between the analytic solution (red) and the numerical solution (black), the piecewise linear solution agrees to a reasonably high degree with the expected result.}
    \label{fig:1DAdvectionGauss}
\end{figure}

Based on the results of this numerical experiment, we now want to more strongly connect the discontinuous Galerkin method to finite volume methods.
It is natural to think of DG as a generalization of finite volume methods.
In finite volume methods, one only tracks the evolution of a single quantity in each cell, the cell average, just like with our piecewise constant representation.
But, we now see there is no reason to restrict ourselves.
We can evolve higher ``moments,'' coefficients corresponding to a higher order representation of our solution, within a cell, and in doing so, obtain a higher accuracy numerical method.

A useful analogy is to connect DG with higher order finite volume methods such as MUSCL schemes \citep{VanLeer:1979} or the piecewise parabolic method \citep{Colella:1984}.
In these higher order finite volume methods, one is still only tracking the evolution of the cell average, but a reconstruction of the solution is done at every time-step to increase the order of accuracy of the scheme, e.g., a linear or quadratic reconstruction of the solution.
In the DG method, instead of generating a reconstruction, we are explicitly evolving something like a reconstruction---we are evolving the higher order representation of the solution inside the cell!
With newfound intuition about how the DG method works, let us now turn to the equation system of interest in this thesis, the Vlasov--Maxwell--Fokker--Planck system of equations.
We will proceed in stages just as with the properties of the VM-FP system of equations in Chapter~\ref{ch:Introduction}, first focusing on the collisionless component of the equation system, the Vlasov--Maxwell system of equations.

\section{The Semi-Discrete Vlasov--Maxwell System of Equations}\label{sec:semi-discrete-Vlasov}

We seek a discretization of the Vlasov--Maxwell system of equations using the discontinuous Galerkin method in all of phase space.
To discretize the Vlasov equation, we introduce a phase space mesh $\mathcal{T}$ with cells $K_j \in \mathcal{T}$, $j=1,\ldots,N$, and a piecewise polynomial approximation space for the distribution function, $f_s(\mvec{z}, t)$,
\begin{align}
  \mathcal{V}_h^p = \{ w : w|_{K_j} \in \mathbb{P}^p, \forall K_j \in
  \mathcal{T} \}, \label{eq:phaseSpaceSolutionSpace}
\end{align}
where $\mathbb{P}^p$ is some space of polynomials of order $p$. 
We then seek $f_h\in \mathcal{V}_h^p$ such that, for all $K_j\in \mathcal{T}$,
\begin{align}
  \int_{K_j} w\pfrac{f_h}{t} \dz + 
  \oint_{\partial K_j}w^- \mvec{n}\cdot\hat{\mvec{F}}  \thinspace dS 
  - \int_{K_j} \gz w \cdot \gvec{\alpha}_h f_h \dz = 0, \label{eq:dis-weak-form}
\end{align}
for all test functions $w\in \mathcal{V}_h^p$.
\eqr{\ref{eq:dis-weak-form}} is commonly referred to as the \emph{discrete-weak form} of the Vlasov equation.
In the derivation of the discrete-weak form of the Vlasov equation, we have used integration by parts on the operator for the flux in phase space, thus producing the surface and volume integrals in \eqr{\ref{eq:dis-weak-form}}.

The pieces of the discrete-weak form of the Vlasov equation again evoke the comparison to finite element and finite volume methods.
The third term, the volume integral, calls to mind the integrals over a cell one performs in a finite element method, while the second term, the surface integral, involves the prescription of a numerical flux function, $\mvec{\hat{F}}$, exactly as in a finite volume method.
The subscript $h$ indicates the discrete solution, the notation $w^-$ ($w^+$) indicates that the function is evaluated just inside (outside) the location on the surface $\partial K_j$, and $\mvec{n}$ is an outward unit vector on the surface of the cell $K_j$.

The discrete distribution function is represented as
\begin{align}
f_h(t, \mvec{z}) = \sum_i f_i(t) w_i(\mvec{z}), \label{eq:distf-expansion}
\end{align}
where $w_i(\mvec{z})$ are a set of polynomials chosen such that they lie in the aforementioned space of polynomials $\mathbb{P}^p$, i.e., we are employing a Galerkin method where the test functions and basis functions are one and the same.
We will avoid specifying the exact polynomial space $\mathbb{P}^p$ for now, as the specific form of the polynomials is not a necessary component of the mathematical formulation of the algorithm.
All that we will require in our mathematical formulation is that the basis set is made up of polynomials.

There are many choices for the numerical flux function, $\mvec{\hat{F}}$, which can be employed for the Vlasov equation.
We will pick the numerical flux function most importantly to be a Godunov flux,
\begin{align}
    \oint_{\partial K_j}w^- \mvec{n}\cdot\hat{\mvec{F}}  \thinspace dS  = -\oint_{\partial K_j}w^+ \mvec{n}\cdot\hat{\mvec{F}}  \thinspace dS. \label{eq:GodunovFlux}
\end{align}
In other words, the flux into the cell $K_j$ along some surface $\partial K_j$ is equal and opposite in sign to the flux out of its neighbor cell along the shared interface. This property likely reads like a sensible and obvious property one would desire of a numerical flux function, as it means that the flux is conserved across the interface, i.e., there is no creation or destruction of the distribution function as it advects in phase space.
Example Godunov fluxes include central fluxes,
\begin{align}
    \mvec{n}\cdot\hat{\mvec{F}}(\gvec{\alpha}^+_h f^+_h, \gvec{\alpha}^-_h f^-_h ) = \frac{1}{2} \mvec{n}\cdot \left (\gvec{\alpha}^+_h f^+_h + \gvec{\alpha}^-_h f^-_h \right ),
\end{align}
the local Lax-Friedrichs flux,
\begin{align}
    \mvec{n}\cdot\hat{\mvec{F}}(\gvec{\alpha}^+_h f^+_h, \gvec{\alpha}^-_h f^-_h ) = \frac{1}{2} \mvec{n}\cdot \left (\gvec{\alpha}^+_h f^+_h + \gvec{\alpha}^-_h f^-_h \right ) - \frac{c}{2} (f^+ - f^-), \label{eq:localLF}
\end{align}
where $c = \max_{\partial K_j}(|\mvec{n} \cdot \gvec{\alpha}^+_h|, |\mvec{n} \cdot \gvec{\alpha}^-_h|)$, and the global Lax-Friedrichs flux\footnote{Note that global Lax-Friedrichs flux applies to a general class of numerical flux functions in which the parameter, $\tau$, is a globally calculated quantity.},
\begin{align}
    \mvec{n}\cdot\hat{\mvec{F}}(\gvec{\alpha}^+_h f^+_h, \gvec{\alpha}^-_h f^-_h ) = \frac{1}{2} \mvec{n}\cdot \left (\gvec{\alpha}^+_h f^+_h + \gvec{\alpha}^-_h f^-_h \right ) - \frac{\tau}{2} (f^+ - f^-), \label{eq:globalLF}
\end{align}
where $\tau = \max_{\mathcal{T}} |\mvec{n} \cdot \gvec{\alpha}_h|$.
Note the difference between the local and global Lax-Friedrichs fluxes, where in the local Lax-Friedrichs flux, \eqr{\ref{eq:localLF}}, the max of the phase space flux is taken along the specific surface $\partial K_j$, while for the global Lax-Friedrichs flux, \eqr{\ref{eq:globalLF}}, the max of the phase space flux is taken over the entire domain $\mathcal{T}$.
Both Eqns.\thinspace (\ref{eq:localLF}) and (\ref{eq:globalLF}) are defined with the motivation to penalize the size of the jumps in the flux so that the discontinuities can be controlled in some fashion.
We will see in Proposition~\ref{prop:discrete-L2-norm} that this penalization naturally leads to some numerical diffusion, thus why we refer to the penalty term as controlling the size of the jumps in the flux.

For the DG discretization of Maxwell's equations, we require the \emph{restriction} of the phase-space mesh, $\mathcal{T}$, to configuration space by $\mathcal{T}_\Omega$. 
The cells in configuration space are denoted by $\Omega_j\in\mathcal{T}_\Omega$, for $i=1,\ldots,N_\Omega$, where $N_\Omega$ are the number of configuration space cells, and we introduce the solution space
\begin{align}
  \mathcal{X}^p_{h} = \{ \varphi : \varphi|_{\Omega_j} \in \mvec{P}^p, \forall \Omega_j \in
  \mathcal{T}_{\Omega} \}.
\end{align}
These basis, and test, functions are defined only on the configuration space domain $\Omega$ and thus contain only dependence on the configuration space variable $\mvec{x}$. 
As with the discrete distribution function, we seek, $\mvec{E}_h, \mvec{B}_h \in \mathcal{X}^p_{h}$ such that, for all $\Omega_j \in \mathcal{T}_\Omega$,
\begin{align}
  \int_{\Omega_j}\varphi \pfrac{\mvec{B}_h}{t}\dx
  &+
  \oint_{\partial\Omega_j} d\mvec{s}\times(\varphi^-\hat{\mvec{E}}_h)
  -
  \int_{\Omega_j} \gx\varphi\times\mvec{E}_h \dx
  =
  0, \label{eq:dis-weak-B} \\
  \epsilon_0\mu_0\int_{\Omega_j}\varphi \pfrac{\mvec{E}_h}{t}\dx
  &-
  \oint_{\partial\Omega_j} d\mvec{s}\times(\varphi^-\hat{\mvec{B}}_h)
  +
  \int_{\Omega_j} \gx\varphi\times\mvec{B}_h \dx
  =
  -\mu_0\int_{\Omega_j} \varphi \mvec{J}_h\dx. \label{eq:dis-weak-E}
\end{align}
Note in the derivation of Eqns.\thinspace(\ref{eq:dis-weak-B}--\ref{eq:dis-weak-E}), we needed to evaluate volume integrals which include terms of the form $\varphi\gx\times\mvec{E}_h$, for $\varphi\in\mathcal{X}^p_{h}$ and likewise for the magnetic field, $\mvec{B}_h$. We have made use of the fact that
\begin{align}
  \int_{\Omega_j} 
  \underbrace{
    \varphi\gx\times\mvec{E}_h
  }_{
    \gx\times(\varphi\mvec{E}_h) - \gx\varphi\times\mvec{E}_h
  }
  \dx.
\end{align}
Gauss' law can then be used to convert one volume integral into a surface integral
\begin{align}
  \int_{\Omega_j} \gx\times(\varphi\mvec{E}_h) \dx
  =
  \oint_{\partial\Omega_j} d\mvec{s}\times(\varphi\mvec{E}_h),
\end{align}
where $d\mvec{s}$ is the (vector) area-element that points in the direction of the outward normal to the configuration space cell  $\Omega_j$. 

As with the discrete-weak form for the Vlasov equation, \eqr{\ref{eq:dis-weak-form}}, we require a prescription for the numerical flux functions $\mvec{\hat{E}}_h, \mvec{\hat{B}}_h$.
We consider two methods of obtaining the cell interface fields needed in the discrete weak-form of Maxwell's equations: central fluxes and upwind fluxes. 
As we will see later, both numerical flux functions have advantages and disadvantages, particularly in terms of the conservation properties the discrete system retains from the continuous system.
For central fluxes, we use averages of values just across the interface, i.e.,
\begin{align}
  \hat{\mvec{E}}_h & = \llbracket \mvec{E}\rrbracket, \label{eq:centralE} \\
  \hat{\mvec{B}}_h & = \llbracket\mvec{B}\rrbracket, \label{eq:centralB}
\end{align}
where $\llbracket\cdot\rrbracket$ represents the averaging operator,
\begin{align}
  \llbracket g \rrbracket \equiv (g^+ + g^-)/2,
\end{align}
for any function $g$.

On the other hand, using upwind fluxes requires solving a Riemann problem in a coordinate system local to that face. 
Consider a local coordinate system $(\mvec{s},\gvec{\tau}_1,\gvec{\tau}_2)$ on the configuration space cell face, i.e., on $\partial\Omega_j$. Here, $\mvec{s}$ is a unit vector normal to $\partial\Omega_j$, and $\gvec{\tau}_1$ and $\gvec{\tau}_2$ are tangent vectors such that $\gvec{\tau}_1\times\gvec{\tau}_2=\mvec{s}$. Let $(E_1,E_2,E_3)$ and
$(B_1,B_2,B_3)$ be electric and magnetic fields in this coordinate system. 
Then, assuming variations only along direction $\mvec{s}$, Maxwell's equations reduce to $\partial B_1/\partial t = 0$, $\partial E_1/\partial t = 0$, and the following uncoupled set of two equations for the tangential field components,
\begin{align}
  \pfrac{B_2}{t} - \pfrac{E_3}{x_1} = 0;\quad \pfrac{E_3}{t} - c^2\pfrac{B_2}{x_1} = 0,
\end{align}
and
\begin{align}
  \pfrac{B_3}{t} + \pfrac{E_2}{x_1} = 0; \quad \pfrac{E_2}{t} + c^2\pfrac{B_3}{x_1} = 0.
\end{align}
Multiplying the first of each pair by $c$ and adding and subtracting from the second of that pair we obtain a set of four uncoupled constant advection equations exactly like the constant advection equation considered in Section~\ref{sec:L2ErrorMinimization},
\begin{align}
  \pfraca{t}(E_3+cB_2) - c \pfraca{x_1}(E_3+cB_2) &= 0, \\
  \pfraca{t}(E_3-cB_2) + c \pfraca{x_1}(E_3-cB_2) &= 0,
\end{align}
and
\begin{align}
  \pfraca{t}(E_2+cB_3) + c \pfraca{x_1}(E_2+cB_3) &= 0, \\
  \pfraca{t}(E_2-cB_3) - c \pfraca{x_1}(E_2-cB_3) &= 0.
\end{align}
Hence, the solution to the Riemann problem with initial conditions is
\begin{align}
  (E_2,E_3) = (E_2^-,E_3^-);\quad (B_2,B_3) = (B_2^-,B_3^-),
\end{align}
for $x_1<0$, and
\begin{align}
  (E_2,E_3) = (E_2^+,E_3^+); \quad (B_2,B_3) = (B_2^+,B_3^+),
\end{align}
for $x_1>0$.
At $x_1=0$, the solution is
\begin{align}
  \hat{E}_3 + c\hat{B}_2 &= E_3^+ + c B_2^+, \\
  \hat{E}_3 - c\hat{B}_2 &= E_3^- - c B_2^-,
\end{align}
and
\begin{align}
  \hat{E}_2 + c\hat{B}_3 &= E_2^- + c B_3^-, \\
  \hat{E}_2 - c\hat{B}_3 &= E_2^+ - c B_3^+.
\end{align}
Rearranging these expressions shows that the upwind fields in the local face coordinate system are
\begin{align}
  \hat{E}_2 = \llbracket E_2 \rrbracket - c\thinspace\{ B_3 \} \label{eq:r-e2} \\
  \hat{E}_3 = \llbracket E_3 \rrbracket + c\thinspace\{ B_2 \} \label{eq:r-e3}
\end{align}
and
\begin{align}
  \hat{B}_2 = \llbracket B_2 \rrbracket + \{E_3\}/c \label{eq:r-b2} \\
  \hat{B}_3 = \llbracket B_3 \rrbracket - \{E_2\}/c \label{eq:r-b3}
\end{align}
where $\{ \cdot \}$ is the jump operator,
\begin{align}
  \{ g \} \equiv (g^+-g^-)/2
\end{align}
for any function $g$, and subscripts 2 and 3 denote the two directions tangent to the surface normal. 
Note that we require the two directions tangent to the surface normal since the surface integral involves a cross product for the discrete version of Maxwell's equations, Eqns.\thinspace(\ref{eq:dis-weak-B})-(\ref{eq:dis-weak-E}). 
The solutions to the Riemann problem given by Eqns.\thinspace(\ref{eq:r-e2})-(\ref{eq:r-b3}) are identical to those presented in previous studies of Maxwell's equations \citep{Barbas:2015}.

Eqns.\thinspace(\ref{eq:dis-weak-form}) and (\ref{eq:dis-weak-B})-(\ref{eq:dis-weak-E}) define the semi-discrete Vlasov--Maxwell system of equations, i.e., a discretization in phase and configuration space, with the time discretization not yet specified. 
Before proceeding to the properties of our semi-discrete system, we note that the discretization of Maxwell's equations given by Eqns.\thinspace(\ref{eq:dis-weak-B}) and (\ref{eq:dis-weak-E}) does not include the constraints given by Eqns.\thinspace(\ref{eq:divE}) and (\ref{eq:divB}), i.e., the divergence constraints in Maxwell's equations, $\gx \cdot \mvec{E} = \rho_c/\epsilon_0$ and $\gx \cdot \mvec{B} = 0$. 
Thus, our algorithm may violate these constraints over the course of the simulation. 
Where appropriate in Chapter~\ref{ch:Benchmarks} as part of the benchmarking of the scheme, we will discuss how the violation of the divergence constraints in Maxwell's equations manifests. 

\section{Properties of the Semi-Discrete Vlasov--Maxwell \\System of Equations}\label{sec:PropertiesDiscreteVM}

We proceed as we did with the continuous system, first considering whether the discrete system conserves mass (or number) density, and then moving through the subsequent conservation properties we studied for the continuous system in Section~\ref{sec:PropertiesKineticEquation}.
An important consideration for the discrete scheme, just like with the continuous system, will be our boundary conditions in configuration and velocity space.
While we can employ similar boundary conditions in configuration space for the discrete system as we did with the continuous system, i.e., periodic or some sort of self-contained boundary like a reflecting wall, velocity space is slightly more subtle.
Since the continuous distribution function was defined on $\mvec{v} \in [-\infty, \infty]$, we could use ``half-open'' cells, where a grid cell in velocity space could span $|\mvec{v}| > \mvec{v}_{max}$, where the absolute value encompasses both positive and negative values for the velocity of the particles.
However, we will instead employ a fixed boundary in velocity space, $\mvec{v} \in [\mvec{v}_{min}, \mvec{v}_{max}]$, and at the velocity space boundary employ zero-flux boundary conditions,
\begin{align}
    \mvec{n}\cdot\hat{\mvec{F}}(\mvec{x}, \mvec{v}_{max}) = \mvec{n}\cdot\hat{\mvec{F}}(\mvec{x}, \mvec{v}_{min}) = 0. \label{eq:zeroFluxBC}
\end{align}
Note that \eqr{\ref{eq:zeroFluxBC}} corresponds to a homogeneous Neumann boundary condition in velocity space. This velocity space boundary condition, along with appropriate boundary conditions in configuration space, will allow us to prove the following properties for the discrete scheme.
\begin{proposition}  \label{prop:discrete-particle-cons}
  The discrete scheme conserves mass,
\begin{align}
    \frac{d}{dt} \sum_j \int_{K_j} m_s f_h \dz = 0.
\end{align}
\end{proposition}
\begin{proof}
Choosing $w=m_s$, a constant, in the discrete weak-form, \eqr{\ref{eq:dis-weak-form}}, and summing over all phase-space cells $K_j$,
\begin{align}
    \sum_j \int_{K_j} m_s \pfrac{f_h}{t}\dz + \sum_j \oint_{\partial K_j} m_s \mvec{n}\cdot\hat{\mvec{F}}  \thinspace dS = 0,
\end{align}
where the volume term vanishes since it involves the gradient of a constant function. If the appropriate boundary conditions are chosen, i.e., zero-flux boundary condition in velocity space and periodic boundary conditions in configuration space, or a similar self-contained boundary condition such as a reflecting wall, then the sum over surface integrals is a telescopic sum and vanishes. This pairwise cancellation of the surface integrals requires no special knowledge of the form of the numerical flux function $\mvec{n}\cdot\hat{\mvec{F}} = \mvec{n}\cdot\hat{\mvec{F}}( \gvec{\alpha}_h^- f_h^-, \gvec{\alpha}_h^+ f_h^+)$; we only require that the numerical flux function is Godunov, \eqr{\ref{eq:GodunovFlux}}, and that the flux at both configuration space and velocity space boundaries vanishes as it does with zero flux boundary conditions in velocity space, plus an appropriate boundary condition in configuration space. We are then left with
\begin{align}
    \sum_j \int_{K_j} m_s \pfrac{f_h}{t}\dz = 0,
\end{align}
and it is thus shown that the semi-discrete scheme in the continuous time limit conserves the total (mass) density.
\end{proof}
Before we move on to the $L^2$ norm, we consider the following Lemma on the compressibility of phase space.
\begin{lemma}
  \label{lem:discrete-phase-space-incompress}
  Phase space incompressibility holds for the discrete system, i.e.,
  \begin{align}
  \gz \cdot \gvec{\alpha}_h = 0. \label{eq:discrete-ps-incompress}
  \end{align}
\end{lemma}
\begin{proof}
For the specific discrete phase space flow in the Vlasov-Maxwell system, $\gvec{\alpha}_h = (\mvec{v}, q_s/m_s [\mvec{E}_h + \mvec{v} \times \mvec{B}_h])$. Within a cell, \eqr{\ref{eq:discrete-ps-incompress}} is zero since, as with the continuous system, $\mvec{v}$ has no configuration space dependence, and $q_s/m_s (\mvec{E}_h + \mvec{v} \times \mvec{B}_h)$ has no divergence in velocity space. The question is whether the jumps in $\gvec{\alpha}_h$ across cell interfaces in phase space are accounted for by the scheme. Integrating \eqr{\ref{eq:discrete-ps-incompress}} over a phase space cell $K_j$, employing the divergence theorem, and summing over cells,
  \begin{align}
  \sum_j \oint_{\partial K_j} \mvec{n} \cdot \gvec{\alpha}_h^- \thinspace dS= 0.
  \end{align}
This result follows for the simple reason that the phase space flow is in fact continuous with respect to the surfaces considered, allowing us to pairwise cancel the integrand upon summation. For example, consider the configuration space component of the flow $\gvec{\alpha_h}$, $\mvec{v}$. The velocity, $\mvec{v}$, is continuous across configuration space surfaces because $\mvec{v}$ has no configuration space dependence. Likewise, the velocity space component of $\gvec{\alpha_h}, q_s/m_s (\mvec{E}_h + \mvec{v} \times \mvec{B}_h)$, is continuous across velocity space surfaces because $\mvec{E}_h$ and $\mvec{B}_h$ have no velocity space dependence, and $\mvec{v}$ in the $\mvec{v} \times \mvec{B}_h$ term is the velocity coordinate, and thus is continuous. We note that this proof is specific to the phase space flow for the Vlasov-Maxwell system and in general may not hold for all systems.
\end{proof}
Using Lemma~\ref{lem:discrete-phase-space-incompress}, we can examine the behavior of the $L^2$ norm of the distribution function.
The exact behavior of the $L^2$ norm will depend on the choice of numerical flux function, and importantly, the fact that the phase space flux, $\gvec{\alpha}_h$, is continuous at the corresponding surface interfaces allows us to simplify the numerical flux functions previously defined,
\begin{align}
    \mvec{n}\cdot\hat{\mvec{F}}(\gvec{\alpha}_h f^+_h, \gvec{\alpha}_h f^-_h ) & = \frac{1}{2} \mvec{n}\cdot \gvec{\alpha}_h \left (f^+_h + f^-_h \right ), \label{eq:simpleCentralVlasov}\\
    \mvec{n}\cdot\hat{\mvec{F}}(\gvec{\alpha}_h
    f_h^-, \gvec{\alpha}_h f_h^+)
    & =
    \begin{cases}
        \mvec{n}\cdot\gvec{\alpha}_h f^- \quad \textrm{if} \quad \sign(\gvec{\alpha}_h) > 0, \\
        \mvec{n}\cdot\gvec{\alpha}_h f^+ \quad \textrm{if} \quad \sign(\gvec{\alpha}_h) < 0,
    \end{cases} \label{eq:simpleUpwindVlasov} \\
    \mvec{n}\cdot\hat{\mvec{F}}(\gvec{\alpha}_h
    f_h^-, \gvec{\alpha}_h f_h^+) & = \frac{1}{2} \mvec{n}\cdot \gvec{\alpha}_h \left (f^+_h + f^-_h \right ) - \frac{\tau}{2} (f^+ - f^-), \label{eq:simpleGlobalLFVlasov}
\end{align}
with $\tau = \max_{\mathcal{T}} |\mvec{n} \cdot \gvec{\alpha}_h|$, the global maximum of the phase space flux over the entire domain $\mathcal{T}$ as before.
Importantly, \eqr{\ref{eq:localLF}} has simplified to an upwind flux because $\gvec{\alpha}_h$ is continuous at the corresponding surface interfaces.
An additional consequence of $\gvec{\alpha}_h$ being continuous at the corresponding surface interfaces: Eqns.\thinspace (\ref{eq:simpleUpwindVlasov}) and (\ref{eq:simpleGlobalLFVlasov}) are now solely penalizing the jump in the distribution function, $f_h$, as opposed to the jump in the flux.
Connecting to our earlier discussion in Section~\ref{sec:semi-discrete-Vlasov}, we now examine the $L^2$ norm of the distribution function in our semi-discrete scheme for the Vlasov equation and determine what effect these numerical flux functions have on the time evolution of the $L^2$ norm.
\begin{proposition} \label{prop:discrete-L2-norm}
  The discrete scheme conserves the $L^2$ norm of the distribution function when central fluxes are employed and decays the $L^2$ norm of the distribution function monotonically when using either upwind fluxes or global Lax-Friedrichs fluxes. 
\end{proposition}
\begin{proof}
Since the distribution function itself lies in the test space, we can set $w = f_h$ in \eqr{\ref{eq:dis-weak-form}}. We then have,
  \begin{align}
    \int_{K_j} f_h \pfrac{f_h}{t} \dz + 
    & \oint_{\partial K_j}f_h^- \mvec{n}\cdot\hat{\mvec{F}}  \thinspace dS 
    -\int_{K_j} \gz f_h \cdot \gvec{\alpha}_h f_h \dz = \notag \\
    &\frac{1}{2} \int_{K_j} \pfrac{f_h^2}{t} \dz + 
    \oint_{\partial K_j}f_h^- \mvec{n} \cdot \left(\hat{\mvec{F}} - \gvec{\alpha}_h \frac{f_h^-}{2} \right ) \thinspace dS = 0,
  \end{align}
where we have used Lemma\thinspace\ref{lem:discrete-phase-space-incompress} to rewrite,
\begin{align}
\gz f_h \cdot \gvec{\alpha}_h f_h = \frac{1}{2} \gz \cdot \left (\gvec{\alpha_h} f_h^2 \right ),    
\end{align}
since phase space is incompressible, even in our discrete system, and then used the divergence theorem.
First, consider the case where $\mvec{\hat{F}}$ is given by \eqr{\ref{eq:simpleCentralVlasov}}, central fluxes.
If we sum over all cells, and group cells pairwise by their common interface, we find,
\begin{align}
    \sum_j \oint_{\partial K_j}f_h^- \mvec{n} \cdot & \left(\hat{\mvec{F}} - \gvec{\alpha}_h \frac{f_h^-}{2} \right ) \thinspace dS \notag \\
    & = \sum_j \oint_{\partial K_j} \mvec{n} \cdot \left (f_h^- \left(\hat{\mvec{F}} - \gvec{\alpha}_h \frac{f_h^-}{2} \right )  - f_h^+ \left(\hat{\mvec{F}} - \gvec{\alpha}_h \frac{f_h^+}{2} \right )\right ) \thinspace dS \notag \\
    & = \sum_j \oint_{\partial K_j} \mvec{n} \cdot \gvec{\alpha}_h \left ( f_h^-f_h^+ - f_h^+f_h^- \right ) = 0.
\end{align}
Thus, central fluxes do not change the $L^2$ norm of the distribution function in our semi-discrete scheme.

We can proceed in a similar fashion for upwind fluxes, \eqr{\ref{eq:simpleUpwindVlasov}},
\begin{align}
    \sum_j \oint_{\partial K_j}f_h^- \mvec{n} \cdot & \left(\hat{\mvec{F}} - \gvec{\alpha}_h \frac{f_h^-}{2} \right ) \thinspace dS \notag \\
    & = \sum_j \oint_{\partial K_j} \mvec{n} \cdot \left (f_h^- \left(\hat{\mvec{F}} - \gvec{\alpha}_h \frac{f_h^-}{2} \right )  - f_h^+ \left(\hat{\mvec{F}} - \gvec{\alpha}_h \frac{f_h^+}{2} \right )\right ) \thinspace dS \notag \\
    & = \sum_j \oint_{\partial K_j} \frac{1}{2} |\mvec{n} \cdot \gvec{\alpha}_h| \left ( (f_h^-)^2 - 2 f_h^- f_h^+ + (f_h^+)^2 \right ) \thinspace dS \notag \\
    & = \sum_j \oint_{\partial K_j} \frac{1}{2} |\mvec{n} \cdot \gvec{\alpha}_h| \left (f^-_h - f^+_h \right )^2 \thinspace dS,
\end{align}
where we have used the fact that if $\gvec{\alpha}_h > 0$,
\begin{align}
     \sum_j \oint_{\partial K_j} \mvec{n} \cdot & \left (f_h^- \left(\hat{\mvec{F}} - \gvec{\alpha}_h \frac{f_h^-}{2} \right )  - f_h^+ \left(\hat{\mvec{F}} - \gvec{\alpha}_h \frac{f_h^+}{2} \right )\right ) \thinspace dS \notag \\
     & = \sum_j \oint_{\partial K_j} \frac{1}{2} \mvec{n} \cdot \gvec{\alpha}_h \left (f^-_h - f^+_h \right )^2 \thinspace dS,
\end{align}
and if $\gvec{\alpha}_h < 0$ we have,
\begin{align}
     \sum_j \oint_{\partial K_j} \mvec{n} \cdot & \left (f_h^- \left(\hat{\mvec{F}} - \gvec{\alpha}_h \frac{f_h^-}{2} \right )  - f_h^+ \left(\hat{\mvec{F}} - \gvec{\alpha}_h \frac{f_h^+}{2} \right )\right ) \thinspace dS \notag \\
     & = -\sum_j \oint_{\partial K_j} \frac{1}{2} \mvec{n} \cdot \gvec{\alpha}_h \left (f^-_h - f^+_h \right )^2 \thinspace dS,
\end{align}
so we can simplify the behavior of the $L^2$ norm irrespective of the sign of $\gvec{\alpha}_h$ by absorbing the minus sign into the $\gvec{\alpha}_h < 0$ case.
But, this means that
\begin{align}
    \frac{1}{2} \int_{K_j} \pfrac{f_h^2}{t} \dz = -\sum_j \oint_{\partial K_j} \frac{1}{2} |\mvec{n} \cdot \gvec{\alpha}_h| \left (f^-_h - f^+_h \right )^2 \thinspace dS,
\end{align}
a negative definite quantity.
Thus, the $L^2$ norm is a monotonically decaying quantity when using upwind fluxes.

We can proceed in a similar fashion to the two previous derivations for the global Lax-Friedrichs flux.
Since one component of the global Lax-Friedrichs flux is exactly equivalent to central fluxes, we know that this component of the global Lax-Friedrichs flux will not contribute to the time evolution of the $L^2$ norm.
Following a similar procedure to what we used for upwind fluxes, we find
\begin{align}
    \frac{1}{2} \int_{K_j} \pfrac{f_h^2}{t} \dz = -\sum_j \oint_{\partial K_j} \frac{\tau}{2} \left (f^-_h - f^+_h \right )^2 \thinspace dS,
\end{align}
a negative definite quantity.
So, global Lax-Friedrichs fluxes also monotonically decay the $L^2$ norm, and they further decay the $L^2$ norm more strongly since,
\begin{align}
    \tau = \max_{\mathcal{T}} |\mvec{n} \cdot \gvec{\alpha}_h| \geq |\mvec{n} \cdot \gvec{\alpha}_h|,
\end{align}
at every surface interface $\partial K_j$.
We can then say that the penalization of the size of the jumps in the distribution function, whether by the use of upwind fluxes, \eqr{\ref{eq:simpleUpwindVlasov}}, or by the use of a global Lax-Friedrichs flux, \eqr{\ref{eq:simpleGlobalLFVlasov}}, introduces numerical diffusion into the scheme by decaying the $L^2$ norm of the distribution function.
\end{proof}
\begin{corollary} \label{cor:discrete-entropy}
If the discrete distribution function $f_h$ remains positive definite, then the discrete scheme conserves the entropy if the $L^2$ norm is conserved, and the discrete scheme grows the discrete entropy monotonically if the $L^2$ norm is a monotonically decaying function\footnote{The behavior of the discrete entropy is due to our convention in the definition of the entropy. If one drops the minus sign in the definition of the entropy, then the discrete entropy is a monotonically \textbf{decreasing} function when the $L^2$ norm is a monotonically decreasing function if the discrete distribution function $f_h$ remains positive definite.},
\begin{align}
\frac{d}{dt} \sum_j \int_{K_j} -f_h \ln(f_h) \dz \geq 0
\end{align}
\end{corollary}
\begin{proof}
Using the well known bound,
\begin{align}
\ln(x) \leq x - 1, \label{eq:log-bound}
\end{align}
we can see that $\ln(f_h) \leq f_h - 1$, so long as $f_h$ remains a positive definite quantity, and thus $\ln(f_h)$ is well-defined. Multiplying by $-f_h$ then gives us the inequality,
\begin{align}
-f_h \ln(f_h) \geq -f_h^2 + f_h.
\end{align}
But, the left-hand side is just the discrete entropy. Integrating over a phase space cell $K_j$, summing over cells, and taking the time-derivative of both sides gives us an expression for the time evolution of the discrete entropy in our scheme,
\begin{align}
\frac{d}{dt} \sum_j \int_{K_j} -f_h \ln(f_h) \dz \geq \frac{d}{dt} \sum_j \int_{K_j} -f_h^2 + f_h \dz.
\end{align}
Now, we note that in Proposition \ref{prop:discrete-L2-norm} we have already proved that the $L^2$ norm of the discrete distribution function is either a conserved quantity or a monotonically decaying function, depending on which numerical flux function we employ. 
Thus, the negative of the $L^2$ norm is either exactly conserved or a monotonically increasing function, and by Proposition \ref{prop:discrete-particle-cons}, the semi-discrete scheme conserves particles. 
Therefore, the discrete entropy is either conserved or a monotonically increasing function depending on our choice of numerical flux function.
\end{proof}
It is worth taking a moment to reflect on the practical consequences of Proposition~\ref{prop:discrete-L2-norm} and Corollary~\ref{cor:discrete-entropy}.
These choices of numerical flux functions, Eqns. \thinspace (\ref{eq:simpleCentralVlasov}--\ref{eq:simpleGlobalLFVlasov}), lead to $L^2$ stable schemes, schemes which do not grow the $L^2$ norm.
In addition, if we employ a numerical flux function that leads to the decay of the $L^2$ norm, then this diffusivity in the $L^2$ norm leads naturally to the growth of the discrete entropy.
In other words, numerical diffusion can manifest in our scheme in the form of the growth of the discrete entropy.
Importantly, as of yet, the numerical flux function only affects the discrete entropy.
We will now examine the conservation of energy in our semi-discrete scheme, first in Maxwell's equations, and then for the complete system.
\begin{lemma}
  \label{lem:em-e-cons}
  The semi-discrete scheme for Maxwell's equations conserves electromagnetic energy exactly when using central fluxes and monotonically decays when using upwind fluxes,
  \begin{align}
    \frac{d}{dt} \sum_k \int_{\Omega_k} \left( \frac{\epsilon_0}{2}|\mvec{E}_h|^2 + \frac{1}{2\mu_0}|\mvec{B}_h|^2 \right) \dx
    \le
    -\sum_k \int_{\Omega_k} \mvec{J}_h\cdot\mvec{E}_h \dx. \label{eq:em-e-cons}
  \end{align}
Note that because $\mvec{J}_h \cdot \mvec{E}_h$ can have either sign, by monotonic decay when using upwind fluxes, we mean that when the right hand side is positive, the electromagnetic energy will increase less than $\left |\sum_k \int_{\Omega_k} \mvec{J}_h\cdot\mvec{E}_h \dx \right |$, and when the right hand side is negative the electromagnetic energy will decay more than $- \left |\sum_k \int_{\Omega_k} \mvec{J}_h\cdot\mvec{E}_h \dx \right |$.
\end{lemma}
\begin{proof}
From the discrete weak-form of Maxwell's equations, we need to compute equations for $|\mvec{E}_h|^2$ and $|\mvec{B}_h|^2$. Since each component of the field lies in the selected test space, we take the $i^{th}$-component of \eqr{\ref{eq:dis-weak-B}} and use $B_{hi}$ as a test function, e.g., choose $\varphi=B_{hx}$. Summing these three equations will give us an expression for the time-derivative of $|\mvec{B}_h|^2$. We follow the same procedure for \eqr{\ref{eq:dis-weak-E}}, which gives an expression for the time-derivative of $|\mvec{E}_h|^2$. With a bit of algebra, we obtain
  \begin{align}
    \frac{d}{dt} \int_{\Omega_j} \frac{1}{2} |\mvec{B}_h|^2 \dx
    +
    \oint_{\partial\Omega_j} d\mvec{s}\cdot\hat{\mvec{E}}_h\times \mvec{B}^-_h
    +
    \int_{\Omega_j} \mvec{E}_h\cdot\gx\times\mvec{B}_h \dx 
    = 0, \label{eq:h-bnum}
  \end{align}
  and
  \begin{align}
    \epsilon_0\mu_0\frac{d}{dt} \int_{\Omega_j} \frac{1}{2} |\mvec{E}_h|^2 \dx
    -
    \oint_{\partial\Omega_j} d\mvec{s}\cdot\hat{\mvec{B}}_h\times \mvec{E}^-_h
    -
    \int_{\Omega_j} \mvec{B}_h\cdot\gx\times\mvec{E}_h \dx
    = -\int_{\Omega_j} \mvec{J}_h\cdot\mvec{E}_h \dx. \label{eq:h-enum}
  \end{align}
We now multiply both equations by $1/\mu_0$ and add them. Since
\begin{align}
    \mvec{E}_h\cdot\gx\times\mvec{B}_h-\mvec{B}_h\cdot\gx\times\mvec{E}_h = \gx\cdot(\mvec{B}_h\times\mvec{E}_h),
\end{align}
we can combine the third terms of Eqns.\thinspace(\ref{eq:h-bnum}) and (\ref{eq:h-enum}),
  \begin{align}
    \int_{\Omega_j} \gx\cdot(\mvec{B}_h\times\mvec{E}_h) \dx
    = 
    \oint_{\partial\Omega_j} d\mvec{s}\cdot\mvec{B}^-_h\times \mvec{E}^-_h.
  \end{align}
In the above result, note that upon integration by parts, we must use the field \emph{just inside} the face of cell $\Omega_j$. Hence, the evolution of the electromagnetic energy in a single cell becomes
  \begin{align}
    \frac{d}{dt} \int_{\Omega_j} & \left( \frac{\epsilon_0}{2}|\mvec{E}_h|^2 + \frac{1}{2\mu_0}|\mvec{B}_h|^2 \right) \dx \notag \\
    & +
    \oint_{\partial \Omega_j} d\mvec{s}\cdot \left(\hat{\mvec{E}}_h\times \mvec{B}^-_h + \mvec{E}^-_h\times \hat{\mvec{B}}_h - {\mvec{E}}^-_h\times \mvec{B}^-_h \right)
    =
    -\int_{\Omega_j} \mvec{J}_h\cdot\mvec{E}_h \dx. \label{eq:er-em-gen}
  \end{align}
{\bf Exact Energy Conservation With Central Flux.} Using central-fluxes to determine the interface fields, i.e., setting $\hat{\mvec{E}}_h = \llbracket \mvec{E}\rrbracket$ and $\hat{\mvec{B}}_h = \llbracket\mvec{B}\rrbracket$, gives us,
  \begin{align}
    \frac{d}{dt} \int_{\Omega_j} & \left( \frac{\epsilon_0}{2}|\mvec{E}_h|^2 + \frac{1}{2\mu_0}|\mvec{B}_h|^2 \right) \dx \notag \\
    & +
    \frac{1}{2} \oint_{\partial \Omega_j} d\mvec{s}\cdot \left( \mvec{E}^+_h\times \mvec{B}^-_h + \mvec{E}^-_h\times \mvec{B}^+_h \right)
    =
     -\int_{\Omega_j} \mvec{J}_h\cdot\mvec{E}_h \dx,
  \end{align}
where the $\mvec{E}^-_h \times \mvec{B}^-_h$ terms cancel upon substitution of central fluxes for the interface fields. Summing over all configuration space cells and assuming appropriate boundary conditions in configuration space, we see that the surface term vanishes because it is symmetric and has opposite signs for the two cells sharing an interface.
This cancellation of the surface term leads to the desired discrete electromagnetic energy conservation equation,
  \begin{align}
    \frac{d}{dt} \sum_k \int_{\Omega_k} \left( \frac{\epsilon_0}{2}|\mvec{E}_h|^2 + \frac{1}{2\mu_0}|\mvec{B}_h|^2 \right) \dx
    = -\sum_k \int_{\Omega_k} \mvec{J}_h\cdot\mvec{E}_h \dx.
  \end{align}
{\bf Monotonic Decay With Upwind Flux.} To see what happens when using upwind fluxes, we transform the fields appearing in surface integral into the $(\mvec{s},\gvec{\tau}_1,\gvec{\tau}_2)$ coordinate system. We can then write the third term in \eqr{\ref{eq:er-em-gen}} as,
\begin{align}
d\mvec{s}\cdot & \left(\hat{\mvec{E}}_h\times \mvec{B}^-_h + \mvec{E}^-_h\times \hat{\mvec{B}}_h - {\mvec{E}}^-_h\times \mvec{B}^-_h \right) \notag \\
& = ds \left [ (\hat{E}_2B_3^- - \hat{E}_3B_2^-) + (E^-_2 \hat{B}_3 - E_3^- \hat{B}_2) - (E_2^- B_3^- - E_3^- B_2^-) \right ].    
\end{align}
Using Eqns.\thinspace(\ref{eq:r-e2})-(\ref{eq:r-b3}) for the interface fields, assuming appropriate boundary conditions, and summing over all configuration space cells, we then obtain
  \begin{align}
    \frac{d}{dt} \sum_k \int_{\Omega_k} &\left( \frac{\epsilon_0}{2}|\mvec{E}_h|^2 + \frac{1}{2\mu_0}|\mvec{B}_h|^2 \right) \dx
    = -\sum_k \int_{\Omega_k} \mvec{J}_h\cdot\mvec{E}_h \dx \notag \\
    &+ \sum_j \oint_{\partial \Omega_j} ds \left( \{E_2\}E_2^-/c+ \{E_3\}E_3^-/c + c \{B_2\}B_2^-+ c \{B_3\}B_3^- \right).
  \end{align}
Note that due to the symmetry of the terms, the central flux terms in Eqns.\thinspace(\ref{eq:r-e2})-(\ref{eq:r-b3}) have vanished on summing over all cells. Now consider the contribution of the term $\{E_2\}E_2^-$ to the two cells adjoining some face. This term will be $(E_2^+-E_2^-)E_2^-/2$ and $(E_2^--E_2^+)E_2^+/2$. On summing over the two cells, this contribution will become $-(E_2^+-E_2^-)^2/2$. Similar results are achieved for the other electric and magnetic field coordinates. Hence, the surface terms, on summation, contribute non-positive quantities to the right-hand side, implying that
  \begin{align}
    \frac{d}{dt} \sum_k \int_{\Omega_k} &\left( \frac{\epsilon_0}{2}|\mvec{E}_h|^2 + \frac{1}{2\mu_0}|\mvec{B}_h|^2 \right) \dx < -\sum_k \int_{\Omega_k} \mvec{J}_h\cdot\mvec{E}_h \dx.
  \end{align}
Note that because the resulting surface terms contribute non-positive quantities, we can say that, despite the sign of $\mvec{J}_h\cdot\mvec{E}_h$ being undetermined, the electromagnetic energy still monotonically decays, i.e., when the right hand side is positive, the electromagnetic energy will increase less than $\left |\sum_k \int_{\Omega_k} \mvec{J}_h\cdot\mvec{E}_h \dx \right |$, and when the right hand side is negative, the electromagnetic energy will decay more than $- \left |\sum_k \int_{\Omega_k} \mvec{J}_h\cdot\mvec{E}_h \dx \right |$.
\end{proof}
\begin{lemma}
  \label{lem:kin-e-cons}
If $|\mvec{v}|^2$ belongs to the approximation space $\mathcal{V}_h^p$, then the semi-discrete scheme satisfies
  \begin{align}
    \frac{d}{dt} \sum_j \sum_s \int_{K_j} \frac{1}{2} m|\mvec{v}|^2 f_h \dz
    -
    \sum_k \int_{\Omega_k} \mvec{J}_h\cdot\mvec{E}_h \dx
    =
    0. \label{eq:lemma-2}
  \end{align}
Note that the species index is implied, the sum over $j$ in the first term is over all phase space cells, and the sum over $k$ in the second term is over all configuration space cells.
\end{lemma}
\begin{proof}
If $|\mvec{v}|^2\in\mathcal{V}_h^p$, we can set $w=m|\mvec{v}|^2/2$ in \eqr{\ref{eq:dis-weak-form}} and obtain
  \begin{align}
    \int_{K_j} \frac{1}{2}m|\mvec{v}|^2 \pfrac{f_h}{t} \dz 
    + 
    \oint_{\partial K_j} \frac{1}{2}m|\mvec{v}|^2 \mvec{n}\cdot\hat{\mvec{F}}\thinspace dS 
    - 
    \int_{K_j} \underbrace{\gz \left( \frac{1}{2}m|\mvec{v}|^{2} \right) \cdot \gvec{\alpha}_h}_{q\mvec{v}\cdot \mvec{E}_h} f_h \dz 
    =
    0.  
  \end{align}
Since $|\mvec{v}|^2$ is continuous at cell interfaces, there is no distinction between the basis function $w$ evaluated just inside and outside the cell surface interface. Upon summing over all cells and the number of species, and the use of appropriate boundary conditions in velocity space and configuration space as in Proposition~\ref{prop:discrete-particle-cons}, we are again able to exploit the fact that the numerical flux function is Godunov and cancel the telescopic sum to obtain,
  \begin{align}
    \frac{d}{dt} \sum_j \sum_s \int_{K_j} \frac{1}{2} m|\mvec{v}|^2 f_h \dz
    -
    \sum_k \int_{\Omega_k} \mvec{J}_h\cdot\mvec{E}_h \dx
    =
    0.
  \end{align}
Note that we have performed the integration in velocity space and substituted in the current density, leaving an integration and sum over only configuration space.
This operation is somewhat subtle, and we will discuss this operation and operations similar in the next section, Section~\ref{sec:WeakEquality}.
\end{proof}
\begin{corollary}
  \label{cor:kin-e-cons-p-1}
Even if only using piecewise linear polynomials and $|\mvec{v}|^2$ does \textbf{not} belong to the approximation space $\mathcal{V}_h^p$, then the semi-discrete scheme satisfies
  \begin{align}
    \frac{d}{dt} \sum_j \sum_s \int_{K_j} \frac{1}{2} m\overline{|\mvec{v}|^2} f_h \dz
    -
    \sum_k \int_{\Omega_k} \overline{\mvec{J}_h}\cdot\mvec{E}_h \dx
    =
    0.
  \end{align}
We again note that the species index is implied, the sum over $j$ in the first term is over all phase space cells, and the sum over $k$ in the second term is over all configuration space cells. In this case, $\overline{g}$ refers to the projection of the prescribed function onto a lower order basis set.
\end{corollary}
\begin{proof}
We define the projection of $|\mvec{v}|^2$ onto piecewise linear basis functions as $\overline{|\mvec{v}|^2}$. Substituting this in for our test function, $w$, in \eqr{\ref{eq:dis-weak-form}} we obtain,
  \begin{align}
    \int_{K_j} \frac{1}{2}m\overline{|\mvec{v}|^2} \pfrac{f_h}{t} \dz 
    + 
    \oint_{\partial K_j} \frac{1}{2}m \overline{|\mvec{v}|^2} \mvec{n}\cdot\hat{\mvec{F}}\thinspace dS 
    - 
    \int_{K_j} \underbrace{\gz \left( \frac{1}{2}m \overline{|\mvec{v}|^2} \right) \cdot \gvec{\alpha}_h}_{q\overline{\mvec{v}}\cdot \mvec{E}_h} f_h \dz 
    =
    0,
  \end{align}
where $\overline{\mvec{v}}$ is the derivative of the piecewise linear representation of $1/2 \thinspace |\mvec{v}|^2$ and is a piecewise constant in each cell\footnote{We can show that $\overline{\mvec{v}}$ is the cell center velocity,
\begin{align}
\gv\left ( \frac{1}{2} \overline{|\mvec{v}|^2} \right ) = \frac{1}{2} (\mvec{v}_{left} + \mvec{v}_{right}) = \mvec{v}_{center},    
\end{align}
since $\overline{|\mvec{v}|^2}$ is continuous. Here $\mvec{v}_{left/right}$ is the value of the velocity on the left (right) edge of the cell, so the average value of the two quantities is the cell center velocity.
}. We note that because $\overline{|\mvec{v}|^2}$ is \textbf{also} continuous at cell interfaces, we can again exploit the fact that the numerical flux function is Godunov, and upon summing over all cells and species, and employing appropriate boundary conditions in velocity and configuration space as in Proposition~\ref{prop:discrete-particle-cons}, cancel the surface integral since it is a telescopic sum. We are then left with, 
  \begin{align}
    \frac{d}{dt} \sum_j \sum_s \int_{K_j} \frac{1}{2} m\overline{|\mvec{v}|^2} f_h \dz
    -
    \int_{K_j} q\overline{\mvec{v}}\cdot \mvec{E}_h f_h \dz
    =
    0.
  \end{align}
Upon substitution of the projected current, $\overline{\mvec{J}_h}$, after performing the velocity integration first, we obtain the desired analogous expression to Lemma~\ref{lem:kin-e-cons} for piecewise linear polynomials. 
\end{proof}
\begin{proposition} \label{prop:DiscreteVlasovEnergyProof}
If central-fluxes are used for Maxwell's equations, and if $|\mvec{v}|^2\in \mathcal{V}_h^p$, the semi-discrete scheme conserves total (particles plus field) energy exactly,
\begin{align}
    \frac{d}{dt} \sum_{j} \sum_s \int_{K_j}\frac{1}{2} m|\mvec{v}|^2 f_h \dz
    +
    \frac{d}{dt} \sum_k \int_{\Omega_k} \left( \frac{\epsilon_0}{2}|\mvec{E}_h|^2 + \frac{1}{2\mu_0}|\mvec{B}_h|^2 \right) \dx
    =
    0.
\end{align}
If upwind fluxes are used for Maxwell's equations, and if $|\mvec{v}|^2\in \mathcal{V}_h^p$, the semi-discrete scheme decays the total (particles plus field) energy,
\begin{align}
    \frac{d}{dt} \sum_{j} \sum_s \int_{K_j}\frac{1}{2} m|\mvec{v}|^2 f_h \dz
    +
    \frac{d}{dt} \sum_k \int_{\Omega_k} \left( \frac{\epsilon_0}{2}|\mvec{E}_h|^2 + \frac{1}{2\mu_0}|\mvec{B}_h|^2 \right) \dx
    <
    0.    
\end{align}
And if only piecewise linear polynomials are used and thus $|\mvec{v}|^2\notin \mathcal{V}_h^p$, then the projected energy will either be conserved or decaying depending on the choice of fluxes for Maxwell's equations,
\begin{align}
    \frac{d}{dt} \sum_{j} \sum_s \int_{K_j}\frac{1}{2} m\overline{|\mvec{v}|^2} f_h \dz
    +
    \frac{d}{dt} \sum_k \int_{\Omega_k} \left( \frac{\epsilon_0}{2}|\mvec{E}_h|^2 + \frac{1}{2\mu_0}|\mvec{B}_h|^2 \right) \dx
    \leq
    0,
\end{align}
so long as the scheme is consistent and the appropriate current $\overline{\mvec{J}_h}$ is incremented on to the electric field in Maxwell's equations.
\end{proposition}
\begin{proof}
The proof of this proposition follows from the substitution of the results of Lemma~\thinspace\ref{lem:em-e-cons} into the results of Lemma~\thinspace\ref{lem:kin-e-cons}, or Corollary~\ref{cor:kin-e-cons-p-1} if $|\mvec{v}|^2$ is not in the solution space.
\end{proof}
We wish to make a few remarks about the results of this section.
Firstly, we emphasize that energy conservation for the Vlasov equation was agnostic on the specific form of the numerical flux function, central, upwind, or global Lax-Friedrichs, so long as the numerical flux is Godunov.
Secondly, we want to point out a subtlety in comparison between the continuous proof of energy conservation, Proposition~\ref{prop:collisionlessEnergyConservation}, and the proof of energy conservation for our semi-discrete system, Proposition~\ref{prop:DiscreteVlasovEnergyProof}.
The continuous proof involves the manipulation of terms which are higher order than $|\mvec{v}|^2$, but we note that the higher order terms in the continuous proof come from the substitution of the explicit expressions for $\gvec{\alpha}$, the phase space flow, whereas in the discrete proof presented here, we have left the discrete phase space flow $\gvec{\alpha}_h$ as is to stress the fact that $\gvec{\alpha}_h$ has its own basis function expansion.
Thus, the higher order terms which are explicit in the continuous energy conservation proof are implicit here in the discrete energy conservation proof. 
If $|\mvec{v}|^2\in \mathcal{V}_h^p$, then $\mvec{v}\in \mathcal{V}_h^p$ as well, and terms in the discrete phase space flow such as $\mvec{v}$, the configuration space component of the phase space flow, can be exactly represented in terms of our basis function expansion.

Finally, we note that, although the total energy decays when using upwind fluxes for Maxwell's equations, this decay is small due to the high order nature of the scheme.
We will demonstrate this explicitly in Chapter~\ref{ch:Benchmarks} as part of the benchmarking of the algorithms.
Other authors have also demonstrated that this loss of energy is small for higher order schemes such as the DG method employed here~\citep{Balsara:2017}.

Before we conclude this section on the properties of the semi-discrete Vlasov--Maxwell system of equations, we would be remiss not to discuss the evolution of the total momentum.
The total momentum, particles plus fields, is conserved in the continuous system of equations, but what about our semi-discrete system?
Our formulation of the DG method for the Vlasov--Maxwell system of equations does \textbf{not} conserve momentum.

We can show momentum non-conservation by choosing $w = m_s\mvec{v}$ and proceeding as we did with the continuous system,
  \begin{align}
    \int_{K_j} m\mvec{v} \pfrac{f_h}{t} \dz 
    + 
    \oint_{\partial K_j} m\mvec{v} \mvec{n}\cdot\hat{\mvec{F}}\thinspace dS 
    - 
    \int_{K_j} \gz \left( m\mvec{v} \right) \cdot \gvec{\alpha}_h f_h \dz 
    =
    0.    
  \end{align}
Since $\mvec{v}$ is continuous, upon summation over all phase space cells and species, we obtain
  \begin{align}
    \sum_j \sum_s \int_{K_j} m_s \mvec{v} \pfrac{f_h}{t} \dz 
    - 
    \sum_k \int_{\Omega_k} \rho_{c_h} \mvec{E}_h + \mvec{J}_h \times \mvec{B}_h  \dx
    =
    0. 
  \end{align}
We can proceed exactly as we did with the continuous proof, but we note a key subtlety,
  \begin{align}
    \frac{d}{dt} & \sum_j \sum_s \int_{K_j} m_s \mvec{v} f_h \dz + \frac{d}{dt} \sum_k \int_{\Omega_k} \left (\epsilon_0 \mvec{E}\times\mvec{B} \right) \dx \notag \\
    & + \int_{\Omega_j} \gx\left( \frac{\epsilon_0}{2}|\mvec{E}_h|^2 + \frac{1}{2\mu_0}|\mvec{B}_h|^2 \right) - \gx\cdot\left( \epsilon_0\mvec{E}_h\mvec{E}_h + \frac{1}{\mu_0}\mvec{B}_h\mvec{B}_h \right) \dx
    =
    0. \label{eq:DiscreteMomentumConservation}
  \end{align} 
Since the electric and magnetic fields are discontinuous across configuration space cell interfaces, we cannot use integration by parts to eliminate the latter two terms.
In other words, integration by parts holds only locally and not over the whole domain due to the jumps in the fields across surfaces. 
However, it is important to note from the form of this equation that momentum conservation depends only weakly on velocity space resolution.
Since the size of the discontinuities in the electric and magnetic fields decrease with increasing configuration space resolution, we can more strongly conserve momentum by increasing configuration space resolution.

So, our semi-discrete Vlasov--Maxwell system of equations using the discontinuous Galerkin finite element method conserves mass, and can conserve the energy, $L^2$ norm, and entropy depending on our choice of numerical flux function, while incurring errors in the total momentum due to our discretization of Maxwell's equations.
We still need to discretize the system in time, but we will delay this discussion for a moment as we move to the semi-discrete discretization of the Fokker--Planck collision operator and the discrete Fokker--Planck collision operator properties.
Before we derive the semi-discrete form of the Fokker--Planck operator, it is useful to go into more detail on a concept we have been surreptitiously employing throughout our discussion of the discontinous Galerkin method: the concept of weak equality.
Weak equality underlies all of our discussion up to this point, but we have not made explicit what it means for two functions to be weakly equal, nor how we can use weak equality to actually compute quantities we require in our algorithm, such as velocity moments and the drag and diffusion coefficients in the Fokker--Planck equation.

\section{An Interlude on Weak Equality and Weak Operators} \label{sec:WeakEquality}

Consider some interval $I$ and some function space $\mathcal{P}$ spanned by basis set $\psi_{\ell}$, $\ell=1,\ldots,N$. We will define two functions $f$ and $g$ to be \textit{weakly equal} if
\begin{align}
    \int_I (f-g) \psi_{\ell} \thinspace dx = 0, \quad \forall \ell=1,\ldots,N. \label{eq:weakEqualityDef}
\end{align}
We will denote weakly equal functions by $f \doteq g$. 
Unlikely strong equality, in which functions agree at all points in the interval, weak equality only assures us that the projection of the functions on a chosen basis set is the same. 
However, the functions themselves may be quite different from each other with respect to their behaviour, e.g, each function's positivity or monotonicity in the interval.

The connection between weak equality and the minimization of the error in the $L^2$ norm in Section~\ref{sec:L2ErrorMinimization} is immediately clear.
In constructing a DG discretization of some operator $G[f]$, we are saying,
\begin{align}
    \pfrac{f}{t} \doteq G[f],
\end{align}
and then we construct a projection of the solution $f$ in the space $\mathcal{P}$, which we chose to be the space of piecewise polynomials of order $p,$, i.e., $\mathbb{P}^p$.
This concept of weak equality is another means of deriving Eqns.\thinspace(\ref{eq:dis-weak-form}) and (\ref{eq:dis-weak-B})-(\ref{eq:dis-weak-E}) for the semi-discrete Vlasov--Maxwell system of equations, and why these forms for the Vlasov equation and Maxwell's equation are referred to as the \emph{discrete weak forms} for these equations.

The real power in the concept of weak equality is the ability to connect functions defined in different spaces.
Consider an operation we performed as part of our proof of energy conservation for the Vlasov equation, Lemma~\ref{lem:kin-e-cons},
\begin{align}
    \sum_j \sum_s \int_{K_j} q_s \mvec{v}\cdot \mvec{E}_h f_h \dz \defeq \sum_k \int_{\Omega_k} \mvec{J}_h \cdot\mvec{E}_h \dx. \label{eq:discreteJdotE}
\end{align}
Note that we are using the $\defeq$ symbol here to emphasize that in the process of proving Lemma~\ref{lem:kin-e-cons}, we took \eqr{\ref{eq:discreteJdotE}} as a definition.

While \eqr{\ref{eq:discreteJdotE}} may seem to follow naturally from our definition of the continuous current density in \eqr{\ref{eq:currentDensity}}, the subtlety here is that the distribution function projection is defined over the full phase space, $f_h \in \mathcal{V}_h^p$, while the current density is defined only in the solution space for configuration space, $\mvec{J}_h \in \mathcal{X}_h^p$.
But, here is where we can leverage weak equality,
\begin{align}
    \mvec{J}_h \doteq \sum_j \sum_s \int_{K_j\setminus \Omega_k} q_s \mvec{v} f_h \dv,
\end{align}
i.e., we project the integral over velocity space of the distribution function, weighted by $q_s \mvec{v}$ in this case, onto configuration space basis functions in the space $\mathcal{X}_h^p$.
Note the change of subscript between the phase space cell $K_j$ and configuration space cell $\Omega_k$ since for the purposes of this operation, we need to sum the contributions from all the velocity space cells for a given configuration space cell.
The full computation for this expression would be
\begin{align}
    \sum_m \mvec{J}_m \int_{\Omega{k}} \varphi_m \varphi_\ell \dx = \sum_j \sum_s \int_{\Omega_k} \int_{K_j\setminus \Omega_k} q_s \mvec{v} f_h \varphi_{\ell} \dv \dx, \label{eq:computeCurrentDetailed}
\end{align}
upon plugging in the phase space expansion of the distribution function.
Note that this operation is performed for all $\varphi_\ell \in \mathcal{X}_h^p$.
This procedure gives us a general means of defining the velocity space moments, such as the current density, which couple the particle dynamics and the electromagnetic fields.

So, the actual operation for proving Lemma~\ref{lem:kin-e-cons} is
\begin{align}
    \sum_j \sum_s \int_{K_j} q_s \mvec{v}\cdot \mvec{E}_h f_h \dz \doteq \sum_k \int_{\Omega_k} \mvec{J}_h \cdot\mvec{E}_h \dx. \label{eq:weakJdotE}
\end{align}
Importantly, for the purposes of using the weak operation to compute the current density in \eqr{\ref{eq:weakJdotE}}, we should have technically substituted $w = 1/2 \thinspace m_s |\mvec{v}|^2 \varphi_\ell(\mvec{x})$, where $\varphi_\ell$ are each of our $\ell$ configuration space basis functions, as our test function $w$ when proving Lemma~\ref{lem:kin-e-cons} (and Corollary~\ref{cor:kin-e-cons-p-1}).
In other words, to actually convert the integral over velocity space of $\mvec{v} \cdot \mvec{E}_h f_h$ to the discrete analog of $\mvec{J}_h \cdot \mvec{E}_h$, we must ensure we are projecting the velocity integral of $\mvec{v} \cdot \mvec{E}_h f_h$ onto the full configuration space expansion.

These procedures, such as the operation defined in \eqr{\ref{eq:computeCurrentDetailed}}, are sometimes referred to as weighted $L^2$ projections, or more generally weighted projections, if the norm of choice is not the $L^2$ norm.
A more mathematically complete discussion of these types of projection operators can be found in textbooks on the foundations of finite element methods, such as \citet{BrennerScott:2008}, and these operators are common throughout the literature \citep{Cockburn:2000}.
In fact, there has been growing interest in leveraging weighted $L^2$ projections in novel ways, especially for wave propagation in heterogeneous media, so that the complexities of the media the wave is propagating in are directly encoded within the discretization \citep{Chan:2017, Chan:2019, Guo:2020, Shukla:2020}.

We will use the concept of weak equality in a similar fashion to the construction of these weighted $L^2$ projections to define other types of weak operators in anticipation of the needed machinery to discretize the Fokker--Planck equation in the VM-FP system of equations.
We will define a new set of notation to make the subsequent discussion a bit more clear,
\begin{align}
    M_{0_h} & \doteq \sum_j \int_{K_j\setminus \Omega_k} f_h \dv, \label{eq:discrete0thMoment} \\
    \mvec{M}_{1_h} & \doteq \sum_j \int_{K_j\setminus \Omega_k} \mvec{v} f_h \dv, \label{eq:discrete1stMoment} \\
    M_{2_h} & \doteq \sum_j \int_{K_j\setminus \Omega_k} |\mvec{v}|^2 f_h \dv,  \label{eq:discrete2ndMoment}
\end{align}
which are related to discrete representations of Eqns.\thinspace(\ref{eq:massDefinition}--\ref{eq:energyDefinition}), but without factors of mass and the relevant constants.
For example, we can compute the discrete charge density and discrete current density from Eqns.\thinspace(\ref{eq:discrete0thMoment}--\ref{eq:discrete1stMoment}),
\begin{align}
    \rho_{c_h} & = \sum_s q_s M_{0_{h_s}}, \\
    \mvec{J}_h & = \sum_s q_s \mvec{M}_{1_{h_s}}. \label{eq:strongEqualityCurrent}
\end{align}

For the drag and diffusion coefficients in the Fokker--Planck equation,
% Line-break after Eqns. for formatting. May need fixing depending on comments from committee
Eqns.\\(\ref{eq:DoughertyDrag}) and (\ref{eq:DoughertyDiffusion}), we require the flow and temperature, Eqns.\thinspace(\ref{eq:flowDefinition}) and (\ref{eq:temperatureDefinition}), which involve a number of different operations applied to the velocity moments, such as the division of two velocity moments in \eqr{\ref{eq:flowDefinition}}.
We might naively expect the discrete representation for the flow to be
\begin{align}
    \mvec{u}_h = \frac{\sum_j \int_{K_j\setminus \Omega_k} \mvec{v} f_h \dv}{\sum_j \int_{K_j\setminus \Omega_k} f_h \dv}. \label{eq:badDiscreteFlowDefinition}
\end{align}
But, we only know the projections of the moments, not the actual functions, so simple division like in \eqr{\ref{eq:badDiscreteFlowDefinition}} is ill-defined.
To make this point more concrete, consider what constructing a polynomial expansion of $\mvec{u}_h$ defined in \eqr{\ref{eq:badDiscreteFlowDefinition}} would require: a polynomial expansion of a rational function, since both the numerator and denominator have their own polynomial expansions in \eqr{\ref{eq:badDiscreteFlowDefinition}}.
We cannot project a rational function onto a polynomial as we would incur \emph{aliasing errors} in the construction of the polynomial because a rational function requires an infinite number of polynomials to represent, and we are already limiting ourselves to a finite subspace of polynomials.

By aliasing errors, we mean errors that arise due to being unable to uniquely determine the representation of the quantity of interest.
Because the flow $\mvec{u}_h$ in \eqr{\ref{eq:badDiscreteFlowDefinition}} has no unique representation in a finite subspace of polynomials, the ultimate computation of \eqr{\ref{eq:badDiscreteFlowDefinition}} can lead to uncontrolled and unbounded errors\footnote{A suitable analogy would be the aliasing that arises in the context of Fourier transforms, where an undersampled signal, a signal which would require a higher sampling rate to resolve the Nyquist frequency, will produce an inaccurate Fourier transform due to power in the higher frequencies being ``aliased'' into the signal\citep[see, e.g.,][]{NumRecipes:2007}. This power aliased into the signal is the same manifestation of the unbounded and uncontrolled errors that arise in trying to construct a polynomial representation of a rational function like the rational function in our naive definition of the flow $\mvec{u}_h$ in \eqr{\ref{eq:badDiscreteFlowDefinition}}.}.
We will find later in Chapter~\ref{ch:ImplementationDGFEM} that the elimination of aliasing errors will prove a critical component of constructing stable discretizations of the VM-FP system of equations.

So, in anticipation of this requirement for our algorithm, how do we eliminate aliasing errors in the computation of the flow and temperature required by the Fokker--Planck equations?
To find $\mvec{u}_h$, we need to invert the \emph{weak-operator} equation,
\begin{align}
    M_{0_h} \mvec{u}_h \doteq \mvec{M}_{1_h}. \label{eq:weakFlowDefinition}
\end{align}
Using the definition of weak-equality in \eqr{\ref{eq:weakEqualityDef}} extended to multiple dimensions, this expression means,
\begin{align}
    \int_I (M_{0_h} \mvec{u}_h - \mvec{M}_{1_h}) \psi_{\ell} \dx = 0,
\end{align}
where in general, for our algorithms, the space $I$ is a configuration space cell $\Omega_j$ (or $K_j$) and the basis expansion is $\varphi \in \mathcal{X}_h^p$ (or $w \in \mathcal{V}_h^p$).
This procedure should determine $\mvec{u}_h$, i.e., the projection of the flow in the function space, so we can write $\mvec{u}_h = \sum_m \mvec{u}_m \psi_m$, leading to the linear system of equations
\begin{align}
\sum_m u_m \int_I M_{0_h} \psi_{\ell} \psi_m \dx = \int_I \mvec{M}_{1_h} \psi_{\ell} \dx,
\end{align}
for $\ell=1,\ldots,N$. 
Inverting this linear system determines $\mvec{u}_m$ and hence the projection of $\mvec{u}_h$ in the function space.
We call this process \emph{weak-division}. 
Note that weak-division only determines $\mvec{u}_h$ to an equivalence class as we can replace the specific $\mvec{u}_h$ in the function space with any other function that is weakly equal to it. 
In addition, note that $M_{0_h}$ and $\mvec{M}_{1_h}$ are determined by \eqr{\ref{eq:discrete0thMoment}} and \eqr{\ref{eq:discrete1stMoment}} and thus themselves have expansions that must be included in this computation.
Therefore, the weight in the weighted $L^2$ projection is the basis expansion of the moment $M_{0_h}$.

We can follow a similar procedure for the temperature,
\begin{align}
    M_{0_h} \frac{T_h}{m} \doteq \frac{1}{3} \left (M_{2_h} - \mvec{M_{1_h}} \cdot \mvec{u}_h \right ),
\end{align}
where the factor of $1/3$ comes from the number of velocity dimensions, since the integrals over each velocity direction all contribute to the temperature.
This procedure requires both \emph{weak-division} and what can be referred to as \emph{weak-multiplication}, because we require the expansion of $\mvec{M_{1_h}} \cdot \mvec{u}_h$,
\begin{align}
    \mathcal{K}_h & \doteq \mvec{M_{1_h}} \cdot \mvec{u}_h, \\
    \sum_m \mathcal{K}_m \int_I \psi_{\ell} \psi_m \dx & = \int_I \mvec{M_{1_h}} \cdot \mvec{u}_h \psi_{\ell} \dx,
\end{align}
where both $\mvec{M_{1_h}}$ and $\mvec{u}_h$ themselves have expansions which must be included in the computation.
These sorts of ``polynomial operations,''  where division and multiplication are extended to act on quantities which have expansions in some basis, have been exploited previously in the literature \citep{Atkins:1998,Lockard:1999}.

Having generalized certain operations such as division and multiplication to situations where all the quantities of interest are projections, we can ask the question: do these operations have similar rules to their elementary counterparts?
For example, does weak-division have the equivalent of divide-by-zero issues?
Consider the interval $[-1,1]$ and the orthonormal linear basis set
\begin{align}
    \psi_0 = \frac{1}{\sqrt{2}}; \qquad \psi_1 = \frac{\sqrt{3}}{\sqrt{2}}x. \label{eq:1DOrthonormalBasis}
\end{align}
In one dimension, where $\mvec{M}_{1_h}$ and $\mvec{u}_h$ have just a single component, let $M_1 = 1$ and $M_0 = n_0\psi_0 + n_1 \psi_1$. 
For this simple case, the result of weak-division is 
\begin{align}
    u = \frac{\sqrt{2}}{n_0^2-n_1^2}(n_0 - \sqrt{3}n_1 x).
\end{align}
Hence, the weak-division is not defined for $n_1 = \pm n_0$\footnote{Formally, the mean density $n_0$ is a positive definite quantity, so this constraint should simply be $n_1 = n_0$}.
This calculation shows that, even if the mean density is positive, the slope cannot become too steep---see Figure~\ref{fig:weak-div-p1} where we plot the trend of steepening the slope of density, $M_0$, and the effect of the steepening on the calculation of the flow, $u$. 
When the ``blow-up'' occurs, i.e., $n_1 = \pm n_0$, $M_0$ has a zero-crossing at either $x = \pm 1/\sqrt{3}$.
Although the function for the flow, $u$, appears well behaved through the steepening of the density, this blow-up corresponds to the situation where the density itself becomes unrealizable with a positive definite function.

\begin{figure}[ht]
    \centering
    \includegraphics[width=\textwidth]{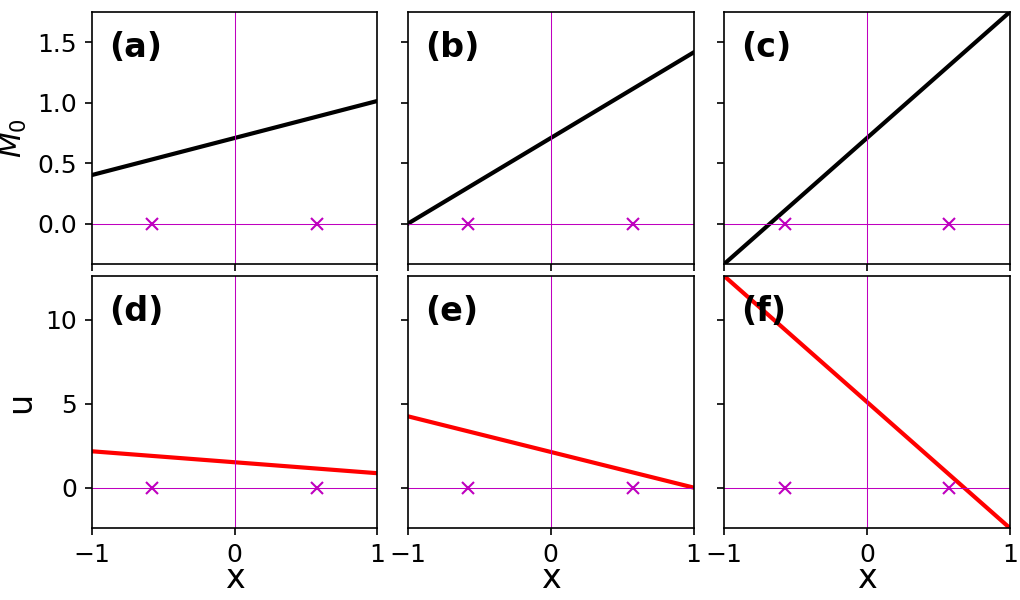}
    \caption{Weak division for the $p=1$ basis, \eqr{\ref{eq:1DOrthonormalBasis}}, to compute $u$ from $M_0 u \doteq M_1$. In this plot, $M_1 = 1$ and the effect of changing $M_0$ (top row) on the flow $u$ (bottom row) is shown. As the density steepens, the velocity becomes larger. If the density becomes too steep, if we increase the slope of the density further so that the density has a zero crossing at $x = \pm 1/\sqrt{3}$ (magenta crosses), the solution for $u$ would blows in the sense that the flow becomes an unrealizable function. Importantly, this blow-up condition corresponds to the situation where the slope of $M_0$ becomes too steep to represent $M_0$ with a positive definite function, which physically corresponds to a situation where the representation of the density is producing negative density functions. Since the value of the particle density can only be positive, this blow-up is highly undesirable.}
    \label{fig:weak-div-p1}
\end{figure}

In this regard, there is nothing necessarily unphysical with a piecewise linear reconstruction $M_0(x)$ having a zero crossing within the domain, and the DG algorithm can result in such solutions.
In principle, there is a physically realizable function that is weakly equivalent $\tilde{M}_0(x) \doteq M_0(x)$ but positive everywhere, as long as $n_1 < \sqrt{3} n_0$.
However, if the slope of $M_0$ becomes too steep, if the density varies too rapidly within a cell, we lose the ability to construct a physically realizable representation for the density, and thus the computation of the flow $u$ would also become physically unrealizable.
In practice we can use constraints like these to limit the slope of the density and thus make the weak division operator always well posed.
This idea of using these constraints to limit the higher order moments in our DG expansion is similar to the philosophy of limiters in high order finite-volume methods.
In smooth regions, we use can the standard calculations and so retain high-order accuracy there, while introducing limiters where the solution is locally varying too quickly to be accurately resolved, in order to robustly preserve certain properties of the solution.

Another application of weak-equality is to recover a continuous function from a discontinuous one.
Say we want to construct a continuous representation $\hat{f}$ on the interval $I=[-1,1]$, from a function, $f$, which has a single discontinuity at $x=0$.
We can choose some function spaces $\mathcal{P}_L$ and $\mathcal{P}_R$ on the interval $I_L = [-1,0]$ and $I_R = [0,1]$ respectively. 
Then, we can reconstruct a continuous function $\hat{f}$ such that
\begin{align}
  \hat{f} &\doteq f_L \quad x \in I_L \quad\mathrm{on}\
      \mathcal{P}_L \label{eq:he1} \\
  \hat{f} &\doteq f_R \quad x \in I_R \quad\mathrm{on}\ \mathcal{P}_R. \label{eq:he2}
\end{align}
where $f = f_L$ for $x\in I_L$ and $f = f_R$ for $x\in I_R$.

As with all our previous discussions about weak equality, this procedure only determines $\hat{f}$ up to its projections in the left and right intervals. 
To determine $\hat{f}$ uniquely, we use the fact that given the $N$ pieces of information in $I_L$ and $N$ pieces of information in $I_R$, where $N$ is the number of basis functions in $\mathcal{P}_{L,R}$, we can construct a polynomial of maximum order $2N-1$. 
We can hence write
\begin{align}
  \hat{f}(x) = \sum_{m=0}^{2N-1} \hat{f}_m x^m.
\end{align}
Using this expression in Eqn.\thinspace (\ref{eq:he1}) and (\ref{eq:he2}) completely determines $\hat{f}$.
In a certain sense, the recovery procedure is a special case of a more general method to go from one basis to another under the restriction of weak equality, just as we constructed the velocity moments, defined only in configuration space, from an operation over the full phase space.

And just as we leveraged weak equality to give us a prescription for the computation of the components of the drag and diffusion coefficients, this procedure to recover a continuous function from discontinuous function foreshadows an additional need we have when discretizing the Fokker--Planck equation: the ability to compute second derivatives.
As an example in one dimension, we wish to compute $g \doteq f_{xx}$ where we know $f$ on a mesh with cells $I_j = [x_{j-1/2}, x_{j+1/2}]$.
Multiply by some test function $\psi \in \mathcal{P}_j$, where $\mathcal{P}_j$ is the function space in cell $I_j$ and integrate to get the following weak-form,
\begin{align}
    \int_{I_j} \psi g \thinspace dx =   \psi \hat{f}_x \bigg|^{x_{j+1/2}}_{x_{j-1/2}}
  -
  \int_{I_j} \psi_{x} f_x \thinspace dx. \label{eq:gfxxp1}
\end{align}
Where we have replaced $f$ by the reconstructed function $\hat{f}$ in the surface term. 
Note that we need \emph{two reconstructions}, one using data in cells $I_{j-1}, I_j$ and the other using data in cells $I_j, I_{j+1}$. 
In the volume term, we continue to use $f$ itself and not the left/right reconstructions as the latter are weakly-equal to the former and can be replaced without changing the volume term.
Once the function space is selected, we have completely determined $g$. 

Notice that one more integration by parts can be performed in \eqr{\ref{eq:gfxxp1}} to obtain another weak-form,
\begin{align}
    \int_{I_j} \psi g \thinspace dx =   (\psi \hat{f}_x - \psi_x \hat{f})\bigg|^{x_{j+1/2}}_{x_{j-1/2}}
  +
  \int_{I_j} \psi_{xx} f \thinspace dx. \label{eq:gfxxp2}
\end{align}
In this form, we need to use both the value and first derivative of the reconstructed functions at the cell interfaces.
Numerically, each of these weak-forms will lead to different update formulas. 
For example, for piecewise linear basis functions, the volume term drops out in \eqr{\ref{eq:gfxxp2}}.
We will find two integration by parts allows us to retain more properties of the continuous Fokker--Planck equation in our semi-discrete formulation of the Fokker--Planck equation in Section~\ref{sec:propertiesDiscreteFP}.

The procedure outlined above is essentially the recovery discontinuous Galerkin (RDG) scheme first proposed in~\citet{VanLeer:2005} and \citet{VanLeer:2007}. 
Extensive study of the properties of the RDG scheme to compute second derivatives is presented in~\citet{Hakim:2014}, where it is shown that the RDG scheme has some advantages compared to the standard local discontinuous Galerkin (LDG) schemes~\citep{Cockburn:1998a,Cockburn:2000} traditionally used to discretize diffusion operators in DG. 
The formulation in terms of weak equality allows systematic extension to higher dimensions just as we developed general formulas for velocity moments irrespective of dimensionality, and we turn now to a semi-discrete formulation of the Fokker--Planck equation given the tools outlined in this section.

\section{The Semi-Discrete Fokker--Planck Equation}

We now want to derive the semi-discrete form of the Fokker--Planck equation using a DG method.
Since the Fokker--Planck equation is solved in tandem with the Vlasov--Maxwell portion of the VM-FP system of equations, we will consider the same phase space mesh, $\mathcal{T}$, with cells $K_j$, and solution space $\mathcal{V}_h^p$ defined in \eqr{\ref{eq:phaseSpaceSolutionSpace}}.
For the Fokker--Planck component, we can integrate by parts once to obtain a discrete weak form, analogous to the collisionless component of the VM-FP system of equations in \eqr{\ref{eq:dis-weak-form}},
\begin{align}
    \int_{K_j} w\pfrac{f_h}{t} \dz = \oint_{\partial K_j}\nu \thinspace w^- \mvec{n}\cdot\hat{\mvec{G}}  \thinspace dS - \int_{K_j} \nu \thinspace \gv w \cdot \left [ (\mvec{v} - \mvec{u}_h)f_h + \frac{T_h}{m} \gv f_h \right ] \dz. \label{eq:schemeNC}
\end{align}
Here, the numerical flux function $\hat{\mvec{G}}$ includes both the drag and diffusion terms,
\begin{align}
    \mvec{n}\cdot\hat{\mvec{G}} = \mvec{n}\cdot \left (\hat{\mvec{F}}_{drag} + \frac{T_h}{m}\gv \hat{f} \right ),
\end{align}
where $\hat{\mvec{F}}_{drag}$ is a numerical flux function for the drag term. 
Our only requirement for the numerical flux function for the drag term will be that, like the collisionless flux in phase space, this numerical flux function for the drag term is a Godunov flux, \eqr{\ref{eq:GodunovFlux}}.
The latter term involves the recovery of the distribution function at a velocity space cell as described in Section~\ref{sec:WeakEquality}.
Note that $T_h$ is unchanged by the recovery process, as the temperature is only a function of configuration space, and thus is continuous across velocity space interfaces.
As with the collisionless phase space flux, example Godunov fluxes for the drag term include,
\begin{align}
    \mvec{n}\cdot\hat{\mvec{F}}_{drag} & = \frac{1}{2} \mvec{n}\cdot (\mvec{v} - \mvec{u}_h) (f^+ + f^-), \label{eq:dragCentral}\\
    \mvec{n}\cdot\hat{\mvec{F}}_{drag} & =     \begin{cases}
        \mvec{n}\cdot(\mvec{v} - \mvec{u}_h) f^- \quad \textrm{if} \quad \sign(\mvec{v} - \mvec{u}_h) > 0, \\
        \mvec{n}\cdot(\mvec{v} - \mvec{u}_h) f^+ \quad \textrm{if} \quad \sign(\mvec{v} - \mvec{u}_h) < 0,
    \end{cases} \\
    \mvec{n}\cdot\hat{\mvec{F}}_{drag} & = \frac{1}{2} \mvec{n}\cdot (\mvec{v} - \mvec{u}_h) (f^+ + f^-) - \frac{\max_{\mathcal{T}}|\mvec{v} - \mvec{u}_h|}{2} (f^+ - f^-), \label{eq:dragGlobalLF}
\end{align}
where we have already exploited the fact that $\mvec{v} - \mvec{u}_h$ is continuous at velocity space interfaces to simplify a central flux, upwind flux, and global Lax-Friedrichs flux to the forms shown in Eqns.\thinspace(\ref{eq:dragCentral}--\ref{eq:dragGlobalLF}).

While \eqr{\ref{eq:schemeNC}} may seem like a perfectly fine DG method for the Fokker--Planck equation, the method as written in \eqr{\ref{eq:schemeNC}} does not retain some of the important properties of the continuous system.
For example, we can show that the method as written in \eqr{\ref{eq:schemeNC}} does not conserve momentum.
To see this lack of conservation, substitute $w = m \mvec{v}$, where we have dropped the species subscript because, as we showed in Section~\ref{sec:PropertiesKineticEquation}, the Fokker--Planck equation conserves the momentum of each species individually.
Upon substitution of $w = m \mvec{v}$ and summing over all cells, we obtain
\begin{align}
    \sum_j \int_{K_j} m \mvec{v} \pfrac{f_h}{t} \dz = 
    -\sum_j \int_{K_j} m \nu \thinspace \gv \mvec{v} \cdot \left [ (\mvec{v} - \mvec{u}_h)f_h + \frac{T_h}{m} \gv f_h \right ] \dz,
\end{align}
where we have already eliminated the surface term due to the assumption of the flux being Godunov and the fact that $\mvec{v}$ is continuous at velocity space interfaces.
Note that implicit in the cancellation of the surface terms is the fact that we are again employing zero-flux boundary conditions in velocity space, similar to \eqr{\ref{eq:zeroFluxBC}},
\begin{align}
        \mvec{n}\cdot\hat{\mvec{G}}(\mvec{x}, \mvec{v}_{max}) = \mvec{n}\cdot\hat{\mvec{G}}(\mvec{x}, \mvec{v}_{min}) = 0.
\end{align}
Additionally, while the numerical flux due to the drag is a Godunov flux, and thus why it can be eliminated upon summation over cells, the reason the $\gv \hat{f}$ term, the gradient of the recovered distribution function, vanishes is for the simple reason that $\gv \hat{f}$ can also be constructed to be continuous at the corresponding interfaces.
When constructing the recovered distribution function, we can make sure that both the value, and the slope, are continuous at the shared interface, a desirable property for the discretization of a diffusion operator!

Substitution of $\gv \mvec{v} = \overleftrightarrow{\mvec{I}}$ allows us to determine under what conditions our discrete scheme conserves momentum.
Firstly, we require that
\begin{align}
    \sum_j \int_{K_j} (\mvec{v} - \mvec{u}_h)f_h = 0, \label{eq:schemeNCFlow}
\end{align}
but this is simply $ M_{0_h} \mvec{u}_h \doteq \mvec{M}_{1_h}$ once the integrals are separated into their configuration space and velocity space components\footnote{Note that \eqr{\ref{eq:weakFlowDefinition}} is a stronger statement than \eqr{\ref{eq:schemeNCFlow}}, because \eqr{\ref{eq:weakFlowDefinition}} is the full projection of the flow onto the configuration space basis expansion. As we mentioned in Section~\ref{sec:WeakEquality} when discussing the equality in \eqr{\ref{eq:weakJdotE}}, for the purposes of discussing conservation relations, we have substituted expressions such as $w = \mvec{v}$ or $w = |\mvec{v}|^2$, but we could have just as easily substituted $w = \mvec{v} \varphi_\ell(\mvec{x})$ or $w = |\mvec{v}|^2 \varphi_\ell(\mvec{x})$, where $\varphi_\ell \in \mathcal{X}_h^p$ are each of the $\ell$ basis functions spanning configuration space. Doing so would not change the algebra and the subsequent proofs and would make the connection between \eqr{\ref{eq:weakFlowDefinition}} and \eqr{\ref{eq:schemeNCFlow}} concrete.}.
So, if we ensure computation of the discrete flow exactly as we described in Section~\ref{sec:WeakEquality}, i.e., that the projection of the discrete flow is consistent and incurs no aliasing errors, this term will vanish.
Unfortunately, the final term,
\begin{align}
    \sum_j \int_{K_j}  \frac{T_h}{m} \gv f_h = \sum_j \oint_{\partial K_j} \frac{T_h}{m} f_h^- dS \neq 0,
\end{align}
since the distribution function is not continuous at cell interfaces.
Thus, in this formulation of the semi-discrete Fokker--Planck equation, we cannot expect to conserve momentum.

A similar argument shows that the semi-discrete Fokker-Planck equation described in \eqr{\ref{eq:schemeNC}} does not conserve energy either.
The lack of conservation of both momentum and energy can be traced to the gradient term, $T_h/m \thinspace \gv f_h$, which regardless of whether one is examining the energy or momentum, will pick up the jumps in the distribution function at cell interfaces.
We are thus motivated to integrate by parts again to obtain a new semi-discrete formulation for the Fokker--Planck equation,
\begin{align}
   \int_{K_j} w\pfrac{f_h}{t} \dz = \oint_{\partial K_j} & \nu \thinspace w^- \mvec{n}\cdot\hat{\mvec{G}}  \thinspace dS - \oint_{\partial K_j} \nu \thinspace \mvec{n}\cdot \gv w^- \frac{T_h}{m} \hat{f} dS \notag \\
   &-\int_{K_j} \nu \left [ \thinspace \gv w \cdot \left ( \mvec{v} - \mvec{u}_h \right )f_h - \gv^2 w \left ( \frac{T_h}{m} f_h \right ) \right ] \dz. \label{eq:dis-weak-FP}    
\end{align}
For this scheme, we require both the value and the slope of the recovered distribution function at the cell interfaces.
\eqr{\ref{eq:dis-weak-FP}} will be the form whose properties we examine in the next section as we determine what we have retained compared to the continuous Fokker--Planck equation.

\section{Properties of the Semi-Discrete Fokker--Planck Equation}\label{sec:propertiesDiscreteFP}

We now proceed as we did in Section~\ref{sec:PropertiesDiscreteVM}, but for the semi-Discrete Fokker--Planck equation, to determine what properties the semi-discrete formulation retains in comparison to the continuous equation.
As with the semi-discrete Vlasov equation, we will assume the boundary conditions in velocity space are zero flux,
\begin{align}
        \mvec{n}\cdot\hat{\mvec{G}}(\mvec{x}, \mvec{v}_{max}) = \mvec{n}\cdot\hat{\mvec{G}}(\mvec{x}, \mvec{v}_{min}) = 0. \notag
\end{align}
In addition, we note that there is an additional boundary term due to our second integration by parts,
\begin{align}
    \oint_{\partial K_j} \nu \thinspace & \mvec{n}\cdot \gv w^- \frac{T_h}{m} \hat{f} dS = \notag \\
    & \int_{\Omega_j} \oint_{\partial V_{max/min}} \nu \thinspace \mvec{n}\cdot \gv w^- \frac{T_h}{m} f_h(\mvec{x}, \mvec{v}_{max/min}, t) \dx \thinspace dS_{V_{max/min}}, \label{eq:additionalFPBoundaryTerm}
\end{align}
where we have separated the surface integral over the edge of velocity space into an integral over configuration space and the specific edge of velocity space surface, and we have substituted for the recovery polynomial at the edge of velocity space the distribution function evaluated at the edge of velocity space.
Since we have no information beyond the edge of velocity space due to the zero flux boundary condition on the numerical flux function $\hat{\mvec{G}}$, choosing the recovery polynomial at the edge of velocity space to be simply the distribution function evaluated at the edge is the most natural choice.
This vector notation may seem somewhat strange, so as a concrete example, this operation for the $v_x$ derivative is
\begin{align}
\oint_{\partial K_j} \nu \thinspace & \hat{\mvec{x}} \cdot \gv w^- \frac{T_h}{m} \hat{f} dS = \notag \\
& \int_{\Omega_j} \oint_{\partial V_{max/min}} \nu \thinspace \nabla_{v_x} w^- \frac{T_h}{m} f_h(\mvec{x}, v_{x_{max/min}}, v_y, v_z, t) \dx \thinspace dv_y dv_z,
\end{align}
i.e., for the edge of $v_x$ in velocity space, we evaluate the distribution function at the maximum or minimum $v_x$ and leave the other dependencies (all of $\mvec{x}$ and $v_y,v_z$) intact to be integrated over.
This particular boundary term will turn out to be important for the conservation properties of the semi-discrete system, in addition to being an explicit boundary term required as part of the complete update formula.
Note that, because the Fokker--Planck equation only involves derivatives in velocity space, our semi-discrete formulation of the Fokker--Planck equation is agnostic to the boundary conditions in configuration space for the following properties.
We will also drop the species subscript from the mass, as we know from Section~\ref{sec:PropertiesKineticEquation} that the continuous Fokker--Planck equation conserves mass, momentum, and energy individually for each species.
%For clarity and flow, we refer the reader to Appendix~\ref{app:multSpeciesModelFP} for the derivation and implementation of the model multi-species Fokker--Planck equation.
\begin{proposition}
The discrete scheme in \eqr{\ref{eq:dis-weak-FP}} conserves mass,
\begin{align}
    \frac{d}{dt} \sum_j \int_{K_j} m f_h \dz = 0.
\end{align}
\end{proposition}
\begin{proof}
Substituting $w = m$ into \eqr{\ref{eq:dis-weak-FP}} and summing over all cells, we obtain
\begin{align}
    \sum_j \int_{K_j} m\pfrac{f_h}{t} \dz = \sum_j \oint_{\partial K_j} \nu m \mvec{n}\cdot\hat{\mvec{G}}  \thinspace dS = 0,
\end{align}
since the gradient of a constant function is zero, and we have chosen the numerical flux function $\hat{\mvec{G}}$ to be a Godunov flux so that the sum over surfaces pairwise cancels the flux.
Combined with a zero flux boundary condition in velocity space, the proof of mass conservation is complete.
Note that because the Fokker--Planck equation only contains derivatives in velocity space, just like the continuous proof in Section~\ref{sec:PropertiesKineticEquation}, we can construct the time evolution of the zeroth moment due to the semi-discrete Fokker--Planck equation,
\begin{align}
    \sum_j \int_{K_j \setminus \Omega_j } m \pfrac{f_h}{t} \dz \doteq \pfrac{\rho_{m_h}}{t} = 0,
\end{align}
where $\rho_{m_h}$ is the projection of the mass density onto configuration space basis functions.
\end{proof}
\begin{proposition}\label{prop:discrete-momentum-cons-FP}
The discrete scheme in \eqr{\ref{eq:dis-weak-FP}} conserves momentum,
\begin{align}
    \frac{d}{dt} \sum_j \int_{K_j} m \mvec{v} f_h \dz = 0, 
\end{align}
if
\begin{align}
    T_h \left (  \sum_j \oint_{\partial V_{max_j}} f_h \thinspace dS_{V_{max}} - \sum_j \oint_{\partial V_{min_j}} f_h \thinspace dS_{V_{min}} \right ) + m \mvec{M}_{1_h} - m M_{0_h} \mvec{u}_h \doteq 0, 
\end{align}
i.e., for each velocity component, for example the $x$ component, we have
\begin{align}
    T_h & \left [  \sum_j \oint_{\partial V_{max_j}} f_h(\mvec{x}, v_{x_{max}}, v_y, v_z) \thinspace dv_y dv_z - \sum_j \oint_{\partial V_{min_j}} f_h(\mvec{x}, v_{x_{min}}, v_y, v_z) \thinspace dv_y dv_z \right ] \notag \\
    & + m M_{1_{x_h}} - m M_{0_h} u_{x_h} \doteq 0, \label{eq:FPMomentumConstraint}
\end{align}
where we have temporarily dropped the time dependence from $f_h$ for ease of notation.
\end{proposition}
\begin{proof}
Substituting $w = m \mvec{v}$ into \eqr{\ref{eq:dis-weak-FP}} and summing over all cells, we obtain
\begin{align}
   \sum_j \int_{K_j} m \mvec{v} \pfrac{f_h}{t} \dz = - \sum_j \oint_{\partial K_j} \nu T_h \hat{f} dS - \sum_j \int_{K_j} \nu m \left ( \mvec{v} - \mvec{u}_h \right )f_h  \dz,    
\end{align}
where we have eliminated the sum over the surface integral involving the numerical flux function $\hat{\mvec{G}}$ since it involves the Godunov flux for the drag and the gradient of the recovered distribution function, both of which cancel upon pairwise summation over the shared surfaces.
Likewise, $\gv^2 \mvec{v} = 0$, so the second volume term vanishes.
These simplifications leave the surface term involving the value of the recovered distribution function, plus the volume term for the drag.
Since the recovered distribution function is continuous at the shared interface, this term also pairwise cancels upon execution of the sum over the surfaces, with the exeception of the contribution at the edge of velocity space.
Thus, to conserve momentum, we require,
\begin{align}
    \int_{\Omega_k} T_h & \left (  \sum_j \oint_{\partial V_{max_j}} f_h \thinspace dS_{V_{max}} - \sum_j \oint_{\partial V_{min_j}} f_h \thinspace dS_{V_{min}} \right )\dx \notag \\
    & + \int_{\Omega_k} \left [\sum_j \int_{K_j \setminus \Omega_k} m \left ( \mvec{v} - \mvec{u}_h \right )f_h \dv \right ] \dx = 0,
\end{align}
where we have used the fact that the Fokker--Planck equation only involves derivatives in velocity space to explicitly separate the configuration space and velocity space integrals, i.e., we have made the conservation of momentum local to a configuration space cell as it must be given the continuous proof in Proposition~\ref{prop:collisionMomentumConservation}.
But we note that this constraint is simply
\begin{align}
    T_h \left (  \sum_j \oint_{\partial V_{max_j}} f_h \thinspace dS_{V_{max}} - \sum_j \oint_{\partial V_{min_j}} f_h \thinspace dS_{V_{min}} \right ) + m \mvec{M}_{1_h} - m M_{0_h} \mvec{u}_h \doteq 0, \notag    
\end{align}
with the caveat that the weak equality will involve a projection over the entire configuration space basis expansion.
To complete the proof, we just redo this calculation with $w = m \mvec{v} \varphi_\ell(\mvec{x})$ for each of our $\ell$ configuration space basis functions, $\varphi_\ell \in \mathcal{X}_h^p$, so that we can substitute
\begin{align}
    \int_{\Omega_k} \left [ \sum_j \int_{K_j \setminus \Omega_k} m \left ( \mvec{v} - \mvec{u}_h \right )f_h \dv \right ] \dx = 0 \notag,    
\end{align}
with
\begin{align}
    m \mvec{M}_{1_h} - m M_{0_h} \mvec{u}_h \doteq 0.    
\end{align}
For clarity, we note that the constraint equation for the flow and temperature required for momentum conservation, \eqr{\ref{eq:FPMomentumConstraint}}, in one spatial dimension and one velocity dimension, 1X1V, is
\begin{align}
    T_h \big[ f_h(v_{max})-f_h(v_{min})\big] + m M_{1_h} - m M_{0_h} u_h \doteq 0.
\end{align}
\end{proof}
\begin{proposition}\label{prop:discrete-energy-cons-FP-p2}
The discrete scheme in \eqr{\ref{eq:dis-weak-FP}} conserves energy,
\begin{align}
    \frac{d}{dt} \sum_j \int_{K_j} \frac{1}{2} m |\mvec{v}|^2 f_h \dz = 0, 
\end{align}
if $|\mvec{v}|^2 \in \mathcal{V}_h^p$, and
\begin{align}
    T_h & \left [  \sum_j \oint_{\partial V_{max_j}} (\mvec{n}\cdot \mvec{v}_{max})f_h \thinspace dS_{V_{max}} - \sum_j \oint_{\partial V_{min_j}} (\mvec{n}\cdot \mvec{v}_{min}) f_h \thinspace dS_{V_{min}} \right ] \notag \\
    & + m M_{2_h} - m \mvec{M}_{1_h} \cdot \mvec{u}_h - 3 M_{0_h} T_h \doteq 0, \label{eq:FPEnergyWeakMomentP2}
\end{align}
where the $\mvec{n}\cdot \mvec{v}_{max/min}$ involves a sum over the contribution from each velocity space surface, i.e.,
\begin{align}
    T_h & \left [  \sum_j \oint_{\partial V_{max_j}} v_{x_{max}} f_h(\mvec{x}, v_{x_{max}}, v_y, v_z) \thinspace dv_y dv_z - \sum_j \oint_{\partial V_{min_j}} v_{x_{min}} f_h(\mvec{x}, v_{x_{min}}, v_y, v_z) \thinspace dv_y dv_z \right ] \notag \\
    + & T_h \left [  \sum_j \oint_{\partial V_{max_j}} v_{y_{max}} f_h(\mvec{x}, v_x, v_{y_{max}}, v_z) \thinspace dv_x dv_z - \sum_j \oint_{\partial V_{min_j}} v_{y_{min}} f_h(\mvec{x}, v_x, v_{y_{min}}, v_z) \thinspace dv_x dv_z \right ] \notag \\
    + & T_h \left [  \sum_j \oint_{\partial V_{max_j}} v_{z_{max}} f_h(\mvec{x}, v_x, v_y, v_{z_{max}}) \thinspace dv_x dv_y - \sum_j \oint_{\partial V_{min_j}} v_{z_{min}} f_h(\mvec{x}, v_x, v_y, v_{z_{min}}) \thinspace dv_x dv_y \right ] \notag \\
    + & m M_{2_h} - m \mvec{M}_{1_h} \cdot \mvec{u}_h - 3 M_{0_h} T_h \doteq 0,
\end{align}
where we have temporarily dropped the time dependence from $f_h$ for ease of notation.
\end{proposition}
\begin{proof}
Since $|\mvec{v}|^2$ is in our approximation space $\mathcal{V}_h^p$, we can substitute $w = 1/2 \thinspace m |\mvec{v}|^2$ into \eqr{\ref{eq:dis-weak-FP}} and sum over all cells to obtain
\begin{align}
   \sum_j \int_{K_j} \frac{1}{2} m |\mvec{v}|^2 \pfrac{f_h}{t} \dz = - \sum_j & \oint_{\partial K_j} \nu T_h (\mvec{n} \cdot \mvec{v}) \hat{f} dS \notag \\
   & - \sum_j \int_{K_j} \nu m  \mvec{v} \cdot \left ( \mvec{v} - \mvec{u}_h \right )f_h  - 3 \nu T_h f_h \dz, \label{eq:FP-energy-step-p2}
\end{align}
where we have again leveraged the fact that the surface integral involving the numerical flux function, $\hat{\mvec{G}}$, is a Godunov flux for the drag, and the gradient of the recovered distribution function is continuous, so that both terms cancel upon pairwise summation over the shared surfaces.
We have also substituted $\gv |\mvec{v}|^2 = 2 \mvec{v}$ and $\gv^2 |\mvec{v}|^2 = 6$.
As with our proof of discrete momentum conservation, Proposition~\ref{prop:discrete-momentum-cons-FP}, the interior summation of the remaining surface terms vanishes since the recovered distribution function is continuous at velocity space surfaces, leaving only the integrals along the surfaces at the edge of velocity space.
To conserve energy, we then must satisfy
\begin{align}
    \int_{\Omega_k} T_h & \left [  \sum_j \oint_{\partial V_{max_j}} (\mvec{n}\cdot \mvec{v}_{max})f_h \thinspace dS_{V_{max}} - \sum_j \oint_{\partial V_{min_j}} (\mvec{n}\cdot \mvec{v}_{min}) f_h \thinspace dS_{V_{min}} \right ] \dx \notag \\
    & + \int_{\Omega_k} \left [\sum_j \int_{K_j \setminus \Omega_k} m \left ( |\mvec{v}|^2 - \mvec{v} \cdot \mvec{u}_h \right )f_h  - 3 T_h f_h \dv \right ] \dx = 0, \label{eq:FPEnergyConservationConstraintP2}
\end{align}
where we have again used the fact that the Fokker--Planck equation only involves derivative in velocity space to explicitly separate the configuration space and velocity space integrals, i.e., we have made the conservation of energy local to a configuration space cell as it must be given the continuous proof in Proposition~\ref{prop:collisionEnergyConservation}.
This constraint is simply
\begin{align}
    T_h & \left [  \sum_j \oint_{\partial V_{max_j}} (\mvec{n}\cdot \mvec{v}_{max})f_h \thinspace dS_{V_{max}} - \sum_j \oint_{\partial V_{min_j}} (\mvec{n}\cdot \mvec{v}_{min}) f_h \thinspace dS_{V_{min}} \right ] \notag \\
    & + m M_{2_h} - m \mvec{M}_{1_h} \cdot \mvec{u}_h - 3 M_{0_h} T_h \doteq 0, \notag 
\end{align}
so long as we recognize that weak equality is a stronger statement than the constraint in \eqr{\ref{eq:FPEnergyConservationConstraintP2}}, and we repeat our calculation with $w = 1/2 \thinspace m |\mvec{v}|^2 \varphi_\ell(\mvec{x})$ for each of our $\ell$ configuration space basis functions, $\varphi_\ell \in \mathcal{X}_h^p$, so that we can substitute
\begin{align}
    \int_{\Omega_k} \left [\sum_j \int_{K_j \setminus \Omega_k} m \left ( |\mvec{v}|^2 - \mvec{v} \cdot \mvec{u}_h \right )f_h  - 3 T_h f_h \dv \right ] \dx = 0,    
\end{align}
with
\begin{align}
    m M_{2_h} - m \mvec{M}_{1_h} \cdot \mvec{u}_h - 3 M_{0_h} T_h \doteq 0.
\end{align}
We note that in one spatial dimension and one velocity dimension, 1X1V, this constraint in \eqr{\ref{eq:FPEnergyWeakMomentP2}} is simply
\begin{align}
    T_h \big[ v_{max} f_h(v_{max})-v_{min} f_h(v_{min})\big] + m M_{2_h} - m M_{1_h} u_h - T_h M_{0_h} \doteq 0, \label{eq:FP-1X1V-p2-constraint}
\end{align}
where the coefficient multiplying $T_h M_{0_h}$ has reduced from three to one because we are now only integrating over one velocity dimension, instead of three velocity dimensions.
\end{proof}
One of the most important consequences of Propositions~\ref{prop:discrete-momentum-cons-FP} and \ref{prop:discrete-energy-cons-FP-p2} is that the constraints, Eqns.\thinspace \ref{eq:FPMomentumConstraint} and \ref{eq:FPEnergyWeakMomentP2}, provide a closed set of equations to determine the components of the drag and diffusion coefficients.
Collecting our constraint equations,
\begin{align}
    T_h & \left (  \sum_j \oint_{\partial V_{max_j}} f_h \thinspace dS_{V_{max}} - \sum_j \oint_{\partial V_{min_j}} f_h \thinspace dS_{V_{min}} \right ) + m \mvec{M}_{1_h} - m M_{0_h} \mvec{u}_h \doteq 0, \notag \\
    T_h & \left [  \sum_j \oint_{\partial V_{max_j}} (\mvec{n}\cdot \mvec{v}_{max})f_h \thinspace dS_{V_{max}} - \sum_j \oint_{\partial V_{min_j}} (\mvec{n}\cdot \mvec{v}_{min}) f_h \thinspace dS_{V_{min}} \right ] \notag \\
    & + m M_{2_h} - m \mvec{M}_{1_h} \cdot \mvec{u}_h - 3 M_{0_h} T_h \doteq 0,  \notag
\end{align}
or in one spatial dimension and one velocity dimension (1X1V),
\begin{align}
    T_h \big[ f_h(v_{max})-f_h(v_{min})\big] + m M_{1_h} - m M_{0_h} u_h & \doteq 0, \notag \\
    T_h \big[ v_{max} f_h(v_{max})-v_{min} f_h(v_{min})\big] + m M_{2_h} - m M_{1_h} u_h - T_h M_{0_h} & \doteq 0, \notag 
\end{align}
we have a system of linear equations which allow us to uniquely determine the temperature, $T_h$, and flow, $\mvec{u}_h$, which can then be substituted into our discrete weak form, \eqr{\ref{eq:dis-weak-FP}}.
These expressions may at first seem surprising, as they are a coupled set of linear equations, involving corrections to the temperature, $T_h$, and flow, $\mvec{u}_h$, based on the value of the distribution function at the boundary of velocity space.
If the distribution function vanishes at the boundary, one can eliminate these boundary conditions and recover what we might naively expect for the constraint equations for the temperature and flow, e.g., \eqr{\ref{eq:weakFlowDefinition}} for the flow.
But critically, because we are using a zero-flux boundary condition in velocity space, the distribution function is not exactly zero at the boundary, and one must account for this correction, however small it may be, to ensure the discrete scheme for the Fokker--Planck equation conserves momentum and energy, both locally within a configuration space cell, and globally.

Additionally, we note that we only discussed the case when $|\mvec{v}|^2 \in \mathcal{V}_h^p$ when examining whether the semi-discrete Fokker--Planck equation conserved energy in Proposition~\ref{prop:collisionEnergyConservation}.
Since we showed in Corollary~\ref{cor:kin-e-cons-p-1} and Proposition~\ref{prop:DiscreteVlasovEnergyProof} that the semi-discrete Vlasov--Maxwell system of equations conserves energy, even if only employing piecewise linear polynomials and thus $|\mvec{v}|^2\notin \mathcal{V}_h^p$, we can examine a similar case, but for the semi-discrete Fokker--Planck equation.
We note that we can at this point connect our discussion about the projection of $\overline{|\mvec{v}|^2}$ onto piecewise linear basis functions using the language of weak equality, i.e.,
\begin{align}
    \overline{|\mvec{v}|^2} \doteq |\mvec{v}|^2. \label{eq:v2-p1}
\end{align}
We emphasize again an important property of this projection: just like $|\mvec{v}|^2, \overline{|\mvec{v}|^2}$ is continuous in velocity space, so that we do not have to worry about discontinuities in the projection of  $|\mvec{v}|^2$ onto piecewise linear basis functions.
\begin{proposition}\label{prop:FP-energy-p1}
The discrete scheme in \eqr{\ref{eq:dis-weak-FP}} conserves energy when using piecewise linear polynomials,
\begin{align}
    \frac{d}{dt} \sum_j \int_{K_j} \frac{1}{2} m \overline{|\mvec{v}|^2} f_h \dz = 0, 
\end{align}
where $\overline{|\mvec{v}|^2}$ is the projection of $|\mvec{v}|^2$ onto piecewise linear basis functions, if
\begin{align}
    T_h & \left [  \sum_j \oint_{\partial V_{max_j}} (\mvec{n}\cdot \overline{\mvec{v}}_{max})f_h \thinspace dS_{V_{max}} - \sum_j \oint_{\partial V_{min_j}} (\mvec{n}\cdot \overline{\mvec{v}}_{min}) f_h \thinspace dS_{V_{min}} \right ] \notag \\
    & + m M^*_{2_h} - m \mvec{M}^*_{1_h} \cdot \mvec{u}_h - 3 M^*_{0_h} T_h \doteq 0. \label{eq:FPEnergyWeakMomentP1}
\end{align}
Here, $\mvec{n}\cdot \overline{\mvec{v}}_{max/min}$ involves a sum over the contribution from each velocity space surface, i.e.,
\begin{align}
    T_h & \left [  \sum_j \oint_{\partial V_{max_j}} \overline{v}_{x_{max}} f_h(\mvec{x}, v_{x_{max}}, v_y, v_z) \thinspace dv_y dv_z - \sum_j \oint_{\partial V_{min_j}} \overline{v}_{x_{min}} f_h(\mvec{x}, v_{x_{min}}, v_y, v_z) \thinspace dv_y dv_z \right ] \notag \\
    + & T_h \left [  \sum_j \oint_{\partial V_{max_j}} \overline{v}_{y_{max}} f_h(\mvec{x}, v_x, v_{y_{max}}, v_z) \thinspace dv_x dv_z - \sum_j \oint_{\partial V_{min_j}} \overline{v}_{y_{min}} f_h(\mvec{x}, v_x, v_{y_{min}}, v_z) \thinspace dv_x dv_z \right ] \notag \\
    + & T_h \left [  \sum_j \oint_{\partial V_{max_j}} \overline{v}_{z_{max}} f_h(\mvec{x}, v_x, v_y, v_{z_{max}}) \thinspace dv_x dv_y - \sum_j \oint_{\partial V_{min_j}} \overline{v}_{z_{min}} f_h(\mvec{x}, v_x, v_y, v_{z_{min}}) \thinspace dv_x dv_y \right ] \notag \\
    + & m M^*_{2_h} - m \mvec{M}^*_{1_h} \cdot \mvec{u}_h - 3 M^*_{0_h} T_h \doteq 0,
\end{align}
and the ``star moments'' are defined as follows,
\begin{align}
    M^*_{0_h} & \doteq \sum_{j\neq j_{max}} \oint_{\partial K_j \setminus \Omega_k} (\mvec{n} \cdot \Delta \mvec{v}) \hat{f} dS, \label{eq:M0Star} \\
    \mvec{M}^*_{1_h} & \doteq \sum_j \int_{K_j \setminus \Omega_k} \overline{\mvec{v}} f_h \dv, \label{eq:M1Star}\\
    M^*_{2_h} & \doteq \sum_j \int_{K_j \setminus \Omega_k} \mvec{v} \cdot \overline{\mvec{v}} f_h \dv, \label{eq:M2Star}
\end{align}
where $\hat{f}$ is the recovery polynomial, $\overline{\mvec{v}}$ is $1/2 \thinspace \gv \overline{|\mvec{v}|^2}$ and equal to the cell center velocity, as previously shown in Corollary~\ref{cor:kin-e-cons-p-1}, and $\Delta \mvec{v}$ is the 1D grid spacing along the direction $\mvec{v}$. Note that $\Delta \mvec{v}$ in the $j^{\textrm{th}}$ cell is related to the cell center velocity, 
\begin{align}
    \Delta \mvec{v}_j = \overline{\mvec{v}}_{j+1} - \overline{\mvec{v}}_j, \label{eq:deltaVCellCenter}
\end{align}
and the sum in \eqr{\ref{eq:M0Star}} is over all surfaces except the edges of velocity space, i.e., the last index $j_{max}$.
\end{proposition}
\begin{proof}
Since we are restricting ourselves to only using piecewise linear polynomials, $|\mvec{v}|^2\notin \mathcal{V}_h^p$, and we must project $|\mvec{v}|^2$ onto our basis set using \eqr{\ref{eq:v2-p1}}.
We can then substitute $w = 1/2 m \overline{|\mvec{v}|^2}$ into \eqr{\ref{eq:dis-weak-FP}} and sum over cells to obtain
\begin{align}
   \sum_j \int_{K_j} \frac{1}{2} m \overline{|\mvec{v}|^2} \pfrac{f_h}{t} \dz = - \sum_j & \oint_{\partial K_j} \nu T_h (\mvec{n} \cdot \overline{\mvec{v}}) \hat{f} dS \notag \\
   & - \sum_j \int_{K_j} \nu m  \overline{\mvec{v}} \cdot \left ( \mvec{v} - \mvec{u}_h \right )f_h \dz, \label{eq:FP-energy-step-p1}
\end{align}
where $1/2 \thinspace \gv \overline{|\mvec{v}|^2} = \overline{\mvec{v}}$, the cell center velocity. Note the differences in \eqr{\ref{eq:FP-energy-step-p1}} compared to \eqr{\ref{eq:FP-energy-step-p2}} in Proposition~\ref{prop:discrete-energy-cons-FP-p2}, when we were employing at least quadratic polynomials and $|\mvec{v}|^2 \in \mathcal{V}_h^p$.
The sum over surface integrals involving the numerical flux function, $\hat{\mvec{G}}$, still vanishes because $\overline{|\mvec{v}|^2}$, despite being a projection, is continuous across velocity space interfaces, and we can thus still leverage the fact that the flux for the drag term is a Godunov flux and the gradient of the recovered distribution function is continuous to pairwise cancel the surface integrals in the sum.
Importantly, the volume term for the diffusion has vanished, since $\gv^2 \overline{|\mvec{v}|^2} = 0$.
Likewise, we cannot cancel the interior sums over the surface in the remaining surface integrals like we did in Proposition~\ref{prop:discrete-energy-cons-FP-p2} because $\overline{\mvec{v}}$ is not continuous at velocity space surfaces---$\overline{\mvec{v}}$ is a piecewise constant function!
To have energy conservation, we then must have
\begin{align}
    \int_{\Omega_k} & T_h \left [  \sum_j \oint_{\partial V_{max_j}} (\mvec{n}\cdot \overline{\mvec{v}}_{max})f_h \thinspace dS_{V_{max}} - \sum_j \oint_{\partial V_{min_j}} (\mvec{n}\cdot \overline{\mvec{v}}_{min}) f_h \thinspace dS_{V_{min}} \right ] \dx \notag \\
    & + \int_{\Omega_k} \left [\sum_j \int_{K_j \setminus \Omega_k } m  \overline{\mvec{v}} \cdot \left ( \mvec{v} - \mvec{u}_h \right )f_h \dv - T_h \sum_{j\neq j_{max}} \oint_{\partial K_j \setminus \Omega_j} (\mvec{n} \cdot \Delta \mvec{v}) \hat{f} dS \right ]\dx= 0,
\end{align}
where we have used \eqr{\ref{eq:deltaVCellCenter}} to simplify the interior surface integrals of the recovered distribution function.
Repeating our calculation for $w = 1/2 m \overline{|\mvec{v}|^2} \varphi_\ell(\mvec{x})$ for each of our $\ell$ configuration space basis functions, $\varphi_\ell \in \mathcal{X}_h^p$, and using Eqns.\thinspace(\ref{eq:M0Star}--\ref{eq:M2Star}), we then have
\begin{align}
    T_h & \left [  \sum_j \oint_{\partial V_{max_j}} (\mvec{n}\cdot \overline{\mvec{v}}_{max})f_h \thinspace dS_{V_{max}} - \sum_j \oint_{\partial V_{min_j}} (\mvec{n}\cdot \overline{\mvec{v}}_{min}) f_h \thinspace dS_{V_{min}} \right ] \notag \\
    & + m M^*_{2_h} - m \mvec{M}^*_{1_h} \cdot \mvec{u}_h - 3 M^*_{0_h} T_h \doteq 0, \notag    
\end{align}
exactly the constraint we expect for energy to be conserved.
We note that in one spatial dimension and one velocity dimension, this constraint simplifies to
\begin{align}
    T_h \big[ \overline{v}_{max}f_h(v_{max})-\overline{v}_{min} f_h(v_{min})\big] + M^*_{2_h} - M^*_{1_h} u_h -T_h M^*_{0_h} \doteq 0,
\end{align}
where, like in \eqr{\ref{eq:FP-1X1V-p2-constraint}} in Proposition~\ref{prop:discrete-energy-cons-FP-p2}, the 1X1V constraint does not have a factor of three multiplying $T_h$.
\end{proof}
So, the semi-discrete Fokker--Planck equation also retains conservation of energy with piecewise linear polynomials, provided one modifies the constraint equations,
\begin{align}
    T_h & \left (  \sum_j \oint_{\partial V_{max_j}} f_h \thinspace dS_{V_{max}} - \sum_j \oint_{\partial V_{min_j}} f_h \thinspace dS_{V_{min}} \right ) + m \mvec{M}_{1_h} - m M_{0_h} \mvec{u}_h \doteq 0, \notag \\
    T_h & \left [  \sum_j \oint_{\partial V_{max_j}} (\mvec{n}\cdot \overline{\mvec{v}}_{max})f_h \thinspace dS_{V_{max}} - \sum_j \oint_{\partial V_{min_j}} (\mvec{n}\cdot \overline{\mvec{v}}_{min}) f_h \thinspace dS_{V_{min}} \right ] \notag \\
    & + m M^*_{2_h} - m \mvec{M}^*_{1_h} \cdot \mvec{u}_h - 3 M^*_{0_h} T_h \doteq 0, \notag    
\end{align}
or in one spatial dimension and one velocity dimension (1X1V),
\begin{align}
    T_h \big[ f_h(v_{max})-f_h(v_{min})\big] + m M_{1_h} - m M_{0_h} u_h & \doteq 0, \notag \\
    T_h \big[ \overline{v}_{max}f_h(v_{max})-\overline{v}_{min} f_h(v_{min})\big] + M^*_{2_h} - M^*_{1_h} u_h -T_h M^*_{0_h} & \doteq 0. \notag   
\end{align}
When using piecewise linear polynomials, one must not only compute $M_{0_h}$ and $\mvec{M}_{1_h}$, the standard moments given in Eqns.\thinspace(\ref{eq:discrete0thMoment}) and (\ref{eq:discrete1stMoment}), but also the ``star moments'' given by Eqns.\thinspace(\ref{eq:M0Star} -- \ref{eq:M2Star}).
Using these coupled constraint equations, we can then uniquely compute the temperature, $T_h$, and flow, $\mvec{u}_h$, for use in our semi-discrete Fokker--Planck equation.

Before we conclude this section, we note that we have not discussed a discrete analogy to the continuous system's Second Law of Thermodynamics, Proposition~\ref{prop:collisionEntropy}, and H-theorem, Corollary~\ref{coro:HTheorem}.
Unfortunately the finite velocity space extents required by our continuum approach complicate the requisite proofs, along with the required gradients of the expansion of $\ln (f_h)$. 
We will instead defer until Chapter~\ref{ch:Benchmarks}, when we demonstrate numerically that the scheme still respects these essential physics properties.

\section{The Time Discretization of the \\Vlasov--Maxwell--Fokker--Planck System of Equations}\label{sec:TimeDiscretization}
Having now constructed a semi-discrete scheme for the VM-FP system of equations for the discretization of the equation system in phase space and configuration space, we seek to complete the discretization with a discussion of how best to numerically integrate the semi-discrete system in time.
We note that the result of the semi-discrete system is a set of ordinary differential equations.
For the Vlasov--Fokker--Planck equation we have
\begin{align}
    \pfrac{f_h}{t} \doteq \mathcal{L}(f_h, \mvec{E}_h, \mvec{B}_h, t), \label{eq:linearOpKinetic}
\end{align}
and likewise for Maxwell's equations, where $\mathcal{L}$ is a linear operator encompassing the evaluation of the integrals in the discrete weak forms, Eqns.\thinspace(\ref{eq:dis-weak-form}) and (\ref{eq:dis-weak-FP}), for all basis functions $w \in \mathcal{V}_h^p$ and all cells $K_j \in \mathcal{T}$.
We will show in Chapter~\ref{ch:ImplementationDGFEM} how we actually construct and evaluate $\mathcal{L}$ in \eqr{\ref{eq:linearOpKinetic}}.
For now, we imagine that we have evaluated $\mathcal{L}$ for the Vlasov--Fokker--Planck equation, and likewise Maxwell's equations, and now need to solve the system of ordinary differential equations for the time derivative of the discrete distribution function and electromagnetic fields.

We consider in this thesis a class of strong stability preserving Runge--Kutta (SSP-RK) methods \citep{Shu:2002,Durran:2010}.
These methods are all multi-stage Runge--Kutta methods.
Defining a forward Euler step as
\begin{align}
    \mathcal{F}(f,t) = f + \Delta t \mathcal{L}(f,t), \label{eq:forwardEuler}
\end{align}
we can construct, for example, the second order SSP-RK,
\begin{align}\begin{aligned}
  f^{(1)} &= \mathcal{F}\left(f^{n},t^{n}\right), \\ f^{n+1} &=
  \frac{1}{2}f^{n} + \frac{1}{2}\mathcal{F}\left(f^{(1)},t^n+\Delta
  t\right),
\end{aligned}\end{align}
the third order SSP-RK,
\begin{align}
\begin{aligned}
  f^{(1)} &= \mathcal{F}\left(f^{n},t^{n}\right), \\ f^{(2)} &=
  \frac{3}{4}f^{n} + \frac{1}{4}\mathcal{F}\left(f^{(1)},t^n+\Delta
  t\right),\\ f^{n+1} &= \frac{1}{3}f^{n} +
  \frac{2}{3}\mathcal{F}\left(f^{(2)},t^n+\Delta t/2\right),
\end{aligned}\label{eq:SSPRK3}
\end{align}
and the four stage third order SSP-RK:
\begin{align}\begin{aligned}
  f^{(1)} &= \mathcal{F}\left(f^{n},t^{n}\right), \\ f^{(2)} &=
  \frac{1}{2}f^{(1)} + \frac{1}{2}\mathcal{F}\left(f^{(1)},t^n+\Delta
  t/2\right),\\ f^{(3)} &= \frac{2}{3}f^{n} + \frac{1}{6}f^{(2)} +
  \frac{1}{6}\mathcal{F}\left(f^{(2)},t^n+\Delta t\right),\\ f^{n+1}
  &= \frac{1}{2}f^{(3)} +
  \frac{1}{2}\mathcal{F}\left(f^{(3)},t^n+\Delta t/2\right).
\end{aligned}\end{align}
There are SSP-RK methods with more stages, as well as higher order, than the methods shown here \citep{Shu:2002}. 
Multi-stage Runge--Kutta methods require a balance between the order of the scheme and the number of stages, and thus the amount of computations required. 
Especially for very high order multi-stage Runge--Kutta methods, it can require increasingly large numbers of intermediate stages to attain marginal improvements to the order of the scheme.
We will most often employ the three-stage, third order SSP-RK method, \eqr{\ref{eq:SSPRK3}}, as a balance between accuracy, computation, and memory footprint for the storage of the intermediate stages.

The result of the SSP-RK-DG space-time discretization for the VM-FP system of equations is a fully explicit scheme, and thus we expect to be restricted in the size of our time-step by a Courant-Friedrich-Lewy (CFL) condition. 
CFL conditions arise due to the restriction that we must be able to integrate the system of ordinary differential equations along the characteristics of the partial differential equation.
In practical terms, imagine propagating a wave with velocity $v$ in a discrete system.
In order to propagate the wave along a discrete grid with some cell spacing $\Delta x$, we must be careful not to take too large of a time-step, lest the wave move multiple grid cells in a single time-step and thus potentially lose amplitude and phase information.
Thus, we require
\begin{align}
    \Delta t \lesssim \frac{\Delta x}{v}.
\end{align}

CFL conditions can be expressed simply in terms of the CFL frequency, the fastest signal in the discrete system,
\begin{align}
    \sum_{i=1}^d \omega_i \Delta t \leq C,
\end{align}
where $d$ is the dimensionality of the problem, $\omega_i$ is the fastest frequency in each of the $i$ dimensions, and $C$ is some additional safety factor which may be required for stability.
One CFL condition will come from solving Maxwell's equations, where we must be able to stably propagate light waves,
\begin{align}
    \sum_{i=1}^{CDIM} c \frac{\Delta t}{\Delta x_i} \leq \frac{1}{2p + 1}. \label{eq:MaxwellCFL}
\end{align}
Here, $c$ is the speed of light, $CDIM$ is the number of configuration space dimensions, and $p$ is the polynomial order of our basis expansion.
We recognize $c/\Delta x_i$ as the largest discrete frequency in the system given some cell spacing $\Delta x_i$ in each of the $i$ configuration space dimensions, since the speed of light is unequivocally the fastest velocity in the system.

Note that we have plugged in for the safety factor $C = 1/(2p+1)$.
This CFL condition is similar to the constraint for the finite-difference-time-domain (FDTD) discretization of Maxwell's equations \citep{Yee:1966}, but with this additional safety factor for stability which depends upon the polynomial order of our basis expansion\citep{Cockburn:2001}.
In fact, \citet{Cockburn:2001} explicitly calculated the required safety factor arising from the polynomial order of the basis expansion for $L^2$ stability, and although it is not exactly $1/(2p + 1)$, the safety factor is approximately this value for a wide variety of polynomial orders, at least to within five to ten percent. 
As such, we use $1/(2p + 1)$ for the safety factor in our computation of the size of the time-step.
In the limit that the grid spacing in each configuration space dimension is equal, $\Delta x = \Delta y = \Delta z$, we can simplify the Maxwell's equation CFL condition to
\begin{align}
    c \frac{\Delta t}{\Delta x} \leq \frac{1/CDIM}{2p + 1}.
\end{align}

We likewise have a CFL condition for the Vlasov--Fokker--Planck equation.
We first note that the Vlasov equation CFL condition can be written as,
\begin{align}
    \Delta t \sum_{i=1}^{CDIM+VDIM} \max_{\mathcal{T}} \left |\frac{\alpha_i}{\Delta z_i} \right | \leq \frac{1}{2p + 1},
\end{align}
where $|\cdot|$ is the absolute value.
It will give us better intuition for this time step constraint by separating the configuration space and velocity space CFL conditions,
\begin{align}
    \Delta t \left [\sum_{i=1}^{CDIM} \max_{\mathcal{T}}\left |\frac{v_i }{\Delta x_i} \right | + \sum_{j=1}^{VDIM} \max_{\mathcal{T}}\left |\frac{q_s/m_s \thinspace (E_h + v \times B_h)_j}{\Delta v_j} \right | \right ] \leq \frac{1}{2p + 1}, \label{eq:VlasovCFLCondition}
\end{align}
where we have abbreviated the number of configuration space dimensions as $CDIM$, as before in \eqr{\ref{eq:MaxwellCFL}}, and the number of velocity space dimensions as $VDIM$.
The first term on the left-hand side of \eqr{\ref{eq:VlasovCFLCondition}} uses the maximum velocity in each direction, i.e., the velocity space edge in each direction, to determine the largest frequency in configuration space from the local configuration space grid spacing $\Delta x_i$.
The second term on the left-hand side of \eqr{\ref{eq:VlasovCFLCondition}} uses the maximum acceleration due to the electromagnetic fields measured in the phase space domain $\mathcal{T}$ to compute the largest frequency in velocity space from the local velocity space grid spacing $\Delta v_j$.

Likewise, for the Fokker--Planck equation we have
\begin{align}
    \Delta t \left [\sum_{j=1}^{VDIM} \max_{\mathcal{T}} \left | \nu \frac{(v - u_h)_j}{\Delta v_j} \right |+ \sum_{j=1}^{VDIM} \max_{\mathcal{T}} \left |\nu \frac{T_h}{m_s} \frac{1}{(\Delta v_j)^2} \right | \right ]\leq \frac{1}{2p + 1}, \label{eq:FPCFLCondition}
\end{align}
where the first term on the left-hand side of \eqr{\ref{eq:FPCFLCondition}} is the maximum frequency due to the drag term, and the second term on the left-hand side of \eqr{\ref{eq:FPCFLCondition}} is the CFL frequency due to the diffusion operator.
Note that the CFL frequency of the diffusive term scales like $(\Delta v_j)^{-2}$, the inverse square of the grid spacing, as it must because the diffusion operator involves two derivatives of the distribution function in velocity space.
Defining
\begin{align}
    CFL^{collisionless} & = \sum_{i=1}^{CDIM} \max_{\mathcal{T}}\left |\frac{v_i }{\Delta x_i} \right | + \sum_{j=1}^{VDIM} \max_{\mathcal{T}}\left |\frac{q_s/m_s \thinspace (E_h + v \times B_h)_j}{\Delta v_j} \right |, \\
    CFL^c & = \sum_{j=1}^{VDIM} \max_{\mathcal{T}} \left | \nu \frac{(v - u_h)_j}{\Delta v_j} \right |+ \sum_{j=1}^{VDIM} \max_{\mathcal{T}} \left |\nu \frac{T_h}{m_s} \frac{1}{(\Delta v_j)^2} \right |, 
\end{align}
we can then say that the total CFL condition for the Vlasov--Fokker--Planck equation is
\begin{align}
    \Delta t (CFL^{collisionless} + CFL^c) \leq \frac{1}{2p + 1}. \label{eq:totalVlasovCFL}
\end{align}

A few remarks on the CFL condition for the Vlasov--Fokker--Planck equation are in order. 
The first remark is that we are being careful to determine the maximum frequency in each dimension.
For Maxwell's equation, the CFL condition, \eqr{\ref{eq:MaxwellCFL}}, could naturally be simplified because the speed of light is the same in each direction.
We could presume a similar restriction for the Vlasov--Fokker--Planck equation, find the maximum characteristic of each of the phase space dimensions, find the maximum of those maximum characteristics, and then include an additional safety factor of $1/d_z$ where $d_z$ is the number of phase space dimensions.
In other words, presuming the acceleration in the $x$ direction, $E_x + v_y B_z - v_z B_y$, is the maximum characteristic in the system, we can use that acceleration divided by the grid spacing $\Delta v_x$ to calculate the CFL frequency, and then divide that CFL frequency by six if one is evolving the Vlasov--Fokker--Planck equation in the full six dimensional phase space.
Of course, this approach could lead to a quite restrictive time-step compared to the combination of CFL frequencies in Eqns.\thinspace(\ref{eq:VlasovCFLCondition}) and (\ref{eq:FPCFLCondition}), depending on how anisotropic the characteristics are.
For example, even if the acceleration is quite large in the $x$ direction, leading to a large CFL frequency in the $v_x$ direction, the acceleration in the other two velocity dimensions, along with the maximum velocity in the three configuration space dimensions, could be lower magnitude and thus lead to smaller contributions to the total CFL frequency.
So long as we are careful to stay within the region of stability for our SSP-RK scheme, there is little reason not to take the largest possible time-step. 

An additional remark is to connect the maximum characteristic, for example the maximum acceleration or the maximum drag, to the numerical flux functions defined previously, Eqns.\thinspace(\ref{eq:simpleGlobalLFVlasov}) and (\ref{eq:dragGlobalLF}).
In the global Lax-Friedrichs fluxes defined for the Vlasov equation and drag component of the Fokker--Planck equation, we required the maximum of the flux, either collisionless or drag, sampled over the whole phase space domain, $\mathcal{T}$.
This term, $\tau$ for example in \eqr{\ref{eq:simpleGlobalLFVlasov}}, is exactly the required component of the CFL frequency in each dimension in Eqns.\thinspace(\ref{eq:VlasovCFLCondition}) and (\ref{eq:FPCFLCondition}).
Historically, the definition of the penalization term has also been done in the opposite direction, with for example
\begin{align}
    \tau_i = \frac{1}{2p + 1} \frac{\Delta z_i}{\Delta t},
\end{align}
as in \citep{Lax:1954}.
Though this particular penalization term is a critical component of some stability bounds proved for the hyperbolic partial differential equations studied in \citep{Lax:1954}, such a large penalization can have unintended consequences for the accuracy of the scheme, leading to a combination of overdiffusion and monotonicity errors in the discrete solution.
We will avoid such an extreme definition and instead continue to use Eqns.\thinspace(\ref{eq:simpleGlobalLFVlasov}) and (\ref{eq:dragGlobalLF}) when we discuss the actual implementation of the method in the next chapter, Chapter~\ref{ch:ImplementationDGFEM}.

With both the CFL constraint for Maxwell's equations and the CFL constraint for the Vlasov--Fokker--Planck equation in hand, we have completed the mathematical formulation of our discrete VM-FP system of equations.
We evaluate the operators defined in our semi-discrete scheme, Eqns.\thinspace(\ref{eq:dis-weak-form}) and (\ref{eq:dis-weak-FP}) for the Vlasov--Fokker--Planck equation, and Eqns.\thinspace (\ref{eq:dis-weak-B}) and (\ref{eq:dis-weak-E}) for Maxwell's equations, and then determine from these evaluations which of the two CFL conditions, \eqr{\ref{eq:totalVlasovCFL}} or \eqr{\ref{eq:MaxwellCFL}}, is more restrictive.
Having calculated both the linear operator $\mathcal{L}$ for the complete semi-discrete VM-FP system of equations and the size of the time step $\Delta t$, we can then plug the results into a forward Euler time step, \eqr{\ref{eq:forwardEuler}}, and repeat the process as desired for a multi-stage SSP-RK method, e.g., SSP-RK3 in \eqr{\ref{eq:SSPRK3}}.
Before we move on from the mathematical foundation we have laid in this chapter to the details of turning this mathematical foundation into algorithms and code, we summarize the results of this chapter in the next section.

\section{Summary of Chapter 2}\label{sec:chapter2Summary}

We now summarize the contents of this chapter, and in doing so, foreshadow some of the most important issues we will have to address in Chapter~\ref{ch:ImplementationDGFEM} when we move from a mathematical formulation of the discrete scheme to an algorithmic formulation of the numerical method.
In this regard, it is worth further driving the point of this chapter home: Eqns.\thinspace(\ref{eq:dis-weak-form}) and (\ref{eq:dis-weak-FP}) for the Vlasov--Fokker--Planck equation, and Eqns.\thinspace (\ref{eq:dis-weak-B}) and (\ref{eq:dis-weak-E}) for Maxwell's equations, followed by an appropriate ordinary differential equation integrator such as an SSP-RK3 method,  \eqr{\ref{eq:SSPRK3}}, are a mathematically complete description of the discrete scheme.
To now be a bit glib, mathematically, we are done.

We have formulated a discrete scheme, which has provably retained properties of the continuous system, with some flexibility in the choice of numerical flux function, e.g., for Maxwell's equations,
central fluxes, Eqns.\thinspace(\ref{eq:centralE})-(\ref{eq:centralB}), or upwind fluxes, Eqns.\thinspace(\ref{eq:r-e2})-(\ref{eq:r-b3}), both of which are perfectly acceptable numerical flux functions for Maxwell's equations which have different, but potentially better properties depending on the problem being tackled.
For example, we showed in Lemma~\ref{lem:em-e-cons} that central fluxes for Maxwell's equations conserves the electromagnetic energy, thus producing a completely conservative scheme in Proposition~\ref{prop:DiscreteVlasovEnergyProof}, while upwind fluxes for Maxwell's equations introduces numerical diffusion in the electromagnetic energy, thus leading to a monotonic decay of the energy.
Central fluxes for Maxwell's equations is not free of numerical errors though, replacing diffusive errors with dispersive errors, errors in the phases of the solutions, e.g., when propagating an electromagnetic wave.
These dispersive errors can be equally problematic \citep{Hesthaven:2004}, but regardless of the choice of numerical flux function, the central point remains: the mathematical formulation of the discrete Vlasov--Maxwell--Fokker--Planck (VM-FP) system of equations using a discontinuous Galerkin finite element method, with a polynomial basis, is completely specified by Eqns.\thinspace(\ref{eq:dis-weak-form}) and (\ref{eq:dis-weak-FP}) for the Vlasov--Fokker--Planck equation, and Eqns.\thinspace (\ref{eq:dis-weak-B}) and (\ref{eq:dis-weak-E}) for Maxwell's equations.
Of course, to go from Eqns.\thinspace(\ref{eq:dis-weak-form}) and (\ref{eq:dis-weak-FP}) for the Vlasov--Fokker--Planck equation, and Eqns.\thinspace (\ref{eq:dis-weak-B}) and (\ref{eq:dis-weak-E}) for Maxwell's equations, to a numerical algorithm and code is its own non-trivial task, which we address in Chapter~\ref{ch:ImplementationDGFEM}.
So, to summarize:
\begin{enumerate}
\item The discontinuous Galerkin finite element method (DG) is a spatial discretization scheme which combines aspects of finite element and finite volume methods and leverages the benefits of both numerical methods to produce high order accurate, robust, physically motivated spatial discretizations of a wide spectrum of partial differential equations.
The essential idea is an $L^2$ minimization of the error after expanding the quantity of interest, for example the distribution function,
\begin{align*}
f(\mvec{z}, t) \approx f_h(\mvec{z}, t) = \sum_{k=1}^N f_k(t) w (\mvec{z}),
\end{align*}
in a basis set $w = w (\mvec{z})$, which we took to be the space of polynomials of order $p,\thinspace \mathbb{P}^p$ throughout this Chapter. 
The $L^2$ minimization of the error can be formulated in the language of \emph{weak equality},
\begin{align*}
\pfrac{f_h}{t} \doteq G[f_h],
\end{align*} 
where $G[f_h]$ is a general operator acting on the quantity of interest, for example the Vlasov--Fokker--Planck spatial operator, and $\doteq$ in the space spanned by $w = w (\mvec{z})$ denotes the operation
\begin{align*}
\int_I \pfrac{f_h}{t} w_\ell (\mvec{z}) \dz = \int_I G[f_h] w_\ell (\mvec{z}), \quad \forall \ell=1,\ldots,N.
\end{align*}
Note that weak equality, unlike strong equality where functions are everywhere equal, determines the solution up to an equivalence class, enforcing that the projections of the left hand side and right hand side on the basis set spanned by $ w_\ell (\mvec{z}), \thinspace \forall \ell=1,\ldots,N$ are equal.
\item With the machinery of weak equality and an $L^2$ minimization of the error, we can formulate the DG discretization of our equation system of interest and derive the discrete-weak forms of the VM-FP system of equations,
 \begin{align*}
   \int_{K_j} & w\pfrac{f_h}{t} \dz + 
  \oint_{\partial K_j}w^- \mvec{n}\cdot\hat{\mvec{F}}  \thinspace dS 
  - \int_{K_j} \gz w \cdot \gvec{\alpha}_h f_h \dz = \\ 
  & \oint_{\partial K_j} \nu \thinspace w^- \mvec{n}\cdot\hat{\mvec{G}}  \thinspace dS - \oint_{\partial K_j} \nu \thinspace \mvec{n}\cdot \gv w^- \frac{T_h}{m} \hat{f} dS \\
  &-\int_{K_j} \nu \left [ \thinspace \gv w \cdot \left ( \mvec{v} - \mvec{u}_h \right )f_h - \gv^2 w \left ( \frac{T_h}{m} f_h \right ) \right ] \dz, \\
  \int_{\Omega_j}& \varphi \pfrac{\mvec{B}_h}{t}\dx
  +
  \oint_{\partial\Omega_j} d\mvec{s}\times(\varphi^-\hat{\mvec{E}}_h)
  -
  \int_{\Omega_j} \gx\varphi\times\mvec{E}_h \dx
  =
  0, \\
  \epsilon_0\mu_0\int_{\Omega_j}& \varphi \pfrac{\mvec{E}_h}{t}\dx
  -
  \oint_{\partial\Omega_j} d\mvec{s}\times(\varphi^-\hat{\mvec{B}}_h)
  +
  \int_{\Omega_j} \gx\varphi\times\mvec{B}_h \dx
  =
  -\mu_0\int_{\Omega_j} \varphi \mvec{J}_h\dx,
\end{align*}
where the first equation is the semi-discrete Vlasov--Fokker--Planck equation, and the second two equations are the semi-discrete Faraday and Ampere-Maxwell equations from Maxwell's equations.
Note that, as we pointed out in Section~\ref{sec:semi-discrete-Vlasov}, the divergence constraint for Maxwell's equations are not an explicit component of our discretization, and thus errors in the divergence of the electric and magnetic fields may arise throughout the numerical integration of our DG discretization of Maxwell's equations.
We will address this subtlety in Chapter~\ref{ch:Benchmarks} when we benchmark our numerical method for the VM-FP system of equations.
The discrete weak forms for the VM-FP system of equations, Eqns.\thinspace(\ref{eq:dis-weak-form}) and (\ref{eq:dis-weak-FP}) for the Vlasov--Fokker--Planck equation, and Eqns.\thinspace (\ref{eq:dis-weak-B}) and (\ref{eq:dis-weak-E}) for Maxwell's equations, are derived using integration by parts on the spatial operators, so that we obtain contributions to the solution from both volume and surface integrals.
These individual pieces make the connection between DG and finite element and finite volume methods concrete, with the volume integral bearing a resemblance to the integrals over the grid cells required in a finite element method, and the surface integral requiring the specification of a numerical flux function for the advection of the quantities of interest across surface interfaces, just as in a finite volume method.
Importantly, the semi-discrete Fokker--Planck equation requires \emph{two} integration by parts on the diffusion operator to ultimately demonstrate the semi-discrete scheme retains some of the properties of the continuous Fokker--Planck equation discussed in Section~\ref{sec:PropertiesKineticEquation}.
\item There are many potential options for numerical flux functions, but a critical property of the numerical flux function to prove our semi-discrete spatial discretization retains properties of the continuous system is that the numerical flux function obeys the \emph{Godunov flux} condition,
   \begin{align*}
      \oint_{\partial K_j}w^- \mvec{n}\cdot\hat{\mvec{F}}  \thinspace dS  = -\oint_{\partial K_j}w^+ \mvec{n}\cdot\hat{\mvec{F}}  \thinspace dS,
   \end{align*}
i.e., the flux into a cell is equal and opposite to the flux out of its neighbor cell along the shared interface.
Example numerical flux functions for the collisionless advection in phase space are
\begin{align*}
    \mvec{n}\cdot\hat{\mvec{F}}(\gvec{\alpha}_h f^+_h, \gvec{\alpha}_h f^-_h ) & = \frac{1}{2} \mvec{n}\cdot \gvec{\alpha}_h \left (f^+_h + f^-_h \right ), \\
    \mvec{n}\cdot\hat{\mvec{F}}(\gvec{\alpha}_h
    f_h^-, \gvec{\alpha}_h f_h^+)
    & =
    \begin{cases}
        \mvec{n}\cdot\gvec{\alpha}_h f^- \quad \textrm{if} \quad \sign(\gvec{\alpha}_h) > 0, \\
        \mvec{n}\cdot\gvec{\alpha}_h f^+ \quad \textrm{if} \quad \sign(\gvec{\alpha}_h) < 0,
    \end{cases} \\
    \mvec{n}\cdot\hat{\mvec{F}}(\gvec{\alpha}_h
    f_h^-, \gvec{\alpha}_h f_h^+) & = \frac{1}{2} \mvec{n}\cdot \gvec{\alpha}_h \left (f^+_h + f^-_h \right ) - \frac{\tau}{2} (f^+ - f^-),
\end{align*}
i.e., central fluxes, upwind fluxes, and global Lax-Friedrichs fluxes.
Note that these forms of the numerical flux function exploit the fact that the discrete phase space flow, $\gvec{\alpha}_h$, is continuous at the corresponding surface interfaces, Lemma~\ref{lem:discrete-phase-space-incompress}.
Likewise, similar flux functions can be defined for the numerical flux function for Maxwell's equations,
\begin{align*}
    \hat{\mvec{E}}_h & = \llbracket \mvec{E}\rrbracket, \\ \hat{\mvec{B}}_h & = \llbracket\mvec{B}\rrbracket,
\end{align*}
or upwind fluxes,
\begin{align*}
  \hat{E}_2 & = \llbracket E_2 \rrbracket - c\thinspace\{ B_3 \}, \\
  \hat{E}_3 & = \llbracket E_3 \rrbracket + c\thinspace\{ B_2 \}, \\
  \hat{B}_2 & = \llbracket B_2 \rrbracket + \{E_3\}/c, \\
  \hat{B}_3 &= \llbracket B_3 \rrbracket - \{E_2\}/c, 
\end{align*}
with
\begin{align*}
  \llbracket g \rrbracket & \equiv (g^+ + g^-)/2, \\
  \{ g \} & \equiv (g^+-g^-)/2,
\end{align*}
and the drag component of the Fokker--Planck equation,
\begin{align*}
    \mvec{n}\cdot\hat{\mvec{F}}_{drag} & = \frac{1}{2} \mvec{n}\cdot (\mvec{v} - \mvec{u}_h) (f^+ + f^-), \\
    \mvec{n}\cdot\hat{\mvec{F}}_{drag} & =     \begin{cases}
        \mvec{n}\cdot(\mvec{v} - \mvec{u}_h) f^- \quad \textrm{if} \quad \sign(\mvec{v} - \mvec{u}_h) > 0, \\
        \mvec{n}\cdot(\mvec{v} - \mvec{u}_h) f^+ \quad \textrm{if} \quad \sign(\mvec{v} - \mvec{u}_h) < 0,
    \end{cases} \\
    \mvec{n}\cdot\hat{\mvec{F}}_{drag} & = \frac{1}{2} \mvec{n}\cdot (\mvec{v} - \mvec{u}_h) (f^+ + f^-) - \frac{\max_{\mathcal{T}}|\mvec{v} - \mvec{u}_h|}{2} (f^+ - f^-),
\end{align*}
where we have used the fact that $\mvec{v} - \mvec{u}_h$ is continuous across velocity space surfaces to simplify a central flux, upwind flux, and global Lax-Friedrichs flux for the drag component of the numerical flux function for the Fokker--Planck equation.
The total numerical flux function for the Fokker--Planck equation is
\begin{align*}
    \mvec{n}\cdot\hat{\mvec{G}} = \mvec{n}\cdot \left (\hat{\mvec{F}}_{drag} + \frac{T_h}{m}\gv \hat{f} \right ),
\end{align*}
where $\hat{f}$ is the distribution function at the interface using the recovery procedure, and we require both the gradient, and the value, of the recovered distribution function since we integrated the diffusion term by parts twice.
\item The recovery procedure for computing the surface terms for the diffusion also leverages weak equality.
We have the distribution function in two neighboring cells sharing an interface,
\begin{align*}
  \hat{f} &\doteq f_L, \\
  \hat{f} &\doteq f_R,
\end{align*}
where $f_L$ is the distribution function in the cell to the ``left'' of the interface and $f_R$ is the distribution function to the ``right'' of the interface.
Defining the recovery polynomial as
\begin{align}
  \hat{f}(x) = \sum_{m=0}^{2N-1} \hat{f}_m x^m,
\end{align}
in one dimension, we can then uniquely compute a continuous polynomial (with continuous first derivatives, too).
Importantly, the recovery procedure is fundamentally one-dimensional, since the discontinuity we are constructing the continuous representation along is a discontinuity at a surface.
A continuous function, with continuous first derivatives, is ``recovered'' using the data that is discontinuous at a given surface, i.e., the discontinuity is along the one dimension that is fixed at that surface.
The reconstruction of the recovery polynomial's functional dependence along the surface in arbitrary dimensions will be addressed as part of our discussion of how to turn the mathematical formulation of the discrete scheme into code in Chapter~\ref{ch:ImplementationDGFEM}.
\item We further utilize weak equality to determine the required velocity space moments for the coupling between the Vlasov--Fokker--Planck equation and Maxwell's equations, as well as the moments required for the drag and diffusion coefficients in the Fokker--Planck equations.
Weak equality allows us to define fundamental operators, e.g., division and multiplication, when the quantities being manipulated are themselves projections.
The velocity space moments are
\begin{align*}
    M_{0_h} & \doteq \sum_j \int_{K_j\setminus \Omega_k} f_h \dv, \\
    \mvec{M}_{1_h} & \doteq \sum_j \int_{K_j\setminus \Omega_k} \mvec{v} f_h \dv, \\
    M_{2_h} & \doteq \sum_j \int_{K_j\setminus \Omega_k} |\mvec{v}|^2 f_h \dv,  
\end{align*}
with the charge and current densities required for coupling to Maxwell's equations given by
\begin{align*}
    \rho_{c_h} & = \sum_s q_s M_{0_{h_s}}, \\
    \mvec{J}_h & = \sum_s q_s \mvec{M}_{1_{h_s}}.
\end{align*}
Note that the charge and current density are strongly equal to the sum over species of the velocity space moments, since we have already projected down to the configuration space expansion.
Likewise, for the flow and temperature in the drag and diffusion coefficients,
\begin{align*}
    T_h & \left (  \sum_j \oint_{\partial V_{max_j}} f_h \thinspace dS_{V_{max}} - \sum_j \oint_{\partial V_{min_j}} f_h \thinspace dS_{V_{min}} \right ) + m \mvec{M}_{1_h} - m M_{0_h} \mvec{u}_h \doteq 0, \\
    T_h & \left [  \sum_j \oint_{\partial V_{max_j}} (\mvec{n}\cdot \mvec{v}_{max})f_h \thinspace dS_{V_{max}} - \sum_j \oint_{\partial V_{min_j}} (\mvec{n}\cdot \mvec{v}_{min}) f_h \thinspace dS_{V_{min}} \right ] \\
    & + m M_{2_h} - m \mvec{M}_{1_h} \cdot \mvec{u}_h - 3 M_{0_h} T_h \doteq 0, 
 \end{align*}
which require weak multiplication and division, or weighted $L^2$ projections, as defined in Section~\ref{sec:WeakEquality}.
These expressions can be modified, for the discrete current density, temperature, and flow, in the case of running with only piecewise linear polynomials, as discussed in Corollary~\ref{cor:kin-e-cons-p-1} and Proposition~\ref{prop:FP-energy-p1} respectively.
\item Using weak equality to construct consistent projections of velocity moments, a Godunov numerical flux function, and appropriate boundary conditions, i.e., zero-flux in velocity space and a self-contained boundary condition in configuration space, like periodic boundary conditions, we can prove that the semi-discrete scheme retains a number of the continuous VM-FP system of equations' properties.
In particular, the whole system conserves mass and energy, even when using piecewise linear polynomials and projecting $|\mvec{v}|^2$ onto linear polynomials, and we can show that even though the collisionless evolution does not obey momentum conservation, the semi-discrete Fokker--Planck equation conserves momentum.
Importantly, the lack of momentum conservation arises from our discretization of Maxwell's equations, and thus only depends on configuration space resolution, a property we will numerically demonstrate in Chapter~\ref{ch:Benchmarks}.
The collisionless component, the semi-discrete Vlasov--Maxwell system of equations, is also $L^2$ stable, either conserving or decaying the $L^2$ norm.
This $L^2$ stability leads to a discrete analogue of the second Law of Thermodynamics for the semi-discrete Vlasov--Maxwell system of equations, with numerical diffusion arising as a production of entropy in our discrete system.
Although we did not analytically prove a discrete second Law of Thermodynamics for the semi-discrete Fokker--Planck equation, we will compare the entropy behavior between collisionless and collisional simulations in Chapter~\ref{ch:Benchmarks}, and show that the collisionless entropy production is small compared to the collisional entropy production. 
Because many of these properties, especially for the semi-discrete Vlasov--Fokker--Planck equation, only depended on the numerical flux function being Godunov and not a specific form of the numerical flux function, we can imagine further flexibility in terms of the mathematical formulation of the scheme.
For example, we could extend the recovery procedure to handle the collisionless and drag components of the discretization and still retain the properties proved.
\item Having specified a spatial discretization and constructed the semi-discrete VM-FP system of equations, we only require an ordinary differential equation integrator for the time integration to complete the discretization and integrate the equation system in time.
Example integrators include strong-stability preserving Runge--Kutta methods, e.g., a three-stage third order method,
\begin{align*}
\begin{aligned}
  f^{(1)} &= \mathcal{F}\left(f^{n},t^{n}\right), \\ f^{(2)} &=
  \frac{3}{4}f^{n} + \frac{1}{4}\mathcal{F}\left(f^{(1)},t^n+\Delta
  t\right),\\ f^{n+1} &= \frac{1}{3}f^{n} +
  \frac{2}{3}\mathcal{F}\left(f^{(2)},t^n+\Delta t/2\right),
\end{aligned}
\end{align*}
with $\mathcal{F}$ defining the complete evaluation of the semi-discrete VM-FP system of equations, Eqns.\thinspace(\ref{eq:dis-weak-form}) and (\ref{eq:dis-weak-FP}) for the Vlasov--Fokker--Planck equation, and Eqns.\thinspace (\ref{eq:dis-weak-B}) and (\ref{eq:dis-weak-E}) for Maxwell's equations.
These explicit time integrators have Courant-Friedrichs-Lewy constraints on the size of the time-step,
\begin{align*}
    CFL^{collisionless} & = \sum_{i=1}^{CDIM} \max_{\mathcal{T}}\left |\frac{v_i }{\Delta x_i} \right | + \sum_{j=1}^{VDIM} \max_{\mathcal{T}}\left |\frac{q_s/m_s \thinspace (E_h + v \times B_h)_j}{\Delta v_j} \right |, \\
    CFL^c & = \sum_{j=1}^{VDIM} \max_{\mathcal{T}} \left | \nu \frac{(v - u_h)_j}{\Delta v_j} \right |+ \sum_{j=1}^{VDIM} \max_{\mathcal{T}} \left |\nu \frac{T_h}{m_s} \frac{1}{(\Delta v_j)^2} \right |, \\
    \Delta t & (CFL^{collisionless} + CFL^c) \leq \frac{1}{2p + 1},
\end{align*}
for the Vlasov--Fokker--Planck equation, and
\begin{align*}
    c \frac{\Delta t}{\Delta x} \leq \frac{1/CDIM}{2p + 1},
\end{align*}
for Maxwell's equations.
Here, we have abbreviated the number of configuration space dimensions as $CDIM$ and the number of velocity space dimensions as $VDIM$.
The more restrictive of the two conditions tells us the maximum stable time-step, and completes the prescription for the numerical integration of the VM-FP system of equations in space and time.
 \end{enumerate}
 
Thus, we can now move to a discussion of how to evaluate the discrete scheme, i.e., how do we turn the math into code, an algorithmic formulation of the discrete scheme that allows one to actually perform numerical experiments.
Throughout this summary, we have emphasized the requirements that components of the discrete scheme be constructed \emph{consistently}, e.g., computing velocity moments using weak equality.
This emphasis is not without merit.
When we first described plasmas as rich in their underlying physics in Chapter~\ref{ch:Introduction}, we alluded to the fact that important physics properties are implicit to the underlying equation system.
For example, we are discretizing the Vlasov--Fokker--Planck equation for the evolution of the particle distribution function, but just as important is that velocity moments such as the zeroth, mass, and second, energy, obey conservation equations.
To actually retain these properties that we painstakingly proved in this Chapter, we will find that the ultimate algorithmic formulation of the scheme requires a comparable amount of precision to the amount of mathematical care that was taken when deriving the discrete scheme.
%Chapter 3

\renewcommand{\thechapter}{3}
\epigraph{Some of the material in this chapter has been adapted from \citet{Juno:2018}, \citet*{Hakim:2019}, and \citet{HakimJuno:2020}.}{}
\chapter{From Math to Code: Efficient Implementation of DG for \\the Vlasov--Maxwell--Fokker--Planck System of Equations}\label{ch:ImplementationDGFEM}

It is now time to undertake the task of translating the discrete scheme described in Chapter~\ref{ch:DGFEM} into an algorithm which can be implemented in a code, in this case, the \gke~simulation framework.
As part of our derivation of the discrete scheme, there were many components of the scheme we left deliberately abstract as they were unnecessary for describing the numerical method mathematically and proving properties of the discretization of the VM-FP system of equations.
We have a long to-do list for converting Eqns.\thinspace(\ref{eq:dis-weak-form}) and (\ref{eq:dis-weak-FP}) for the Vlasov--Fokker--Planck equation, and Eqns.\thinspace (\ref{eq:dis-weak-B}) and (\ref{eq:dis-weak-E}) for Maxwell's equations, into code.

We have restricted ourselves to basis sets of polynomials as part of the proofs of the various conservation properties that our discrete scheme retains from the continuous system, such as conservation of mass and energy, but we have made no mention yet of what specific form this polynomial basis takes.
We likewise must now evaluate these integrals in the discrete weak forms of the VM-FP system of equations in some fashion, including a potential transformation from a more convenient computational space to the physical domain on which the equations are defined.
Finally, in tandem with actually performing the integrals in the discrete weak forms, we must determine algorithmically how to compute the various components of the scheme, such as velocity moments for the coupling between Maxwell's equations and the Vlasov--Fokker--Planck equation and the recovery of the distribution function for the Fokker--Planck equation.
With a prescription for how to perform these operations, we will then be able to bring the whole algorithm together and evaluate computationally the spatial discretization.
Combined with the time discretization described in Section~\ref{sec:TimeDiscretization}, we will then have completed the conversion from the mathematical machinery described in Chapter~\ref{ch:DGFEM} to the computational machinery required to perform the numerical integration of the VM-FP system of equations in our simulation framework \gke.

\section{Polynomial Bases in 1D: Nodal versus Modal}\label{sec:1DBases}

Even in one dimension, there is tremendous freedom in the definition of the polynomial basis.
The definition of the function space, $\mathbb{P}^p$, only restricts us to polynomials of, at most, order $p$.
We could, for example, take our basis set to be simply
\begin{align}
    \psi_k(x) = x^k, \quad k = 0,\dots,p, \quad x \in [-1,1], \label{eq:bad1DBasis}
\end{align}
where we have defined the polynomials on the interval $[-1,1]$ for convenience.

We could define the polynomials with respect to the local grid cell immediately, as we did in the brief one dimensional DG example in Section~\ref{sec:L2ErrorMinimization} in \eqr{\ref{eq:1DPolynomialDGExample}} wherein the linear polynomial included the local grid cell volume and cell center coordinate.
However, as we will show in Section~\ref{sec:ComputationToPhysical}, we can always transform our computational domain to the physical domain on which the equations are defined.
We will find certain properties of the polynomial basis are ultimately more intuitive by defining the polynomials on a reference element, in this case the element $[-1, 1]$ in one dimension.
By defining the polynomials on a reference element, we also afford ourselves greater flexibility, especially with respect to the physical coordinate system and the overall structure of the grid on which the physical domain is defined.

So, with these caveats about defining the polynomial basis on a reference element aside, the basis set defined in \eqr{\ref{eq:bad1DBasis}} seems perfectly acceptable.
Indeed, \eqr{\ref{eq:bad1DBasis}} is, mathematically, a completely reasonable basis.
We could employ this basis and the basis would lead to the discrete scheme retaining all the properties of the continuous system proved in Chapter~\ref{ch:DGFEM} and the discrete scheme would still be $L^2$ stable.
However, the basis defined in \eqr{\ref{eq:bad1DBasis}} is a very bad choice for our basis expansion because the basis has serious computational issues.

To see why \eqr{\ref{eq:bad1DBasis}} forms a bad basis computationally, consider the following operation that will be required as part of our discretization,
\begin{align}
    \int_{K_j} \pfrac{f_h(\mvec{z}, t)}{t} w_\ell(\mvec{z}) \dz = \sum_k \frac{df_k(t)}{dt} \int_{K_j} w_k(\mvec{z}) w_\ell(\mvec{z}) \dz = \mvec{M} \frac{d \mvec{f}}{dt}, \label{eq:firstStepMassMatrix}
\end{align}
where the matrix $\mvec{M}$ has entries
\begin{align}
    M_{k\ell} = \int_{K_j} w_k (\mvec{z}) w_\ell(\mvec{z}) \dz, \label{eq:firstMassMatrixDefinition}
\end{align}
and we have added back in the spatial dependence to the basis functions to make the meaning of evaluation of entries of the matrix $\mvec{M}$ more clear.
In other words, each combination of basis functions, integrated over the cell $K_j$, produces a matrix with size $N_p \times N_p$, where $N_p$ is the number of basis functions in the expansion within a cell.
This matrix, \eqr{\ref{eq:firstMassMatrixDefinition}}, is often called the mass matrix in the DG and finite element literature \citep{Hesthaven:2007}.
Note that \eqr{\ref{eq:firstStepMassMatrix}} implies that we will require the inverse of the mass matrix, $\mvec{M}$, to ultimately discretize the system of ordinary differential equations for $\mvec{f}$, the vector of expansion coefficients within a cell.

Now, this mass matrix in one dimension on the reference cell is simply
\begin{align}
    M_{k\ell} & = \int_{-1}^1 \psi_k (x) \psi_\ell(x) \thinspace dx.
\end{align}
To make this example concrete, for the basis defined in \eqr{\ref{eq:bad1DBasis}}, consider the mass matrix in one dimension for polynomial order four:
\begin{align}
  M_{k\ell} = \int_{-1}^1 x^k x^\ell \thinspace dx
  = \begin{pmatrix} 
    2 & 0 & \frac{2}{3} & 0 & \frac{2}{5} \\ 0 &
    \frac{2}{3} & 0 & \frac{2}{5} & 0 \\ \frac{2}{3} & 0 & \frac{2}{5}
    & 0 & \frac{2}{7} \\ 0 & \frac{2}{5} & 0 &\frac{2}{7} & 0
    \\ \frac{2}{5} & 0 & \frac{2}{7} & 0 & \frac{2}{9}
  \end{pmatrix}. \label{eq:badMassMatrix1D}
\end{align}
Perhaps unremarkable, but let us examine the condition number for the matrix in \eqr{\ref{eq:badMassMatrix1D}},
\begin{align}
    \kappa^\infty(\mvec{M}) \defeq || \mvec{M}^{-1} ||_\infty || \mvec{M} ||_\infty = \frac{8211}{16},
\end{align}
where $|| \cdot ||_\infty$ is the $L^\infty$ matrix norm\footnote{Note the condition number can be defined with any suitable matrix norm, such as the Frobenius norm,
\begin{align*}
    || \mvec{A} || = \sqrt{\sum_{k=1}^N \sum_{\ell=1}^N |A_{kl}|^2}.
\end{align*}
},
\begin{align}
    || \mvec{A} || = \max_{1\leq k \leq N} \sum_{\ell=1}^N |A_{kl}|.
\end{align}

The condition number measures the sensitivity of the solution to small changes in the initial data.
Because we require the inverse of the mass matrix, $\mvec{M}$, before we can discretize the system of ordinary differential equations for the time evolution of $\mvec{f}$ a large condition number for the mass matrix is very bad.
A rough rule of thumb is that for $\kappa^\infty(\mvec{A}) = 10^n$, we expect to lose $n$ digits of accuracy due to a loss of precision from the inversion of the matrix \citep{NumRecipes:2007}. 
So, for the matrix in \eqr{\ref{eq:badMassMatrix1D}}, we would expect to lose $\log_{10}(\kappa^\infty(\mvec{M})) \sim 2.7$ digits of accuracy.
As we go to higher and higher polynomial order with the simple monomial basis defined in \eqr{\ref{eq:bad1DBasis}}, the loss of accuracy becomes quite high.

A standard means of ameliorating this issue of poor conditioning of the component matrices, such as the mass matrix, in the DG discretization is to perform a Gram-Schmidt orthogonalization process on \eqr{\ref{eq:bad1DBasis}}.
We would thus obtain a basis of orthogonal polynomials, which can then be made orthonormal.
As part of the Gram-Schmidt procedure, we first define a projection operator,
\begin{align}
    \textrm{proj}_{\upsilon}(\psi) = \frac{(\psi, \upsilon)_{L^2}}{(\upsilon, \upsilon)_{L^2}} \upsilon,
\end{align}
where the $L^2$ inner product, $(\cdot, \cdot)_{L^2}$, is the inner product we have been continually employing,
\begin{align*}
    (\psi, \upsilon)_{L^2} = \int_{-1}^1 \psi(x) \upsilon(x) \thinspace dx,
\end{align*}
with natural generalizations to higher dimensions.
We then use this projection operator to transform the monomial basis in \eqr{\ref{eq:bad1DBasis}} into a set of orthogonal polynomials. Proceeding sequentially through the polynomial set,
\begin{align} \begin{aligned}
    \upsilon_0 & = \psi_0 = 1, \\
    \upsilon_1 & = \psi_1 - \textrm{proj}_{\upsilon_0}(\psi_1) = x, \\
    \upsilon_2 & = \psi_2 - \textrm{proj}_{\upsilon_0}(\psi_2) - \textrm{proj}_{\upsilon_1}(\psi_2) = \frac{3x^2 - 1}{3}, \\
    \upsilon_3 & = \psi_3 - \textrm{proj}_{\upsilon_0}(\psi_3) - \textrm{proj}_{\upsilon_1}(\psi_3) - \textrm{proj}_{\upsilon_2}(\psi_3) = \frac{x (5x^2 - 3)}{5} \\
    \upsilon_4 & = \psi_4 - \textrm{proj}_{\upsilon_0}(\psi_4) - \textrm{proj}_{\upsilon_1}(\psi_4) - \textrm{proj}_{\upsilon_2}(\psi_4) - \textrm{proj}_{\upsilon_3}(\psi_4)= \frac{35 x^4 - 30 x^2 + 3}{35}.
    \end{aligned}
\end{align}
This procedure generalizes to higher polynomial orders as we might expect, with
\begin{align}
    \upsilon_k = \psi_k - \sum_{j=1}^{k-1} \textrm{proj}_{\upsilon_{j-1}}(\psi_k).
\end{align}
We can make these polynomials orthonormal using
\begin{align}
    \hat{\upsilon} = \frac{\upsilon}{\sqrt{(\upsilon, \upsilon)_{L^2}}},
\end{align}
i.e., dividing by the $L^2$ norm of the polynomials.
This procedure gives us the following set of orthonormal polynomials for the one dimensional, $p=4$, basis,
\begin{align}\begin{aligned}
    \hat{\upsilon}_0 & = \frac{1}{\sqrt{2}},  \\
    \hat{\upsilon}_1 & = \sqrt{\frac{3}{2}} x, \\
    \hat{\upsilon}_2 & = \sqrt{\frac{5}{8}} (3x^2 - 1), \\
    \hat{\upsilon}_3 & = \sqrt{\frac{7}{8}} (5x^3 - 3 x), \\
    \hat{\upsilon}_4 & = \frac{3}{8\sqrt{2}} (35 x^4 - 30 x^2 + 3).
    \end{aligned} \label{eq:1DOrtho}
\end{align}
Because these polynomials are orthonormal,
\begin{align}
    \int_{-1}^1 \hat{\upsilon}_k \hat{\upsilon}_\ell \thinspace dx = \delta_{k\ell},
\end{align}
\eqr{\ref{eq:badMassMatrix1D}} reduces to
\begin{align}
    \mvec{M} = \overleftrightarrow{\mvec{I}},
\end{align}
the identity matrix, whose condition number is trivially $\kappa^\infty(\mvec{M}) = 1$.

As an aside, we can gain intuition for why the conditioning of the mass matrix improves so dramatically when employing orthonormal polynomials by examining the behavior of our two choice of basis sets on the interval $[-1, 1]$, shown in Figure~\ref{fig:basis-comp}. 
We can understand the poor conditioning of \eqr{\ref{eq:badMassMatrix1D}} because the monomial basis defined in \eqr{\ref{eq:bad1DBasis}} becomes less linearly independent as we go to higher order, i.e., the higher order polynomials become indistinguishable from each other, implying that the representation is more sensitive to changes in the solution.
In other words, we have trouble actually obtaining an accurate representation of the solution from the monomial basis because of the behavior of the monomials on the interval $[-1,1]$.
In contrast, the orthonormal basis maintains good coverage of the interval as we increase the order of the polynomials, and thus we expect the accuracy of the representation continually improves as we go to higher and higher order.

\begin{figure}
    \centering
    \includegraphics[width=\textwidth]{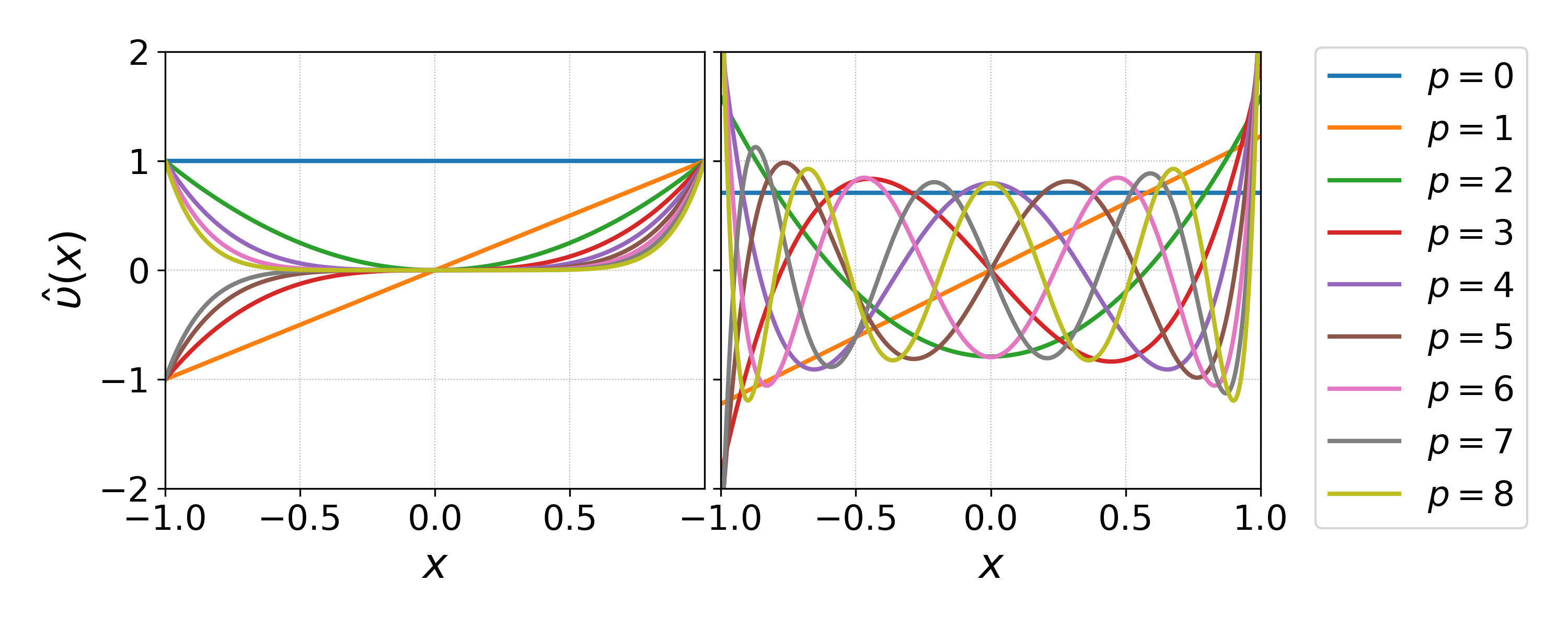}
    \caption{The simple monomial basis (left) defined in \eqr{\ref{eq:bad1DBasis}} and the orthonormal basis obtained by a Gram-Schmidt orthogonolization (orthonormalization) applied to the monomial basis (right). We can see that in the limit of high polynomial order, the monomial basis becomes less linearly independent, i.e., the higher order polynomials are essentially indistinguishable. On the other hand, the orthonormal basis maintains better ``coverage'' of the space on the interval from $[-1,1]$ so that it is easy to imagine why higher order orthonormal polynomials do actually improve the accuracy of the representation.}
    \label{fig:basis-comp}
\end{figure}

Also, we note the similarities between the orthonormal polynomials defined and Legendre polynomials, which are an orthogonal set of polynomials defined on the interval $[-1,1]$ with an identical inner product to the inner product we have been employing, \eqr{\ref{eq:L2innerproduct}}. 
Legendre polynomials are normalized to be equal to $\pm 1$ at the edges of the interval. Although Legendre polynomials are orthogonal and very similar to the orthogonal polynomials we first found with our Gram-Schmidt procedure, they are not orthonormal,
\begin{align}
    \int_{-1}^1 P_n(x) P_m(x) \thinspace dx = \frac{2}{2n + 1} \delta_{mn}.
\end{align}

Importantly, we have our first instance justifying out choice to define the polynomials on a reference element $[-1,1]$: we are performing a Gram-Schmidt orthogonalization (orthonormalization) process on polynomials defined on this interval.
Thus, the polynomials will be orthogonal and orthonormal on this interval, and potentially only on this interval.
We will see that this restriction does not cause any issues for the purposes of transforming from the reference element, or computational space, to physical space in Section~\ref{sec:ComputationToPhysical}.

The set of one dimensional orthonormal polynomials derived in this chapter, e.g., \eqr{\ref{eq:1DOrtho}} for polynomial order four, define what is called a \emph{modal} basis for our DG discretization.
This terminology follows from the fact that in the projection of a quantity of interest onto our basis set, we are projecting onto a set of modes.
An alternative prescription is called a \emph{nodal} basis, wherein the basis set is defined by a set of polynomials whose values are known at nodes.
In other words, a basis such as
\begin{align}
    f(x, t) \approx f_h(x, t) \defeq \sum_{k=0}^{N_p-1} f_k(\xi_k, t) \ell_k(x), \label{eq:1DNodal}
\end{align}
where $\ell_k$ are the Lagrange interpolating polynomials,
\begin{align}
     \ell_k(x) = \prod_{j=0, j\neq k}^{N_p-1} \frac{x - \xi_j}{\xi_k - \xi_j},
\end{align}
and $\xi_k$ are the $k$ nodes by which the polynomials are defined.
In other words, in this basis set, the polynomials take the value of one at one node and zero at all other nodes, thus the coefficients $f_k$ in \eqr{\ref{eq:1DNodal}} are known at the nodes $\xi_k$.

Just as \eqr{\ref{eq:bad1DBasis}} was related mathematically to the orthonormal, modal basis by the Gram-Schmidt orthogonalization (orthonormalization) process, so too do the one dimensional modal and nodal bases have a mathematical connection.
Using the Vandermonde matrix,
\begin{align}
    \mathcal{V}_{k\ell} = \hat{\upsilon}_\ell(\xi_k),
\end{align}
i.e., the matrix whose entries are each of the $\ell$ orthonormal polynomials evaluated at the nodes $\xi_k$, we can transform the coefficients in the \emph{modal} basis to the coefficients in the \emph{nodal} basis,
\begin{align}
    \mathcal{V}_{k\ell} f_\ell(t) = f_k(\xi_k, t).
\end{align}
And just as with \eqr{\ref{eq:bad1DBasis}}, both the modal and nodal bases are perfectly mathematically acceptable basis sets for implementing the DG scheme for the VM-FP system of equations described in Chapter~\ref{ch:DGFEM}, but they have quite different computational properties.
Before we can explore the full extent of the computational consequences for a modal versus a nodal basis set, we should first discuss the generalization of these basis sets to higher dimensions.

\section{Polynomial Bases in Higher Dimensions: \\The ``Curse of Dimensionality'' and Serendipitous Basis Choices}\label{sec:HigherDimensionBases}

From the beginning, we have been interested in the numerical integration of an equation system which is high-dimensional, up to six dimensions plus time.
This high dimensionality of the VM-FP system of equations presents a special set of challenges for the design and implementation of our numerical method.
The ``curse of dimensionality,'' the exponential cost scaling of a numerical method with the dimensionality of the problem, is not a ``curse'' to be taken lightly.
This exponentially increasing cost scaling with dimensionality is in fact one of the principal reasons for the popularity of the particle-in-cell method discussed in Chapter~\ref{ch:Introduction}, as it is argued that the integration of particles on a three-dimensional grid, instead of the integration of the particle distribution function on a six-dimensional grid, is inevitably more cost effective.
Of course, we have strong motivation for the direct discretization approach, so we instead want to focus on whether this burden of cost can be overcome.

The standard higher dimensional generalization of the one dimensional bases defined in Section~\ref{sec:1DBases} is a tensor basis constructed from a tensor product of the one dimensional basis sets for each dimension of interest.
For example, in two dimensions, the generalization of the monomial basis is simply
\begin{align}
    \mathbb{Q}^p_2 = \spn_{0\leq m, n \leq p} \{ x^m y^n \}.
\end{align}
Due to the nature of the tensor product, the number of basis functions within a cell scales like $(p+1)^d$, exactly the exponential scaling we predicted at the beginning of this section.
We seek reductions then of this tensor product basis.

The first reduction we consider is known as the Serendipity basis set \citep{Arnold:2011}. 
The Serendipity basis set is obtained by dropping all monomial terms which have ``super-linear'' degree greater than the specified polynomial order $p$.
For example, for the piecewise quadratic, two dimensional, Serendipity basis expansion, we would have
\begin{align}
    \mathcal{S}_2^2 = \{1, x, y, xy, x^2, y^2, x^2 y, x y^2, \stkout{x^2 y^2} \},
\end{align}
because the ``super-linear'' degree of $x^2 y^2$ is four, which is greater than the specified polynomial order of two. 
We could then apply the appropriate higher dimensional generalization of the Gram-Schmidt orthonormalization procedure described in the previous section, Section~\ref{sec:1DBases}.
In two dimensions, this generalization of the inner product would be
\begin{align}
    (\upsilon, \psi)_{L^2} = \int_{-1}^1 \int_{-1}^1 \upsilon(x,y) \psi(x,y) \thinspace dxdy,
\end{align}
so that we would find the two dimensional, piecewise quadratic, orthonormal, modal, Serendipity basis to be
\begin{align}\begin{aligned}
  \hat{\upsilon}_0(x,y) &= \frac{1}{2}, \\
  \hat{\upsilon}_1(x,y) &= \frac{\sqrt{3}x}{2}, \\
  \hat{\upsilon}_2(x,y) &= \frac{\sqrt{3}y}{2}, \\
  \hat{\upsilon}_3(x,y) &= \frac{3xy}{2}, \\
  \hat{\upsilon}_4(x,y) &= \frac{\sqrt{5}(3x^2-1)}{4}, \\
  \hat{\upsilon}_5(x,y) &= \frac{\sqrt{5}(3y^2-1)}{4}, \\
  \hat{\upsilon}_6(x,y) &= \frac{\sqrt{15}(3x^2-1)y}{4}, \\
  \hat{\upsilon}_7(x,y) &= \frac{\sqrt{15}(3y^2-1)x}{4}. \\
\end{aligned}\label{eq:2DSerendipQuadratic}
\end{align}

The general scaling of the Serendipity basis set is given by
\begin{align}
N_p = \sum_{i=0}^{\min(d, p/2)} 2^{n-i} \binom{d}{i} \binom{p - i}{i},
\end{align}
where $N_p$ is the number of polynomials, $d$ is the dimensionality of the basis set, and $p$ is the polynomial order.
This particular reduced basis set has been extensively studied in the literature, and found to have the same formal convergence order as the tensor basis, though the generalization of the Serendipity basis to unstructured grids requires care as arbitrary refinements of an unstructured grid will destroy the convergence order of the Serendipity expansion \citep{Arnold:2002}.
By convergence order, we mean the rate of convergence to the true solution of the continuous system in the limit that the grid spacing goes to zero.
So a second order method corresponds to a method where the errors decrease as $(\Delta x)^2$ as $\Delta x \rightarrow 0$.
Although we have not said so explicitly up to this point, all the work of this thesis uses structured grids, specifically structured quadrilaterals.

We can consider a further reduction on top of the Serendipity basis to drop all monomials of \emph{total} degree greater than the polynomial order specified, which we call the maximal order basis set. 
For this reduced basis set, we would only retain polynomials zero through five in \eqr{\ref{eq:2DSerendipQuadratic}}, since polynomials six and seven have total degree three.
The general scaling of the maximal order basis set is
\begin{align}
    N_p = \frac{(p+d)!}{p!d!}.
\end{align}

Elsewhere in the finite element literature, these three basis sets, the tensor basis, Serendipity, and what we are calling maximal order, are sometimes abbreviated as the $\mathbb{Q}, \mathcal{S},$ and $\mathbb{P}$ spaces respectively.
Like the Serendipity basis set, the maximal order basis set has been the subject of a large number of studies to examine its convergence order and accuracy relative to the tensor basis.
While maintaining the same convergence order, the maximal order basis set is generally found to be less accurate, and this basis can have further detrimental consequences for the physicality of the solution.
For example, \citet{Cheng:2013} found the maximal order basis to have more serious issues with artificial dissipation compared to the tensor basis in a Vlasov--Poisson study using the discontinuous Galerkin method.

For reference, the number of degrees of freedom in a cell for a variety of polynomial orders and up to six dimensions for the three basis sets, tensor, Serendipity, and maximal order, is included in Tables \ref{table:TensorTable}, \ref{table:SerendipityTable}, and \ref{table:MaxOrderTable}, respectively.
\begin{table}[!htb]
\begin{center}
\begin{tabular}{| l | l | l | l | l | l | l | l |}
\hline
$\mathbb{Q}$ & \textbf{Polynomial Order} & 1 & 2 & 3 & 4  & 5 & 6  \\ \hline
\textbf{Dimension} & $(p+1)^d$& & & & & & \\ \hline
2 & & 4 & 9 & 16 & 25 & 36 & 49  \\ \hline
3 & & 8 & 27 & 64 & 125 & 216 & 343 \\ \hline
4 & & 16 & 81 & 256 & 625 & 1296 & 2401 \\ \hline
5 & & 32 & 243 & 1024 & 3125 & 7776 & 16807 \\ \hline
6 & & 64 & 729 & 4096 & 15625 & 46656 & 117649 \\ \hline
\end{tabular}
\caption{Number of degrees of freedom internal to a cell in the tensor product basis set.}
\label{table:TensorTable}
\end{center}
\end{table}
\begin{table}[!htb]
\begin{center}
\begin{tabular}{| l | l | l | l | l | l | l | l |}
\hline
$\mathcal{S}$ & \textbf{Polynomial Order} & 1 & 2 & 3 & 4 & 5 & 6 \\ \hline
\textbf{Dimension} & $\sum_{i=0}^{\min(d, p/2)} \binom{d}{i} \binom{p - i}{i} $& & & & & & \\ \hline
2 & & 4 & 8 & 12 & 17 & 23 & 30 \\ \hline
3 & & 8 & 20 & 32 & 50 & 74 & 105 \\ \hline
4 & & 16 & 48 & 80 & 136 & 216 & 328 \\ \hline
5 & & 32 & 112 & 192 & 352 & 592 & 952 \\ \hline
6 & & 64 & 256 & 448 & 880 & 1552 & 2624 \\ \hline
\end{tabular}
\caption{Number of degrees of freedom internal to a cell in the Serendipity basis set.}
\label{table:SerendipityTable}
\end{center}
\end{table}
\begin{table}[!htb]
\begin{center}
\begin{tabular}{| l | l | l | l | l | l | l | l |}
\hline
$\mathbb{P}$ & \textbf{Polynomial Order} & 1 & 2 & 3 & 4  & 5 & 6 \\ \hline
\textbf{Dimension} & $\frac{(p+d)!}{p!d!}$& & & & & & \\ \hline
2 & & 3 & 6 & 10 & 15 & 21 & 28 \\ \hline
3 & & 4 & 10 & 20 & 35 & 56 & 84 \\ \hline
4 & & 5 & 15 & 35 & 70 & 126 & 210 \\ \hline
5 & & 6 & 21 & 56 & 126 & 252 & 462 \\ \hline
6 & & 7 & 28 & 84 & 210 & 462 & 924 \\ \hline
\end{tabular}
\caption{Number of degrees of freedom internal to a cell in the maximal order basis set.}
\label{table:MaxOrderTable}
\end{center}
\end{table}
We can see the aforementioned exponential increase in the number of polynomials, and thus the cost, with the tensor product basis in Table~\ref{table:TensorTable}.
We note the rather dramatic reduction in the number of degrees of freedom, especially for the higher dimensional cases, for the Serendipity and maximal order basis sets.

We conclude this section with a brief discussion of how these reduced modal basis sets in higher dimensions can also be converted to their nodal counterparts.
While there is no known nodal configuration for the maximal order basis, there are nodal configurations for the Serendipity basis with potentially favorable computational properties, such as the nodal configuration in one, two, and three dimensions (1D, 2D, 3D) discussed in \citet{Arnold:2011} and shown in Figure~\ref{fig:SerendipitySchematic}.
In this case, we will have a unique polynomial for each node, which has variation throughout the entire multi-dimensional reference cell, that takes the value of one at one node and zero at the other nodes.
\begin{figure}
    \centering
    \includegraphics[width=\textwidth]{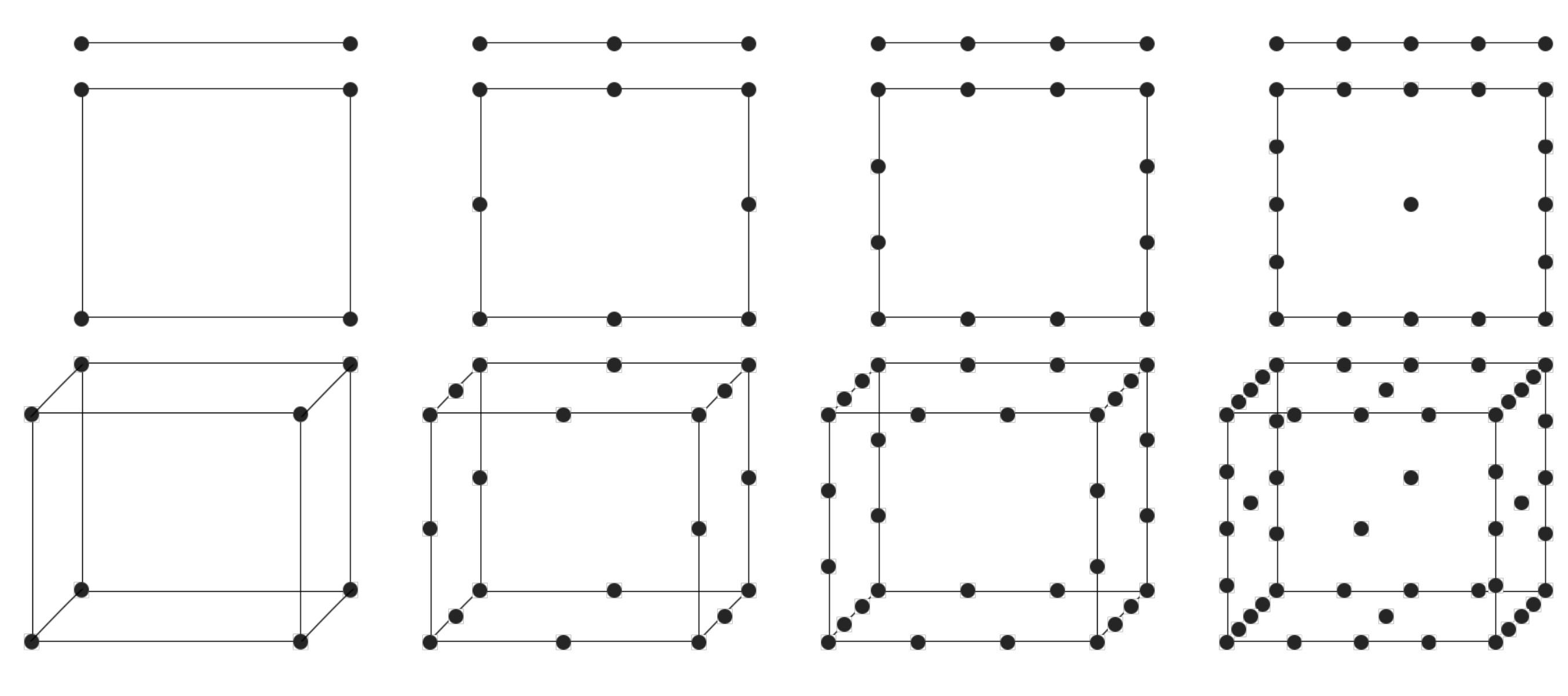}
    \caption{Schematic drawing of the nodal locations for the Serendipity basis in 1D (top), 2D (middle), and 3D (bottom) for polynomial orders one (far left), two (middle left), three (middle right), and four (far right).}
    \label{fig:SerendipitySchematic}
\end{figure}

This particular nodal layout is constructed such that every higher dimensional reference quadrilateral element is built from the lower dimensional reference quadrilateral elements, so that the lower dimensional faces of a reference quadrilateral element also form a unisolvent expansion, i.e., the polynomials local to the face form a complete basis of the solution space. 
For example, consider the pictorial representation of the 3D reference element. 
Each 2D face of the reference 3D element is exactly the 2D reference element for that particular polynomial order. 
This same recursive approach can be applied to higher dimensions as well\footnote{This fact is true in general, but higher polynomial orders may modify the lower dimensional reference elements such that the recursive algorithm is not quite as obvious as the one presented here.
Just as polynomial order four introduces an interior node to a reference 2D element, so can higher polynomial orders introduce interior nodes to higher dimensional reference elements which would have to be taken into account in the recursive generation of the reference element. 
For up to polynomial order four though, every higher dimensional object can be easily generated as described, with 2D reference elements making up the faces of a 3D reference element, 3D reference elements making up the faces of a 4D reference element, and so on. 
Considering that the Serendipity basis in four, five, and six dimensions, with polynomial order four, involves the solution of a large number of degrees of freedom per cell, we will not consider further extensions of this recursive algorithm due to the same performance and cost considerations that motivated the use of the Serendipity basis---we seek to avoid evolving thousands of degrees of freedom per cell.}, with the reference 4D element being comprised of reference 3D elements for each of the 4D element's eight 3D faces, a reference 5D element consisting of a reference 4D element for all ten 4D faces of a 5D element, and so on. 
This approach has the advantage of greatly simplifying surface integral calculations. Since every higher dimensional element is recursively generated from lower dimensional elements, every face of a higher dimensional element, the 2D faces in 3D or the 4D faces in 5D, forms a unisolvent expansion for that surface. 
We thus only require the nodal information local to that face and can reduce the number of multiplications in the evaluation of the surface integrals by a somewhat sizable fraction. In 5D for instance, to advance the solution of the distribution function in time, one performs one 5D volume integral and ten 4D surface integrals, so the 4D surface integrals can be performed with this reduced number of degrees of freedom.

So, we have mitigated the malediction normally imposed on us by solving a higher dimensional partial differential equation system like the VM-FP system of equations by choosing reduced basis sets such as the Serendipity and maximal order basis set.
In addition, we have prescriptions for both nodal and modal bases for the Serendipity basis set.
Having defined our basis sets, we will now move to the actual evaluation of the integrals in the discrete weak form, first focusing on the transformation from the reference elements on which we have chosen to define the polynomials to the actual physical domain on which the VM-FP system of equations is defined, and then moving to a procedure to evaluate the integrals in totality.
The latter procedure will prove especially subtle and lead to the two most critical algorithmic advancements in this thesis.

\section{Transforming from Computational Space to Physical Space}\label{sec:ComputationToPhysical}

Having defined suitable polynomial basis sets for the full spectrum of dimensionality of interest, for arbitrary polynomial order, we return to an issue discussed in Section~\ref{sec:1DBases}.
We require integrals over the physical domain, i.e., a physical cell $K_j$ in phase space, such as in \eqr{\ref{eq:firstMassMatrixDefinition}}, but we have defined the polynomials on the interval $[-1,1]^d$, where $d$ is the dimensionality of the reference element.
To transform \eqr{\ref{eq:firstMassMatrixDefinition}}, we can make a change of variables,
\begin{align}
    M_{k\ell} & = \int_{K_j} w_k (\mvec{z}) w_\ell(\mvec{z}) \dz \notag \\
            & = \int_I w_k(\mvec{z}(\gvec{\eta})) w_\ell(\mvec{z}(\gvec{\eta})) \left | \frac{d\mvec{z}}{d\gvec{\eta}}\right |\thinspace d\gvec{\eta} \notag \\
            & = \int_I \hat{\upsilon}_k(\gvec{\eta}) \hat{\upsilon}_\ell(\gvec{\eta}) \left | \frac{d\mvec{z}}{d\gvec{\eta}}\right |\thinspace d\gvec{\eta}, \label{eq:transformedMassMatrix}
\end{align}
where
\begin{align}
    \left (\frac{d\mvec{z}}{d\gvec{\eta}} \right)_{ij} =  \frac{dz_i}{d\eta_j}
\end{align}
is the Jacobian matrix, and we require its determinant to perform the transformation.
In this procedure, we have transformed the basis functions $w(\mvec{z})$ defined on the physical phase space mesh to $\hat{\upsilon}(\gvec{\eta})$, the orthonormal basis set defined on the reference element $I = [-1,1]^d$, where $d$ is the dimensionality of the reference element.
We could also just as easily transform the phase space basis functions $w(\mvec{z})$ to the nodal basis defined on the reference element $I = [-1,1]^d$.

To determine the Jacobian matrix and its determinant, we must know the functional form for the change of variables from the coordinate $\mvec{z}$ to the coordinate $\gvec{\eta}$.
To take a simple example, we could transform from a uniform, structured, Cartesian grid to the reference element with the formula
\begin{align}
    \mvec{z} = \gvec{\eta} \frac{\Delta \mvec{z}}{2} + \mvec{z}_{\textrm{center}},
\end{align}
where $\Delta \mvec{z}$ is the grid spacing in each direction of phase space, and $\mvec{z}_{\textrm{center}}$ is the cell center.
The entries of the Jacobian matrix would then be
\begin{align}
    \frac{dz_i}{d\eta_j} = \frac{\Delta z_i}{2} \delta_{ij},
\end{align}
and since this matrix is diagonal, the determinant is straightforwardly
\begin{align}
    \left | \frac{d\mvec{z}}{d\gvec{\eta}}\right | = \frac{1}{2^d} \prod_{i=1}^d \Delta z_i. \label{eq:structuredJacobian}
\end{align}

The change of variables need not be so simple. 
But, so long as the Jacobian for the change of variables is known, we can map the reference element onto as complex a physical grid as we can imagine.
For example, we can construct a non-orthogonal coordinate system which follows magnetic field lines, as is done with the simulation framework the VM-FP solver is built in, \gke, for other applications \citep{Bernard:2019,Shi:2019,Mandell:2020,Bernard:2020,Francisquez:2020}.
Depending on the complexity of the Jacobian though, e.g., if the transformation itself varies in space, further modification of the integrals may be required, especially for the terms involving gradients.

Let us now, in the lead up to the next section, return to the explicit expression for the discrete weak form of the Vlasov equation, \eqr{\ref{eq:dis-weak-form}}, and attempt to reveal exactly the integrals we need to compute.
Substituting the expansions of the distribution function, $f_h$ and the phase space flow, $\gvec{\alpha}_h$, into \eqr{\ref{eq:dis-weak-form}}, we obtain
\begin{align}
    \sum_k \frac{d f_k(t)}{dt} \int_{K_j} & w_k (\mvec{z}) w_\ell(\mvec{z}) \dz + \sum_m \hat{\mvec{F}}_m (t) \cdot \oint_{\partial K_j} \mvec{n} w_\ell^- (\mvec{z}) w_m(\mvec{z})  \thinspace dS \notag \\
    & - \sum_{m,n} f_m(t) \gvec{\alpha}_n(t) \cdot \int_{K_j} \gz w_\ell (\mvec{z}) w_m(\mvec{z}) w_n (\mvec{z})\dz = 0.
\end{align}
Assuming our grid is uniform, structured, and Cartesian we can rearrange this expression using the procedure in \eqr{\ref{eq:transformedMassMatrix}}, as well as the determinant of the Jacobian matrix in \eqr{\ref{eq:structuredJacobian}}, to obtain
\begin{align}
    \sum_k & \frac{d f_k(t)}{dt} \frac{1}{2^d} \prod_{i=1}^d \Delta z_i \int_I \hat{\upsilon}_k (\gvec{\eta}) \hat{\upsilon}_\ell(\gvec{\eta}) d\gvec{\eta} + \left (\frac{1}{2^d} \prod_{i=1, i\neq j}^d \Delta z_i \right )\sum_m \hat{\mvec{F}}_m (t) \cdot \oint_{\partial I_j} \mvec{n} \hat{\upsilon}_\ell^- (\gvec{\eta}) \hat{\upsilon}_m(\gvec{\eta})  \thinspace dS \notag \\
    & - \left (\frac{1}{2^d} \prod_{i=1}^d \Delta z_i \right )\sum_{m,n} f_m(t) \gvec{\alpha}_n(t) \cdot \int_I \frac{2}{\Delta \mvec{z}} \nabla_{\gvec{\eta}} \hat{\upsilon}_\ell (\gvec{\eta}) \hat{\upsilon}_m(\gvec{\eta}) \hat{\upsilon}_n (\gvec{\eta}) d\gvec{\eta} = 0. \label{eq:inelegantVlasovDiscrete}
\end{align}
Note the slight change in notation, where we are denoting the surface $\partial I_j$ as the surface with constant $j$ dimension, where $j=x,y,z,v_x,v_y,v_z$, since the determinant of the Jacobian for the surface integral will not have the volume factor for that dimension.
In addition, we have obtained an additional factor of $2/\Delta \mvec{z}$ in transforming the gradient from $\gz$ to $\nabla_{\gvec{\eta}}$.
Importantly, this term is still a vector, and one only picks up the factor of $2/\Delta \mvec{z}$ for the particular gradient being transformed.

Since \eqr{\ref{eq:inelegantVlasovDiscrete}} must be solved for every $\hat{\upsilon}_\ell$ in our basis expansion, we can make \eqr{\ref{eq:inelegantVlasovDiscrete}} more elegant by rearranging it to be a linear system,
\begin{align}
    \frac{d f_k}{dt} = (M_{k\ell})^{-1} \left [\sum_{m} \mathcal{U}_{\ell m} \cdot \hat{\mvec{F}}_m (t) + \sum_{m,n} \mathcal{C}_{\ell m n} \cdot \gvec{\alpha}_n (t) f_m (t) \right ], \label{eq:linearVlasovExp}
\end{align}
where $(M_{k\ell})^{-1}$ is the inverse of the transformed mass matrix,
\begin{align}
    M_{k\ell} = \int_I \hat{\upsilon}_k (\gvec{\eta}) \hat{\upsilon}_\ell(\gvec{\eta}) d\gvec{\eta}, \label{eq:VlasovMassMatrix}
\end{align}
and the tensors $\mathcal{U}_{\ell m}$ and $\mathcal{C}_{\ell m n}$ are
\begin{align}
    \mathcal{U}_{\ell m} & = \frac{2}{\Delta z_j} \oint_{\partial I_j} \mvec{n} \hat{\upsilon}_\ell^- (\gvec{\eta}) \hat{\upsilon}_m(\gvec{\eta})  \thinspace dS, \label{eq:VlasovSurfaceMatrix} \\
    \mathcal{C}_{\ell m n} & = \int_I \frac{2}{\Delta \mvec{z}} \nabla_{\gvec{\eta}} \hat{\upsilon}_\ell (\gvec{\eta}) \hat{\upsilon}_m(\gvec{\eta}) \hat{\upsilon}_n (\gvec{\eta}) d\gvec{\eta}. \label{eq:VlasovVolumeMatrix}
\end{align}
A few remarks on these matrices and tensors are in order.
The first remark is the implicit sum in retaining the dot products in \eqr{\ref{eq:linearVlasovExp}}, i.e., we have to perform the surface integrals for each of the $j$ surfaces and sum over the contribution, and likewise we must sum over each contribution from the phase space flux, $\gvec{\alpha}_h$, in the volume term.
In addition, we remark that the contribution from the determinant of the Jacobian matrix has been cancelled when going from \eqr{\ref{eq:inelegantVlasovDiscrete}} to \eqr{\ref{eq:linearVlasovExp}}.
The only coordinate transform contributions that survive are the factor from transforming the gradient $\gz$ to $\nabla_{\gvec{\eta}}$, and the remaining inverse volume factor, $2/\Delta z_j$, in the surface integral for the dimension which is constant at the corresponding surface, $\partial I_j$.

While we chose to illustrate the change of coordinates and construction of the linear system with the orthonormal modal basis expansion, i.e., $\hat{\upsilon}_\ell$ for each of the $\ell$ basis functions in the expansion, we could have just as easily illustrated these transformations with the nodal basis expansion.
Importantly, a key operation we must perform to be able to construct the linear system shown in \eqr{\ref{eq:linearVlasovExp}} is to project the numerical flux function onto our basis expansion.
For example, if we employ central fluxes, then using the machinery of weak equality from Section~\ref{sec:WeakEquality}, we have
\begin{align}
    \hat{\mvec{F}} \doteq \frac{1}{2} \gvec{\alpha}_h (f_h^+ + f_h^-), \label{eq:projectCentralFluxes}
\end{align}
where the projection is done over the full basis expansion, but the phase space flux $\gvec{\alpha}_h$ and the distribution function $f^{\pm}_h$ are evaluated at the corresponding surface.

Similar manipulations which produced \eqr{\ref{eq:linearVlasovExp}} can also be performed for our semi-discrete forms of Maxwell's equations and the Fokker-Planck equation.
The essential idea is always to construct the mass matrix which multiplies the time derivative, and the two tensors which encode the spatial discretization, one for the surface integral contributions, for each surface on the reference element, and one for the volume integral contribution.
Note that in the construction of the tensor for the surface integral contributions, we must project the flux functions onto the corresponding basis set, i.e., we must project central, Eqns.\thinspace(\ref{eq:centralE})-(\ref{eq:centralB}), or upwind fluxes, Eqns.\thinspace(\ref{eq:r-e2})-(\ref{eq:r-b3}), for Maxwell's equations onto configuration space basis functions.
Likewise, we must project the two surface fluxes for the semi-discrete Fokker--Planck equation onto phase space basis functions.

The evaluation of all of these linear operations in each cell $K_j$ in phase space and $\Omega_j$ in configuration space then completes the algorithm for the spatial discretization.
To actually evaluate these linear operations though, we now need to construct these tensors for the surface integrals and volume integral by specifying how to compute the integrals in Eqns.\thinspace(\ref{eq:VlasovMassMatrix}--\ref{eq:VlasovVolumeMatrix}).
What may seem relatively straightforward belies a subtlety that is of singular consequence for the construction of the algorithm.

\section{Evaluating the Integrals: The Importance of an \\ Alias-Free Scheme}\label{sec:alias-free-scheme}

At first glance, there is nothing remarkable about the integrals which must be performed in the construction of Eqns.\thinspace(\ref{eq:VlasovMassMatrix}--\ref{eq:VlasovVolumeMatrix}).
They are products of polynomials; we could either use Gaussian quadrature of an appropriate degree, or even exactly integrate the combinations of polynomials and store the entries of the matrices and tensors defined in Eqns.\thinspace(\ref{eq:VlasovMassMatrix}--\ref{eq:VlasovVolumeMatrix}) for the Vlasov equation and the analogous matrices and tensors for the Fokker--Planck equation and Maxwell's equations.

Consider what the application of Gaussian quadrature to \eqr{\ref{eq:VlasovVolumeMatrix}} would entail.
In one dimension, the numerical integration of a function with Gaussian quadrature is done via
\begin{align}
  \int_{-1}^1 f(x) \thinspace dx \approx \sum_{i=1}^{N_q} \mathcal{W}_i f(x_i),  
\end{align}
where $\mathcal{W}_i$ and $x_i$ are the $i$ weights and abscissas for the Gaussian quadrature rule.
The extension to higher dimensions is done using a tensor product of one dimensional weights and abscissas, e.g., in two dimensions,
\begin{align}
      \int_{-1}^1 \int_{-1}^1 f(x, y) \thinspace dxdy \approx \sum_{i=1}^{N_q} \sum_{j=1}^{N_q} \mathcal{W}_i \mathcal{W}_j f(x_i, y_j).
\end{align}
An example Gaussian quadrature rule, Gauss-Legendre, is shown in Table~\ref{tab:model:quadrature}.
\begin{table}[!htb]
 \label{tab:model:quadrature}
 \begin{center}
   \begin{tabular}{| l | l | l | l |}
   \hline
    $N_q$ & $x_i$ & $\mathcal{W}_i$ & Order of Accuracy $(2 N_q - 1)$ \\ \hline 
    1 & 0 & 2 & 1\\ \hline
    2 & $\pm\frac{1}{\sqrt{3}}$ & 1 & 3\\ \hline
    \multirow{2}{*}{3} &
     0 & $\frac{8}{9}$ & \\ & $\pm\sqrt{\frac{3}{5}}$ &$
     \frac{5}{9}$ & 5 \\ \hline
    \multirow{2}{*}{4} &
     $\pm\sqrt{\frac{3}{7}-\frac{2}{7}\sqrt{\frac{6}{5}}}$ &
     $\frac{18+\sqrt{30}}{36}$ & \\ &
     $\pm\sqrt{\frac{3}{7}+\frac{2}{7}\sqrt{\frac{6}{5}}}$&
     $\frac{18-\sqrt{30}}{36}$ & 7 \\ \hline
    \multirow{3}{*}{5} & 0 &
     $\frac{128}{225}$ & \\ &
     $\pm\frac{1}{3}\sqrt{5-2\sqrt{\frac{10}{7}}}$ &
     $\frac{322+13\sqrt{70}}{900}$ & \\ &
     $\pm\frac{1}{3}\sqrt{5+2\sqrt{\frac{10}{7}}}$ & 
     $\frac{322-13\sqrt{70}}{900}$ & 9 \\ \hline
   \end{tabular}
   \caption{The weights and abscissas for the Gauss-Legendre quadrature rule.  The nodes (abscissas) are the roots
   of the Legendre polynomial $P_{N_q}(x)$ and the weights
   $\mathcal{W}_i=2/[(1-x_i^2)(P_{N_q}'(x_i))^2]$ \citep{Abramowitz:1985}.}
 \end{center}
\end{table}
To perform Gaussian quadrature on integrals such as \eqr{\ref{eq:VlasovVolumeMatrix}}, we require a tensor product of $N_q$ quadrature points in each direction for every dimension we wish to integrate.
This approach will integrate exactly monomials of a particular order, e.g., $2N_q - 1$ for Gauss-Legendre or $2N_q - 3$ for Gauss-Lobatto, regardless of the dimension in which the monomial varies. 

Even with the added accuracy of Gauss-Legendre,  this strategy quickly becomes untenable for the same reason the tensor product basis is prohibitively expensive for solving the VM-FP system of equations: the ``curse of dimensionality.'' For example, consider integrating the volume term in five dimensions with second order polynomials. Naively, one expects this to require the integration of monomials with degree $3p = 6$ in each dimension, because both $\gvec{\alpha_h}$ and $f_h$ have polynomial expansions, thus requiring at least 4 quadrature points in each dimension, or a total of $4^5 = 1024$ quadrature points, to avoid under-integrating the volume term in \eqr{\ref{eq:VlasovVolumeMatrix}}. Given that the scaling of the computation of the volume integral in a cell is $\mathcal{O}(N^{tot}_q N_p)$, where $N_q^{tot}$ is the total number of quadrature points, the number of operations per phase space cell becomes quite large for modest polynomial orders in high dimensions.

Leveraging the fact that \eqr{\ref{eq:VlasovVolumeMatrix}} is just a triple product of polynomials and exactly integrating each term in the tensor to some specified precision, e.g., double precision, is not guaranteed to produce a more favorable computational complexity.
If every degree of freedom within a phase space cell is coupled, the resulting tensor would be dense and the computational complexity of evaluating this tensor convolution would then be $\mathcal{O}(N_p^3)$, where $N_p$ is the number of basis functions in our phase space expansion.
It is perhaps the case that the scaling would not be as dire as $\mathcal{O}(N_p^3)$, since the phase space flux, $\gvec{\alpha}_h$, requires the expansions of the electromagnetic fields, which live in the configuration space subspace of our phase space expansion, but $\gvec{\alpha}$ does vary linearly in velocity space via the $\mvec{v} \times \mvec{B}$ component of the Lorentz force. Thus, we expect the computational complexity would be between $\mathcal{O}(N_p^3)$ and $\mathcal{O}(N_c N_p^2)$, where $N_c$ is the number of configuration space basis functions, and not the full reduction to the more favorable $\mathcal{O}(N_c N_p^2)$ scaling.

An approach that is standard with nodal bases is to reduce the cost of the scheme by only evaluating the terms in these integrals, such as \eqr{\ref{eq:VlasovSurfaceMatrix}} and \eqr{\ref{eq:VlasovVolumeMatrix}}, at the specified nodes that define the polynomials \citep{Hesthaven:2007,Hindenlang:2012}.
In doing so, the required number of operations would be significantly decreased, as the values of the coefficients at the nodes are known by the definition of the nodal basis, reducing the computational complexity to $\mathcal{O}(N_p^2)$.
But, this approach incurs the very same aliasing errors we warned about in Section~\ref{sec:WeakEquality}.
Even if the values of the various quantities such as $\gvec{\alpha}_h$ and $f_h$ are known at the nodes, the product of the two quantities required for the volume term is not known at the nodes because the product of the two quantities is higher order.
Thus, we will be unable to determine the nonlinear term uniquely if we evaluate $\gvec{\alpha}_h$ and $f_h$ at the nodes and multiply the result.

We now make concrete one of the principal advancements of this thesis: the intolerable consequences of aliasing errors in a DG discretization of an equation system such as the VM-FP system of equations.
We emphasized in Section~\ref{sec:PropertiesKineticEquation} and Appendix~\ref{app:proofsContinuous} for the continuous system, and again when we discussed the properties of the discrete system in Sections~\ref{sec:PropertiesDiscreteVM} and \ref{sec:propertiesDiscreteFP}, that many properties of the VM-FP system of equations are implicit to the equation system.
The Vlasov--Fokker--Planck equation is a conservation equation for the particle distribution function, and the fact that it is a conservation equation makes certain properties explicit, such as phase space incompressibility for the collisionless component of the equation system.
However, other properties are contained in velocity moments of the equation system.
For example, it is the second velocity moment of the Vlasov--Fokker--Planck equation, combined with Maxwell's equations, that gives us total energy conservation.

When proving that the discrete scheme maintains properties of the continuous system such as conservation of mass and energy, we substituted for the test functions, $w$, expressions we presumed we would be able to integrate.
In one case, we substituted $w = 1/2\thinspace m |\mvec{v}|^2$ and evaluated the integrals to massage the volume term into forms which determined the conditions for which energy would be conserved.
While at first glance this may seem like an obvious assertion: we have to evaluate the integrals correctly to actually retain properties such as conservation of mass and energy, it is important to realize why this is the case.
Were we evaluating explicit conservation relations, such as the conservation of mass, momentum, and energy equations in the Euler equations, the Navier-Stokes equations, or the equations of magnetohydrodynamics, aliasing errors could be problematic, but they would not destroy conservation relations.
The aliasing errors arising from not exactly representing the fluid equation solution in a DG algorithm exactly might cause anomalous energy transport, but the aliasing induced transport would not destroy energy conservation of the equation system.

We do not wish to be overly uncharitable on this point. It is well known within the DG computational fluid dynamics community that aliasing errors can lead to stability issues \citep{Kirby:2003}; however, because the aliasing errors manifest in the smallest scales and highest wavenumbers, techniques such as filtering and artificial dissipation are commonly employed to ameliorate these errors \citep{Fischer:2001, Gassner:2013b, Flad:2016, Moura:2017}.
And because fluids equations such as the Euler equations, the Navier-Stokes equations, or the equations of magnetohydrodynamics involve the discretization of explicit conservation relations for mass, momentum, and energy, there is far less concern that such filtering or artificial dissipation will destroy the quality of the solution, at least at scales above the resolution of the simulation.
There are other means of alleviating or eliminating aliasing errors using split-form formulations of the DG method\footnote{In the split-form forumulation, conservative and non-conservative forms of the equation at the continuous level are averaged to produce a different, but ultimately more computationally favorable, equation to discretize.} \citep{Gassner:2013a, Gassner:2014, Gassner:2016a, Gassner:2016b, Flad:2017}, and overintegration, essentially the idea we already discussed of adding sufficient quadrature points to exactly integrate the nonlinear term \citep{Mengaldo:2015, Kopriva:2018, Fehn:2019}.
For a comparison of these two approaches, see \citet{Winters:2018}.
Importantly, with the exception of overintegration\footnote{And only overintegration in specific circumstances, as overintegration of expressions such as \eqr{\ref{eq:badDiscreteFlowDefinition}} for computing the discrete flow, $\mvec{u}_h$ will always incur aliasing errors unless you apply overintegration to the linear operation defined in \eqr{\ref{eq:weakFlowDefinition}}, because \eqr{\ref{eq:badDiscreteFlowDefinition}} involves integration of a rational function, which Gaussian quadrature cannot integrate exactly.}, techniques such as filtering and the split-form formulation are attempts to reduce aliasing errors, not completely eliminate them.
For certain equation systems, the split form formulation has been shown to balance robustly two sources of aliasing: too much energy in the small scales due to under-integration of the conservative form and too little energy in the small scales caused by under-integration of the non-conservative form.
These errors then roughly cancel and produce a more favorable method; however, formulating the equations in split-form is still principally a means of controlling aliasing errors, not removing aliasing errors entirely \citep{Winters:2018}. 

Critically, we must eliminate aliasing errors from our DG discretization of the VM-FP system of equations, lest these aliasing errors manifest themselves as the ``energy content'' of the velocity moments being transported in uncontrolled and undesirable ways.
Because the physics content of the velocity moments of the particle distribution function are directly encoded in our DG discretization, we cannot allow aliasing errors to change the behavior of these moments in our basis expansion.
We are explicitly evolving a polynomial expansion in velocity space that corresponds directly to evolving velocity moments like mass and energy, so any anomalous transport of the ``energy content'' of our expansion will inevitably destroy the conservation relations implicit to the VM-FP system of equations.

The very same structure of our basis expansion we leveraged to demonstrate the discrete VM-FP system of equations retained key properties of the continuous system imposes the constraint that we eliminate aliasing errors from the evaluation of the discrete weak forms for the VM-FP system of equations.
If we do not respect this restriction on our discrete scheme, we by no means guarantee the VM-FP system of equations retains these properties of the continuous system, and thereby risk not just the physicality of the solution, but the overall stability of the numerical method.
It would be nigh impossible to correct the rearrangement of the ``energy content'' of the basis expansion in a physically reasonable way, much less a stable way.
If we cannot safely apply standard techniques such as filtering to mitigate aliasing errors, we must then eliminate these errors in their entirety.

So, we return to the computational complexity we found for the naive means of eliminating aliasing errors with exact integration. 
For exact numerical integration, the computational complexity will inevitably be $\mathcal{O}(N^{tot}_q N_p)$, while exact analytic integration will produce an algorithm we expect will lie between $\mathcal{O}(N_c N_p^2)$ and $\mathcal{O}(N_p^3)$, at least if one assumes that every degree of freedom couples to every other degree of freedom in the expansion.
We can ask the question if there is any way to reduce this cost, and indeed for numerical integration, some savings can be obtained by use of an anisotropic quadrature scheme.
For example, if we consider the advection in velocity space,
\begin{align}
\int_{K_j} \gv w_\ell \cdot \gvec{\alpha}^{v}_h f_h  \dz = \int_{K_j} \gv w_\ell \cdot \frac{q}{m} (\mvec{E}_h + \mvec{v} \times \mvec{B}_h) f_h  \dz,
\end{align}
for each of the $\ell$ basis functions in our phase space expansion, we can see that, while we require integrating monomials of degree $3p$ in configuration space, in velocity space we require at most integrating monomials with degree $2p + 1$.
Table \ref{table:AnisotropicQuadratureTable} considers the impact anisotropic quadrature, using only the minimum number of quadrature points required along each direction of integration, has on a few combinations of velocity space and configuration space dimensions. 
\begin{table}[!htb]
\begin{center}
\begin{tabular}{| l | l | l | l | l | l |}
\hline
& \textbf{Polynomial Order} & 1 & 2 & 3 & 4 \\ \hline
\textbf{Dimension} & $((3p+1)/2)^{CDIM}\times((2p+2)/2)^3$& & & & \\ \hline
1X3V & & 16 & 108 & 320 & 875 \\ \hline
2X3V & & 32 & 432 & 1600 & 6125 \\ \hline
3X3V & & 64 & 1728 & 8000 & 42875 \\ \hline
\end{tabular}
\caption{Number of quadrature points required to integrate the volume term for the advection of the distribution function in velocity space as a function of dimension.}
\label{table:AnisotropicQuadratureTable}
\end{center}
\end{table}
\begin{table}[!htb]
\begin{center}
\begin{tabular}{| l | l | l | l | l | l |}
\hline
\textbf{Cost(Original/New)} & \textbf{Polynomial Order} & 1 & 2 & 3 & 4 \\ \hline
\textbf{Dimension} & & & & & \\ \hline
1X3V & & 1 & $\sim 2.37$ & $\sim 1.95$ & $\sim 2.74$ \\ \hline
2X3V & & 1 & $\sim 2.37$ & $\sim 1.95$ & $\sim 2.74$ \\ \hline
3X3V & & 1 & $\sim 2.37$ & $\sim 1.95$ & $\sim 2.74$ \\ \hline
\end{tabular}
\caption{Reduction in the number of quadrature points,  relative to isotropic quadrature, required to integrate the volume term for the advection of the distribution function in velocity space.}
\label{table:AnisotropicQuadratureReduction}
\end{center}
\end{table}
While there is no gain for polynomial order one, there is a moderate improvement relative to isotropic quadrature for other combinations, as shown in Table \ref{table:AnisotropicQuadratureReduction}.
A similar reduction in the number of quadrature points required can be demonstrated for the surface integrals.

Although we could individually examine each component of the semi-discrete Vlasov--Fokker--Planck equation and determine the minimum amount of quadrature required to integrate each term exactly, it is worth pointing out that, inevitably, the computational complexity of this algorithm remains $\mathcal{O}(N^{tot}_q N_p)$.
There are some exceptions: for example, we can rewrite the phase space flux in configuration space to exploit the fact that the we are employing structured, Cartesian grids, 
\begin{align}
    \int_{K_j} \gx w_\ell \cdot \mvec{v} f_h  \dz = \int_{K_j} \gx w_\ell \cdot (\mvec{v} - \mvec{v}_{\textrm{center}}) f_h d\mvec{z} + \int_{K_j} \gx w_\ell \cdot \mvec{v}_{\textrm{center}} f_h  \thinspace d\mvec{z},
\end{align}
for each of the $\ell$ basis functions in our phase space expansion, where $\mvec{v}_{\textrm{center}} = \overline{\mvec{v}}$ is the cell center velocity.
These integrals can be pre-computed on the phase space reference elements because they are only coordinate weighted matrices, independent of one's exact position in velocity space, thus reducing their computational complexity to $\mathcal{O}(N_p^2)$.

However, the rearrangement of the phase space flux in configuration space to reduce the cost is the exception and not the norm.
The individual pieces of the semi-discrete Fokker--Planck equation will be limited in cost by the number of quadrature points required to integrate exactly the semi-discrete form because the Fokker--Planck equation is nonlinear, just like the advection in velocity space due to the electromagnetic fields.

So, numerical quadrature will be inescapably expensive if we are to satisfy our constraint that we must integrate the semi-discrete VM-FP system of equations exactly to prevent aliasing errors from destroying the quality of our solution.
As stated above, at first glance, the analytical integration to pre-compute and construct the tensors, for example \eqr{\ref{eq:VlasovVolumeMatrix}}, for convolution as part of the update, are very dense.
The convolution of these dense tensors will lead to an unfavorable computational complexity, similar to the numerical quadrature approach, between $\mathcal{O}(N_c N_p^2)$ and $\mathcal{O}(N_p^3)$.
However, if we could sparsify these tensors in some way, thereby reducing the couplings between all of the polynomials in our basis expansion, we may dramatically improve the computational complexity, and thus reduce the cost, of our numerical method for the VM-FP system of equations.

It is no coincidence we have drawn continual attention to the modal, orthonormal basis in our discussion of the specific forms our polynomial bases might take.
We now emphasize the second of our most important algorithmic advances in our implementation of our DG discretization of the VM-FP system of equations: employing a modal, orthonormal basis set for our polynomial basis expansion.
This judicious choice of basis functions allows us to significantly sparsify the requisite tensors needed to evaluate the spatial discretization of the VM-FP system of equations, while still respecting the requirement that our algorithm be alias-free for stability and accuracy.

To get a sense for just how sparse the update with a modal, orthonormal basis is, we consider again the collisionless update, the Vlasov equation, and the volume term defined in \eqr{\ref{eq:VlasovVolumeMatrix}}.
Now, we will project the phase space flux, $\gvec{\alpha}_h$, onto this modal, orthonormal basis,
\begin{align}
    \gvec{\alpha}^{x}_j(t) & = \int_{I} (\mvec{v} - \mvec{v}_{\textrm{center}}) \hat{\upsilon}_j(\gvec{\eta}) d\gvec{\eta} + \int_{I} \mvec{v}_{\textrm{center}} \hat{\upsilon}_j(\gvec{\eta}) d\gvec{\eta} \label{eq:orthonormalAlphaConfig} \\
    \gvec{\alpha}^{v}_j(t) & = \sum_i \int_{I} \frac{q}{m} \left[ \mvec{E}_i(t) + \mvec{v}_{\textrm{center}} \times  \mvec{B}_i(t) \right]  \hat{\vartheta}_i(\gvec{\zeta}) \hat{\upsilon}_j(\gvec{\eta}) d\gvec{\eta} \notag \\
    & + \sum_i \int_{I} \frac{q}{m} (\mvec{v} - \mvec{v}_{\textrm{center}}) \times \mvec{B}_i(t) \hat{\vartheta}_i(\gvec{\zeta}) \hat{\upsilon}_j(\gvec{\eta}) d\gvec{\eta}, \label{eq:orthonormalAlphaVelocity}
\end{align}
where we have denoted the orthonormal basis expansion in phase space as $\hat{\upsilon}(\gvec{\eta})$ and the orthonormal expansion in configuration space as $\hat{\vartheta}(\gvec{\zeta})$. 
Importantly, these expressions have already leveraged the fact that the mass matrix is the identity matrix, up to the volume factor in a cell, to simplify the resulting expressions so that the index $\gvec{\alpha}_j$ maps to the $j^{\textrm{th}}$ basis function on the right hand side.
By separating $\mvec{v} \rightarrow (\mvec{v} - \mvec{v}_{\textrm{center}}) + \mvec{v}_{\textrm{center}}$, we can cleanly separate the velocity dependence into the piecewise constant basis function and a piecewise linear basis function.
In other words, we can clearly see that we require only a small fraction of the full basis expansion's dependence in velocity space to represent both the configuration space and velocity space phase space flux $\gvec{\alpha}^{x,v}_h$.

These expressions for the phase space flux can be plugged in for the coefficients in \eqr{\ref{eq:linearVlasovExp}}, and the whole update evaluated, after exploiting a similar sparsity in the collisionless numerical flux function and the other components of the discrete weak forms of the VM-FP system of equations.
To actually evaluate matrices such as \eqr{\ref{eq:VlasovVolumeMatrix}}, we can use a computer algebra system, for example Maxima \citep{maxima}, and compute the explicit form of the sums in \eqr{\ref{eq:linearVlasovExp}}.
In other words, by evaluating
\begin{align}
    \textrm{out}_k = \sum_{m,n} \mathcal{C}_{kmn} \cdot \gvec{\alpha}_n f_m,
\end{align}
where $\textrm{out}_k$ is a component of the update for $df_k/dt$, and using the fact that the mass matrix is the identity matrix to change variables $\ell \rightarrow k$, we obtain the update shown in Figure~\ref{fig:1X2V-p1-vol-update} for the piecewise linear tensor product basis in one spatial and two velocity dimensions (1X2V).
\begin{figure}[!htb]
    \centering
    \includegraphics[width=\textwidth]{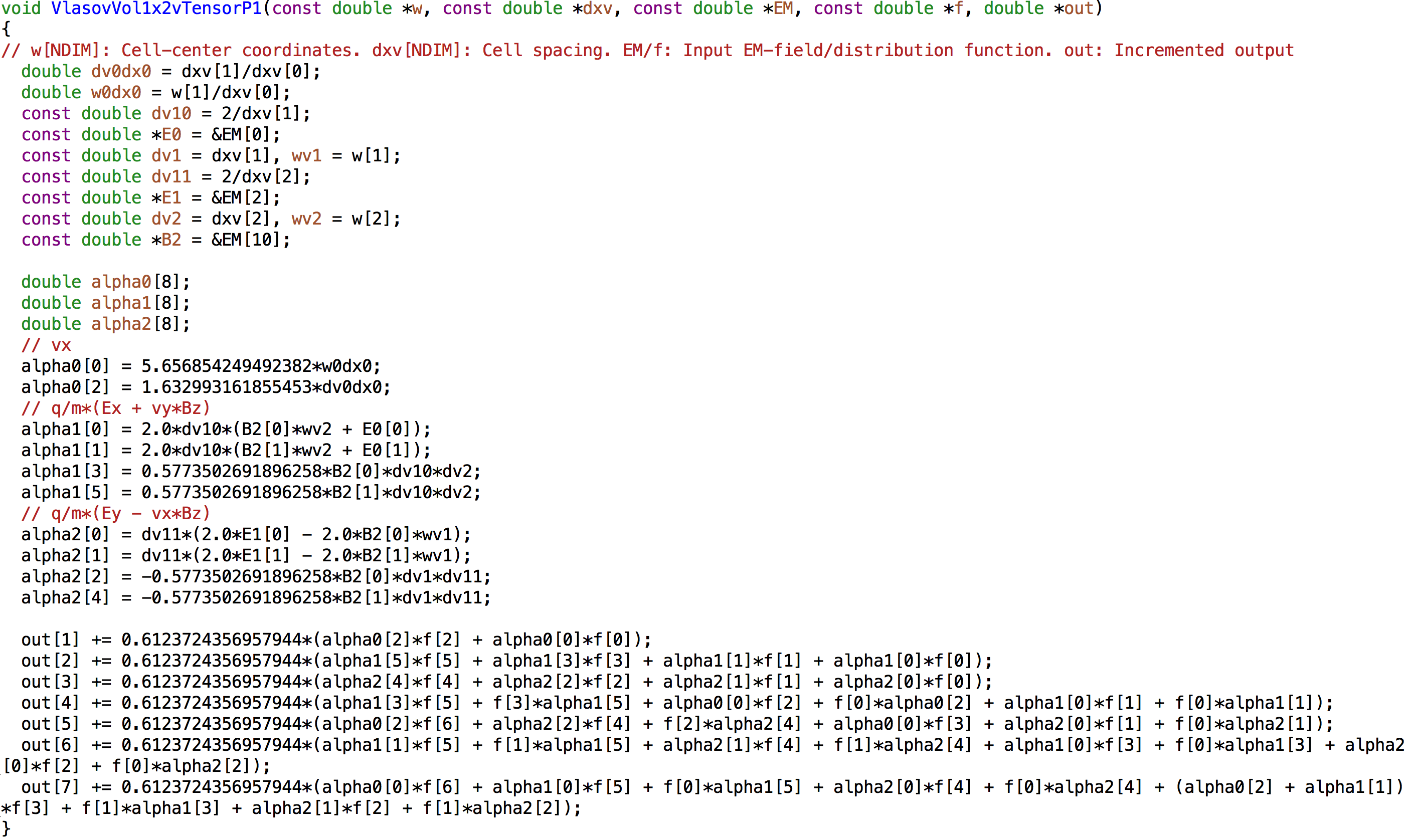}
    \caption{The computational kernel for the volume integral, \eqr{\ref{eq:VlasovVolumeMatrix}}, for the collisionless advection in phase space of the particle distribution function in one spatial dimension and two velocity dimensions (1X2V) for the piecewise linear tensor product basis. Note that this computational kernel takes the form of a C++ kernel that can be called repeatedly for each grid cell $K_j$ depending on the local cell center coordinate and the local grid spacing. Here, the local cell coordinate is the input ``const double w'' and the local grid spacing is the input ``const double dxv''. The out array is the increment to the right hand side due this volume integral contribution in a forward Euler time-step, i.e., a piece of \eqr{\ref{eq:linearOpKinetic}} for the Vlasov--Fokker--Planck equation. To complete the right hand side of \eqr{\ref{eq:linearOpKinetic}} for the evolution of the particle distribution function, for a given phase space cell, we require the surface contributions for the collisionless advection, as well as the computational kernels for the corresponding tensors encoding the spatial discretization of the Fokker--Planck equation.}
    \label{fig:1X2V-p1-vol-update}
\end{figure}

Figure~\ref{fig:1X2V-p1-vol-update} shows a C++ computational kernel that can be called for every cell $K_j$ of a structured, Cartesian grid in phase space, as we are passing all the information required to the kernel to determine where we are physically in phase space, i.e., the local cell center coordinate and grid cell size.
The output of this computational kernel, the out array, is a piece of \eqr{\ref{eq:linearOpKinetic}} for the Vlasov--Fokker--Planck equation, the volume integral of the collisionless advection in phase space. 
To complete the right hand side of \eqr{\ref{eq:linearOpKinetic}} for a given phase space cell, we require the surface contributions for the collisionless advection, as well as the computational kernels for the corresponding tensors encoding the spatial discretization of the Fokker--Planck equation.
We will likewise have computational kernels for Maxwell's equations which completely specify the volume and surface contributions, and allow for the incrementing of the solution in a forward Euler time-step.

Notably, the computational kernel in Figure~\ref{fig:1X2V-p1-vol-update} has no matrix data structure, much less the requirement to perform quadrature since we have already analytically evaluated the integrals in \eqr{\ref{eq:VlasovVolumeMatrix}} with a computer algebra system and written out the results to double precision.
We refer to this as a ``quadrature- and matrix-free'' implementation of the DG method.
Such quadrature-free methods using orthogonal (orthonormal) polynomials were studied in the early days of the DG method \citep{Atkins:1998,Lockard:1999} and are still applied to a variety of linear hyperbolic equations, such as the acoustic wave equation for studies of seismic activity, the level set equation, and Maxwell's equations \citep{Kaser:2006,Marchandise:2006,Koutschan:2012,Kapidani:2020}.
Even for alternative formulations of DG which do not seek to eliminate aliasing errors by exactly integrating the components of the discrete weak form, matrix-free implementations are desirable to reduce the memory footprint of the scheme \citep{Fehn:2019}.
Minimizing the memory footprint can lead to performance gains even beyond the reduction in the number of operations required to take a time-step.

We emphasize again the novelty of our approach.
Using a modal, orthonormal basis, we produce a ``quadrature- and matrix-free'' method that respects our requirement that our algorithm be alias-free by analytically evaluating the integrals in the discrete weak forms of the VM-FP system of equations, thus the quadrature-free component.
And the matrix-free component follows from the fact that the resulting integrals produce sparse tensors whose convolutions can be unfolded in their entirety, eliminating the need for a matrix data structure to actually evaluate the tensor-tensor convolutions.
All that is required is entry-by-entry evaluation of the results of these convolutions, as demonstrated in Figure~\ref{fig:1X2V-p1-vol-update} by the out array.

As a frame of reference the sparseness of our ``quadrature- and matrix-free'' method, the computational kernel in Figure~\ref{fig:1X2V-p1-vol-update} has $\sim 70$ multiplications; whereas, the update for numerical quadrature applied to a nodal basis has $\sim 250$ multiplications.
The potential gains from a nodal basis by only requiring the expansion local to a surface in the surface integrals do not provide enough computational savings to compete with the sparsity of the orthonormal, modal expansion.
We will do a thorough computational complexity experiment in Section~\ref{sec:computationalComplexity} to determine both exactly what the computational complexity of the sparse, orthonormal, modal basis expansion is, as well as compare in totality the performance of a sparse, orthonormal, modal basis expansion to an optimized nodal basis expansion using anisotropic quadrature with high performance linear algebra libraries.
Before we do this comparison though, it is worth going through the final details of the algorithm.
We must now discuss how we compute the recovery polynomial in generality, and how we compute velocity moments, to complete the implementation of our numerical method for the VM-FP system of equations.

\section{Extending the Recovery Scheme to Higher Dimensions}\label{sec:recoveryHigherDimensions}

As stated above in Section~\ref{sec:ComputationToPhysical}, many of the components of the surface integrals, for example the numerical flux functions for the collisionless advection and drag term, are simple enough to project onto our phase space basis expansion, compute the coefficients in our modal, orthonormal basis expansion, and then convolve tensors such as \eqr{\ref{eq:VlasovSurfaceMatrix}} to evaluate the surface integral contributions in our discretization of the VM-FP system of equations.
However, we require a prescription for computing the recovery polynomial in generality so we can evaluate the corresponding surface integrals in the discrete Fokker--Planck equation.
Whereas projections such as \eqr{\ref{eq:projectCentralFluxes}} for central fluxes applied to the collisionless advection naturally retain the spatial dependence at the surface, and thus the high order nature of our scheme, we have not yet described a procedure for the non-recovered spatial dependence in our computation of the recovery polynomial.

We said in the summary of Chapter~\ref{ch:DGFEM}, Section~\ref{sec:chapter2Summary}, that the recovery procedure is fundamentally one dimensional: we are only generating a recovery polynomial across the surface where the function has a discontinuity.
So, let us consider the operation of projecting a two dimensional function, $f(x,y)$, onto a one-dimensional basis,
\begin{align}
    \int_{-1}^1 g(x,y) \psi_k(x) \thinspace dx & = \int_{-1}^1 f(x,y) \psi_k(x) \thinspace dx, \\
    g_k(y) & = \int_{-1}^1 f(x,y) \psi_k(x) \thinspace dx, \label{eq:2DTo1DProjection}
\end{align}
i.e., each of the $k$ coefficients for the component expansion in the $x$ dimension retain their $y$ variation.
Note that the simplified form of \eqr{\ref{eq:2DTo1DProjection}} assumes the basis $\psi_k$ is our modal, orthonormal basis expansion to simplify the left hand side, and that as part of this operation $f(x,y)$ has a two dimensional basis expansion in $x$ and $y$.

Although we characterized the recovery procedure mathematically in Section~\ref{sec:WeakEquality}, we should now explicitly compute the recovery polynomial in a specific test case to make apparent how to use \eqr{\ref{eq:2DTo1DProjection}} to compute the recovery polynomial in generality.
Let us use the piecewise linear, one dimensional, modal, orthonormal basis for this demonstration,
\begin{align}
    \begin{aligned}
    \hat{\upsilon}_1(x) & = \frac{1}{\sqrt{2}}, \\
    \hat{\upsilon}_2(x) & = \sqrt{\frac{3}{2}} x,
    \end{aligned}
\end{align}
but on a slightly different reference element, $K_L=[-2,0]$ on the left, and $K_R=[0,2]$ on the right, so that the left and right cells each have the same volume as our original reference element $[-1, 1]$.
The discontinuity will still be located $x = 0$.
These shifted basis functions are then
\begin{align}
    \begin{aligned}
    \hat{\upsilon}_{L1}(x) & = \frac{1}{\sqrt{2}}, \\
    \hat{\upsilon}_{R1}(x) & = \frac{1}{\sqrt{2}}, \\
    \hat{\upsilon}_{L2}(x) & = \sqrt{\frac{3}{2}} (x + 1), \\
    \hat{\upsilon}_{R2}(x) & = \sqrt{\frac{3}{2}} (x - 1),
    \end{aligned}
\end{align}
so that the full basis expansions in each cell are,
\begin{align}
    \begin{aligned}
    f_L(x) & = \frac{1}{\sqrt{2}} f_{L1} + \sqrt{\frac{3}{2}} (x + 1) f_{L2}, \\
    f_R(x) & = \frac{1}{\sqrt{2}} f_{R1} + \sqrt{\frac{3}{2}} (x - 1) f_{R2}.
    \end{aligned}
\end{align}
Since we are using piecewise linear polynomials in the left and right cells, two basis functions in each cell, four basis functions total, we can represent a cubic function across the interface,
\begin{align}
    h(x) = h_1 + h_2 x + h_3 x^2 + h_4 x^3.
\end{align}
We then solve the following set of equations
\begin{align}
    \begin{aligned}
    \int_{-2}^0 [h(x) - f_L(x)]  \hat{\upsilon}_{L1}(x) \thinspace dx & = 0, \\
    \int_{-2}^0 [h(x) - f_L(x)]  \hat{\upsilon}_{L2}(x) \thinspace dx & = 0, \\
    \int_0^2 [h(x) - f_R(x)]  \hat{\upsilon}_{R1}(x) \thinspace dx & = 0, \\
    \int_0^2 [h(x) - f_R(x)]  \hat{\upsilon}_{R2}(x) \thinspace dx & = 0,
    \end{aligned}
\end{align}
using a computer algebra system to analytically evaluate each integral and invert the matrix equation for the coefficients,
\begin{align}
    \begin{aligned}
    h_1 & = \frac{\sqrt{2}\left(-2\sqrt{3}f_{R2} + 2\sqrt{3}f_{L2} + 3f_{R1} + 3f_{L1} \right)}{12},\\
    h_2 & =- \frac{\sqrt{2}\left(5\sqrt{3}f_{R2} +    5\sqrt{3}f_{L2} - 9f_{R1} + 9f_{L1} \right)}{16},\\
    h_3 & =- \frac{\sqrt{3}\left(f_{L2} - f_{R2} \right)}{\sqrt{2}^5},\\
    h_4 & = \frac{\sqrt{2}\left(5\sqrt{3}f_{R2} + 5\sqrt{3}f_{L2} - 5f_{R1} + 5f_{L1} \right)}{32}. 
    \end{aligned} \label{eq:1DRecoveryEvaluated}
\end{align}
Now we can use \eqr{\ref{eq:2DTo1DProjection}} to modify the individual pieces of \eqr{\ref{eq:1DRecoveryEvaluated}}.
For example, if the original function $f = f(x,y)$, we can compute in the right cell
\begin{align}
    \begin{aligned}
    f_{R1}(y) & = \int_0^2 f(x,y) \hat{\upsilon}_{R1}(x) \thinspace dx, \\
    f_{R2}(y) & = \int_0^2 f(x,y) \hat{\upsilon}_{R2}(x) \thinspace dx,
    \end{aligned}
\end{align}
and likewise for the left cell. 

This procedure, combining the one dimensional recovery in \eqr{\ref{eq:1DRecoveryEvaluated}} with the projection from the higher dimensional space onto the one dimensional basis, \eqr{\ref{eq:2DTo1DProjection}}, to determine how the coefficients vary in the other dimensions, is general and can be extended to as high dimensionality and as high polynomial order as we choose.
Notably, regardless of the specific form of the recovery polynomial, we emphasize that we only require the first and second coefficients, $h_1$ and $h_2$ in \eqr{\ref{eq:1DRecoveryEvaluated}}, because we are evaluating the recovery polynomial and its first derivative at the $x=0$ surface.
In other words, the value of the recovery polynomial at the surface of the reference element is $h_1$, and the value of the gradient of the recovery polynomial at the surface of the reference element is $h_2$, at least for piecewise linear polynomials.
We have thus completely specified the required recovered function, e.g., the recovered distribution function in the discrete Fokker--Planck equation, the value and the gradient of the recovered function, and the recovered function's variation along the surface across which we are constructing the recovered function.
We can then project the results of this recovery process onto phase space basis functions, and construct a similar tensor to \eqr{\ref{eq:VlasovSurfaceMatrix}} to convolve and evaluate the surface contributions in the discrete Fokker--Planck equation.

\section{Computing the Coupling Moments}

The final component of our implementation is a means of computing the velocity moments which close our equation system, such as $\mvec{J}_h$ for coupling to Maxwell's equations.
In the same way we demonstrated how one leverages weak equality to actually calculate the recovery polynomial in arbitrary dimensions in Section~\ref{sec:recoveryHigherDimensions}, the goal of this section is to illustrate the use of weak equality to compute the coupling moments, and the form these computational kernels take. 
Recall the operations we defined in Eqns.\thinspace(\ref{eq:discrete0thMoment}--\ref{eq:discrete1stMoment}), which we here write out explicitly transformed to the reference element on which the modal, orthonormal basis sets are defined,
\begin{align}
    \sum_m M_{0_m} \int_{I_\Omega} & \hat{\vartheta}_\ell(\gvec{\zeta}) \hat{\vartheta}_m(\gvec{\zeta}) d\gvec{\zeta} \notag \\
    & = \left( \frac{1}{2^{VDIM}} \prod_{i=1}^{VDIM} \Delta v_i \right )\sum_n \sum_j \int_{I_j \setminus I_\Omega} f_n(t) \hat{\upsilon}_n(\gvec{\eta}) \hat{\vartheta}_\ell(\gvec{\zeta}) d\gvec{\eta}, \\
    \sum_m \mvec{M}_{1_m} \int_{I_\Omega} & \hat{\vartheta}_\ell(\gvec{\zeta}) \hat{\vartheta}_m(\gvec{\zeta}) d\gvec{\zeta} \notag \\
    & = \left( \frac{1}{2^{VDIM}} \prod_{i=1}^{VDIM} \Delta v_i \right ) \sum_n \sum_j \int_{I_j \setminus I_\Omega} \mvec{v} f_n(t) \hat{\upsilon}_n(\gvec{\eta}) \hat{\vartheta}_\ell(\gvec{\zeta}) d\gvec{\eta}, \\
    \sum_m M_{2_m} \int_{I_\Omega} & \hat{\vartheta}_\ell(\gvec{\zeta}) \hat{\vartheta}_m(\gvec{\zeta}) d\gvec{\zeta} \notag \\
    & = \left( \frac{1}{2^{VDIM}} \prod_{i=1}^{VDIM} \Delta v_i \right ) \sum_n \sum_j \int_{I_j \setminus I_\Omega} |\mvec{v}|^2 f_n(t) \hat{\upsilon}_n(\gvec{\eta}) \hat{\vartheta}_\ell(\gvec{\zeta}) d\gvec{\eta}.
\end{align}
We note that the matrix on the left hand side is simply the mass matrix in configuration space, and since we have already canceled the configuration space volume factor, the matrix is simply the identity matrix.
However, we require a means to make the integrals on the reference element independent of our location in phase space, and so we perform a similar transform as done in Eqns.\thinspace(\ref{eq:orthonormalAlphaConfig}) and (\ref{eq:orthonormalAlphaVelocity}),
\begin{align}
    M_{0_\ell} & = \left( \frac{1}{2^{VDIM}} \prod_{i=1}^{VDIM} \Delta v_i \right )\sum_n \sum_j \int_{I_j \setminus I_\Omega} f_n(t) \hat{\upsilon}_n(\gvec{\eta}) \hat{\vartheta}_\ell(\gvec{\zeta}) d\gvec{\eta}, \label{eq:fullM0Calc} \\
    \mvec{M}_{1_\ell} & = \left( \frac{1}{2^{VDIM}} \prod_{i=1}^{VDIM} \Delta v_i \right )\sum_n \sum_j \int_{I_j \setminus I_\Omega} (\mvec{v} - \mvec{v}_{\textrm{center}}) f_n(t) \hat{\upsilon}_n(\gvec{\eta}) \hat{\vartheta}_\ell(\gvec{\zeta}) d\gvec{\eta} \notag \\
    & +  \left( \frac{1}{2^{VDIM}} \prod_{i=1}^{VDIM} \Delta v_i \right )\sum_n \sum_j \int_{I_j \setminus I_\Omega} \mvec{v}_{\textrm{center}} f_n(t) \hat{\upsilon}_n(\gvec{\eta}) \hat{\vartheta}_\ell(\gvec{\zeta}) d\gvec{\eta} \\
    M_{2_\ell} & = \left( \frac{1}{2^{VDIM}} \prod_{i=1}^{VDIM} \Delta v_i \right )\sum_n \sum_j \int_{I_j \setminus I_\Omega} |\mvec{v} - \mvec{v}_{\textrm{center}}|^2 f_n(t) \hat{\upsilon}_n(\gvec{\eta}) \hat{\vartheta}_\ell(\gvec{\zeta}) d\gvec{\eta} \notag \\
    & + \left( \frac{1}{2^{VDIM}} \prod_{i=1}^{VDIM} \Delta v_i \right )\sum_n \sum_j \int_{I_j \setminus I_\Omega} 2 \mvec{v}_{\textrm{center}} \cdot (\mvec{v} - \mvec{v}_{\textrm{center}}) f_n(t) \hat{\upsilon}_n(\gvec{\eta}) \hat{\vartheta}_\ell(\gvec{\zeta}) d\gvec{\eta} \notag \\
    & + \left( \frac{1}{2^{VDIM}} \prod_{i=1}^{VDIM} \Delta v_i \right )\sum_n \sum_j \int_{I_j \setminus I_\Omega} |\mvec{v}_{\textrm{center}}|^2 f_n(t) \hat{\upsilon}_n(\gvec{\eta}) \hat{\vartheta}_\ell(\gvec{\zeta}) d\gvec{\eta}, \label{eq:fullM2Calc}
\end{align}
which can be further simplified to,
\begin{align}
    \mvec{M}_{1_\ell} & = \mvec{v}_{\textrm{center}} M_{0_\ell} \notag \\
    & + \left( \frac{1}{2^{VDIM}} \prod_{i=1}^{VDIM} \Delta v_i \right )\sum_n \sum_j \int_{I_j \setminus I_\Omega} (\mvec{v} - \mvec{v}_{\textrm{center}}) f_n(t) \hat{\upsilon}_n(\gvec{\eta}) \hat{\vartheta}_\ell(\gvec{\zeta}) d\gvec{\eta}, \label{eq:M1Simplification} \\
    M_{2_\ell} & = 2 \mvec{M}_{1_\ell} \cdot \mvec{v}_{\textrm{center}} - |\mvec{v}_{\textrm{center}}|^2 M_{0_\ell} \notag \\
    & + \left( \frac{1}{2^{VDIM}} \prod_{i=1}^{VDIM} \Delta v_i \right )\sum_n \sum_j \int_{I_j \setminus I_\Omega} |\mvec{v} - \mvec{v}_{\textrm{center}}|^2 f_n(t) \hat{\upsilon}_n(\gvec{\eta}) \hat{\vartheta}_\ell(\gvec{\zeta}) d\gvec{\eta}. \label{eq:M2Simplification}
\end{align}
We can then generate a computational kernel to compute these coupling moments sequentially, and the needed quantities such as the current density can be computed from the results, e.g., via \eqr{\ref{eq:strongEqualityCurrent}}.
Using a 1X2V, one configuration space dimension and two velocity space dimensions, piecewise linear, tensor product basis again as an example, we show the results of a computer algebra system evaluating the integrals in Eqns.\thinspace(\ref{eq:fullM0Calc}--\ref{eq:fullM2Calc}), with the simplifications outlined in Eqns.\thinspace(\ref{eq:M1Simplification}) and (\ref{eq:M2Simplification}), in Figure~\ref{fig:1X2V-p1-Moment-Calc}.
\begin{figure}
    \centering
    \includegraphics[width=\textwidth]{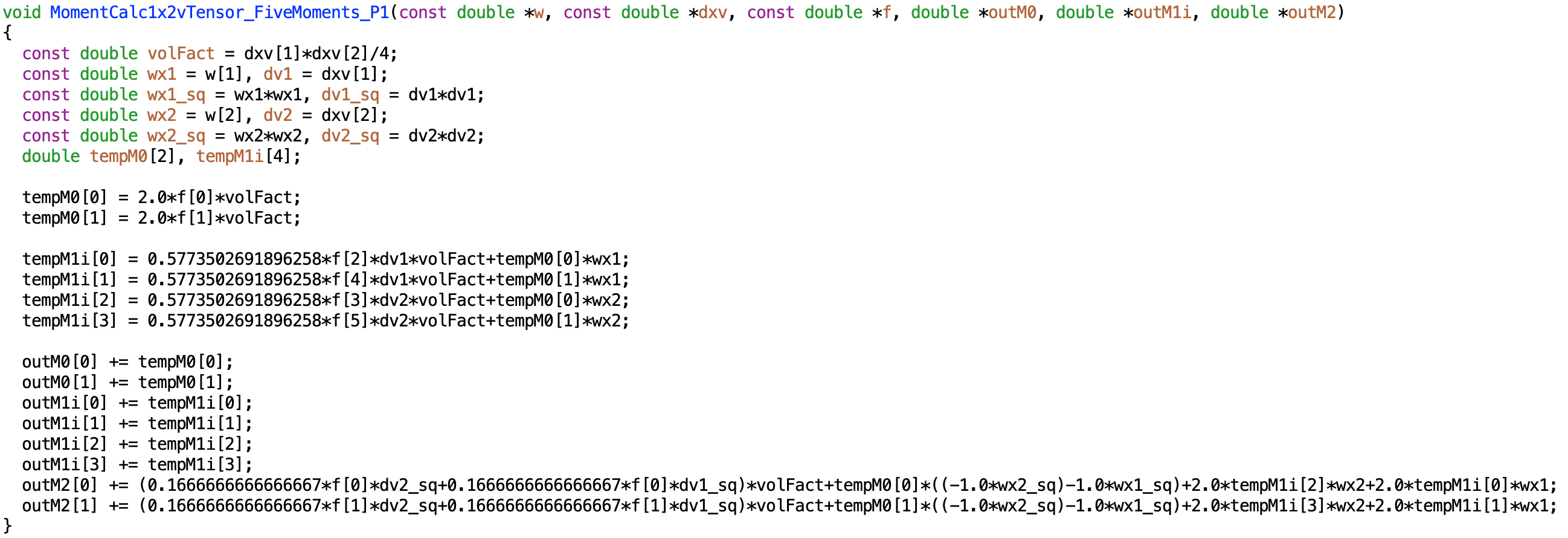}
    \caption{Example computational kernel for the calculation of the zeroth through second moments using weak equality in one spatial dimension and two velocity dimensions (1X2V) with piecewise linear, tensor product, modal, orthonormal polynomials. Note that this computational kernel is called inside a loop over velocity space for a given configuration space cell, as we are integrating over velocity space.}
    \label{fig:1X2V-p1-Moment-Calc}
\end{figure}
It is critical to note that the computational kernel in Figure~\ref{fig:1X2V-p1-Moment-Calc} is called for every velocity space cell associated with a given configuration space cell, i.e., these kernels form a reduction operation across velocity space, as expected since we are integrating over velocity space at a given configuration space cell.
The beauty of Eqns.\thinspace(\ref{eq:fullM0Calc}--\ref{eq:fullM2Calc}), with the simplifications outlined in Eqns.\thinspace(\ref{eq:M1Simplification}) and (\ref{eq:M2Simplification}), is that this same computational kernel can be called irrespective of our location in phase space, so long as we pass the correct cell center coordinate and local grid cell size.
Analogous to the updates for the Vlasov--Fokker--Planck equation and Maxwell's equations, the computation of the coupling moments is also free of both quadrature and matrix data structures.

We note in concluding this section that these procedures can be, and within \gke~are, extended to other diagnostic moments, for example the stress tensor and heat flux,
\begin{align}
    \overleftrightarrow{\mvec{S}}_h & \doteq \sum_j \int_{K_j\setminus \Omega_k} \mvec{v}\mvec{v} f_h \dv, \\
    \gvec{\mathcal{Q}}_h & \doteq \frac{1}{2} \sum_j \int_{K_j\setminus \Omega_k} |\mvec{v}|^2 \mvec{v} f_h \dv,
\end{align}
which can be rearranged similarly with the same variable manipulation as before, $\mvec{v} \rightarrow (\mvec{v} - \mvec{v}_{\textrm{center}}) + \mvec{v}_{\textrm{center}}$.
In general, the mathematical machinery of weak equality can be straightforwardly converted to linear equations which can be computed to determine the desired projection of some quantity, whether it is a velocity moment, a numerical flux function, or a more complicated constraint equation for quantities such as $\mvec{u}_h$ and $T_h$.
The components of the linear equation, the integrals over complex combinations of basis functions, can then be analytically evaluated using a computer algebra system such as Maxima \citep{maxima}, and with the help of the modal, orthonormal polynomial basis, significantly sparsified, reducing the number of operations required to evaluate and solve the linear equations.

Although we have focused on the components of the discretization which are both quadrature- and matrix-free, we should briefly discuss the parts of the discretization which are not necessarily matrix-free.
For example, the solution to the set of linear equations for the discrete flow and temperature, $\mvec{u}_h$ and $T_h$, e.g., Eqns.\thinspace \ref{eq:FPMomentumConstraint} and \ref{eq:FPEnergyWeakMomentP2} when using at least piecewise quadratic polynomials, is not matrix-free because of the coupling between the projections of $\mvec{u}_h$ and $T_h$ due to the boundary corrections from finite velocity space extents.
All the computational machinery we have outlined, i.e., the analytic evaluation of the integrals using a computer algebra system, is still the procedure for evaluating $\mvec{u}_h$ and $T_h$.
Now though, instead of completely unrolling the evaluation of the matrix equations and eliminating the need for a matrix data structure by evaluating every individual term in the linear equation, we construct the relevant matrix and invert the linear system to obtain our solution for $\mvec{u}_h$ and $T_h$.
In Figure~\ref{fig:u-T-calc} we show an example computational kernel to solve the coupled linear system for $\mvec{u}_h$ and $T_h$, using the Eigen linear algebra library \citep*{eigen}, in one configuration space and one velocity space dimension (1X1V) with piecewise quadratic Serendipity polynomials.
\begin{figure}[!htb]
    \centering
    \includegraphics[width=\textwidth]{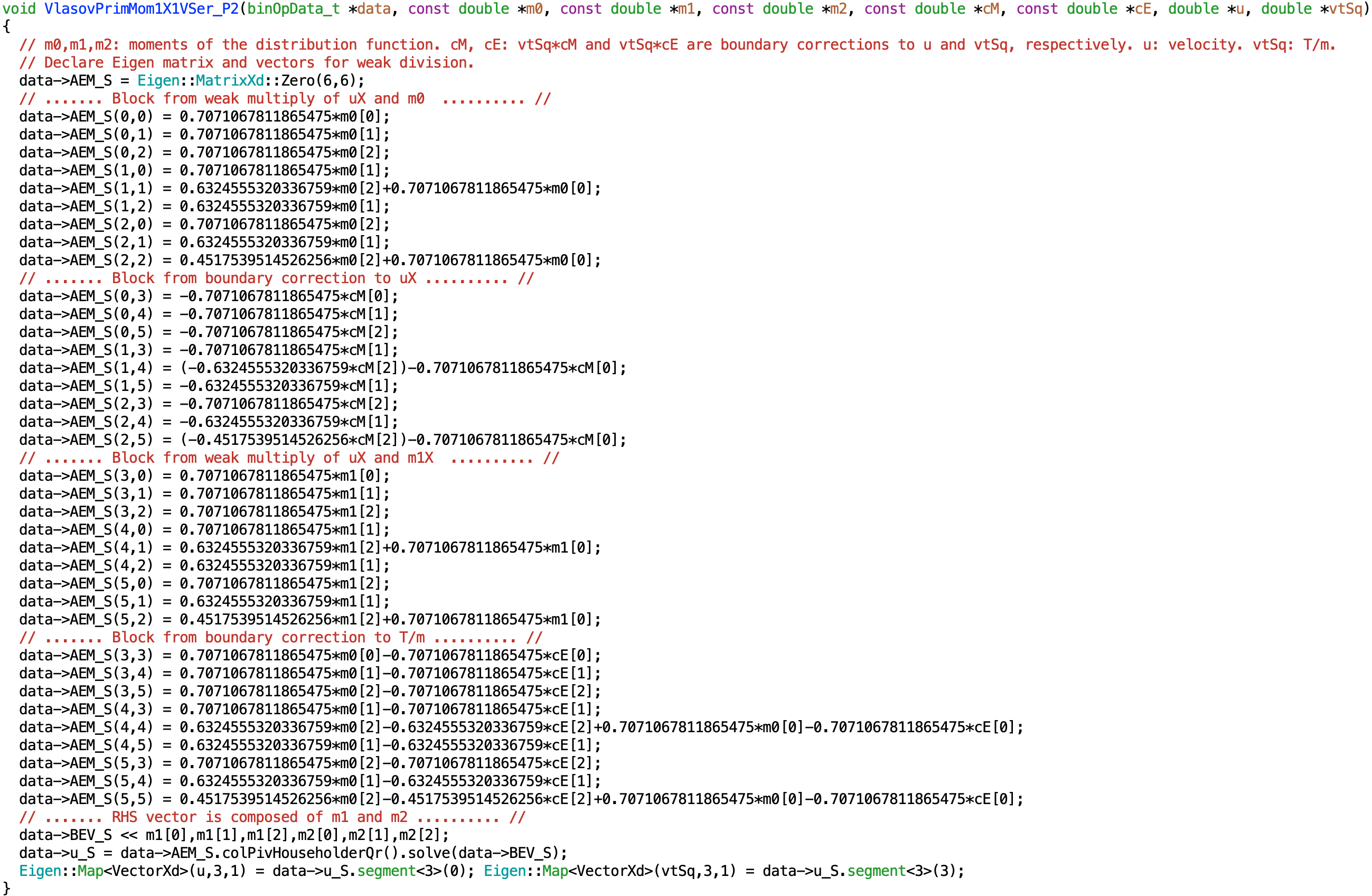}
    \caption{A C++ computational kernel for the construction and inversion of the matrix to solve the coupled linear system for the discrete flow and temperature, $\mvec{u}_h$ and $\mvec{T}_h$. Here, we show the form of the matrix in one spatial and one velocity dimension (1X1V) using piecewise quadratic Serendipity polynomials. Since both $u_h$ and $T_h$ have three degrees of freedom, i.e., three basis functions, which describe their projection, the coupled linear system is six by six. We construct the individual terms in the matrix using a combination of weak multiplication, weak division, and the corrections at the boundary due to our finite velocity space extents. We can then use a linear algebra library, in this case Eigen, to solve the linear system and determine the discrete flow and temperature required in the evaluation of the drag and diffusion coefficients in the discrete Fokker--Planck equation.}
    \label{fig:u-T-calc}
\end{figure}
Importantly, the fact that our basis is modal and orthonormal reduces the number of terms in the matrix we have to invert.
These computational kernels can then be called in every configuration space cell to calculate the local expansion of $\mvec{u}_h$ and $T_h$ required for the discretization of the Fokker--Planck equation.

Now that all the pieces of our discrete scheme are complete, including the means of computing the coupling moments between Maxwell's equations and the Vlasov--Fokker--Planck equation, the implementation of our discrete scheme is finished.
We turn now to the question of the computational complexity of our discrete scheme.
Although we expect the modal, orthonormal basis to have significantly decreased the cost of numerically integrating our DG discretization of the VM-FP system of equations, we require quantitative proof of this cost reduction.

\section{A Computational Complexity Experiment}\label{sec:computationalComplexity}

We know the choice of a modal, orthonormal polynomial basis leads to the tensors over which we need to sum, such as \eqr{\ref{eq:VlasovVolumeMatrix}}, being sparse, and we have evidence from the computational kernel presented in Figure~\ref{fig:1X2V-p1-vol-update} that the number of operations is indeed reduced compared to the use of numerical quadrature.
We would like to determine generally how sparse the tensors required to update our discrete VM-FP system of equations are.
In Figure~\ref{fig:algorithmscaling}, we plot the results of a numerical experiment using the computational kernels for updating the collisionless component of the VM-FP system of equations.
\begin{figure}[!htb]
    \centering
    \includegraphics[width=0.49\textwidth]{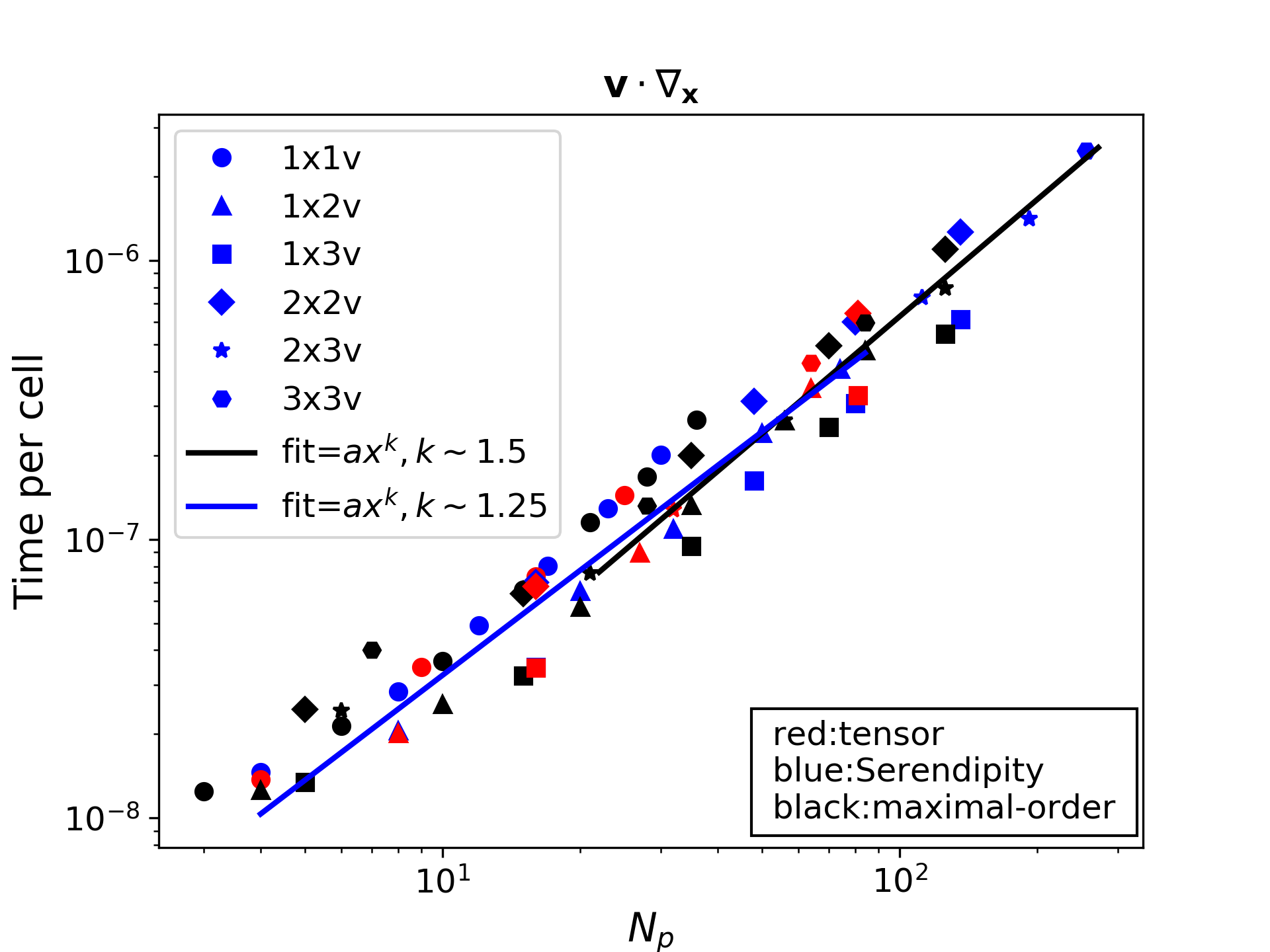}
    \includegraphics[width=0.49\textwidth]{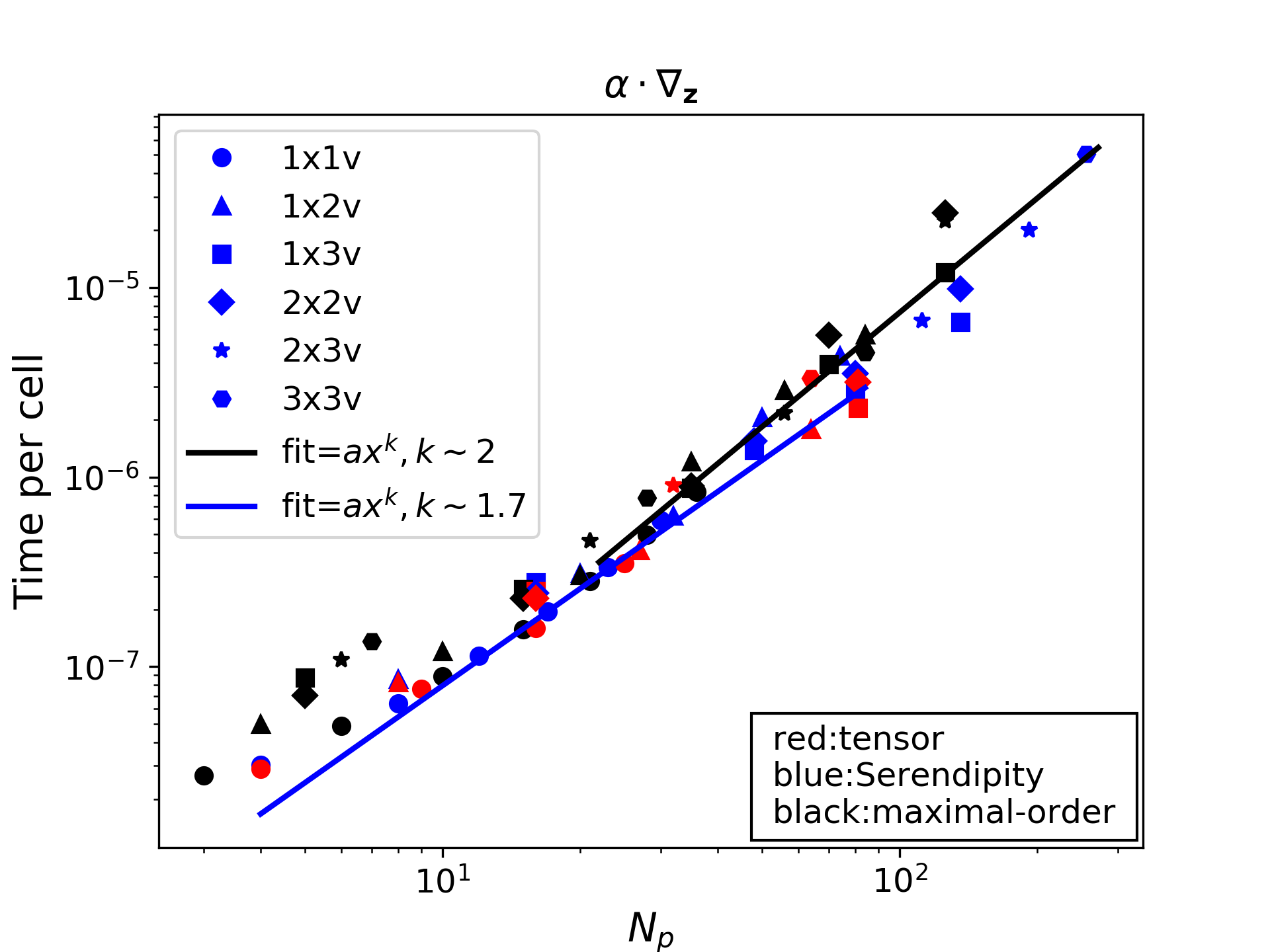}
    \caption{Scaling, i.e., the time to evaluate the update versus the number degrees of freedom, $N_p$, in a cell, of just the streaming term, $\gvec{\alpha}_h^x = \mvec{v}$, (left) and the total, streaming and acceleration, update (right) for the Vlasov solver. The dimensionality of the solve is denoted by the relevant marker, and the three colors correspond to three different basis expansions: black:maximal-order, blue:Serendipity, and red:tensor. Importantly, this is the scaling of the \emph{full} update, for every dimension, i.e., the 3x3v points include the six dimensional volume integral and all twelve five dimensional surface integrals. \label{fig:algorithmscaling}}
\end{figure}
We show the time to evaluate the computational kernels for just the streaming term, $\gvec{\alpha}_h^x = \mvec{v}$, in the left plot of Figure~\ref{fig:algorithmscaling}, and the evaluation of the full phase space update, streaming and acceleration, in the right plot.
From the scaling of the cost to evaluate these computational kernels we can determine the computational complexity of the algorithm with respect to the number of degrees of freedom per cell, i.e., the number of basis functions in our expansion, $N_p$.

It is immediately apparent that even with the steepening of the scaling as the number of degrees of freedom increases there is at least some gain over the use of direct quadrature to evaluate the integrals in the discrete weak form because, at worst, the total, streaming plus acceleration, update scales roughly as $\mathcal{O}(N_p^2)$.
In fact, this scaling of, at worst $\mathcal{O}(N_p^2)$, is exactly the scaling obtained by under-integrating the nonlinear term in a nodal basis, as mentioned in Section~\ref{sec:alias-free-scheme} \citep{Hesthaven:2007,Hindenlang:2012}.
But critically, we have obtained this computational complexity while eliminating aliasing errors from our scheme, as we require for stability and accuracy!
We can explicitly evaluate the gain compared to the anisotropic quadrature shown in Table~\ref{table:AnisotropicQuadratureTable}.
For example, for piecewise quadratic basis functions in six dimensions, the Serendipity space has 256 degrees of freedom in a cell but requires 1728 quadrature points to evaluate the nonlinear term

However, the improvement in the scaling is actually better than it first appears.
The scaling shown in Fig.~\ref{fig:algorithmscaling} is the cost scaling of the full update to perform a forward Euler step in a phase space cell, i.e., in six dimensions, three spatial and three velocity, the total update time in the right plot of Fig.~\ref{fig:algorithmscaling} is the time to compute the six dimensional volume integral plus the twelve required five dimensional surface integrals. 
This means the scaling we are quoting is irrespective of the dimensionality of the problem, unlike in the case of the nodal basis, where the quadrature must be performed for every integral and there is a hidden dimensionality factor in the scaling.
In other words, in six dimensions, what at first may only seem like a factor of $\sim 7$ improvement moving from a nodal to an orthonormal, modal representation is in fact a factor of $\sim 40$ improvement in the scaling once one includes the dimensionality factor, up to the constant of proportionality of the scaling. 
Of course, one must also compare the size of the constant of proportionality multiplying both scalings to accurately compare the reduction in the number of operations and improvement in the overall performance, since said constant of proportionality can either tell us the picture is much rosier, that in fact the improvement in performance is larger than we expected, or much more dire, that the improvement in the scaling is offset by a larger constant of proportionality.

To determine the constant of proportionality, we will perform a more thorough numerical experiment and compare the cost of the alias-free nodal scheme and alias-free modal scheme for a complete collisionless Vlasov--Maxwell simulation.
We consider the following test: a 2X3V computation done with both the nodal and the modal algorithms, with a detailed timing breakdown of the most important step of the algorithm, the Vlasov time step. 
The reader is referred Table \ref{table:VlasovSummaryTable} for a summary of the following two paragraphs if they wish to skip the details of the computer architecture and optimizations employed.
Both computations are performed in serial on a Macbook Pro with an \textbf{Intel Core i7-4850HQ (``Crystal Well'')} chip, the same architecture on which the scaling analysis was performed. 
The only optimization in the compilation of both algorithms is \textbf{``O3''} and both versions of the code are compiled with the C++ \textbf{Clang 9.1} compiler. 

Specific details of the computations are as follows: a $16^2 \times 16^3$ grid, with polynomial order two, and the Serendipity basis, 112 degrees of freedom per cell. 
The two simulations were run for a number of time-steps to allow us to more accurately compute the time per step of just the Vlasov solver, as well as the time per step of the complete simulation. 
The time-stepper of choice for this numerical experiment is the three-stage, third order, SSP-RK method, \eqr{\ref{eq:SSPRK3}}.
To make the simulations as realistic as possible in terms of memory movement, we also evolve a ``proton'' and ``electron'' distribution function, i.e., we evolve the Vlasov--Maxwell system of equations for two plasma species.

To make the comparison as favorable as possible for the nodal algorithm, we also employ the Eigen linear algebra library, \textbf{Eigen 3.3.4} \citep*{eigen}, to perform the dense matrix-vector multiplies required to evaluate the higher order quadrature needed to eliminate aliasing errors in the nodal DG discretization. 
And we note that the nodal algorithm is optimized to use only the surface basis functions in the surface integral evaluations, so we are doing as much as possible to reduce the cost of the alias-free nodal scheme.

The results are as follows: 
for the \emph{nodal} basis, the computation required \textbf{1079.63} seconds per time step, of which \textbf{1033.89} seconds were spent solving the Vlasov equation. 
The remaining time is split between the computation of Maxwell's equations, the computation of the current from the first velocity moment of the distribution function to couple the particles and the fields, and the accumulation of each Runge-Kutta stage from our three stage Runge-Kutta method. 
For the \emph{modal} basis, the computation required \textbf{67.4312} seconds per time step, of which \textbf{60.3431} seconds were spent solving the Vlasov equation.

In the nodal case, we emphasize that we achieve a reasonable CPU efficiency, and the nodal timings are not a matter of poor implementation.
We estimate the number of multiplications in the alias-free nodal algorithm required to perform a full time-step is $\sim 3e12$, three trillion, once one considers the fact that we are evolving two distribution functions with a three-stage Runge--Kutta method.
One thousand seconds to perform three trillion multiplications corresponds to an efficiency of $\sim 3e9$ flops per second (3 GFlops/s).
This estimate is within 50 percent of the measured efficiencies of Eigen's matrix-vector multiplication routines for \textbf{Eigen 3.3.4} on a similar CPU architecture to the one employed for this test \citep*{eigen}, so we argue that the cost of the alias-free nodal algorithm is due to the number of operations required and not an inefficient implementation of the algorithm.

It is then worth discussing how this improvement in the timings using the modal algorithm compares with our expectations.
Given the scaling of the modal basis, we would anticipate the gain in efficiency in five dimensions would be around a factor of twenty, a factor of four from the reduction in the scaling from $\mathcal{O}(N_q N_p)$ to $\mathcal{O}(N_p^2)$, and a factor of five from the latter scaling containing all of the five dimensional volume integrals and the ten four dimensional surface integrals.
We can see that the gain in just the Vlasov solver is $\sim 17$, while the gain in the overall time per step is $\sim 16$, not quite as much as we would naively expect, but still a sizable increase in the speed of the Vlasov solver. 
The reduction in the overall time is due to the fact that, while the time to solve Maxwell's equations and compute the currents to couple the Vlasov equation and Maxwell's equations is reduced, these other two costs, in addition to the cost to accumulate each Runge-Kutta stage, is not reduced as dramatically as the time to solve the Vlasov equation is.
\begin{table}[!htb]
\begin{center}
\begin{tabular}{| l | l | l |}
\hline
\textbf{Computer} & \textbf{Architecture} & \textbf{Compiler} \\ \hline
MacBook Pro & Intel Core i7-4850HQ  & Clang 9.1 C++ \\
(High Sierra OS) & (``Crystal Well'') & \\ \hline
\textbf{Optimization Flags} & \textbf{Grid Size} & \textbf{Polynomial Order}  \\ \hline
``O3,'' & $16^2 \times 16^3$ & Serendipity quadratic, \\ 
Eigen 3.3.4 for nodal & & 112 degrees of freedom \\\hline
\textbf{Nodal Total Time} & \textbf{Modal Total Time} & \textbf{Total Time Reduction} \\ \hline
\textbf{1079.63} $\frac{\textrm{seconds}}{\textrm{time-step}}$ & \textbf{67.4312} $\frac{\textrm{seconds}}{\textrm{time-step}}$ & $ \sim 16 $ \\ \hline
\textbf{Nodal Vlasov Time} & \textbf{Modal Vlasov Time} & \textbf{Vlasov Time Reduction} \\ \hline
\textbf{1033.89} $\frac{\textrm{seconds}}{\textrm{time-step}}$ & \textbf{60.3431} $\frac{\textrm{seconds}}{\textrm{time-step}}$ & $ \sim 17 $ \\ \hline
\end{tabular}
\caption{Summary of the parameters for the numerical experiment to compare the full cost of an alias-free nodal and orthonormal, modal algorithm.}
\label{table:VlasovSummaryTable}
\end{center}
\end{table}
Again, the details of this comparison are summarized in Table \ref{table:VlasovSummaryTable}.

So, we have achieved our goal of respecting the requirement that our DG method for the VM-FP system of equations be alias-free, while measurably reducing the cost to attain the computational complexity of other common DG schemes which tolerate or simply attempt to control aliasing errors.
Because there were ultimately many pieces to the evaluation of the DG method for the VM-FP system of equations, we summarize in the next section the complete algorithm for computing the spatial discretization and taking a forward Euler time-step.

\section{Summary of the Algorithm}

The focus of this chapter has been principally on the evaluation of the linear operator in \eqr{\ref{eq:linearOpKinetic}} which goes into a forward Euler time-step, \eqr{\ref{eq:forwardEuler}}.
We summarize now all the steps in the evaluation of this linear operator, for the discrete Vlasov--Fokker--Planck equation and Maxwell's equations, so that we can perform a forward Euler time-step.
\begin{enumerate}
    \item Loop over configuration space cells, and for each configuration space cell, compute the needed coupling moments from the distribution functions for each species at the old time-step, $f_h^n$, where superscript $n$ denotes the known time-step.
    \begin{itemize}
        \item Within each configuration space cell, loop over velocity space to compute velocity moments using computational kernels, such as the 1X2V kernel shown in Figure~\ref{fig:1X2V-p1-Moment-Calc}. These kernels will give $M^n_{0_h}, \mvec{M}^n_{1_h}$ and $M^n_{2_h}$, Eqns.\thinspace(\ref{eq:fullM0Calc}--\ref{eq:fullM2Calc}).
        \item Calculate the current density from $\mvec{M}^n_{1_h}$ for each plasma species,
        \begin{align*}
            \mvec{J}^n_{h} = \sum_s q_s \mvec{M}^n_{1_{h_s}}.
        \end{align*}
        \item Calculate the discrete flow and temperature, $\mvec{u}^n_h$ and $T_h^n$, from $M^n_{0_h}, \mvec{M}^n_{1_h}$ and $M^n_{2_h}$, as well as the boundary corrections in velocity space, using computational kernels such as the one shown in Figure~\ref{fig:u-T-calc} for a 1X1V, polynomial order two, simulation. Note that if using piecewise linear polynomials, we require the additional ``star moments'' in the computation of $\mvec{u}^n_h$ and $T_h^n$, Eqns.\thinspace(\ref{eq:M0Star}--\ref{eq:M2Star}).
    \end{itemize}
    \item Loop over configuration space cells and update the electromagnetic fields, $\mvec{E}_h^n, \mvec{B}_h^n$, forward in time.
    \begin{itemize}
        \item Project the chosen numerical flux function for the electric and magnetic fields, central fluxes, Eqns.\thinspace(\ref{eq:centralE})-(\ref{eq:centralB}), or upwind fluxes, Eqns.\thinspace(\ref{eq:r-e2})-(\ref{eq:r-b3}), onto the modal, orthonormal configuration space basis expansion.
        \item Evaluate the volume and surface integrals using the corresponding computational kernels, analogous to the volume and surface tensors for the collisionless Vlasov equation, Eqns.\thinspace(\ref{eq:VlasovSurfaceMatrix}) and (\ref{eq:VlasovVolumeMatrix}), but note that these computational kernels only involve the configuration space basis expansion. After evaluation of the volume and surface integrals, increment the electromagnetic fields with this contribution multiplied by the size of the time-step $\Delta t$,
        \begin{align*}
            \mvec{E}_h^{n+1} = \mvec{E}_h^n + \Delta t \mathcal{L}_{EM}(\mvec{E}_h^n, \mvec{B}_h^n),
        \end{align*}
        and likewise for the magnetic field.
        \item Increment the current density at the known time-step onto the electric field,
        \begin{align}
            \mvec{E}_h^{n+1} = \mvec{E}_h^n + \frac{\Delta t}{\epsilon_0} \mvec{J}_h^n.
        \end{align}
    \end{itemize}
    \item Loop over phase space cells and update the particle distribution function for each species, $f_h^n$, forward in time.
    \begin{itemize}
        \item Project the chosen numerical flux functions for both the collisionless advection and the drag term in the Fokker--Planck equation, e.g., central fluxes, \eqr{\ref{eq:simpleCentralVlasov}}, or global Lax-Friedrichs fluxes, \eqr{\ref{eq:simpleGlobalLFVlasov}}, for the collisionless advection and central fluxes, \eqr{\ref{eq:dragCentral}}, or global Lax-Friedrichs fluxes, \eqr{\ref{eq:dragGlobalLF}}, for the drag term in the Fokker--Planck equation.
        \item Determine the recovered distribution function from the general recovery procedure described in Section~\ref{sec:recoveryHigherDimensions}, i.e., recover a continuous function across the interface from the distribution function in the two neighboring cells, while retaining the phase-space dependence of the distribution function representation on the surface. Compute the value and gradient of the recovered distribution function at the surface, and add these contributions to the numerical flux functions for computing the surface integral contributions to the discrete Fokker--Planck equation.
        \item Evaluate the volume and surface integrals in the DG discretization of the Vlasov--Fokker--Planck equation, e.g., the volume kernel in Figure~\ref{fig:1X2V-p1-vol-update} for a piecewise linear, 1X2V, simulation, and increment these contributions multiplied by the size of the time-step $\Delta t$ onto the old values of the particle distribution function,
        \begin{align*}
            f_h^{n+1} = f_h^n + \Delta t \mathcal{L}_{VFP}(f_h^n, \mvec{E}_h^n, \mvec{B}_h^n, \mvec{u}^n_h, T^n_h).
        \end{align*}
        \item Repeat each calculation, the flux function project, the recovery procedure, and the evaluation of volume and surface integrals, for each species in the plasma.
    \end{itemize}
\end{enumerate}

The above steps form the core of forward Euler time-step, which can then be combined into a multi-stage Runge--Kutta method, such as our preferred three-stage, third order, SSP-RK3 scheme, \eqr{\ref{eq:SSPRK3}}.
Note that for computing the size of the time-step, while the CFL condition for Maxwell's equation at each stage will remain fixed since the speed of light is a constant, we can evaluate the CFL constraint for the Vlasov--Fokker--Planck equation at each stage.
The general structure of this forward Euler method is unchanged, even if we modify components of the update, for example applying the recovery procedure for the update of the advective terms such as the collisionless update of the Vlasov equation.
However, we could modify this update to separate the collisionless and collision operators if an operator split would provide a more favorable time-stepping scheme.
For example, as the collisionality increases and the collision operator becomes the more restrictive component of taking a time-step, standard operator splits that employ Runge--Kutta-Legendre multi-stage methods for advection-diffusion equations are an option \citep{Meyer:2014}.

So, we have formulated and implemented a Runge--Kutta discontinuous Galerkin discretization of the Vlasov--Maxwell--Fokker--Planck system of equations---a sizable effort!
But now we turn to the equally important question: does the code give the right answer?
In the next chapter, Chapter~\ref{ch:Benchmarks}, we will pursue an extensive benchmarking endeavor to determine the validity of our numerical method.
%Chapter 4

\renewcommand{\thechapter}{4}
\epigraph{Some of the material in this chapter has been adapted from \citet{Juno:2018}, \citet*{Hakim:2019}, and \citet{HakimJuno:2020}.}{}
\chapter{Benchmarking our DG Vlasov--Maxwell--Fokker--Planck Solver in \gke}\label{ch:Benchmarks}

We will proceed on three different fronts to determine the validity of our implemented DG scheme for the VM-FP system of equations. 
First, we will examine just the Vlasov--Fokker--Planck equation, in the absence of electromagnetic fields.
Then, we will benchmark the collisionless Vlasov--Maxwell system of equations, with special focus on self-consistent simulations including the feedback between the plasma and the electromagnetic fields.
Finally, we will bring it all together for a benchmark of the complete equation system, a validation of the VM-FP system of equations in their entirety.

We reiterate a few definitions for convenience here.
We will make use of the Maxwellian velocity distribution as a common initial condition,
\begin{align}
    f_s(\mvec{x}, \mvec{v}, t=0) = n_s(\mvec{x}) \left (\frac{m_s}{2 \pi T_s(\mvec{x})} \right )^{\frac{VDIM}{2}} \exp \left (-m_s \frac{|\mvec{v} - \mvec{u}_s(\mvec{x})|^2}{2 T_s(\mvec{x})} \right ), \label{eq:ICMaxwellian}
\end{align}
where $VDIM$ is the number of velocity dimensions.
We note that we will have to project \eqr{\ref{eq:ICMaxwellian}} onto our basis expansion at the start of any simulation.
Although this distribution function defines local thermodynamic equilibrium, as we discussed in Corollary~\ref{coro:HTheorem} in Chapter~\ref{ch:Introduction} and in Appendix~\ref{app:proofsContinuous}, the Maxwellian velocity distribution might have some configuration space dependence that is unstable to perturbations.
The system will then rearrange itself to a different energy state in a collisionless system, and to a higher entropy state in the presence of collisions.
\eqr{\ref{eq:ICMaxwellian}} is thus often a convenient initial condition, though we will make clear when we employ different initial plasma distributions.
We will also use consistently the definition of the thermal velocity $v_{th_s} = \sqrt{T_s/m_s}$, especially to define the extents in velocity space.

Although we will reiterate many of the specifics for every benchmark, we note here a few details which will be unchanged throughout our benchmarks.
We will consistently use the Serendipity element space for our polynomial basis as an optimal middle ground of cost and accuracy between the tensor product basis and the maximal order basis.
As an optimization of the computation and memory required in a multi-stage method, and accuracy of the time integration, we will also employ the three stage, third order, SSP-RK3 method for the time integration of all benchmarks presented.
Importantly, we will use the same numerical flux functions for all the presented benchmarks, upwinding, \eqr{\ref{eq:simpleUpwindVlasov}} for $\gvec{\alpha}^x = \mvec{v}$, the streaming term, and global Lax-Friedrichs for both the acceleration $\gvec{\alpha}^v = q/m \thinspace (\mvec{E}_h + \mvec{v} \times \mvec{B}_h)$, \eqr{\ref{eq:simpleGlobalLFVlasov}}, and the drag term, \eqr{\ref{eq:dragGlobalLF}}.
Finally, we will uniformly use zero-flux boundary conditions in velocity space, with the additional boundary term we must evaluate in the Fokker--Planck operator due to integrating by parts twice, \eqr{\ref{eq:additionalFPBoundaryTerm}}, so as to retain the proved conservation properties in Chapter~\ref{ch:DGFEM}. 
When we refer to zero flux boundary conditions in velocity space in all of the forthcoming boundary conditions, and when numerically integrating the discrete Fokker--Planck equation in Sections~\ref{sec:benchmarkVFP} and \ref{sec:benchmarkFullVMFP}, we are implicitly also taking into account this additional boundary condition in \eqr{\ref{eq:additionalFPBoundaryTerm}}.

\section{Benchmarks of the Vlasov--Fokker--Planck Equation}\label{sec:benchmarkVFP}

\subsection{Collisional Relaxation to a Discrete Maxwellian}

In the absence of streaming and body forces, any initial distribution function should relax to a Maxwellian.
Although we did not demonstrate this to be the case via analytic examination of our discretization of the Fokker--Planck equation, we now consider a numerical demonstration of a discrete analog to the H-theorem proved in Corollary~\ref{coro:HTheorem} in Chapter~\ref{ch:Introduction}.
Importantly, a proper implementation of the discrete Fokker--Planck equation has a maximum entropy state, which \textit{by definition} is the discrete Maxwellian.
However, such a discrete Maxwellian is not necessarily the projection of \eqr{\ref{eq:ICMaxwellian}} onto basis functions, as \eqr{\ref{eq:ICMaxwellian}} is a continuous function defined on all of velocity space, $\mvec{v} \in (-\infty, \infty)$, and we are employing finite velocity space extents.
Nevertheless, these two quantities, the projection of \eqr{\ref{eq:ICMaxwellian}} and the maximum entropy state of our discrete Fokker--Planck operator will converge towards each other as the grid is refined.

In this first test, the relaxation of an initial non-Maxwellian distribution function to a discrete Maxwellian, due to collisions, is studied.
We will avoid the use of a species index in this test since the electromagnetic fields are zero, and we are only studying the effects of the collision operator.
The initial distribution function is a step-function in velocity space,
\begin{align}
f_0(x,v,t=0) = \begin{cases}
1/(2v_0) &\qquad |v|<v_0 \\
0 &\qquad |v| \geq v_0,
\end{cases}    
\end{align}
where $v_0 = \sqrt{3} v_{th}$. 
Piecewise linear and quadratic Serendipity basis sets on 16 and 8 velocity space cells, respectively, are used.
Note that there is no variation in configuration space in this problem, so only one configuration space cell is required.
Velocity space extents, ($v_{\min},v_{\max}$), are placed at $\pm6 v_{th}$ the simulation is run to $\nu t = 5$, five collisional periods, and zero flux boundary conditions are used in velocity space. 

In each case, the relative change in density and energy are close to machine precision, demonstrating excellent conservation properties of the scheme. 
In Figure~\ref{fig:sq-relax-cons}, the time-history of the error in normalized energy change is plotted.
The errors per time-step of the conservative scheme are machine precision, and the small change in energy is due to the numerical diffusion inherent to the SPP-RK3 scheme. 
\begin{figure}[!htb]
    \centering
    \includegraphics[width=0.49\textwidth]{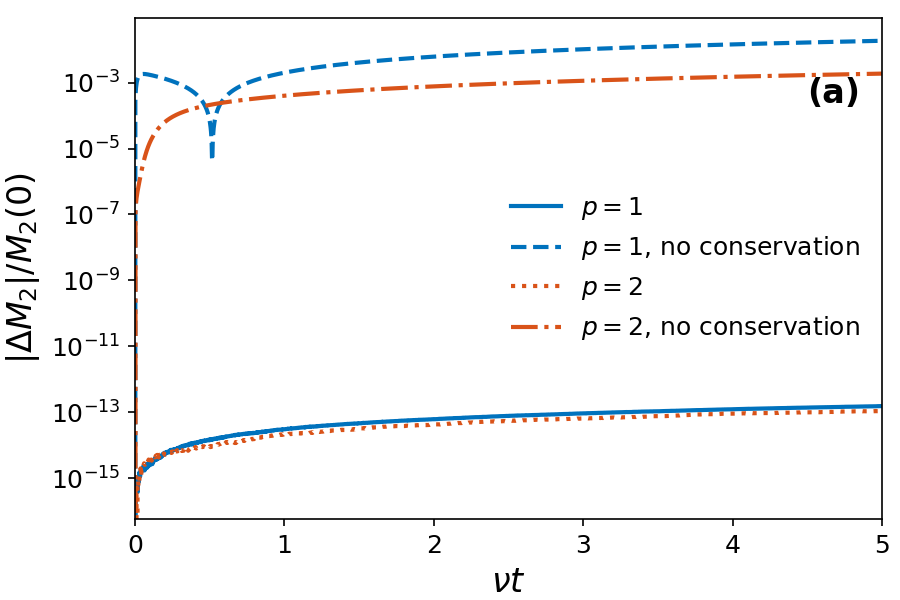}
    \includegraphics[width=0.49\textwidth]{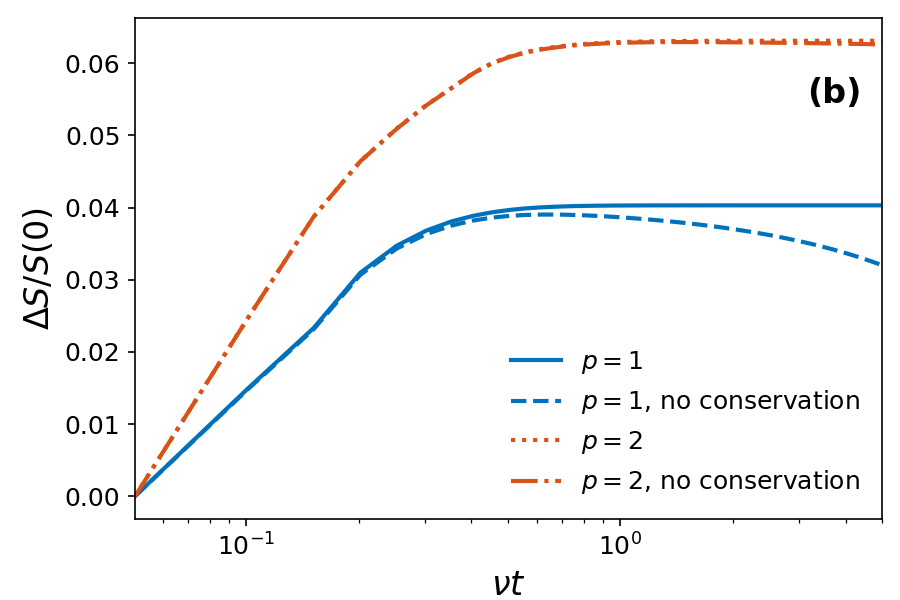}
    \caption{(a) Relative change in energy, $\Delta M_2/M_2 (t=0) = [M_2(t) - M_2 (t=0)]/M_2(t=0)$, for $p=1$, $N=16$ (solid and dashed blue) and $p=2$, $N=8$ (dotted and dash-dot orange) cases for relaxation of a square distribution to a discrete Maxwellian. The decrease in energy in our conservative scheme is close to machine precision. The curves labeled `no conservation' omit the boundary correction terms and use regular moments instead of ``star moments'' (for $p=1$) needed for momentum and energy conservation. (b) Time-history of relative change in entropy. When using the conservative scheme, the entropy rapidly increases and remains constant once the distribution function becomes a discrete Maxwellian.}
    \label{fig:sq-relax-cons}
\end{figure}
For fixed time-step size, changing resolution or polynomial order has little impact on the magnitude of energy errors, and they always remain close to machine precision. 

Figure~\ref{fig:sq-relax-cons} also shows that as the distribution function relaxes, the entropy rapidly increases and then remains constant once the discrete Maxwellian state is obtained.
The change in entropy between $p=1$ and $p=2$ is indicative that different discrete Maxwellians will be obtained depending on grid resolution and polynomial order. 
The same figure shows that neglecting the boundary corrections and ``star moments'' (for $p=1$) needed for conservation degrade energy conservation by many orders of magnitude, and in the $p=1$ case, can even lead to decreasing entropy.
In fact, the violation of the second law of thermodynamics when neglecting the boundary corrections to the drag and diffusion coefficients provides solid evidence that the care taken in accounting for the finite velocity space extents in formulating the scheme in Chapter~\ref{ch:DGFEM} produces a more reliable scheme for the physics content of the equation system.
Note that this is not a good test for momentum conservation, because the initial momentum is zero.

We now consider relaxation in a 1X2V setting. 
For this test, the initial condition is selected as a sum of two Maxwellians, the first with drift velocity $\mvec{u} = (3 v_{th},0)$ and the second with drift velocity $\mvec{u} = (0,3v_{th})$. 
Both Maxwellians have a thermal speed of $v_{th} = 1/2$. 
A $16^2$ grid in velocity space with $p=2$ Serendipity basis functions is used.
Again, there is no variation in configuration space in this problem, so only one configuration space cell is required.

As the particles collide, the distribution function will relax to a new Maxwellian with non-zero drift and different temperature, thus allowing us to test momentum conservation. 
The simulation is run to $\nu t= 5$, five collisional periods. 
Figure~\ref{fig:bi-relax-cons} shows the initial and final distribution function demonstrating the relaxation to the discrete Maxwellian. 
The errors in the energy and the $x$- and $y$-components of momentum are close to machine precision for our conservative scheme, as shown in panel (c). 
Neglecting boundary correction terms degrades conservation by many orders of magnitude. 
Also, panel (d) demonstrates that the entropy increases monotonically, reaching its steady-state value once the discrete Maxwellian is obtained.
\begin{figure}[!htb]
    \centering
    \includegraphics[width=0.9\textwidth]{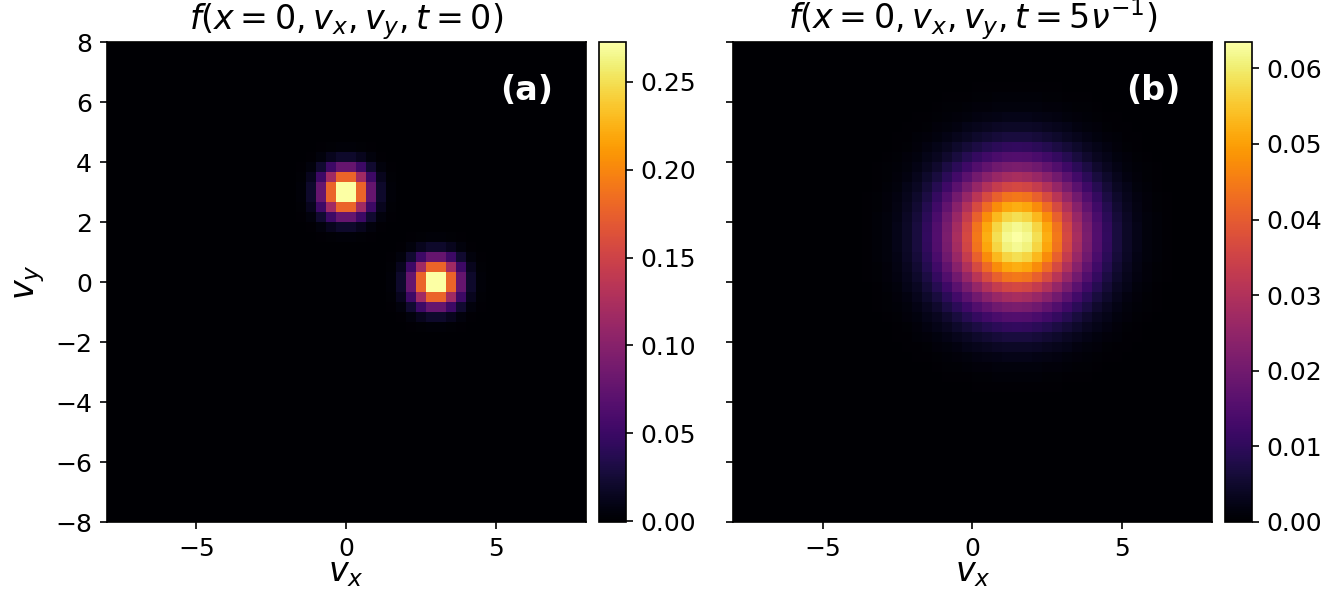}
    \includegraphics[width=0.49\textwidth]{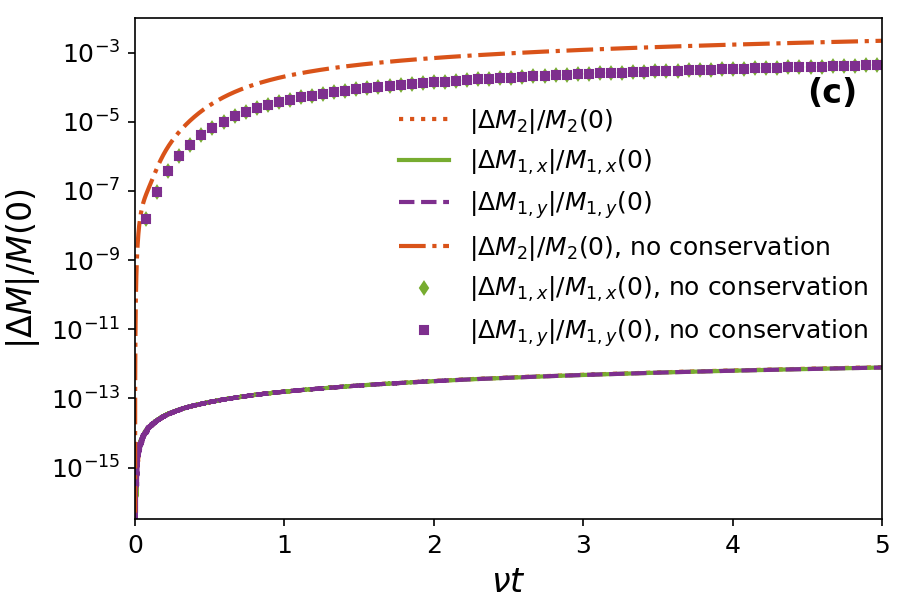}
    \includegraphics[width=0.49\textwidth]{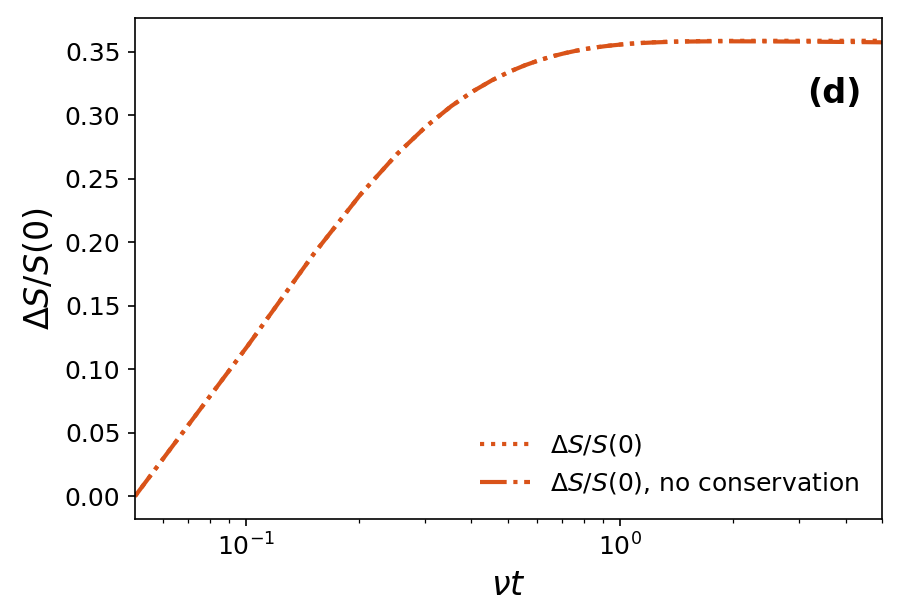}
    \caption{The initial (a), relaxed (b) distribution function in a 1X2V relaxation test. Conservation (c) of energy (orange) and momentum (green, purple) is at machine precision for our conservative scheme. Neglecting boundary corrections breaks conservation by more than 8 orders of magnitude. Purple and green curves overlay each other on this scale. (d) The entropy increases rapidly and then remains constant once the discrete Maxwellian is obtained.}
    \label{fig:bi-relax-cons}
\end{figure}
These tests demonstrate the high accuracy with which the moments are conserved as well as providing empirical evidence that entropy is a non-decreasing function of time, so long as we are careful to include the corrections in the computation of the moments and additional boundary condition which arise from solving the Vlasov--Fokker--Planck equation on a finite velocity grid.

\subsection{Kinetic Sod-Shock}\label{sec:KineticSodShock}

We now add in the streaming of particles in configuration space, $\gvec{\alpha}^x_h = \mvec{v}$, while keeping the electromagnetic fields zero, to test the accuracy of our DG Vlasov--Fokker--Planck equation in the presence of spatial gradients.
In this benchmark, we study shock structure in the kinetic regime with the classic Sod-shock \citep{Sod:1978} initial conditions in one spatial dimension and one velocity dimension (1X1V),
\begin{align}
  \left[
    \begin{matrix}
      \rho_l \\
      u_l \\
      p_l
    \end{matrix}
  \right]
  = 
  \left[
    \begin{matrix}
      1 \\
      0.0 \\
      1.0
    \end{matrix}
  \right],
  \qquad
  \left[
    \begin{matrix}
      \rho_r \\
      u_r \\
      p_r
    \end{matrix}
  \right]
  = 
  \left[
    \begin{matrix}
      0.125 \\
      0.0 \\
      0.1
    \end{matrix}
  \right],
\end{align}
where this mass density, flow, and pressure are used to initialize the Maxwellian velocity distribution defined in \eqr{\ref{eq:ICMaxwellian}} on the left, $l$, and right, $r$, sides of the domain.
The phase space domain is $[0,L]$ in configuration space and $[-6v_{th,l},6v_{th,l}]$ in velocity space, with $v^2_{th,l}=p_l/\rho_l=1$ since $p_l = n_l T_l = 1$ and $\rho_l = m n_l = 1$, and we initialize the discontinuity to be at $x=L/2$.
Note that for this 1X1V system, the gas adiabatic constant is $\gamma = 3$ because the internal energy is defined as $p/(\gamma-1) = N \rho v_{th}^2/2$, $N=1$ in one dimension, and upon rearranging, we find $\gamma = 3$.
The simulations were run on a $64\times 16$ grid, with piecewise quadratic Serendipity elements, $L=1$, and $t_{end}=0.1$.
Zero flux boundary conditions are used in velocity space and copy boundary conditions are used for configuration space, where the value of the distribution function at $x = 0$ and $x = L$ is copied into the ghost layer for the computation of the fluxes at the configuration space boundary.
Note that this copy boundary condition copies the full expansion of the distribution function from the skin cells at $x = 0$ and $x = L$ into the ghost layer, and so is not the same as a homogeneous Neumann boundary condition, but more akin to a perfectly matched layer, i.e., an open boundary condition.

The Knudsen number ($\mathrm{Kn} = \lambda_\textrm{mfp}/L$, where $\nu = v_{th}/\lambda_\textrm{mfp}$) is varied between $1/10$, $1/100$ and $1/500$. 
In the first case, the gas is close to collisionless on the time-scale of the simulation, as the box size is not much greater than the mean free path of particle collisions, but in the last case, the gas is highly collisional, because the particles undergo many collisions while propagating through the box, $L \gg \lambda_{mfp}$.
Hence, in the last case the solution should match, approximately, the solution from the Euler equations for the evolution of a fluid\footnote{We note that the Euler equations are formally derived with the full Boltzmann collision operator accounting for hard sphere collisions of gas particles, and then taking viscosity and heat conduction to be zero. In this case, even the simplified Fokker--Planck operator leads to a high collisionality limit. However, the transport coefficients for matching a Navier-Stokes solution with finite viscosity and heat conduction, i.e., finite momentum and heat transport, would need to be modified to account for this particular collision operator.}.

Figure~\ref{fig:sod-shock-moments} shows the density, velocity, temperature and gas frame, or kinetic, heat-flux,
\begin{align}
    q_h(x,t) \doteq \sum_j \int_{K_j \setminus \Omega_k} (v-u_h(x,t))^3\, f_h(x,v,t) \thinspace dv,
\end{align}
obtained from the kinetic simulations. 
\begin{figure}[!htb]
    \centering
    \includegraphics[width=\textwidth]{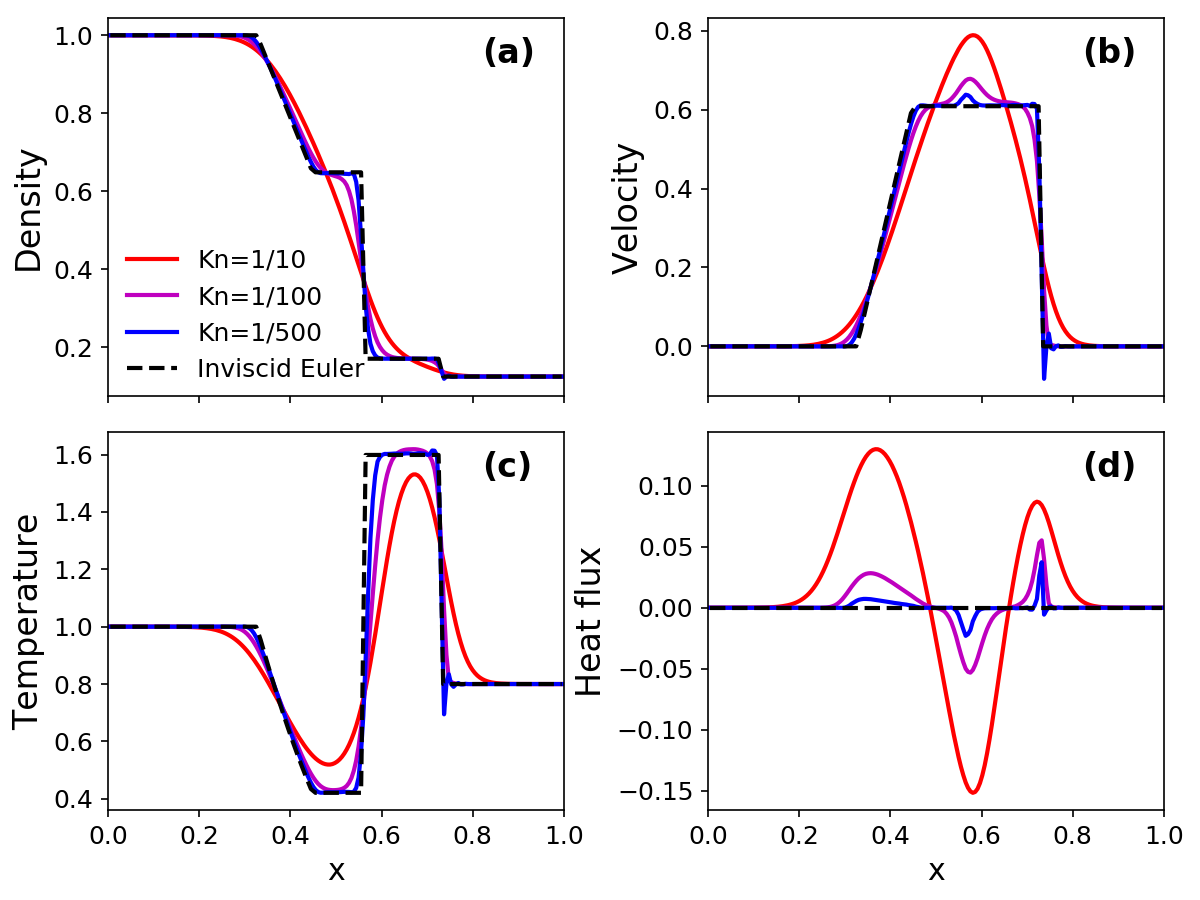}
    \caption{Density (a), velocity (b), temperature (c) and gas frame, or kinetic, heat-flux (d) from a Sod-Shock problem. Plotted are results with Knudsen numbers of $1/10$ (red), $1/100$ (magenta), and $1/500$ (blue), with the inviscid Euler results (black dashed) shown for comparison. As the gas becomes more collisional, i.e., decreasing Knudsen number, the solutions tend to the Euler result. Note that there is no heat-flux in the inviscid limit.}
    \label{fig:sod-shock-moments}
\end{figure}
For comparison, the exact solution to the corresponding inviscid Euler Riemann problem is also shown.
It is observed, as expected, that as the gas becomes more collisional, the moments tend to the Euler solution.
An interesting aspect of the kinetic results, though, are the viscosity, heat-conductivity and other transport effects which smooth the shock structures that are sharp in the Euler solution. 
In particular, the lower-right plot of Figure~\ref{fig:sod-shock-moments} shows that the heat-flux is completely absent in the inviscid equations. 
There is significant heat-flux in the low collisionality case, but this heat flux vanishes as the collisionality increases.
It is a testament to the accuracy of our discrete Vlasov--Fokker--Planck implementation that we can transition from the low to high collisionality limit, comparing favorably with the Euler equation solution in the high collisionality limit. 

We next consider a Sod-shock with a sonic point in the rarefaction wave. The initial conditions are selected as
\begin{align}
  \left[
    \begin{matrix}
      \rho_l \\
      u_l \\
      p_l
    \end{matrix}
  \right]
  = 
  \left[
    \begin{matrix}
      1 \\
      0.75 \\
      1.0
    \end{matrix}
  \right],
  \qquad
  \left[
    \begin{matrix}
      \rho_r \\
      u_r \\
      p_r
    \end{matrix}
  \right]
  = 
  \left[
    \begin{matrix}
      0.125 \\
      0.0 \\
      0.1
    \end{matrix}
  \right],    
\end{align}
and this mass density, flow, and pressure are again used to construct an initial Maxwellian velocity distribution, \eqr{\ref{eq:ICMaxwellian}}.
We employ the same $64 \times 16$ grid with piecewise quadratic Serendipity elements, $[-6 v_{th,l}, 6 v_{th,l}]$ velocity space extents, and zero-flux boundary conditions in velocity space.
In contrast to the standard Sod-shock, this problem is run on a periodic domain $[-1,1]$, with the ``left'' state applied for $|x|<0.3$. 
The Knudsen number is $1/200$ and the simulation is run to $t=0.1$. 
As the domain is periodic, the total momentum and energy should remain constant, thereby testing conservation properties in a more complex setting. 
Note that the net momentum is not zero in this problem, which, combined with the configuration space variation that develops in this benchmark, makes this a more strenuous test of momentum conservation compared to the relaxation test. 

Figure~\ref{fig:sod-distf} shows the density, velocity, and distribution function at $t=0.1$.
\begin{figure}[!htb]
    \centering
    \includegraphics[width=\textwidth]{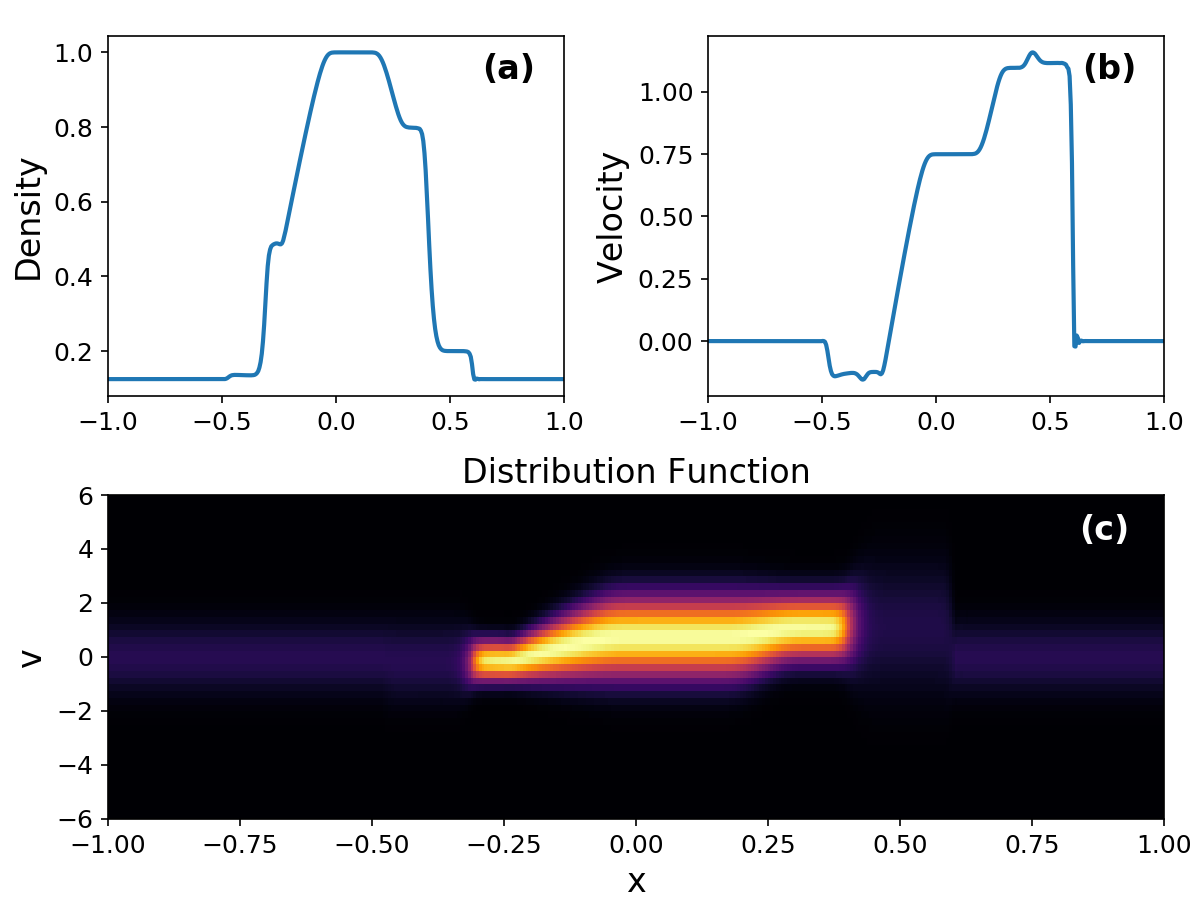}
    \caption{The density (a), velocity (b), and distribution function (c) for the Sod-shock problem with a sonic point in the rarefaction. Complicated shock structures are formed and are visible both in the moments as well as the distribution function.}
    \label{fig:sod-distf}
\end{figure}
Complex shock structures are visible both in the moments and the distribution function.
Figure~\ref{fig:sod-cons} shows the errors in momentum and energy as a function of time for $p=1$ and $p=2$ cases.
In each case, the errors are close to machine precision when using our conservative scheme, but neglecting boundary corrections and using regular moments instead of `star moments' (for $p=1$) leads to errors many orders of magnitude greater.
\begin{figure}[!htb]
    \centering
    \includegraphics[width=0.49\textwidth]{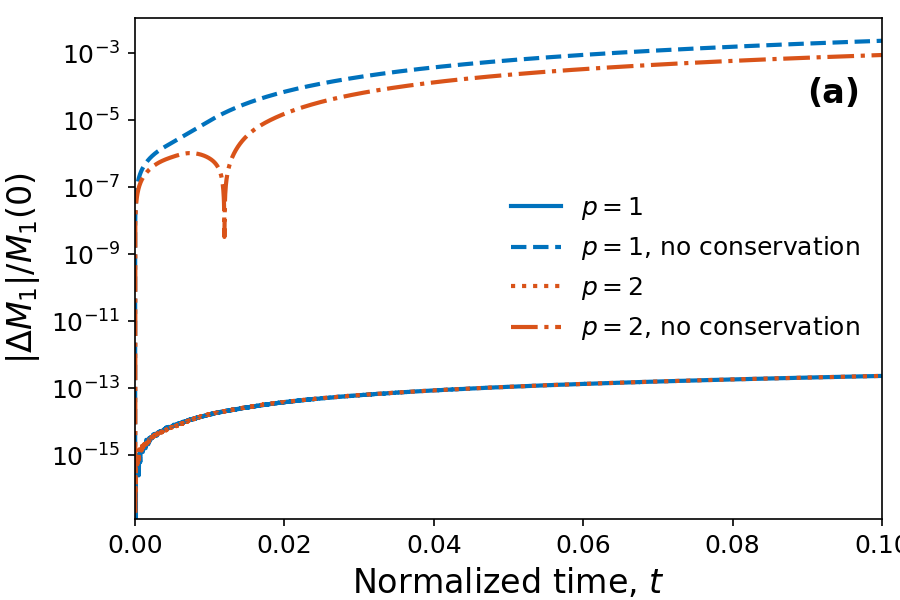}
    \includegraphics[width=0.49\textwidth]{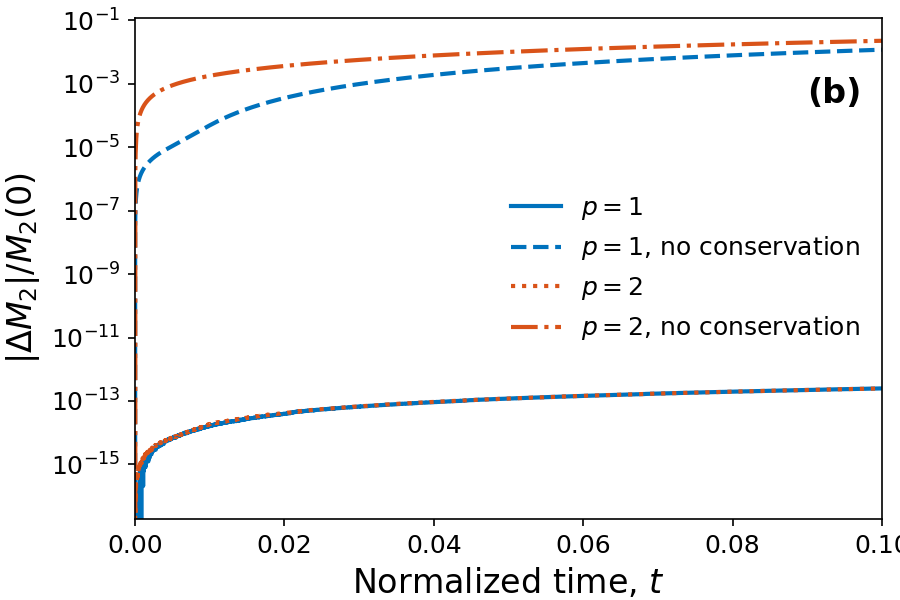}
    \caption{The relative change in the momentum (a) and energy (b) for $p=1$ (blue) and $p=2$ (orange) cases for the Sod-shock problem with a sonic point in the rarefaction. Our conservative scheme gives us machine precision errors in momentum and energy errors that are nearly independent of polynomial order and only depend on the number of time-steps taken in each simulation. However, neglecting the boundary corrections needed for conservation leads to errors orders of magnitude greater.}
    \label{fig:sod-cons}
\end{figure}
We note that, even in the presence of spatial gradients, the errors are independent of polynomial order and only depend on the number of time-steps taken in the simulations, as we expect from our mathematical formulation in the algorithm in Chapter~\ref{ch:DGFEM}.
So, not only do we converge to the inviscid Euler solution in the limit of high collisionality as we expect, but momentum and energy are conserved to a high precision by the scheme.
Importantly, while we did not discuss the limit of no electromagnetic fields in Chapter~\ref{ch:DGFEM} when we discussed momentum conservation in the discrete scheme, we did note that the errors in momentum conservation arose from our discretization of Maxwell's equations.
Thus, we find here by numerical demonstration that our DG discretization of the Vlasov--Fokker--Planck equation in the limit of $\mvec{E} = \mvec{B} = 0$ exactly conserves the momentum, in addition to the energy, as the momentum conservation errors in the relaxation test and kinetic Sod-shock benchmark are only a function of the size of the time-step.

\section{Benchmarks of the Collisionless Vlasov--Maxwell \\ System of Equations}

\subsection{Conservation Test for the Vlasov--Maxwell System of Equations}

To test the conservation properties of the discrete Vlasov--Maxwell system of equations, we set up a drifting electron-proton plasma with a large density gradient in both species to drive strong asymmetric flows.
We initialize a Maxwellian velocity distribution, \eqr{\ref{eq:ICMaxwellian}}, for both protons and electrons with a density gradient,
\begin{align}
n(x, t = 0) & = n_0 (1 + 4\exp(-\beta_l (x - x_m)^2)) \qquad x < x_m, \notag \\
& =  n_0 (1 + 4\exp(-\beta_r (x - x_m)^2)) \qquad x>x_m, \label{eq:conservationTestDensityProfile},
\end{align}
in a 1X1V box.
The phase space domain is $L_x = 96 \lambda_D$ with velocity space extents $[-5.0 v_{th_e}, 7.0 v_{th_e}]$ and $[-6.0 v_{th_p}+v_{th_e}, 6.0 v_{th_p}+v_{th_e}]$ for the electrons and protons respectively.
Here, $\lambda_D$ is the Debye length, \eqr{\ref{eq:DebyeLengthDef}}, and $v_{th_e}$ and $v_{th_p}$ are the electron and proton thermal velocities.

We set $\beta_l = 0.5 \lambda_D^{-2}$, $\beta_r = 0.03125 \lambda_D^{-2}$, $x_m = L_x/4 = 24 \lambda_D$, and $n_0 = 1$ in \eqr{\ref{eq:conservationTestDensityProfile}}.
There is a constant drift in both the protons and electrons, $u(x, t=0) = v_{th_e}$, and the following parameters are chosen: $m_p/m_e = 1836$, $T_p/T_e = 1.0$, and $v_{th_e} = 1.0$. 
The latter is a normalization such that the velocity normalization in the system is the electron thermal velocity, a reasonable choice in 1X1V when Maxwell's equations reduce to just Ampere's Law,
\begin{align}
    \pfrac{\mvec{E}}{t} = \frac{\mvec{J}}{\epsilon_0},
\end{align}
and thus there are no light waves in the system. 

We employ periodic boundary conditions in $x$ and zero-flux boundary conditions in $v_x$, though we note that this density gradient is not periodic.
However, the value of the gradient at the edge of configuration space is small, far below machine precision.
To demonstrate energy conservation, irrespective of configuration space resolution or polynomial order, we perform a number of simulations with $N_x = 4$, $\Delta x = 24 \lambda_D$, and $N_v = 12$, $\Delta v = 1 v_{th_s}$. 
Simulations are run for $1000 \omega_{pe}^{-1}$, where $\omega_{pe}$ is the electron plasma frequency, \eqr{\ref{eq:plasmaFrequency}}.
\begin{figure}[!htb]
    \centering
    \includegraphics[width=0.49\textwidth]{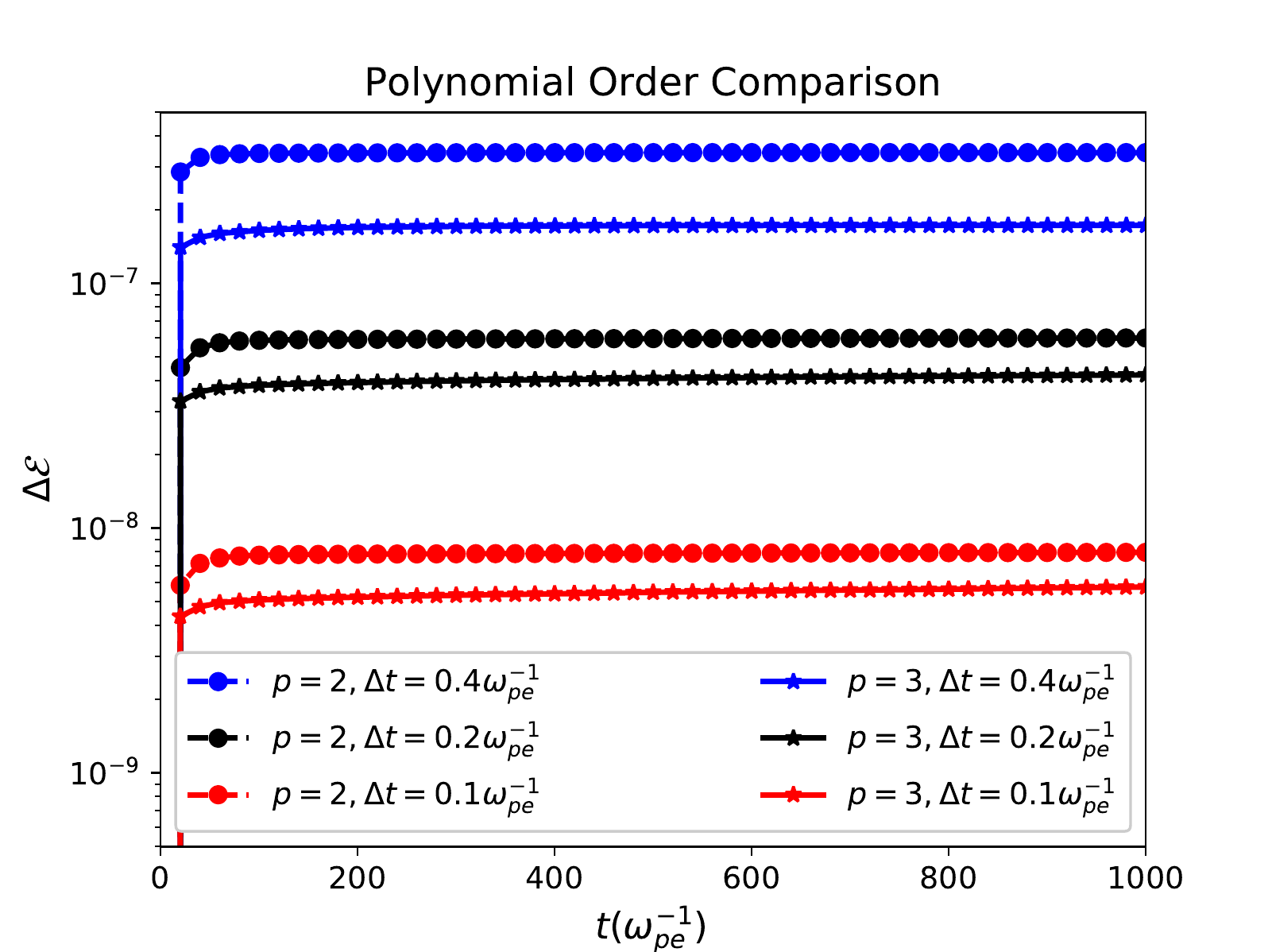}
    \includegraphics[width=0.49\textwidth]{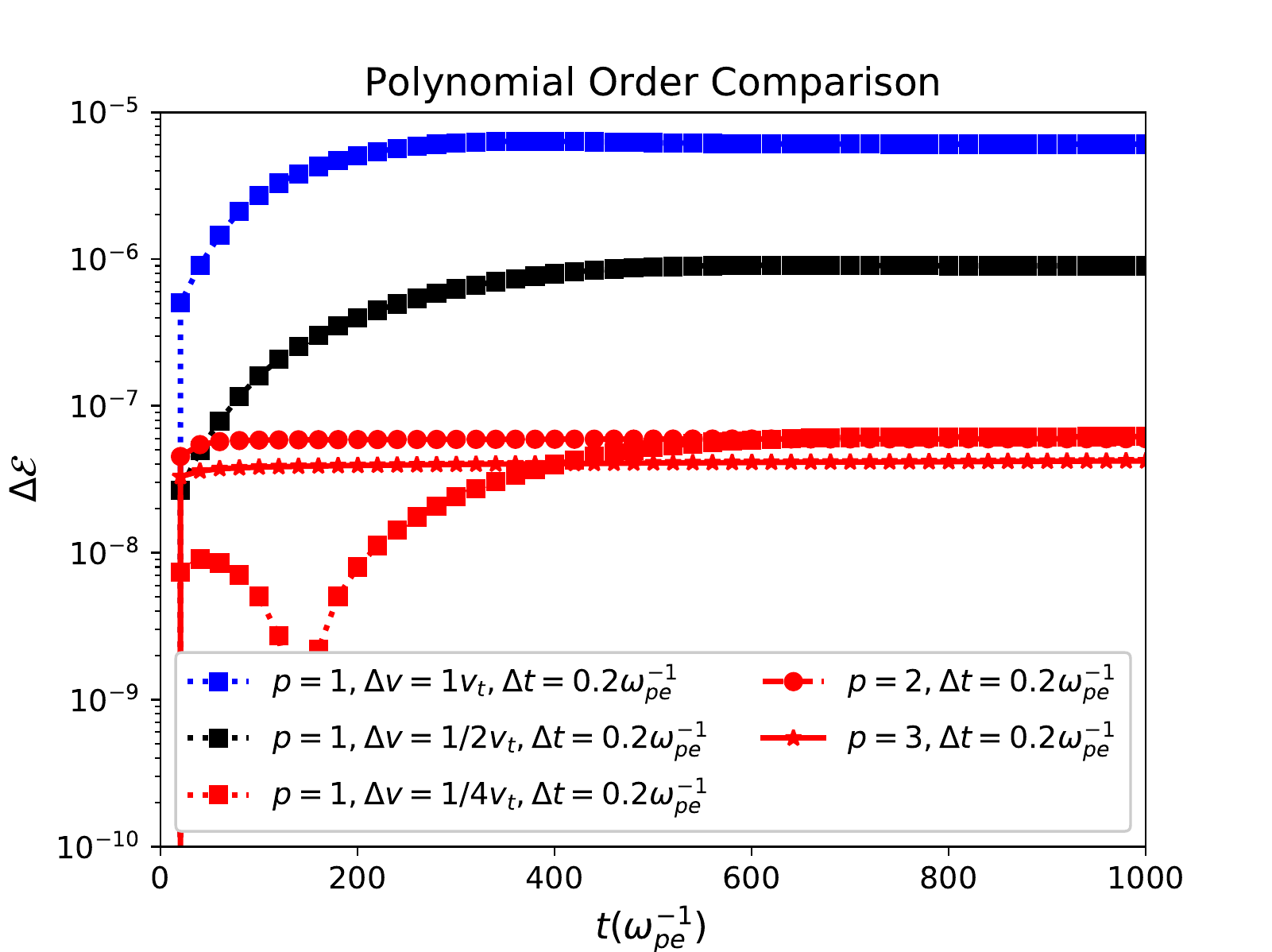}\\
    \includegraphics[width=0.49\textwidth]{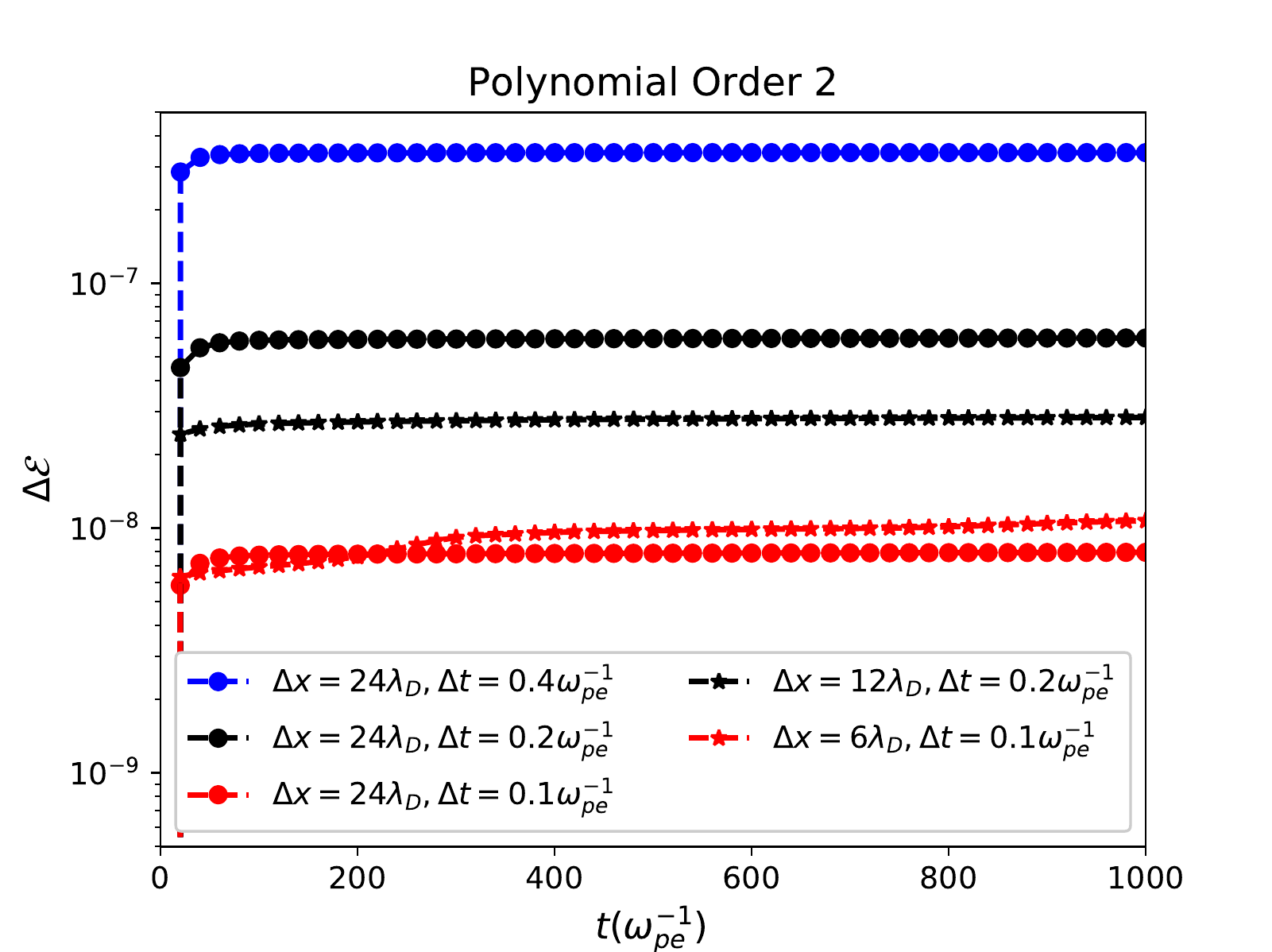}
    \includegraphics[width=0.49\textwidth]{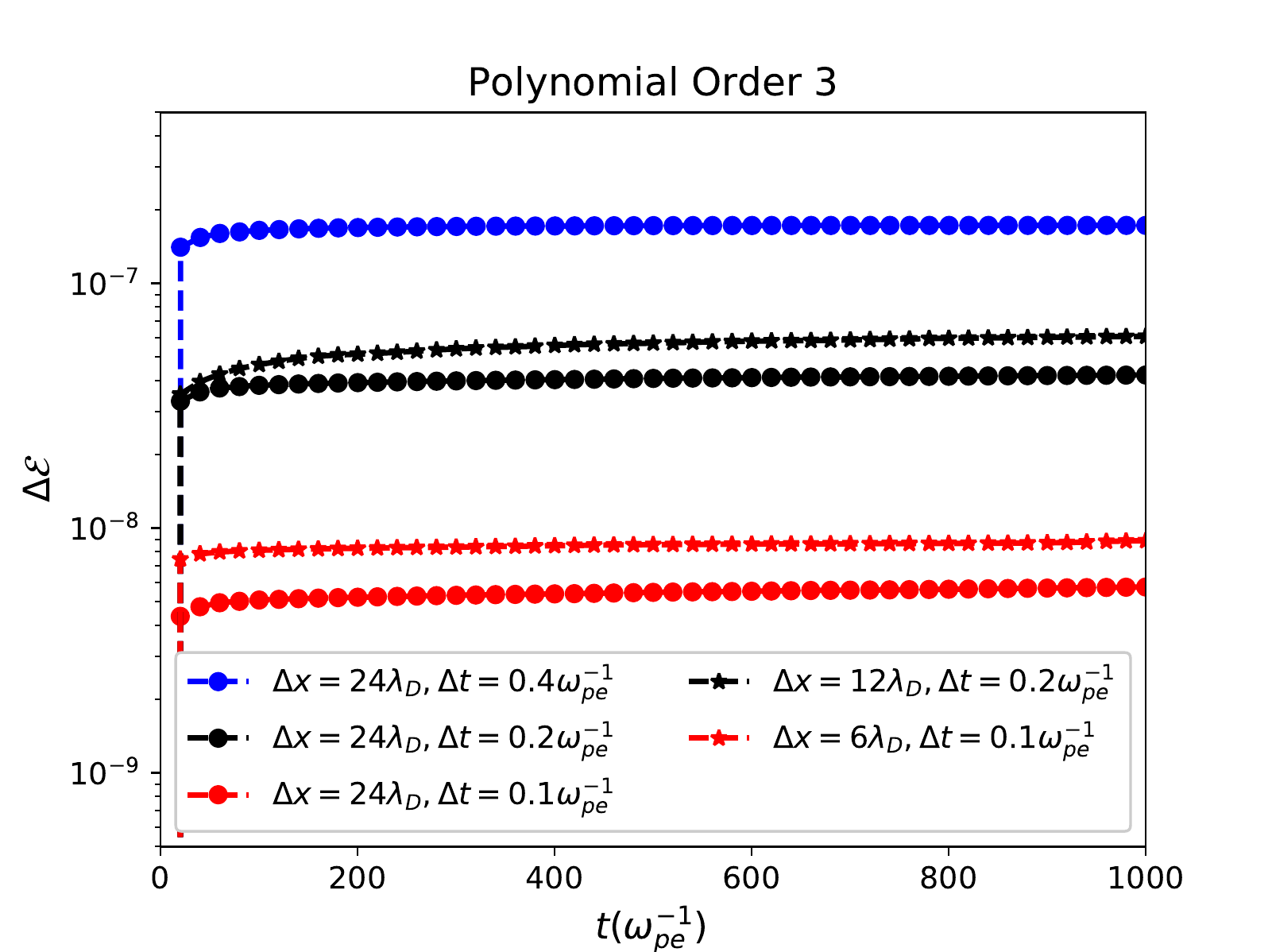}
    \caption{The change in the total, electron plus proton and electromagnetic, energy for a number of simulations to demonstrate the robustness of our energy conserving scheme. The scheme's energy conservation is independent of the polynomial order (top left/right), with the caveat that the choice of polynomial order 1 requires sufficient velocity resolution to reduce the projection errors in projecting $|\mvec{v}|^2$. The latter caveat of projection errors in the polynomial order 1 simulations is also the reason for the dip in the most resolved polynomial order 1 calculation, where the computation of errors is the most sensitive and we must be careful about finite precision effects. We note though that for fixed time-step we recover the energy conservation result of $p=2$ and $p=3$ if we use enough velocity space resolution with the $p=1$ simulations. Likewise, the scheme's energy conservation depends only the size of the time-step, not the configuration space resolution (bottom left/right). The convergence of the energy errors in the top left plot match our expectations for a third order time-stepping method, 2.5 and 2.9 for $p=2$, and 2.0 and 2.9 for $p=3$.} \label{fig:EnergyConservation}
\end{figure}
Results are plotted in Figure~\ref{fig:EnergyConservation}, where the change in the total energy is defined as
\begin{align}
    \Delta \mathcal{E} = \left | \frac{\int_0^{L_x} \mathcal{E}(t) - \mathcal{E}(t=0) \thinspace d\mvec{x}}{\int_0^{L_x} \mathcal{E}(t=0) \thinspace d\mvec{x}} \right |, 
\end{align}
with
\begin{align}
    \mathcal{E} = \frac{1}{2} m_e \int_{\mvec{v}_{min}}^{\mvec{v}_{max}} |\mvec{v}|^2 f_e \thinspace d\mvec{v} + \frac{1}{2} m_p \int_{\mvec{v}_{min}}^{\mvec{v}_{max}} |\mvec{v}|^2 f_p \thinspace d\mvec{v} + \frac{1}{2} \epsilon_0 |\mvec{E}|^2.
\end{align}
Note that the absolute value in the definition of the relative energy change is due to the fact that the total energy decreases with time. 

We emphasize a number of results. Defining the convergence order as,
\begin{align}
    \mathcal{C}(\mathcal{E}_1, \mathcal{E}_2) = \log_2\left ( \frac{E_1}{E_2} \right ) =  \frac{\log(\mathcal{E}_1) - \log(\mathcal{E}_2)}{\log(2)}, \label{eq:convergenceOrder}
\end{align}
we find the order of convergence with decreasing time-step to match our expectations for a third-order Runge-Kutta method, 2.5 and 2.9 for $p=2$, and 2.0 and 2.9 for $p=3$.
In addition, the energy conservation errors are independent of choice of polynomial order.
We note in particular that energy can be conserved with polynomial order 1, but depending on the size of the time-step, one may require more velocity resolution so that projection errors from projecting $|\mvec{v}|^2$ onto linear basis functions do not dominate the error in the computation of the energy. 
Finally, as expected, the conservation of energy is determined by the error in the time-stepping scheme, and refining the grid and increasing the configuration space resolution from $N_x = 4$ to $N_x = 8, 16$ does not improve the energy conservation compared to decreasing the size of the time-step.

We can likewise examine the extent to which momentum is conserved, even though our algorithm does not formally conserve the total momentum.
\begin{figure}[!htb]
    \centering
    \includegraphics[width=0.49\textwidth]{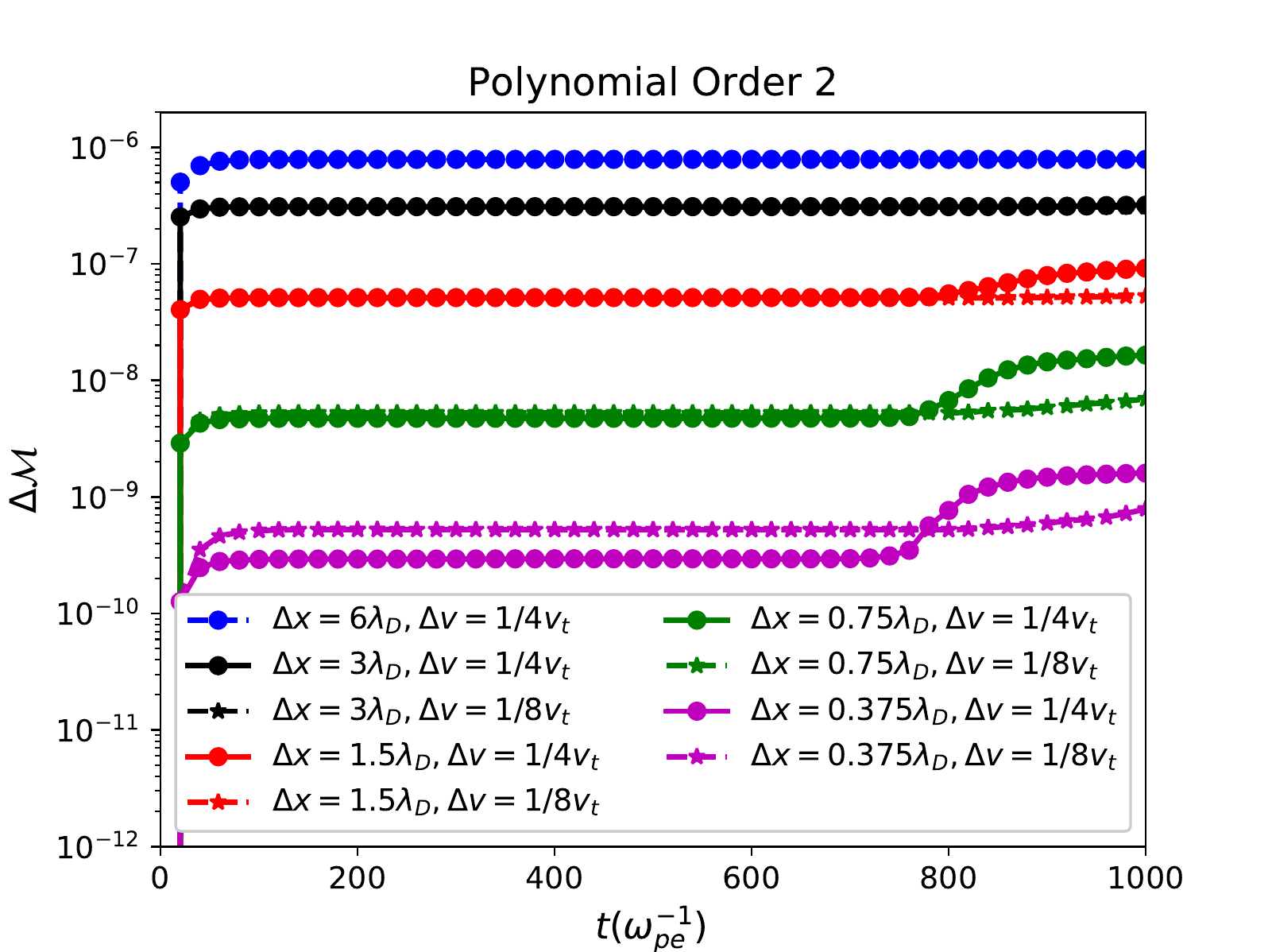}
    \includegraphics[width=0.49\textwidth]{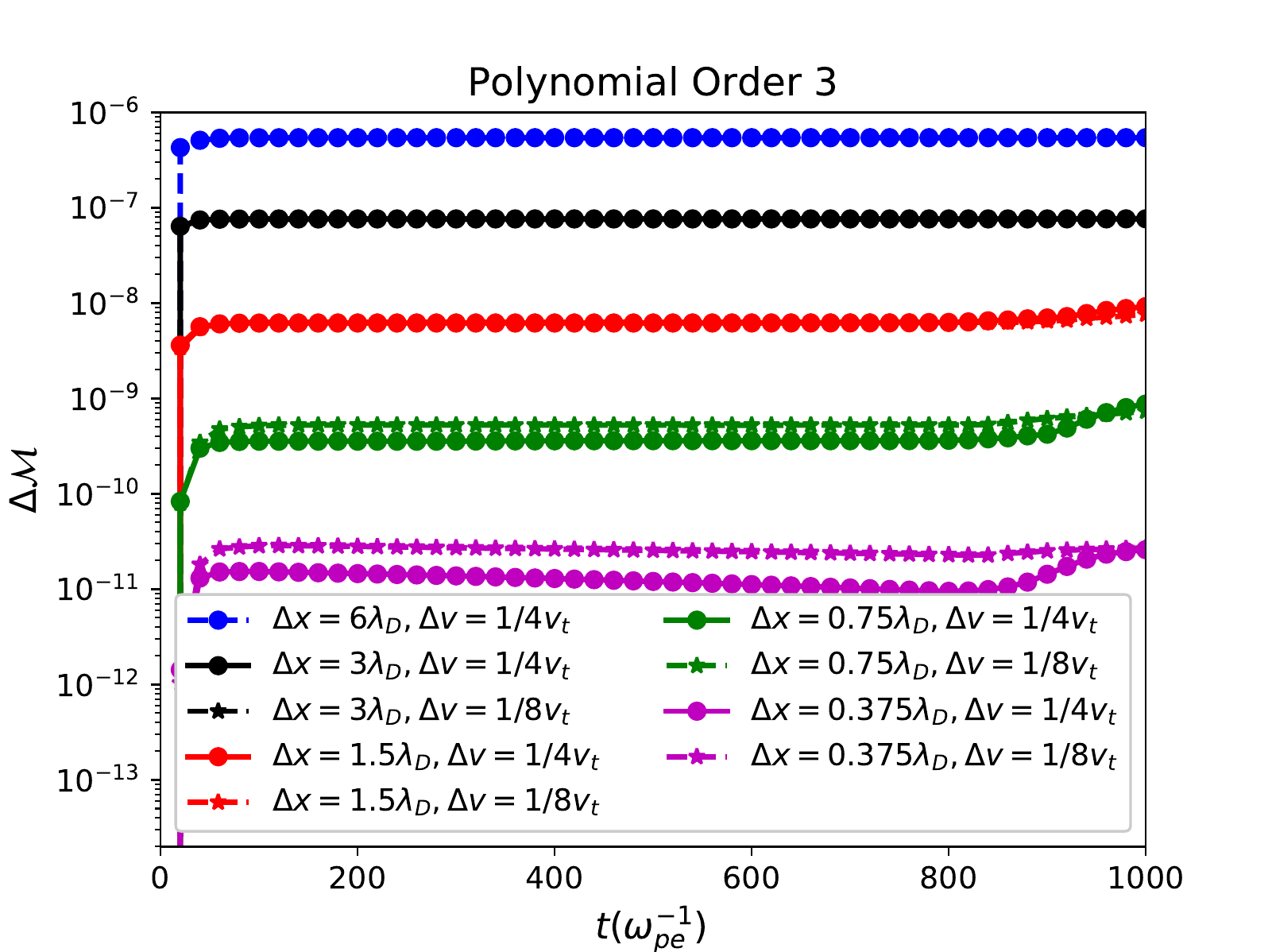}
    \caption{The change in the total, electron plus proton, momentum in a number of simulations. Simulations with polynomial order 2 (left) and polynomial order 3 (right) are performed with increasing configuration space and velocity space resolution to demonstrate that errors in the total momentum decrease with increasing configuration space resolution, while only weakly depending on velocity space resolution. The convergence orders of the polynomial order 2 simulations are 1.35, 2.55, 2.93, and 3.14, and the convergence orders of the polynomial order 3 simulations are 2.83, 3.32, 3.38, and 4.76, and these convergence orders are calculated using the higher velocity resolution results. We note the convergence orders are largely unaffected by using the lower velocity resolution simulations to compute them.} \label{fig:MomentumConservation}
\end{figure}
In Figure~\ref{fig:MomentumConservation}, we plot the integrated total momentum, relative to the total momentum at the beginning of the simulation,
\begin{align}
    \Delta \mathcal{M} = \left | \frac{\int_0^{L_x} \mathcal{M}(t) - \mathcal{M}(t=0) \thinspace d\mvec{x}}{\int_0^{L_x} \mathcal{M}(t=0) \thinspace d\mvec{x}} \right | \label{eq:TotalMomenumChange},
\end{align}
where
\begin{align}
    \mathcal{M} = m_e \int_{\mvec{v}_{min}}^{\mvec{v}_{max}} |\mvec{v}| f_e \thinspace d\mvec{v} + m_p \int_{\mvec{v}_{min}}^{\mvec{v}_{max}} |\mvec{v}| f_p \thinspace d\mvec{v},
\end{align}
is the total, electron plus proton, momentum. We note again the absolute value in \eqr{\ref{eq:TotalMomenumChange}} is due to the fact that the total momentum decreases with time. 
While we cannot show that our scheme conserves the total momentum, the errors in the total momentum converge rapidly with increasing configuration space resolution, and depend only weakly on resolution in velocity space. 
The convergence order as defined by \eqr{\ref{eq:convergenceOrder}} are 1.35, 2.55, 2.93, and 3.14 for $p=2$, and 2.83, 3.32, 3.38, and 4.76 for $p=3$, calculated using the higher velocity resolution results, though one can use the lower velocity resolution results and obtain virtually identical convergence rates.
We have thus demonstrated one aspect of the scheme that is high-order: the convergence of the errors in the total momentum with our orthonormal, modal, DG algorithm are super-linear in polynomial order.

Finally, we examine two additional convergence metrics for our discretization of the Vlasov--Maxwell system with this initial condition: the behavior of the $L^2$ norm of the distribution function and the divergence errors in Gauss' law for the electric field.
We expect with our choice of numerical flux function, upwinding, \eqr{\ref{eq:simpleUpwindVlasov}} for $\gvec{\alpha}^x$, the streaming term, and global Lax-Friedrichs for the acceleration $\gvec{\alpha}^v$, that the $L^2$ norm of the distribution function is a monotonically decaying function.
\begin{figure}[!htb]
    \centering
    \includegraphics[width=0.49\textwidth]{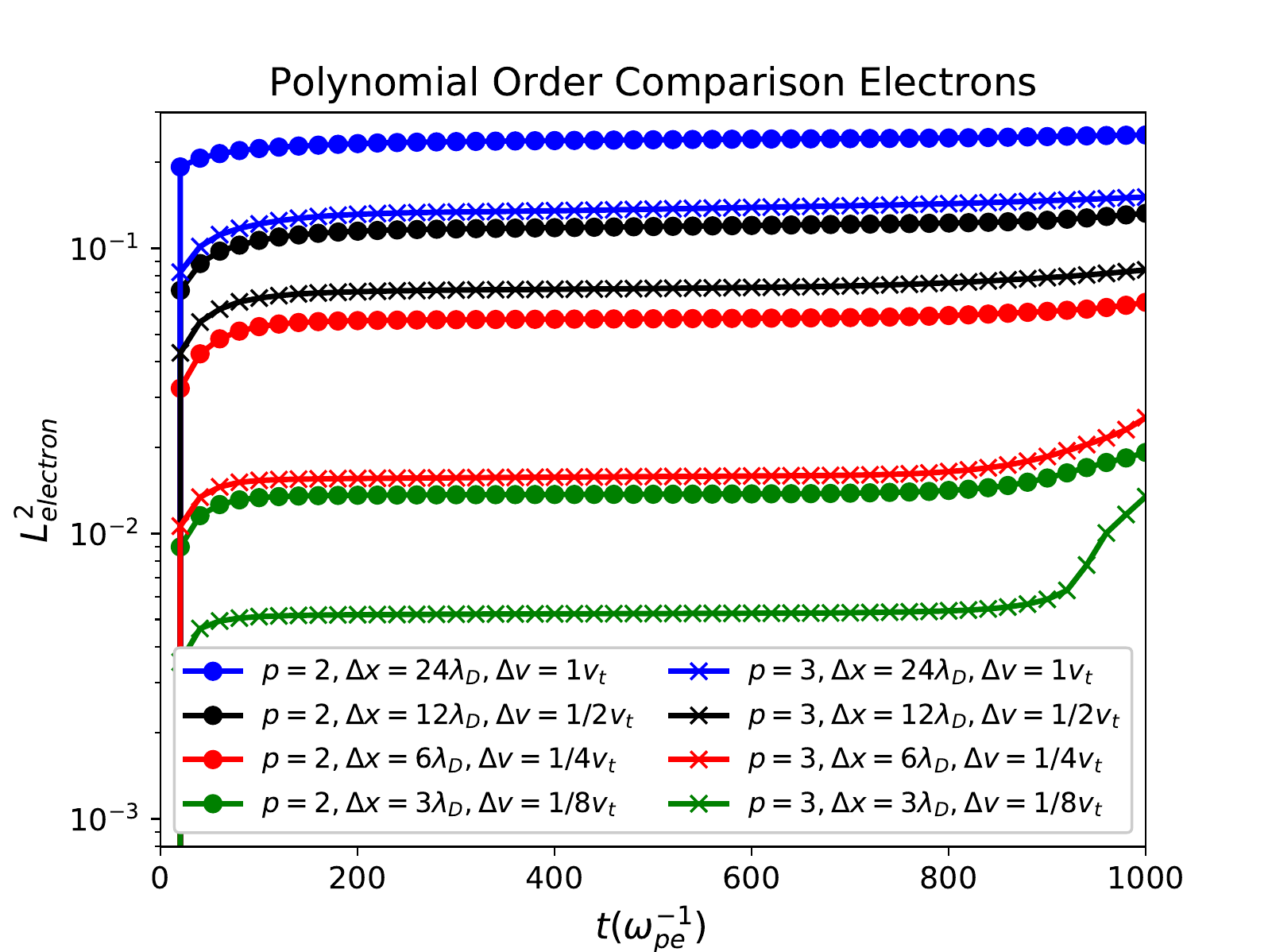}
    \includegraphics[width=0.49\textwidth]{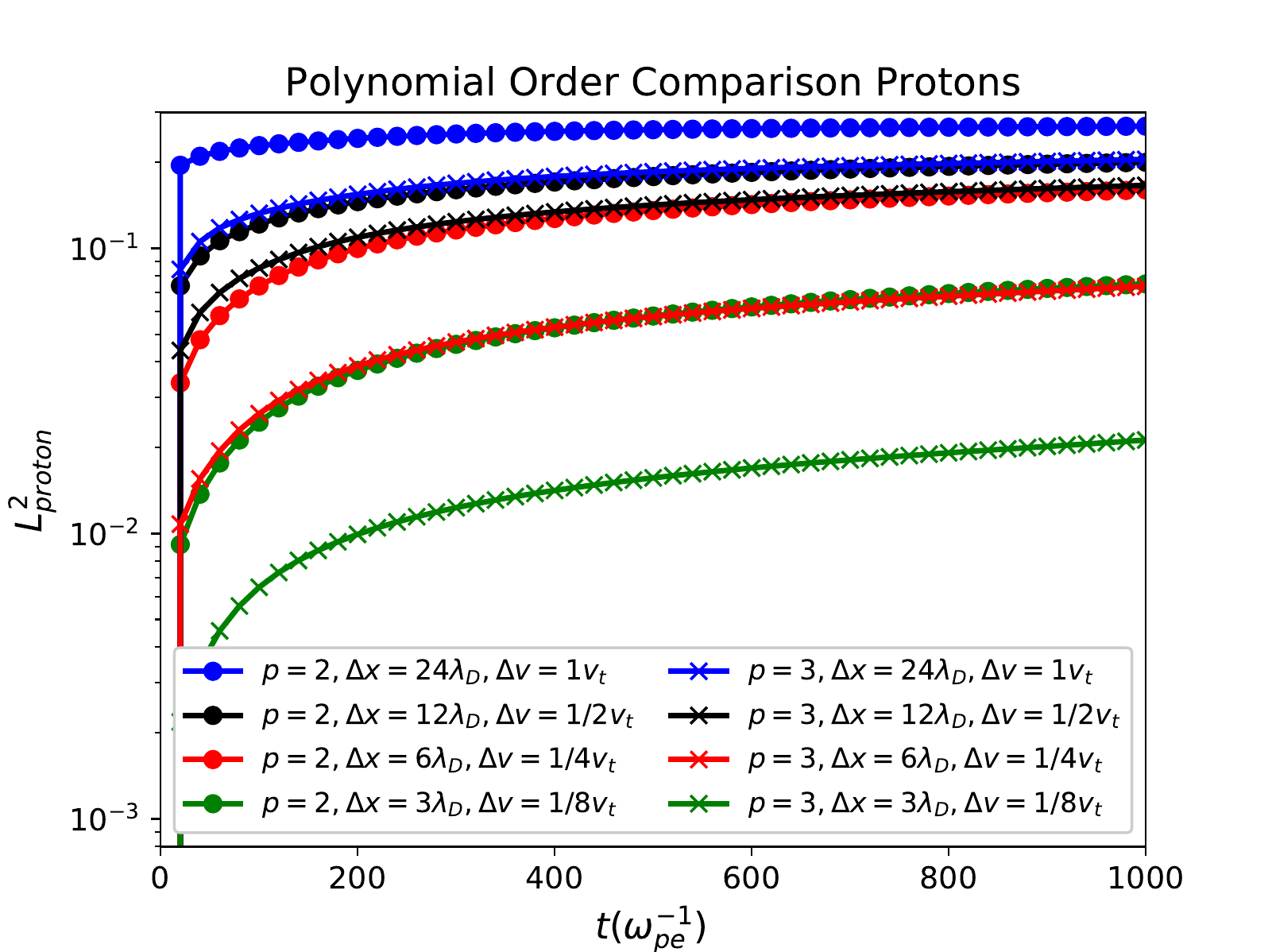}
    \caption{The change in the $L^2$ norm of the electron (left) and proton (right) distribution function with increasing resolution and polynomial order. As expected, the behavior of the $L^2$ norm of the distribution function is monotonic and decays in time. We note as well that increasing the polynomial order from 2 to 3 corresponds extremely well with a doubling of the resolution, providing direct evidence for the often assumed benefit of a high order method.} \label{fig:L2Norm}
\end{figure}
We present numerical evidence for this proof in Figure~\ref{fig:L2Norm} for both the protons and electrons by plotting the relative change in the $L^2$ norm,
\begin{align}
    L^2_s = \left | \frac{\int_0^{L_x} \int_{\mvec{v}_{min}}^{\mvec{v}_{max}} f_s^2(t) - f_s^2(t=0)\thinspace d\mvec{x}d\mvec{v}}{\int_0^{L_x} \int_{\mvec{v}_{min}}^{\mvec{v}_{max}} f_s^2(t=0)\thinspace d\mvec{x}d\mvec{v}} \right |.
\end{align}
It is interesting to note the behavior of polynomial order 3 compared to polynomial order 2, which provides anecdotal evidence that increasing the polynomial order of the simulation is analogous to increasing the resolution in configuration and velocity space. 
Although this behavior is often touted as prima facie for employing high order methods, such behavior is difficult to demonstrate analytically for nonlinear equation systems, if it is demonstrable at all.

Likewise, we consider how well Gauss' law for the electric field is satisfied in a discrete sense. In one dimension, \eqr{\ref{eq:divE}} becomes
\begin{align}
    \pfrac{E_x(x)}{x} = |e|\frac{n_p(x) - n_e(x)}{\epsilon_0},
\end{align}
where we have already substituted in for the charge density, $\rho_c = |e| (n_p - n_e)$.
We plot the results for the suite of simulations considered above, polynomial order 2 and 3, in Figure~\ref{fig:divEErrors} at the end of the simulations, $t = 1000\omega_{pe}^{-1}$.
\begin{figure}[!htb]
    \centering
    \includegraphics[width=0.49\textwidth]{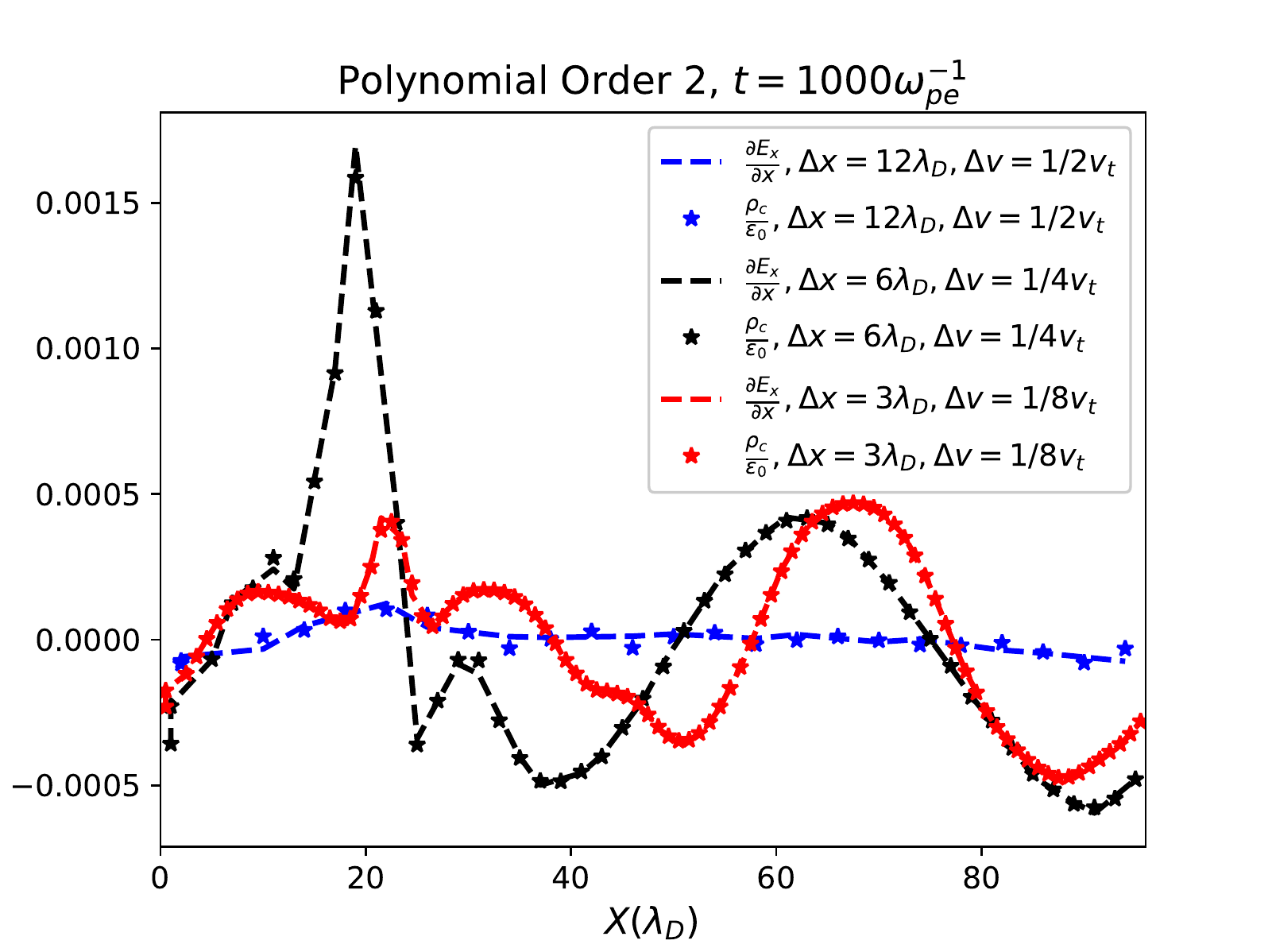}
    \includegraphics[width=0.49\textwidth]{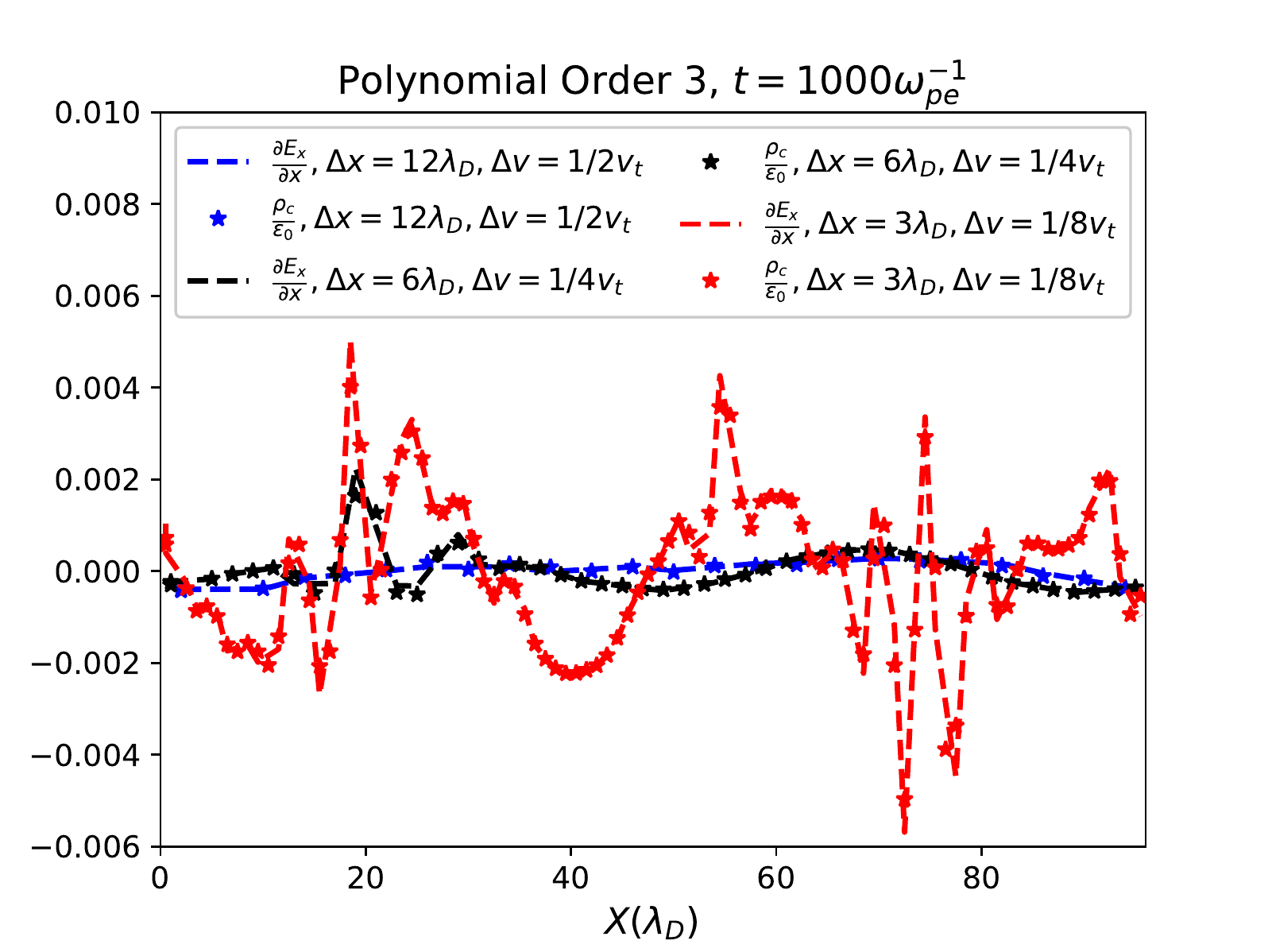}
    \caption{Comparison of the divergence of the electric field (dashed line) and the charge density (stars) for polynomial order 2 (left) and polynomial order 3 (right) simulations at the end of the simulation, $t=1000 \omega_{pe}^{-1}$. We can see that the two quantities agree reasonably well, especially as we refine the grid. Even as higher amplitude, smaller scale, electric fields are excited in the higher resolution simulations, the two quantities track each other well, despite the fact that we do not enforce this condition, and the charge density does not appear anywhere in evolved system of equations.} \label{fig:divEErrors}
\end{figure}
We note that, while the agreement is not perfect, the two quantities track remarkably well, even as larger amplitude, smaller scale, electric fields are formed with increasing resolution.
Especially for the finest resolution, polynomial order 3, when very fine scale structure forms in the electric field as the resolution approaches the Debye length, the characteristic length scale of these simulations, the charge density and divergence of the electric field agree very well. 
We reiterate that we currently do not enforce this condition, as the charge density $\rho_c$ does not appear anywhere in the Vlasov equation or Ampere's law, and thus it is a testament to the robustness of our numerical method that we do not observe large divergence errors in Gauss's law for the electric field.  

\subsection{Advection in Specified Electromagnetic Fields}

We now turn our attention to another simple, yet subtle, test of the Vlasov--Maxwell solver: advection in specified electromagnetic fields. 
Since charged particles circulate around magnetic fields, and we are employing a Cartesian mesh, we check that our numerical method can handle the advection of the distribution function in phase space.
In other words, we are checking that our algorithm can handle corner transport across cells. Consider a constant magnetic field in the $z$ direction, $\mvec{B} = B_0 \mvec{e}_z$ and an oscillating electric field of the form,
\begin{align}
    \mvec{E}(t) = E_0 \cos(\omega t) \mvec{e}_x. \label{eq:constAdvectionEField}
\end{align}
The evolution of charged particles in such a system can be solved analytically.
Assuming no spatial variation of the electric and magnetic fields, we have two ordinary differential equations for the evolution of the particles' velocities,
\begin{align}
    \frac{d v_x}{d t} & = \frac{q_s}{m_s} E_0 \cos(\omega t) + \Omega_c v_y, \\
    \frac{d v_y}{d t} & = -\Omega_c v_x,
\end{align}
where $\Omega_c = q_s B_0/m_s$ is the cyclotron frequency of the particles in this particular magnetic field. 
For simplicity, let us normalize the time and frequency to the inverse cyclotron frequency and cyclotron frequency respectively so that our two ordinary differential equations become,
\begin{align}
    \frac{d v_x}{d \tilde{t}} & = \frac{E_0}{B_0} \cos(\tilde{\omega} \tilde{t}) + v_y, \\
    \frac{d v_y}{d \tilde{t}} & = -v_x,    
\end{align}
where tildes indicate normalized quantities.

We can convert this system of coupled first-order ordinary differential equations into a set of uncoupled second order ordinary differential equations and solve for the particular solutions of each to obtain,
\begin{align}
    v_x(\tilde{t}) & = w_x(\tilde{t}) + v_x(0) \cos(\tilde{t}) + v_y(0) \sin(\tilde{t}), \\
    v_y(\tilde{t}) & = w_y(\tilde{t}) - v_x(0) \sin(\tilde{t}) + v_y(0) \cos(\tilde{t}),
\end{align}
where,
\begin{align}
    w_x(\tilde{t}) & = 
    \begin{cases}
    \frac{E_0}{B_0(1 - \tilde{\omega}^2)}[\sin(\tilde{t}) - \tilde{\omega}\sin(\tilde{\omega}\tilde{t})] & \quad \tilde{\omega} \neq 1, \\
    \frac{E_0}{2 B_0} [\tilde{t} \cos(\tilde{t}) + \sin(\tilde{t})] & \quad \tilde{\omega} = 1,
    \end{cases} \label{eq:wxAnalytic}
    \\
    w_y(\tilde{t}) & = 
    \begin{cases}
    \frac{E_0}{B_0(1 - \tilde{\omega}^2)}[\cos(\tilde{t}) - \cos(\tilde{\omega}\tilde{t})] & \quad \tilde{\omega} \neq 1, \\
    -\frac{E_0}{2 B_0} \tilde{t} \sin(\tilde{t}) & \quad \tilde{\omega} = 1.
    \end{cases}
\end{align}
Note that $\tilde{\omega} = 1$ means that the denormalized frequency is equal to the cyclotron frequency, i.e., when $\tilde{\omega} = 1$, that is the resonant case for the particles. 
Since the motion of a distribution of particles is constant along characteristics, we know that, given an initial distribution $f_0 (v_x, v_y)$, the distribution of particles at any later time is
\begin{align}
    f(v_x(t), v_y(t), t) = f_0(v_x(0), v_y(0), 0).
\end{align}
Consider an initial Maxwellian distribution of electrons in one spacial dimension and two velocity dimensions, 1X2V, \eqr{\ref{eq:ICMaxwellian}}.
Using our solution for the particles' velocities, we can see that,
\begin{align}
    [v_x(\tilde{t}) - w_x(\tilde{t})]^2 + [v_y(\tilde{t}) - w_y(\tilde{t})]^2 = v_x(0)^2 + v_y(0)^2. \label{eq:analyticConstFlow}
\end{align}
So, the exact solution for an initial Maxwellian distribution of particles is just a Maxwellian with drift velocities $w_x(\tilde{t}), w_y(\tilde{t})$ for all future times. 

We simulate the evolution of an initially Maxwellian distribution function of electrons under the influence of a constant magnetic field in the $z$ direction, $\mvec{B} = B_0 \mvec{e}_z$, and a time-varying electric field given by \eqr{\ref{eq:constAdvectionEField}}, one simulation with $\tilde{\omega} = 0.5, E_0/B_0 = 1.0$, a non-resonant case, and one simulation with $\tilde{\omega} = 1.0, E_0/B_0 = 0.5$, a resonant case. 
We compare the analytic solution from Eqns. (\ref{eq:wxAnalytic})--(\ref{eq:analyticConstFlow}) to simulations using our Vlasov--Maxwell solver in Figures \ref{fig:constAdvectionF} and \ref{fig:constAdvectionMoments}. 
\begin{figure}[!htb]
    \centering
    \includegraphics[width=0.72\textwidth]{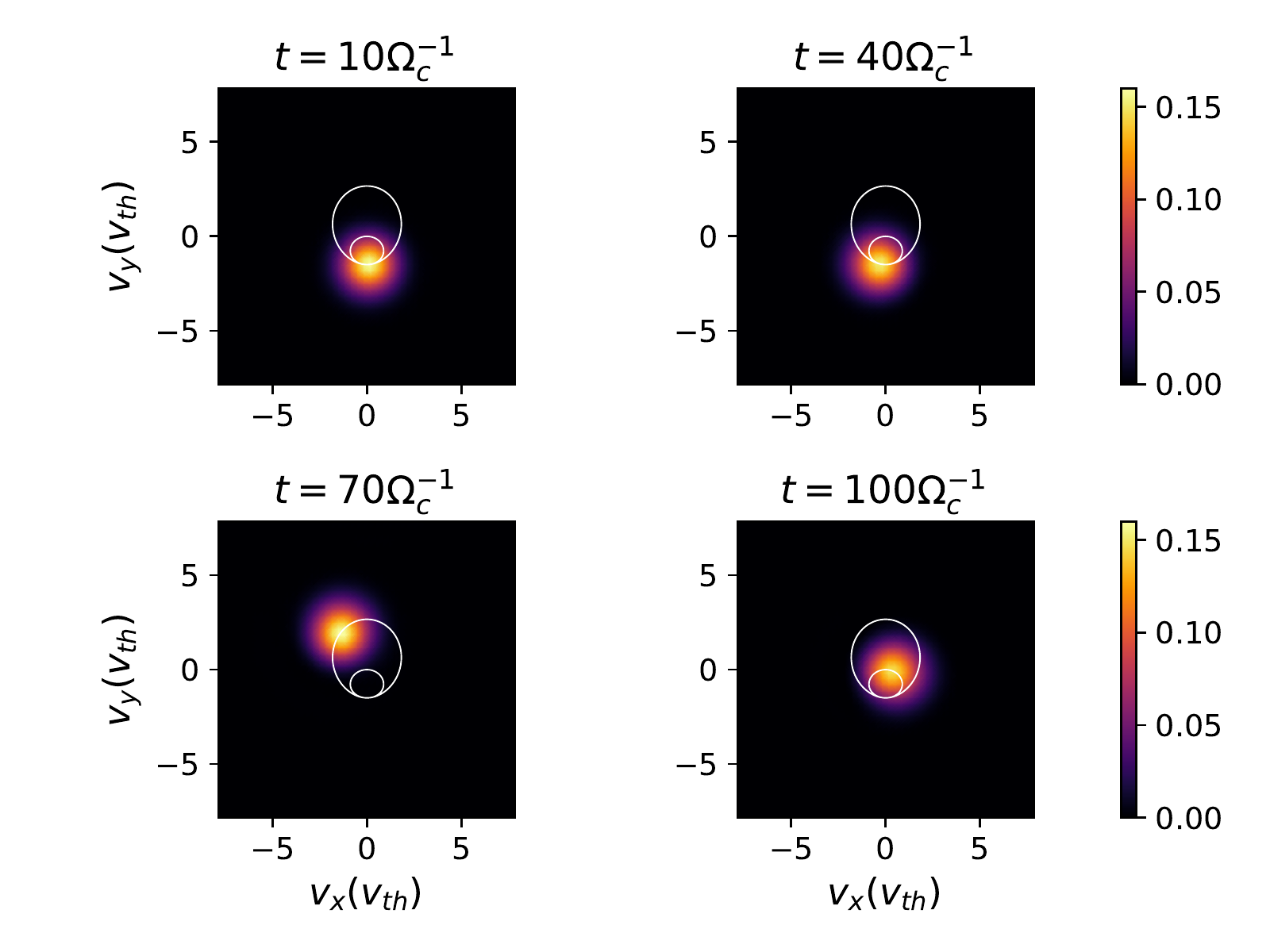}
    \includegraphics[width=0.72\textwidth]{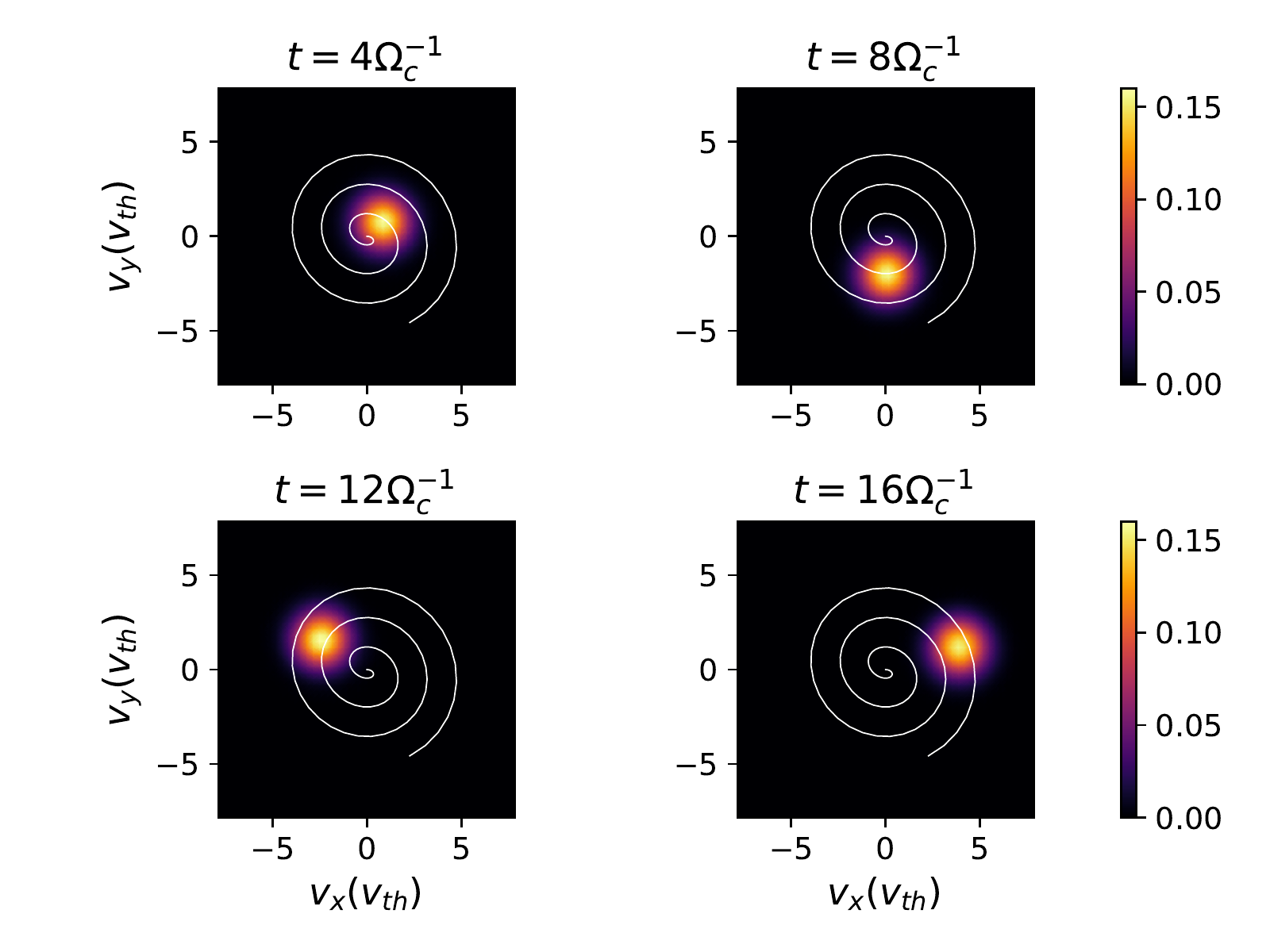}
    \caption{The non-resonant (top) and resonant (bottom) advection of a distribution of electrons in phase space, over-plotted with the analytical solution. The electron distribution function is plotted at $f(x=\pi, v_x, v_y)$. We can see that in both cases the distribution function's evolution is well described by our derived analytical solution, and that in the non-resonant case, where the distribution function is advected for a large number of inverse cyclotron periods, there is no noticeable diffusion of the distribution function in phase space. We emphasize that these simulations are performed with polynomial order 2 on a relatively coarse velocity space mesh, $N_{v_x} = N_{v_y} = 16$ with velocity space extents $[-8 v_{th_e}, 8 v_{th_e}]$ in both the $v_x$ and $v_y$ dimensions, so $\Delta v_x = \Delta v_y = 1 v_{th_e}$.} \label{fig:constAdvectionF}
\end{figure}
\begin{figure}[!htb]
    \centering
    \includegraphics[width=0.85\textwidth]{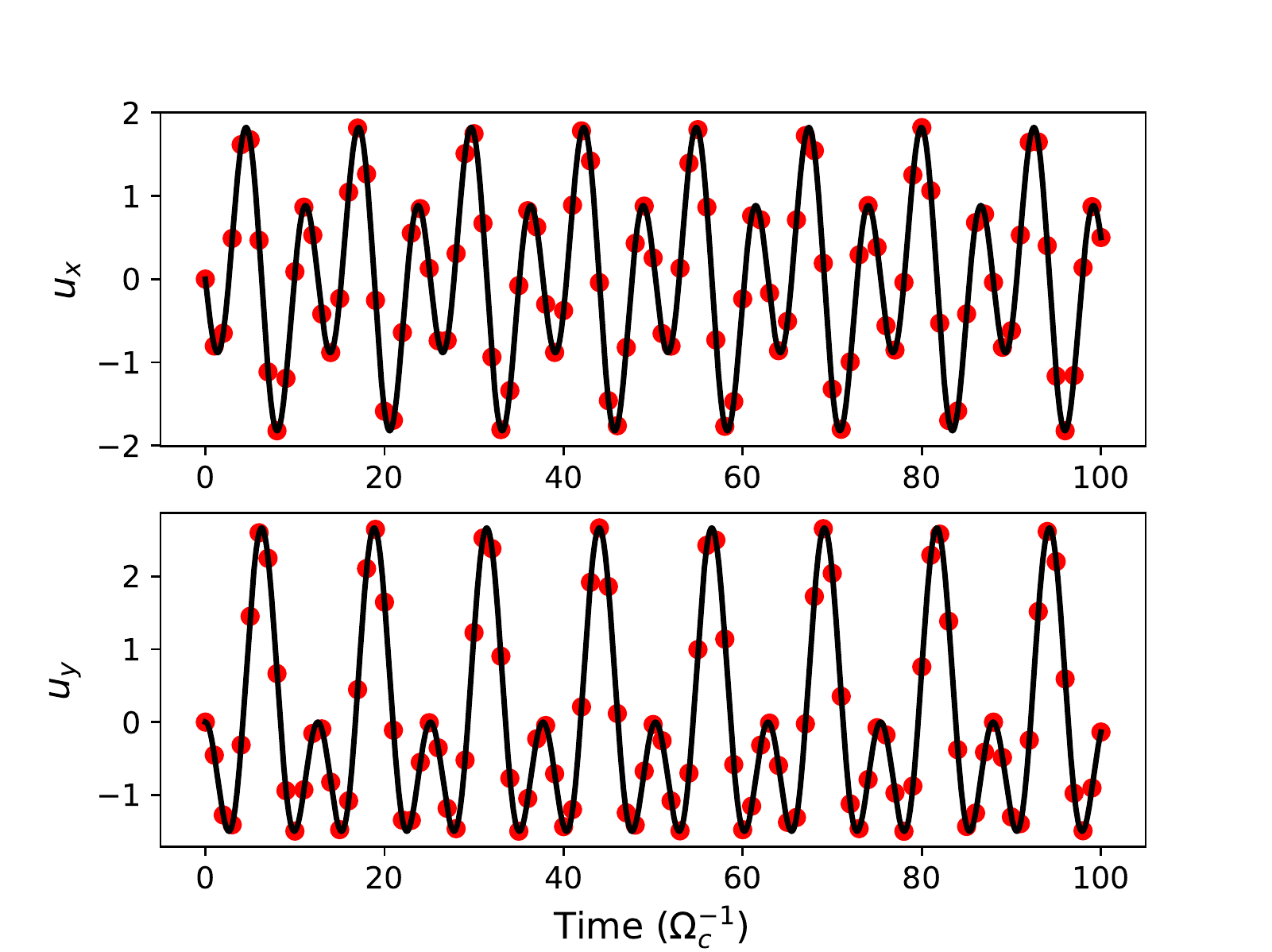}
    \includegraphics[width=0.85\textwidth]{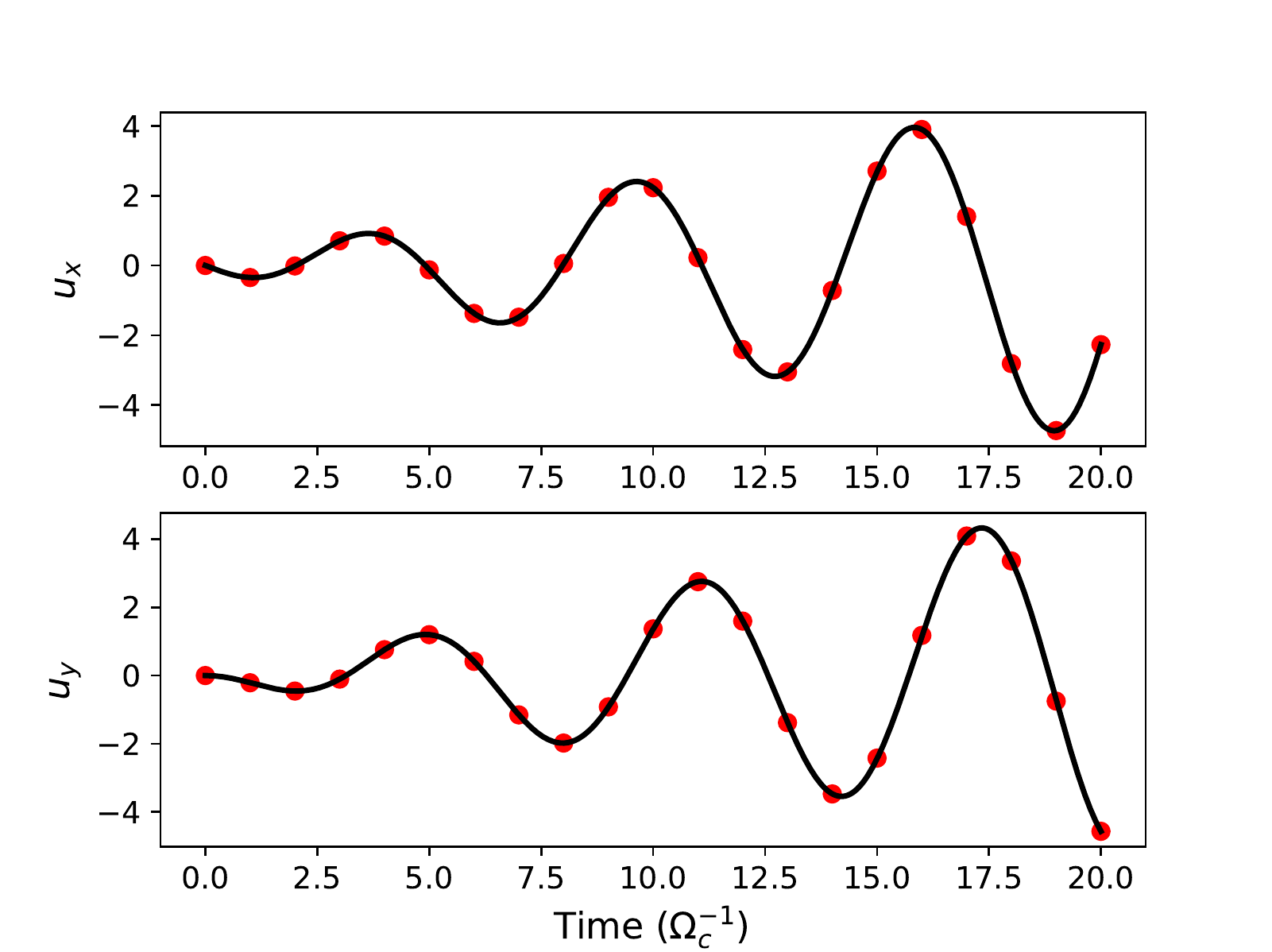}
    \caption{The value of the flow computed from the simulations (red dots) over-plotted with the analytic solution (black line) for non-resonant (top) and resonant (bottom) cases. The values of the flow are plotted at $u_x(x=\pi), u_y(x=\pi)$.} \label{fig:constAdvectionMoments}
\end{figure}
Both simulations are performed on a 1X2V grid with $L_x = 2 \pi$, and velocity space extents $[-8 v_{th_e}, 8 v_{th_e}]$ in both the $v_x$ and $v_y$ dimensions. We use polynomial order 2, $N_x = 2$, and $N_{v_x} = N_{v_y} = 16$, so $\Delta v_x = \Delta v_y = 1 v_{th_e}$. 
Periodic boundary conditions are employed in configuration space, and zero flux boundary conditions are employed in velocity space.
Even on a coarse velocity space mesh, the evolution of the distribution function is well-described by our analytic solution, with very little diffusion as electrons circulate around the magnetic field. 
Additionally, we run the non-resonant case, $\tilde{\omega} = 0.5, E_0/B_0 = 1.0$, to $t = 1000 \Omega_{c}^{-1}$ and plot the final distribution function in Figure~\ref{fig:polyOrderCompConstAdvectionF}.
\begin{figure}[!htb]
    \centering
    \includegraphics[width=\textwidth]{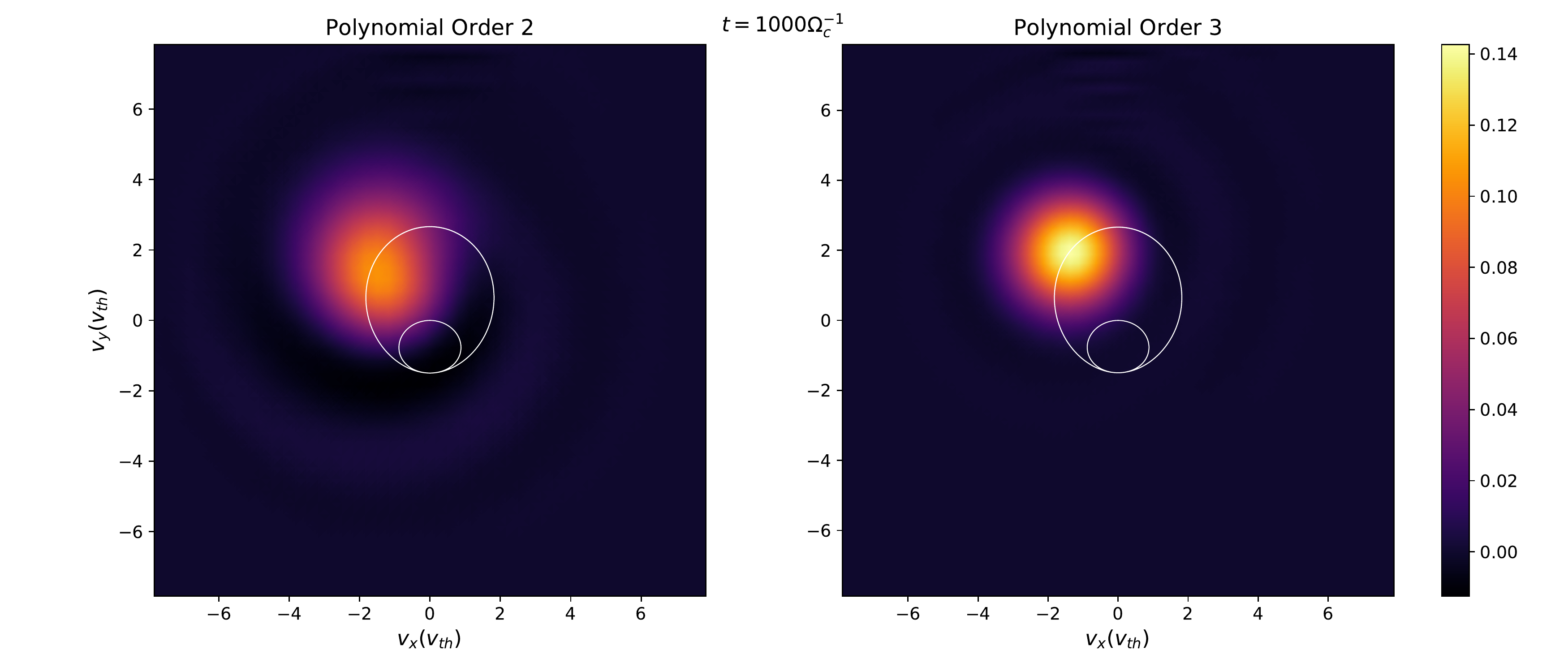}
    \caption{Comparison of a polynomial order 2 (left) and polynomial order 3 (right) simulation of the non-resonant case at $t = 1000 \Omega_{c}^{-1}$. The electron distribution function is plotted at $f(x=\pi, v_x, v_y)$. On this coarse mesh, $N_{v_x} = N_{v_y} = 16$ with velocity space extents $[-8 v_{th_e}, 8 v_{th_e}]$ in both the $v_x$ and $v_y$ dimensions, so $\Delta v_x = \Delta v_y = 1 v_{th_e}$, the diffusion of the distribution function in phase space starts to become noticeable for the polynomial order 2 case after running the simulation for a long enough time. But, we note that for the same coarse mesh, the distribution function in the polynomial order 3 simulation remains pristine at this late time.} \label{fig:polyOrderCompConstAdvectionF}
\end{figure}
While we note some noticeable diffusion in the polynomial order 2 simulation, by increasing to polynomial order 3 on the same grid, we virtually eliminate this diffusion, again illustrating the virtues of a high-order method applied to the discretization of the Vlasov--Maxwell system. 

It is worth emphasizing an inherent flexibility we have in our Vlasov--Maxwell solver in \gke: we can choose whatever polynomial order is ultimately necessary for the required dynamics. 
While the polynomial order 3 simulation of the non-resonant case is slightly more expensive, an 80 percent increase in cost for a $t=1000 \Omega_c^{-1}$ simulation for the specified grid resolution of $N_x = 2$, $N_{v_x} = N_{v_y} = 16$, this freedom to increase the polynomial order as needed ultimately allows us to tackle a wider range of problems. 
And, we wish to point out that an 80 percent increase in cost is actually better than we would naively expect, as there are 60 percent more basis functions, $32/20 = 1.6$, going from polynomial order 2 to 3, and we require 50 percent more time-steps for the high polynomial order simulation from a more restrictive CFL condition. 
This back-of-the-envelope calculation suggests that polynomial order 3 should be 2.5 times more expensive for the same grid resolution and end time. 
The improvement over the naive cost scaling occurs because the higher polynomial order computational kernels obtain better efficiency in terms of arithmetic intensity, i.e., the number of floating point operations per byte of memory moved.

\subsection{Landau Damping of Langmuir Waves}\label{sec:CollisionlessLandauDamping}

Consider a plasma, or Langmuir, wave propagating in a plasma of protons and electrons whose distribution functions are given by Maxwellians, \eqr{\ref{eq:ICMaxwellian}}. 
Langmuir waves are dispersive waves, with a dispersion relation given by
\begin{align}
1 - \frac{1}{2 k^2 \lambda_{De}^2} Z'\left(\frac{\omega}{\sqrt{2} v_{th_e} k} \right ) = 0, \label{eq:plasmaDisp}
\end{align} 
in the limit that the proton mass is much larger than the electron mass and the protons can thus be considered immobile.
$Z(\zeta)$ is the plasma dispersion function, defined as
\begin{align}
Z(\zeta) = \frac{1}{\sqrt{\pi}} \int_{-\infty}^\infty \frac{e^{-x^2}}{x - \zeta} dx, \label{eq:plasmaZFunc}
\end{align}
with the derivative of the plasma dispersion function given by
\begin{align}
Z'(\zeta) = -2[1 + \zeta Z(\zeta)].
\end{align}
An application of complex integration techniques shows that depending on the sign of the largest imaginary component of the frequency $\omega = \omega_r + i \gamma$, the wave is either unstable and will grow with time, or will damp away, a phenomenon known as Landau damping. 

For Langmuir waves propagating in a Maxwellian plasma of protons and electrons, the waves quickly damp. Using a 1X1V setup, we can initialize Langmuir waves in the Vlasov--Maxwell system with a small density perturbation and the corresponding electric field to support this density perturbation,
\begin{align}
n_e(x) & = n_0[1 + \alpha \cos(kx)] \label{eq:electronLangmuirInit} \\
n_p(x) & = n_0 \\
E_x(x) & = -|e| \alpha \frac{\sin(kx)}{\epsilon_0 k}, \label{eq:ExLangmuirInit}
\end{align}
where $n_0 = 1.0$, $\alpha$ is the size of the perturbation, and $k$ is the wavenumber of the wave.
The electric charge $e$ and permittivity of free space $\epsilon_0$ are included in the electric field to satisfy \eqr{\ref{eq:divE}}. 
Choosing $\alpha \ll 1$ allows us to compare with the linear analytical theory described above. 
The box size is set to $L_x = 2\pi/k$ so exactly one wavelength fits in the domain. 
Specific parameters for these runs are: $\alpha = 10^{-4}$, $m_p/m_e = 1836$, $T_p/T_e = 1.0$, and  $v_{th_e}/c = 0.1$. 
For the proton species, the velocity space extents are $\pm 6 v_{th_p}$, and for the electrons, the velocity space extents are $\pm 6 v_{th_e}$. 
The boundary conditions in configuration space are periodic, while the boundary conditions in velocity space are zero flux. 

The resolution is chosen for each simulation to adequately resolve the Debye length in configuration space and to mitigate numerical recurrence in velocity space.
By numerical recurrence, we refer to the process by which the collisionless system artificially ``un-mixes'' if the distribution function forms structure at the velocity space grid scale, see, e.g., \citet{Cheng:2013} for a discussion of numerical recurrence in DG schemes. 
Numerical recurrence is inevitable with finite velocity resolution for this particular problem, because the Landau damping of the wave will create smaller and smaller velocity space structure through the phase-mixing of the wave. 
We could completely eliminate this issue with a diffusive process in velocity space, such as a collision operator, and we will explore the effects of collisions on the Langmuir wave in Section~\ref{sec:collisionLangmuirWave}. 
Here, we choose ample velocity resolution so that the wave damps enough for us to extract a clean damping rate and frequency for the initialized wave. 
We find for the longest wavelengths, using polynomial order 2, a resolution of 64 points in configuration space adequately resolves the Debye length, and 128 points in velocity space permits the wave to phase-mix sufficiently to extract damping rates. 

The evolution of the electromagnetic energy, as well as the other components of the energy, in a prototypical simulation is given in Figure~\ref{fig:LangmuirWaveEnergy}. 
\begin{figure}[!htb]
    \centering
    \includegraphics[width=0.49\textwidth]{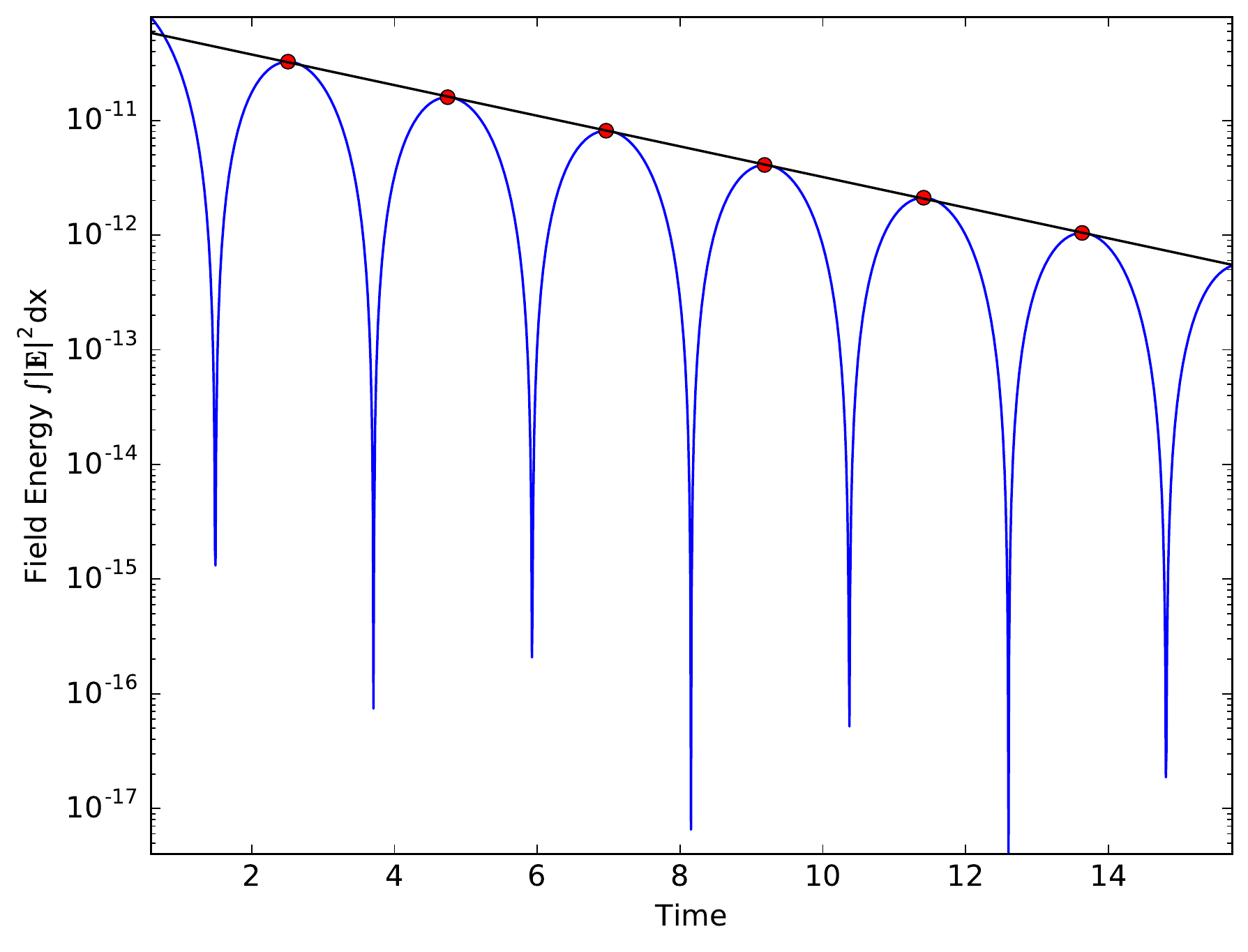}
    \includegraphics[width=0.49\textwidth]{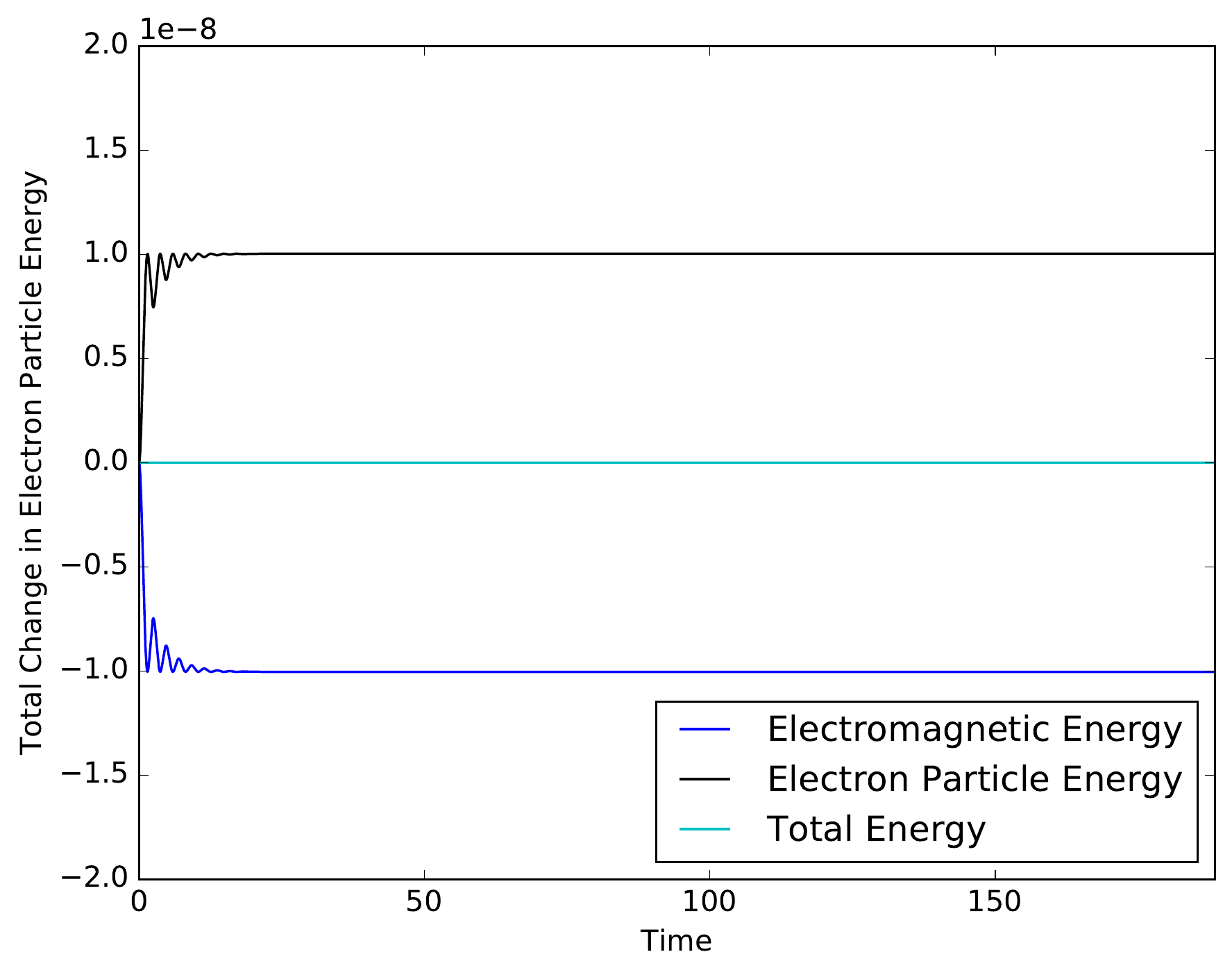}
    \caption{Prototypical evolution of the electromagnetic energy (blue), $\frac{\epsilon_0}{2} \int |\mvec{E}|^2 dx$, for the damping of a Langmuir wave, in this case $k \lambda_D = 0.5$, for a number of plasma periods (left), and the evolution of various components of the energy for the full length of the simulation (right). The right plot is the relative change in the energy component compared to the total energy at $t=0$, i.e, $\Delta E_{comp}/E_0$. The local maxima (red circles) of the evolution in the left plot are used to determine both the damping rate and frequency of the excited wave via linear regression, with the black line being our reference fit for the damping rate. We note that energy is very well conserved, and, as expected, the plasma waves damp on the electrons, converting electromagnetic energy to electron thermal energy.}
    \label{fig:LangmuirWaveEnergy}
\end{figure}
Comparisons of a number of Vlasov--Maxwell simulations with theory for both the damping rates and the frequencies of the waves are given in Figure~\ref{fig:LangmuirWaveTheory}. 
\begin{figure}[!htb]
    \centering
    \includegraphics[width=0.49\textwidth]{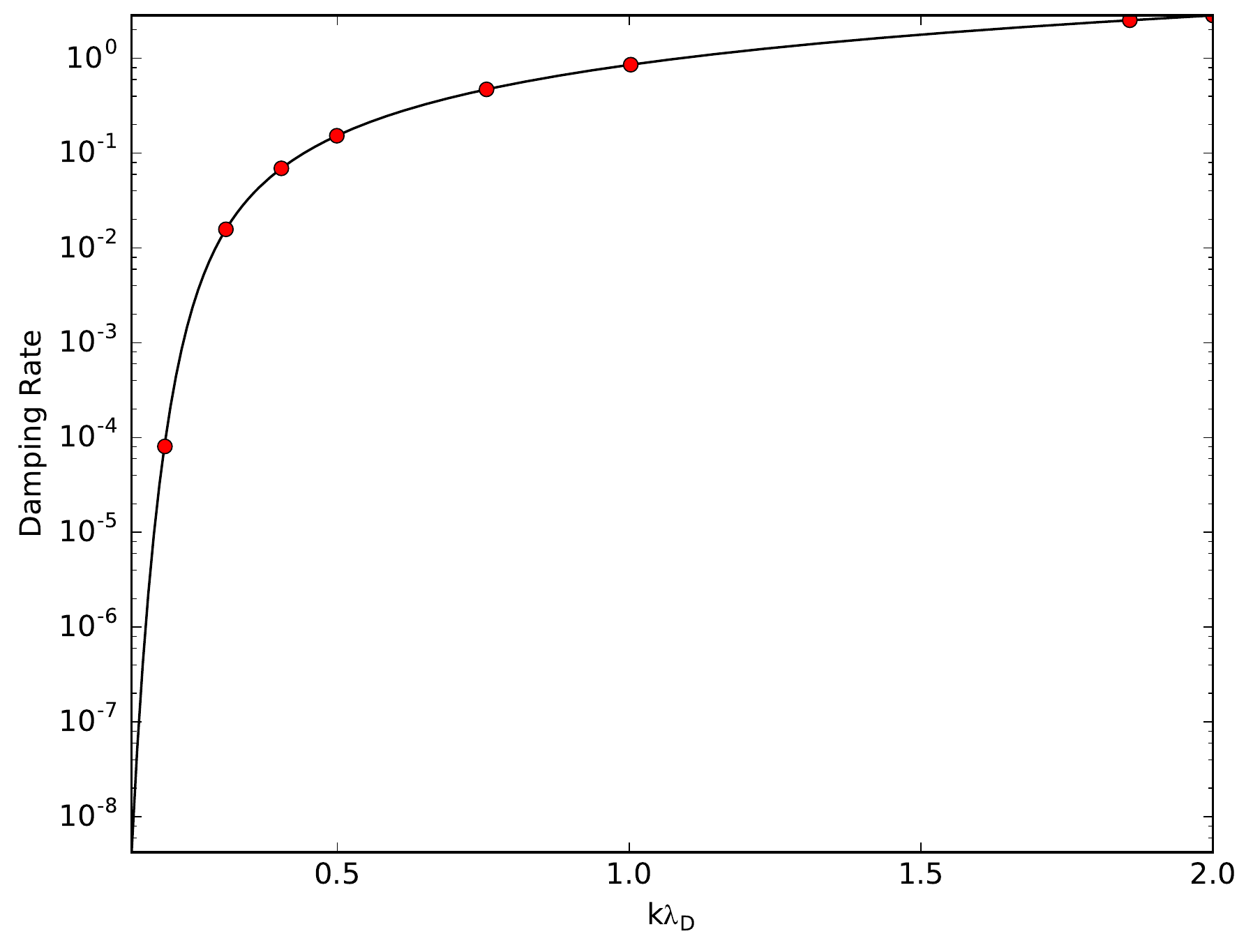}
    \includegraphics[width=0.49\textwidth]{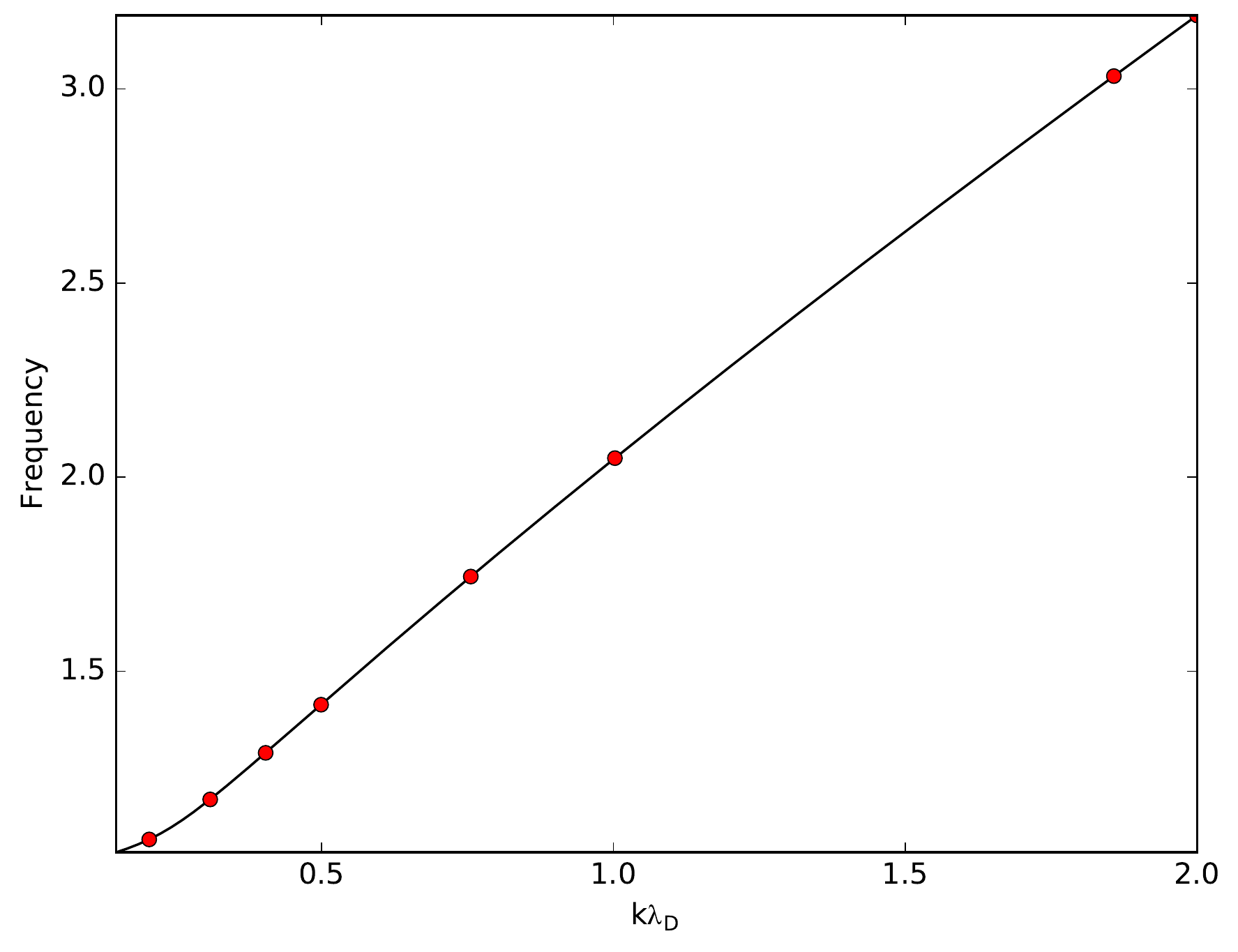}
    \caption{Damping rates (left) and frequencies (right) of Langmuir waves from theory (solid line) and for a number of Vlasov--Maxwell simulations (red circles). The solid lines are obtained using a root finding technique applied to \eqr{\ref{eq:plasmaDisp}}. The x-axis of both figures is normalized to the Debye length, $\lambda_D$, and the y-axis of both figures is normalized to the plasma frequency, $\omega_{pe}$.}
    \label{fig:LangmuirWaveTheory}
\end{figure}
For the theoretical result, we solve \eqr{\ref{eq:plasmaDisp}} using a root-finding technique. 
We emphasize that we solve the Vlasov--Maxwell system in its entirety, including the nonlinear term, for both the protons and electrons. With the above simulation parameters, the plasma waves damp entirely on the electron species, so the approximation that the protons are essentially immobile in our dispersion relation holds to high precision.
We also wish to note that the resolution of 64 points in configuration space is not required for every simulation. For example, the prototypical simulation presented in Figure~\ref{fig:LangmuirWaveEnergy} uses only 16 points in configuration space, or approximately one grid cell per Debye length. As long as the gradients are properly resolved, the Vlasov--Maxwell discretization is extremely robust.

\subsection{Three-Species Collisionless Electrostatic Shock}\label{sec:threeSpeciesShock}

We turn now to benchmarking the flexibility of our Vlasov--Maxwell solve in \gke~by considering the evolution of a plasma with more than two species.
In \citet{Pusztai:2018}, a semi-analytic model for electrostatic collisionless shocks was derived and then checked against the results of a number of fully nonlinear Vlasov--Maxwell calculations.
The Vlasov--Maxwell simulations performed in \citet{Pusztai:2018} were done with an initially alias-free nodal scheme implemented and described in \citet{Juno:2018}, before the algorithm was improved with an orthonormal, modal basis---see Chapter~\ref{ch:ImplementationDGFEM} for details on the othornormal, modal basis compared to the nodal basis.
In the following test, we employ the orthonormal, modal basis algorithm for the three-species shock problem and reproduce the results of \citet{Pusztai:2018} with our new and improved implementation of the DG scheme for the VM-FP system of equations.

The three-species collisionless shock setup described in \citet{Pusztai:2018} is repeated here for clarity.
A Maxwellian, \eqr{\ref{eq:ICMaxwellian}}, with a density gradient in 1X1V in all three species is initialized and allowed to evolve freely, as in Section~\ref{sec:KineticSodShock}, but now allowing the electromagnetic fields to evolve as well.
This density gradient is a step function, with $n_L = n_0$, and $n_R = 2 n_0$, where $n_0$ is the density normalization, and the subscripts $L$ and $R$ denote the left and right values of the density in the 1D configuration space domain.

The three species in the plasma are electrons, fully ionized aluminum, and a proton impurity species. 
The real mass ratios of the various species are employed so that $m_p/m_e = 1836, m_i/m_p = 27$, where the subscript $i$ denotes the mass of the aluminum ion species. 
Note that $Z_i = 13$ for fully ionized aluminum.
Since the proton species is an impurity, we choose $n_p/n_i = 0.01$. The electrons are much hotter than either ion species, $T_e/T_p = 45, T_p = T_i$. The configuration space domain has length $L_x = 100 \lambda_D$.
Note that the jump in the density is initialized at $x = 50 \lambda_D$, the middle of the domain. 
The velocity space extents of the electrons, aluminum ions, and proton impurity are $[-6 v_{th_e}, 6 v_{th_e}], [-18 v_{th_i}, 54 v_{th_i}]$, and $[-6 v_{th_p}, 18 v_{th_p}]$ respectively, with $v_{th_s}$ denoting the thermal velocity of the specified species.
We use the same resolution as \citet{Pusztai:2018}, $N_x = 256$ and $N_v = 96$ for all three species, and $p=2$ Serendipity elements. 
Copy boundary conditions are employed in the $x$ dimension as in Section~\ref{sec:KineticSodShock}, i.e., we employ a perfectly matched layer in configuration space to allow the electromagnetic fields and distribution function to evolve freely at $x = 0$ and $x = 100 \lambda_D$, and zero flux boundary conditions are employed in velocity space.

We plot the aluminum and proton distribution functions in the vicinity of the shock in Figure~\ref{fig:ThreeSpeciesShockDistributionFunctions}. 
We note that this figure is similar to Figure 9 in \citet{Pusztai:2018}.
\begin{figure}[!htb]
    \centering
    \includegraphics[width=\textwidth]{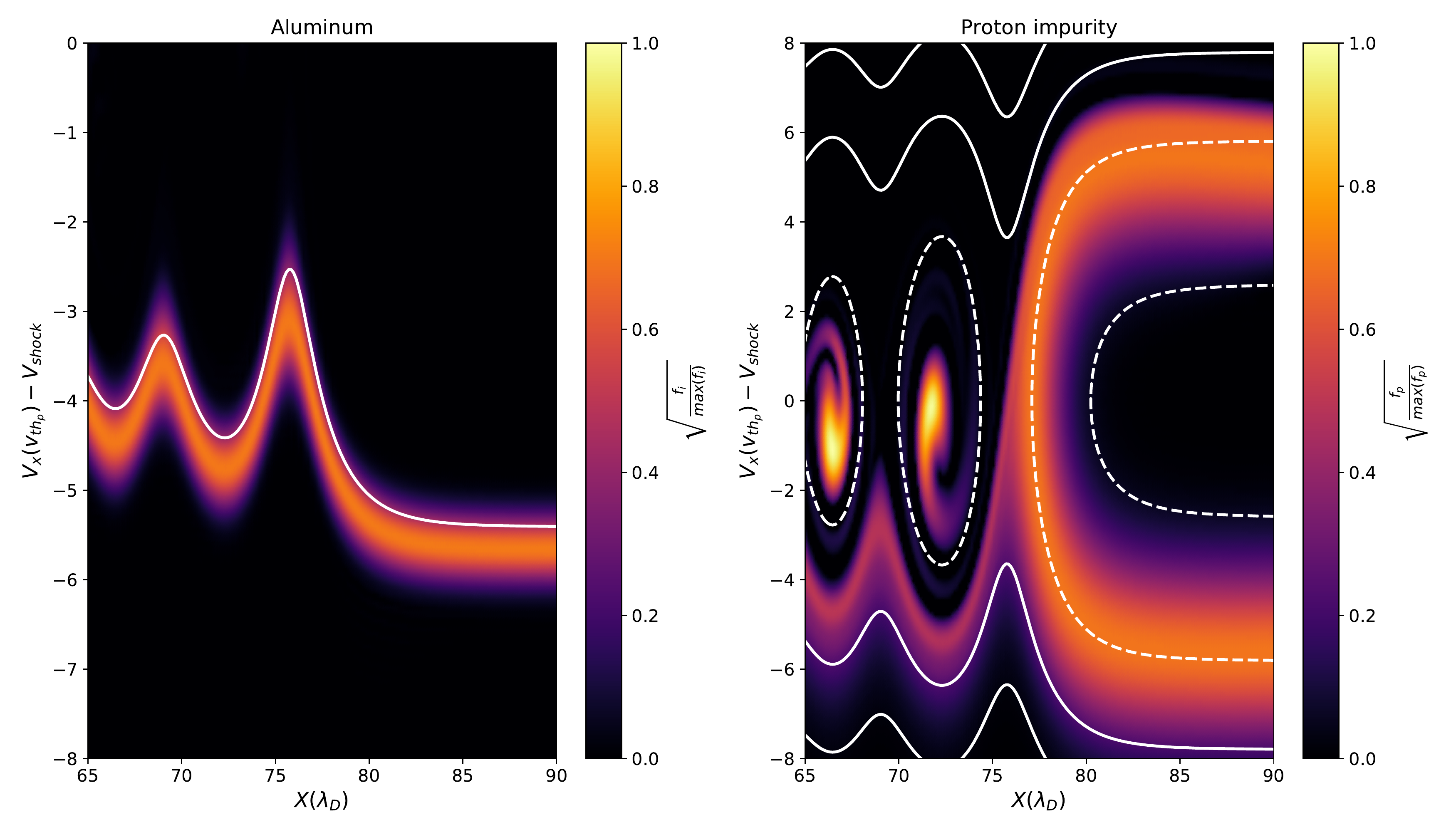}
    \caption{The aluminum (left) and proton impurity (right) distribution functions in the vicinity of the shock at $t = 35 \sqrt{m_e/m_p} \omega_{pe}^{-1} \sim 1500 \omega_{pe}^{-1}$. Over-plotted in white are contours of constant $H(x,v) = \frac{1}{2} m_s v^2 + q_s \phi(x)$, the Hamiltonian. We note that the Hamiltonian has been transformed to the rest frame of the shock, $\hat{v} = v - V_{shock}$, and there is some freedom in computing $\phi(x)$ from the electric field in our simulations. We choose $\phi(x = 0) = 0$ on the left edge of the domain, and then integrate $E_x$ along the 1D domain to determine the electrostatic potential. We draw attention to the trapped particle regions in the proton distribution function just down-stream of the shock, which amplify the cross-shock potential and lead to a large reflected population of protons. Note that we are plotting a normalized value for the distribution function, as in \citet{Pusztai:2018}, and that the v-axes are different for the two species.}
    \label{fig:ThreeSpeciesShockDistributionFunctions}
\end{figure}
These distribution functions are plotted at $t = 35 \sqrt{m_e/m_p} \omega_{pe}^{-1} \sim 1500 \omega_{pe}^{-1}$ and over-plotted in white are contours of constant $H(x,v) = \frac{1}{2} m_s v^2 + q_s \phi(x)$, the Hamiltonian. 
We note that the Hamiltonian has been transformed to the rest frame of the shock, $\hat{v} = v - V_{shock}$, and there is some freedom in computing $\phi(x)$ from the electric field in our simulations. 
We choose $\phi(x = 0) = 0$ on the left edge of the domain, and then integrate $E_x$ along the 1D domain to determine the electrostatic potential. 

We find similar results to \citet{Pusztai:2018} for the value of the shock velocity, $V_{shock} = 5.66 v_{th_p}, M = 1.216$, where $M = V_{shock}/\sqrt{Z_i T_e/m_i}$ is the mach number, the value of the maximum normalized electrostatic potential, $\hat{\phi}_{max} = 23.9$, where $\hat{\phi} = e \phi/T_p$, and the measured ratio of the reflected population of the proton impurity species, $\alpha_p = 0.874$, computed from integrating the density in the upstream and reflected components of the proton distribution function at $x=85 \lambda_D$. 
These results are in good agreement with the semi-analytic model derived in \citet{Pusztai:2018}, especially for the reflected proton ratio, $\alpha_p \sim 0.889$.

The utility of a continuum discretization of the Vlasov-Maxwell system is made manifest by the clean representation of the proton impurity distribution function in Figure~\ref{fig:ThreeSpeciesShockDistributionFunctions}.
The trapped particles in the downstream region amplify the cross-shock potential and lead to a large reflected population.
We make no claims of the effort that may be required to reproduce these features with a particle code.
We merely wish to emphasize here that a continuum representation can be useful for elucidating features of the particle distribution function relevant to the overall dynamics. 

\subsection{Lower Hybrid Drift Instability}

The Vlasov--Maxwell system of equations supports a large zoo of instabilities.
Many of these instabilities are fundamentally ``kinetic'' in nature, meaning their ultimate evolution is challenging to model with fluid systems of equations.
In other words, the actual collisionless dynamics of the plasma is a critical component to the evolution of the instability, and equations that evolve a truncated set of of velocity moments of the Vlasov--Maxwell system of equations will have difficulty modeling these instabilities.

Determining whether an extended two-fluid model could capture the dynamics of current sheets unstable to modes such as the lower-hybrid drift instability (LHDI) \citep{Hirose:1972,Davidson:1977,Yoon:2002} was the focus of a recent paper, \citet{Ng:2019} (see also \citet{NgThesis:2019}).
Due to the inhomogeneities in the magnetic field and density in the vicinity of the current, diamagnetic effects may become important and drive instabilities such as the LHDI.
As part of this study, Vlasov--Maxwell simulations of the LHDI were performed with \gke~to compare both the linear and nonlinear stages of the evolution of the unstable current sheet in a fully kinetic model and the aforementioned extended two-fluid models. 
A simulation of a current sheet unstable to the LHDI is reproduced here as evidence our modal, orthonormal DG discretization of the Vlasov--Maxwell system of equations provides a fiducial representation of the dynamics of this kinetic instability.

We use the same parameters as \citet{Ng:2019}. In 2X2V, two spatial, $(x,y)$, and two velocity, $(v_x, v_y)$, dimensions, we initialize a gradient in an out-of-plane magnetic field,
\begin{align}
    B_z(x, y) & = B_0 (y) + \delta B(x,y), \\
    B_0(y) & = -C_0 \tanh\left ( \frac{y}{\ell} \right ), \\
    \delta B(x,y) & = C_1 \cos \left ( \frac{\pi y}{L_y} \right ) \sin \left ( \frac{2\pi m x}{L_x}\right ), 
\end{align}
where $\ell = \rho_p$ and $m = 8$, i.e., a current sheet of width $\rho_p$ and an $m=8$ perturbation to the current sheet. 
Here, $\rho_p$ is the proton Larmor radius, $\rho_p = v_{th_p}/\Omega_{cp}$. 
The box size is $L_x \times L_y = 6.4 \rho_p \times 12.8 \rho_p$. 
The velocity space extents for electrons are $[-8 v_{th_e}, 8 v_{th_e}]^2$, and the velocity space extents for the protons are  $[-6 v_{th_e}, 6 v_{th_e}]^2$. 
Zero flux boundary conditions are used in velocity space, periodic boundary conditions are used in $x$, and reflecting boundary conditions are used in $y$.
By reflecting, we mean that the particles reflect off the $y$-boundary, and the boundary condition for Maxwell's equations is that of a perfect conductor, zero tangent electric field and zero normal magnetic field. 
The grid resolution is $N_x \times N_y = 128 \times 256$, with $N_v^2 = 32^2$ grid points in velocity space for the electrons, and $N_v^2 = 24^2$ for the protons, with piecewise quadratic Serendipity elements.

Additional parameters are $v_{th_e}/c = 0.06, m_p/m_e = 36, T_p/T_e = 10$, and $\beta_{tot} = 1.0$. 
Since $\beta_{tot} = 1.0$ and the protons are 10 times hotter than the electrons, we have $\beta_p = 10.0/11.0$ and $\beta_e = 1.0/11.0$. 
The system is normalized such that the constants are $C_0 = v_{th_e}/\sqrt{\beta_e} = v_{Ae}$, the electron Alfv\'en velocity, and $C_1 = 10^{-4}/m$ where $m$ is the mode number being initialized.
Note that with the chosen parameters, the resolution is such that $\Delta x \approx \rho_e$, where $\rho_e$ is the electron gyroradius, $\rho_e = v_{th_e}/\Omega_{ce}$.

Finally, we note two critical components to initializing the system.
First, the astute reader will notice that the the initial magnetic field has non-zero curl, and therefore there must be a supporting current in the plasma, thus we refer to this initial condition as a current sheet,
\begin{align}
    J_x & = -\frac{C_0}{\ell} \sech^2\left ( \frac{y}{\ell} \right )  - C_1 \frac{\pi}{L_y} \sin\left ( \frac{\pi y}{L_y} \right ) \sin \left ( \frac{2\pi m x}{L_x}\right ), \\
    J_y & = -C_1 \frac{2 \pi m}{L_x} \cos \left ( \frac{\pi y}{L_y} \right ) \cos \left ( \frac{2\pi m x}{L_x}\right ).
\end{align}
Since the protons are 10 times hotter than the electrons, we give the appropriate fraction of the current to the protons and electrons, 10.0/11.0 to the protons and 1.0/11.0 to the electrons. 
Second, to initialize the particle distribution functions, we initialize both a current carrying and background Maxwellian, the sum of two instances of \eqr{\ref{eq:ICMaxwellian}}, for each species,
\begin{align}
    f_s (x, y, v_x, v_y) = \frac{m_s n_0 \sech^2\left ( \frac{y}{\ell} \right ) }{2 \pi T_s} \exp & \left (- m_s \frac{(v_x - u_{x_s})^2 + (v_y - u_{y_s})^2}{2 T_s} \right ) \notag \\
    & + \frac{m_s n_B}{2 \pi T_s} \exp \left (- m_s \frac{v_x^2 + v_y^2}{2 T_s} \right ),
\end{align}
where,
\begin{align}
    u_{x_s} & = T_{frac} \frac{J_x}{q_s \sech^2\left ( \frac{y}{\ell} \right ) }, \\
    u_{y_s} & = T_{frac} \frac{J_y}{q_s \sech^2\left ( \frac{y}{\ell} \right ) },
\end{align}
and $n_0 = 1.0$ and $n_B = 10^{-3}$. 
Note that $T_{frac}$ is the aforementioned fraction of the current given to the protons and electrons, 10.0/11.0 and 1.0/11.0 respectively. 
This background density is for numerical stability, so that the density does not go to zero away from the current sheet.

We plot the results of this simulation in Figures \ref{fig:lhdiElectricField} and \ref{fig:lhdiProtonDistributionFunction}, focusing on the late linear stage when the traditional mode structure of the LHDI is most visually evident.
In Figure \ref{fig:lhdiElectricField}, we see the logarithmic growth of the electric field associated with the LHDI\footnote{Note that we use a slightly different coordinate system from \citet{Ng:2019}, who instead define the 2X2V domain as $(y, z, v_y, v_z)$. This is why the equivalent mode structure found in \citet{Ng:2019} is in the y-electric field, as opposed to here, where the LHDI mode structure is found in the x-electric field.}, with a growth rate found $\gamma \sim 1.1 \Omega_{ci}$, in agreement with linear theory and \citet{Ng:2019}'s computation, as well as the mode structure expected for an $m=8$ perturbation.
\begin{figure}[!htb]
    \centering
    \includegraphics[width=\textwidth]{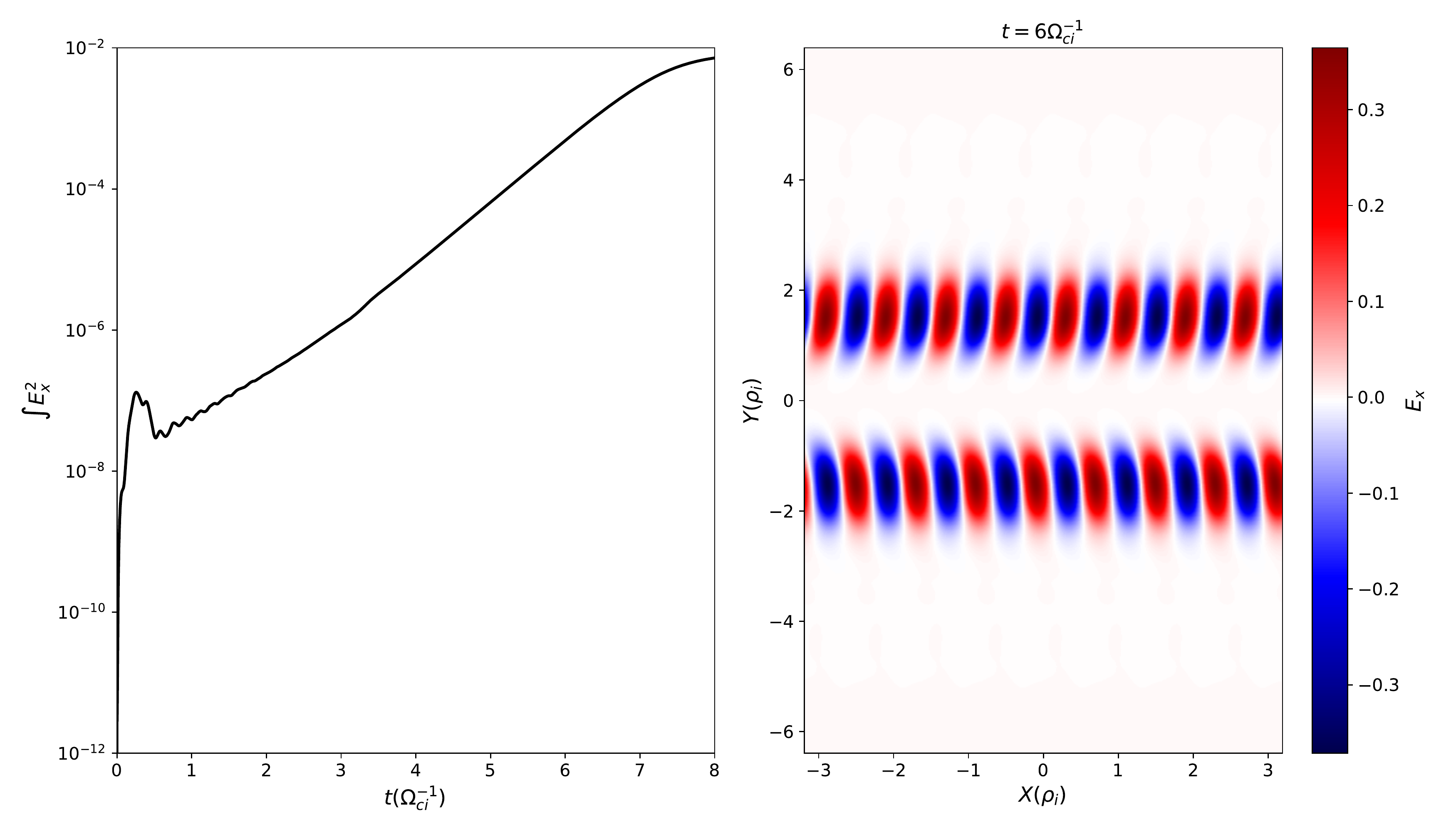}
    \caption{The exponential growth of the LHDI electric field (left) and the LHDI electric field visualized in configuration space late in the linear stage at $t=6 \Omega_{ci}^{-1}$ (right). The growth rate, $\gamma \sim 1.1 \Omega_{ci}$, compares well with linear theory and the results presented in \citet{Ng:2019}. Likewise, the mode structure in a snapshot of the LHDI electric field corresponds to the typical LHDI electric field for an $m=8$ perturbation, with the electric field localized to the edge of the current sheet where the density gradient is largest. The LHDI electric field magnitude is normalized to $B_0 v_{A_0} = B_0^2/\sqrt{\mu_0 n_0 m_p}$ where $B_0$ is the asymptotic magnetic field and $n_0$ is the density in the current layer.}
    \label{fig:lhdiElectricField}
\end{figure}
Likewise the structure is concentrated away from the current sheet centered at $y=0$, as expected since it is the edge of the current sheet where the density gradient is largest and thus most unstable to the LHDI.

In Figure \ref{fig:lhdiProtonDistributionFunction}, we present the proton distribution function at the edge of the current sheet and confirm the presence of the proton resonance expected for the LHDI.
\begin{figure}[!htb]
    \centering
    \includegraphics[width=\textwidth]{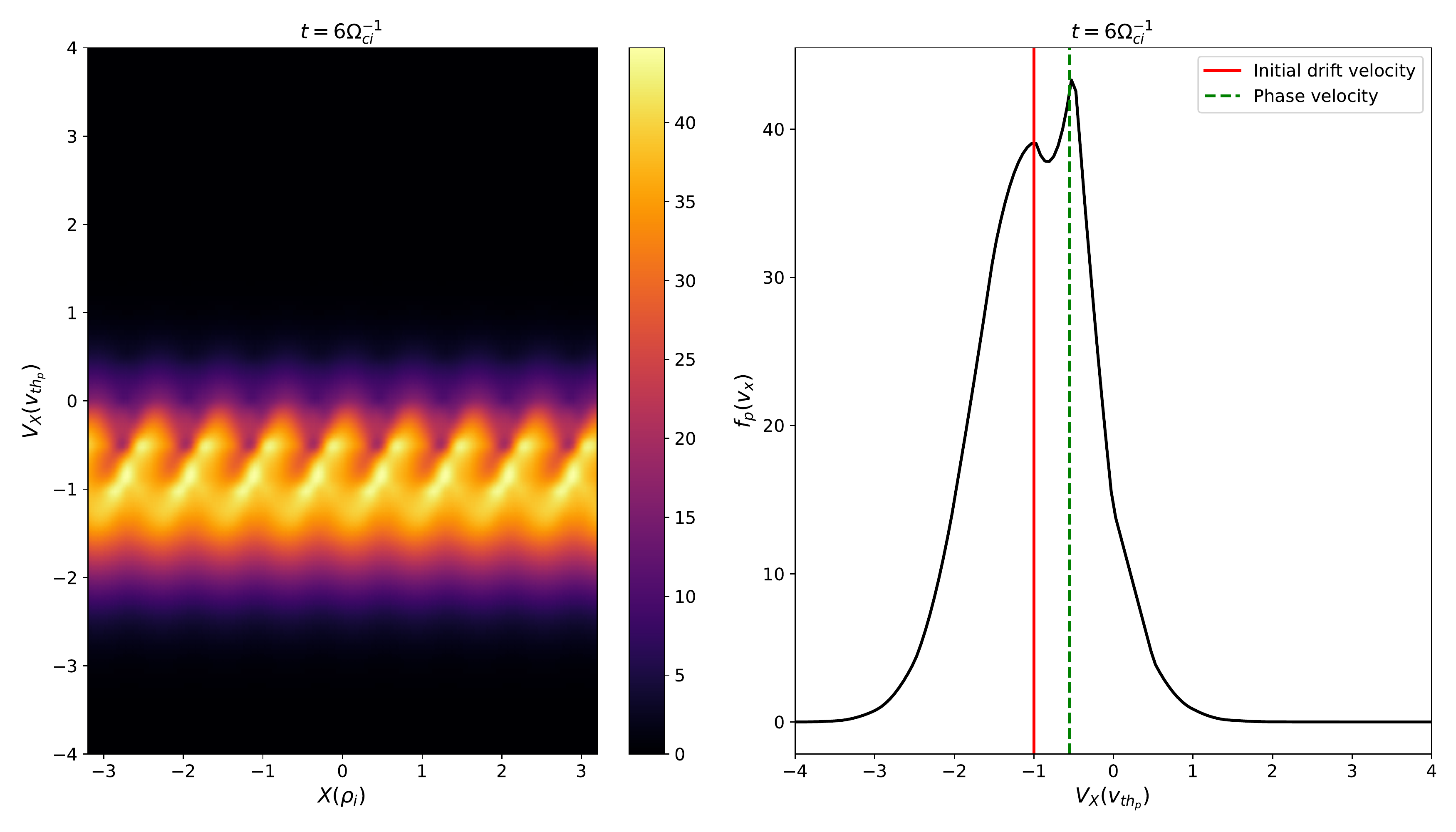}
    \caption{The distribution function for the protons plotted at $f(x, y = -1.7 \rho_p, v_x, v_y = 0.0 v_{th_p})$, at the edge of the current sheet (left), and a further cut of the 2D distribution function, $f(x = 2.3 \rho_p, y = -1.7 \rho_p, v_x, v_y = 0.0 v_{th_p})$ (right). The mode structure for an $m=8$ perturbation is again easily seen in the 2D visualization of the proton distribution function, as the protons at the edge of the current sheet are resonant with the growing electric field from the LHDI. We have over-plotted the initial drift velocity (red solid) and the phase velocity for the resonance condition (green dashed) on top of the 1D cut of the distribution function at $x = 2.3 \rho_p$.}
    \label{fig:lhdiProtonDistributionFunction}
\end{figure}
Both the initial drift and the phase velocity for the ion resonance condition are over-plotted with a cut of the distribution function through $x = 2.3 \rho_p, y = -1.7 \rho_p, v_y = 0.0 v_{th_p}$\footnote{Note that \citet{Ng:2019} contains a sign difference in the initial magnetic field profile, which manifests as a difference in the sign of the proton flow. The growth rate, mode structure, and resonant velocity are manifestly unaffected, because in 2D a change in sign of the initial flow profile is analogous to a rotation of the whole system by 180 degrees, and the Vlasov-Maxwell system has rotational symmetry.}. 
The resonant velocity is computed by solving Eq. (18) in \citet{Ng:2019}. 
The clear resonance structure in the ion distribution function, used as proof of the importance of ion kinetics in the dynamics of the instability in \citet{Ng:2019}, is again a prominent aspect of the algorithm presented here in this thesis.
While there have been numerous particle-in-cell studies of the LHDI \citep{Lapenta:2002, Lapenta:2003, Daughton:2003, Roytershteyn:2012}, the phase space structure lucidly provided by a continuum approach presents an alternative means of understanding the plasma physics of these small scale, kinetic, instabilities.

\subsection{Hybrid Two-stream/Filamentation Instability}\label{sec:hybridTSW}

Our final benchmark of our collisionless Vlasov--Maxwell solver is in the same vein as the previous section and concerns the modeling of small scale, kinetic instabilities.
In astrophysical settings, interpenetrating beams, or flows, of plasma are quite common, as they can serve as a free energy source for a myriad of instabilities.
In particular, in the unmagnetized case, the two-stream instability, filamentation instability \citep{Fried:1959}, and a hybrid mode of the two-stream and filamentation referred to as the electromagnetic oblique mode \citep{Bret:2009} are of interest for a variety of astrophysical systems from gamma ray bursts \citep{Medvedev:1999} to pulsar wind outflows \citep{Kazimura:1998} to cosmological scenarios \citep{Schlickeiser:2003, Lazar:2009}.
It is of particular interest in these astrophysical contexts if the filamentation instability, or filamentation-like instabilities, are efficient enough to produce dynamically important magnetic fields, and, for example, explain the observed emission or the presence of a magnetic field in the system.

The dynamics of these instabilities, especially their competition, served as the motivation for a recent study using the Vlasov--Maxwell solver in \gke~\citep{Skoutnev:2019}.
\citet{Skoutnev:2019} found that in a certain parameter regime, as the beams internal temperature was decreased and $v_{th}/u_d$, the ratio of the thermal velocity to the drift speed of the beam, became smaller, the electromagnetic oblique modes had comparable growth rates to the two-stream instability.
These modes thus saturated on similar time scales, leading to the dynamics of a single mode having a manifestly different final nonlinear state in comparison to an initialization of a spectrum of modes. 

We will consider the results of these nonlinear simulations from \citet{Skoutnev:2019} in Chapter~\ref{ch:Leverage}, but here we focus on the ability of the DG Vlasov--Maxwell solver to accurately capture the linear growth of these modes, two-stream, filamentation, and electromagnetic oblique.
For the purposes of demonstrating that the algorithm adequately captures the growth of these modes, we will focus on single mode simulations, in contrast to the simulations presented in \citet{Skoutnev:2019}, which were initialized from a bath of random fluctuations.
We will focus particular attention on the electromagnetic oblique modes in anticipation of how their unique physics will prove a critical component of the nonlinear evolution of a spectrum of modes discussed in Chapter~\ref{ch:Leverage}.

To initialize these single mode simulations, we consider an electron-proton plasma in 2X2V, but with the protons forming a stationary, charge-neutralizing background\footnote{For the purposes of the simulation, this limit is achieved by not adding a proton contribution to the current in Maxwell's equations so that the only contribution to the current comes from the dynamic electron species.}. 
The electrons are initialized as two drifting Maxwellians, \eqr{\ref{eq:ICMaxwellian}},
\begin{align}
    f_e (x, y, v_x, v_y) = \frac{m_e n_0 }{2 \pi T_e} \exp & \left (- m_e \frac{(v_x)^2 + (v_y - u_d)^2}{2 T_e} \right ) \notag \\
    & + \frac{m_e n_0 }{2 \pi T_e} \exp \left (- m_e \frac{(v_x)^2 + (v_y + u_d)^2}{2 T_e} \right ), \label{eq:TSWElcInit}
\end{align}
where $n_0 = 0.5$ and the drift velocity is chosen to be $u_y = 0.3c$, with $c$ being the speed of light.
The electron temperature is chosen so that $v_{th_e}/u_d = 1/3$, $v_{th_e} = 0.1c$.
The simulations are performed with $N_x \times N_y \times N_v^2 = 8 \times 8 \times 8^2$ configuration and velocity space resolution, with polynomial order 3 and the Serendipity element basis. 
The box size in configuration space is chosen to fit exactly one wave mode in the box $L_x \times L_y = 2\pi/k_x \times 2\pi/k_y$, and the velocity space extents are $[-3 u_y, 3 u_y]^2$, with periodic boundary conditions in configuration space and zero-flux boundary conditions in velocity space.
A small perturbation is seeded in the electric and magnetic fields of the form
\begin{align}
    E_x & = -\frac{\delta \sin(k_x x + k_y y)}{k_x + k_y \alpha}, \\
    E_y & = \alpha E_x, \\
    B_z & = k_x E_y - k_y E_x,
\end{align}
where $\delta$ is the size of the perturbation and $\alpha$ is a coefficient determined by the eigenfunctions of the linear theory and corresponds to the ratio of the y-electric field to the x-electric field. 

In the notation of \citet{Skoutnev:2019}, we define an angle $\theta$ with respect to x-axis so that the wave vector, $\mvec{k} = (k_x \mvec{\hat{x}}, k_y \mvec{\hat{y}})$, corresponds to a pure filamentation mode when $\theta = 0$ degrees, and a pure two-stream mode when $\theta = 90$ degrees. 
In other words, a pure $k_x$ mode is a filamentation mode, and a pure $k_y$ mode is a two-stream mode, with all the intermediate angles defining the aforementioned oblique modes.
We note in both cases the initial condition simplifies, as a filamentation mode reduces to a perturbation in $B_z$, and a two-stream mode reduces to a perturbation in $E_y$. 
For all of the simulations, $\delta$ is chosen to be sufficiently small to maximize the linear regime of the simulation and insure a reasonable fit of the growth rate.

In Figure~\ref{fig:linearTheoryTSW}, we compare the results of the linear theory with a sequence of Vlasov-Maxwell simulations using \gke~for a variety of initial perturbations.
\begin{figure}[!htb]
    \centering
    \includegraphics[width=\textwidth]{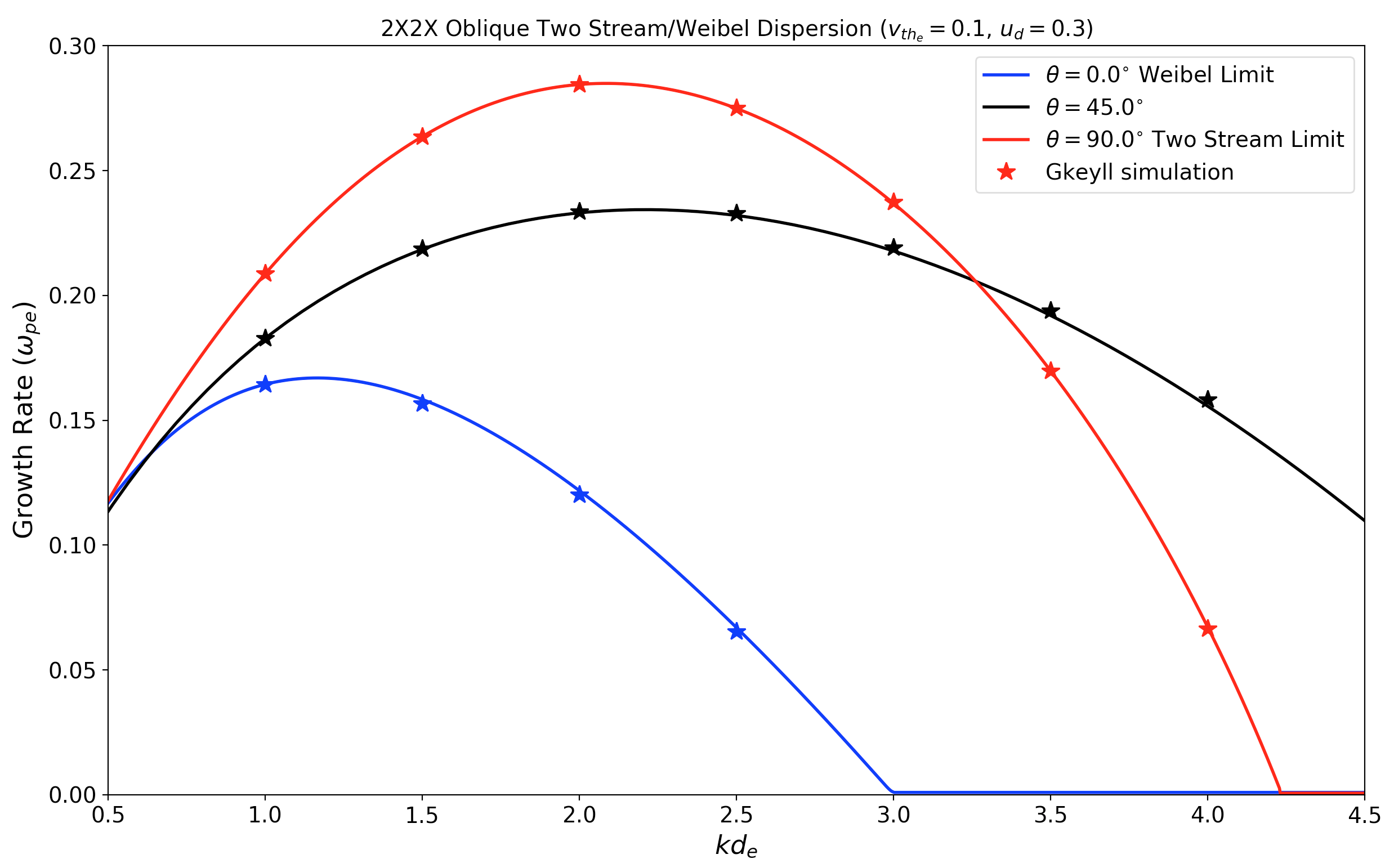}
    \caption{Comparison of linear theory (solid line) calculated from the dispersion relation in \eqr{\ref{eq:generalTSWDispersion}} after rotation to the coordinate system aligned with $\mvec{k}$, Eqns.\thinspace(\ref{eq:diag11Dispersion}--\ref{eq:diag22Dispersion}), with a number of \gke~simulations (stars) for the filamentation limit, $\theta = 0^{\circ}$, an oblique mode at $\theta = 45^{\circ}$, and the two-stream limit, $\theta = 90^{\circ}$. We observe good agreement between the linear theory and our DG Vlasov-Maxwell solver.}
    \label{fig:linearTheoryTSW}
\end{figure}
The linear theory solution is found by linearizing the Vlasov--Maxwell system of equations to obtain the dispersion matrix,
\begin{align}
    D_{ij}=\frac{\omega^2}{c^2}\left( k_ik_j-k^2\delta_{ij}\right) +\epsilon_{ij}, \label{eq:generalTSWDispersion}
\end{align}
where,
\begin{align}
    \epsilon_{ij}=\left(1-\sum_s\frac{\omega_{p_s}^2}{\omega^2}\right)\delta_{ij}+\sum_s \frac{\omega_{p_s}^2}{\omega^2}\int_{-\infty}^{\infty} v_iv_j\frac{\mvec{k}\cdot \gv f_{0_s}}{\omega-\mvec{k}\cdot \mvec{v}} \dv.
\end{align}
It is most convenient to rotate the dispersion matrix to the coordinate system aligned with the wave vector $\mvec{k}$, i.e., a rotation by the angle $\theta$ previously defined,
\begin{align}
    D=\left(
    \begin{array}{cc}
    D_{11} & D_{12}\\
    D_{21} & D_{22} 
    \end{array}\right),
\end{align}
where
\begin{align}
    D_{11} & = 1-\frac{\omega_{pe}^2}{4k^2v_{th}^2}\left[Z'(\xi_{+})+ Z'(\xi_{-})\right], \label{eq:diag11Dispersion} \\
    D_{12} = D_{21} & =\frac{\omega_{pe}^2 u_d \cos\theta}{4\omega kv_{th}^2}\left[Z'(\xi_{+})- Z'(\xi_{-})\right], \\
    D_{22} & = 1-\frac{\omega_{pe}^2}{\omega^2}-\frac{k^2c^2}{\omega^2}- \frac{\omega_{pe}^2\left(u_d^2 \cos^2\theta+v_{th}^2\right)}{4\omega^2 v_{th}^2}\left[Z'(\xi_{+})+ Z'(\xi_{-})\right]. \label{eq:diag22Dispersion}
\end{align}
Here, $Z(\xi_{\pm})$ is the plasma dispersion function previously employed in Section~\ref{sec:CollisionlessLandauDamping}, \eqr{\ref{eq:plasmaZFunc}}, but now with $\xi_{\pm}=\frac{\omega \pm ku_d \sin \theta}{\sqrt{2}kv_{th}}$.
The linear solution, the solid lines in Figure~\ref{fig:linearTheoryTSW}, are eigenmodes of the system found by solving $\text{det} (D)=0$ for $\omega$ with the corresponding eigenvectors satisfying $R^TDR\mvec{E}=0$, where $R$ is the rotation matrix for the angle $-\theta$.

We now turn to the evolution of an electromagnetic oblique mode in the nonlinear regime.
We repeat the oblique mode calculation with $\theta = 45^{\circ}$ with an increased resolution, $N_x \times N_y \times N_v^2 = 48 \times 48 \times 64^2$, and slightly larger velocity extents, $[-10 v_{th_e}, 10 v_{th_e}]^2$, running the simulation for $t = 500 \omega_{pe}^{-1}$, deep into the nonlinear evolution of the mode, with wave-vector $k_x = k_y = 2.0$. 
In Figure~\ref{fig:obliqueEarlyNonlinear}, we plot the three field components, $E_x, E_y$, and $B_z$, as well as the particle distribution function at $(y=L_y/2, v_y = 0), (x=L_x/2, v_x = 0)$, and $(x=L_x/2, y=L_y/2)$ at $t=125 \omega_{pe}^{-1}$ at the initial nonlinear phase.
\begin{figure}[!htb]
    \centering
    \includegraphics[width=\textwidth]{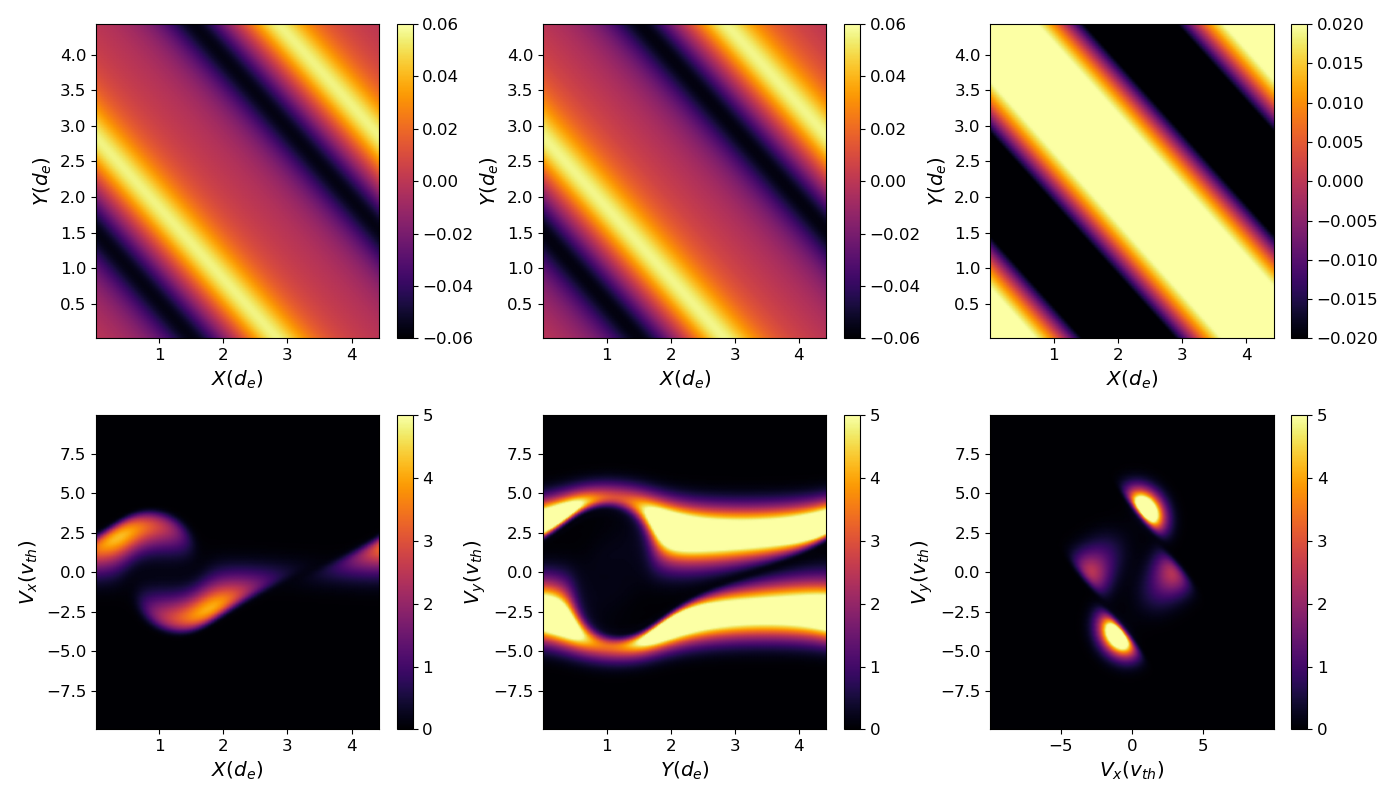}
    \caption{The evolution of the electromagnetic fields, $E_x$ (top left), $E_y$ (top middle), and $B_z$ (top right), as well as the electron distribution function at $(y=L_y/2, v_y = 0)$ (bottom left), $(x=L_x/2, v_x = 0)$ (bottom middle), and $(x=L_x/2, y=L_y/2)$ (bottom right) at $t=125 \omega_{pe}^{-1}$ as the oblique mode, $\theta=45^{\circ}$, instability is going nonlinear. We observe the growth of all three components of the initial electromagnetic fields, with standard signatures of both two-stream- and filamentation modes in the distribution function: the phase space vortices in the $x-v_x$ and $y-v_y$ plane, and the deflection of the beams in the $v_x-v_y$ plane respectively.} \label{fig:obliqueEarlyNonlinear}
\end{figure}
We can see that the oblique mode grows all three components of the field that are initialized, as well as the standard signatures of the the two-stream and filamentation instability, the phase space vortices in the $x-v_x$ and $y-v_y$ plane, and the deflection of the beams in the $v_x-v_y$ plane respectively. 
Late in time at $t=500 \omega_{pe}^{-1}$ in Figure~\ref{fig:obliqueLateNonlinear}, we see that the saturated state has little if any magnetic field, as potential wells have formed in the electric fields that have scattered the particles to a fairly isotropic state in the $v_x-v_y$ plane and depleted the phase space structure required to support a magnetic field.
\begin{figure}[!htb]
    \centering
    \includegraphics[width=\textwidth]{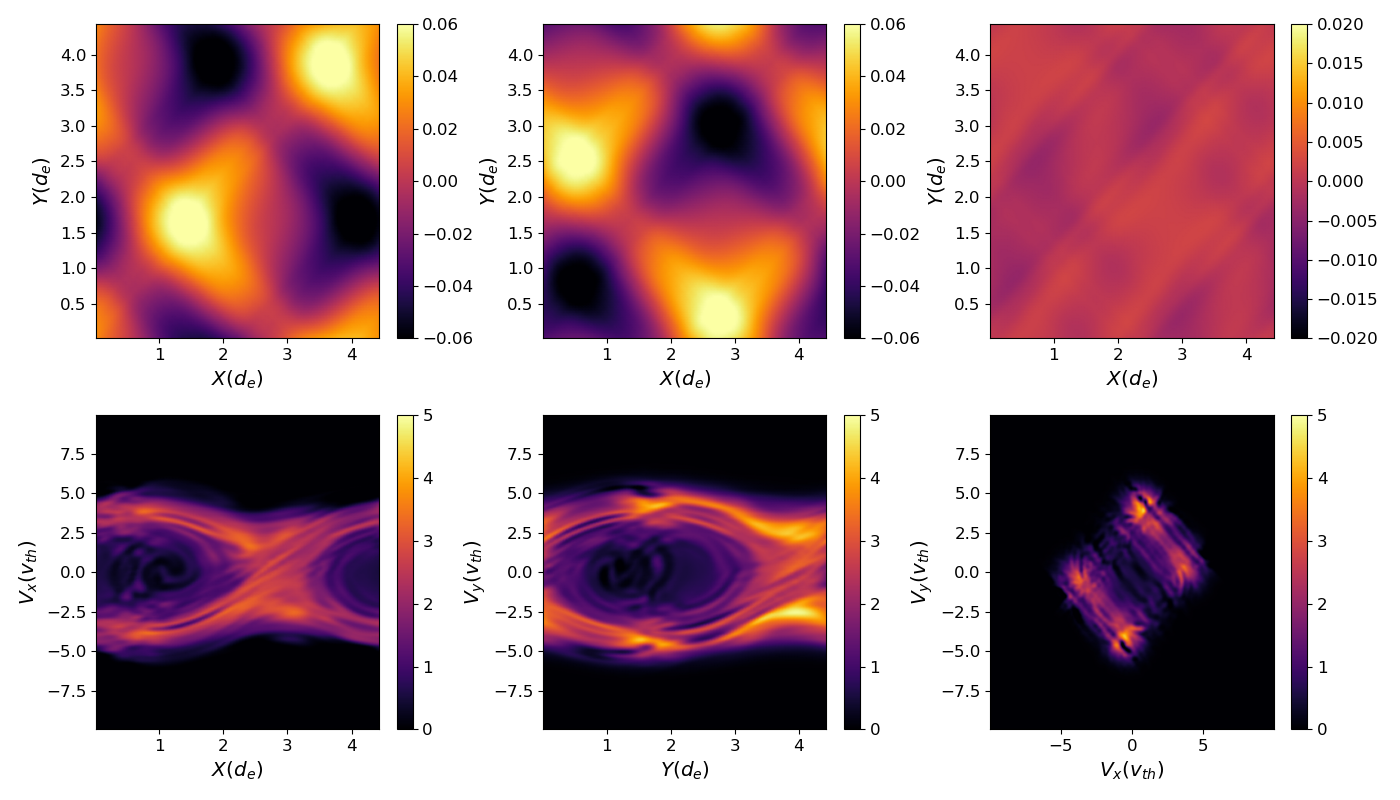}
    \caption{The evolution of the electromagnetic fields, $E_x$ (top left), $E_y$ (top middle), and $B_z$ (top right), as well as the electron distribution function at $(y=L_y/2, v_y = 0)$ (bottom left), $(x=L_x/2, v_x = 0)$ (bottom middle), and $(x=L_x/2, y=L_y/2)$ (bottom right) at $t=500 \omega_{pe}^{-1}$ of the oblique mode, $\theta=45^{\circ}$, instability deep in the nonlinear phase of the dynamics. Here, we observe little, if any, magnetic field, as the electrostatic wells forming in the electric field components scatter particles to a nearly isotropic state in the $v_x-v_y$ plane and deplete the phase space structure required to support the magnetic field.} \label{fig:obliqueLateNonlinear}
\end{figure}
This particle scattering will prove to be an important component of the nonlinear evolution of a spectrum of unstable modes in Chapter~\ref{ch:Leverage}.

\section{Benchmarking the Complete Vlasov--Maxwell-Fokker--Planck \\ System of Equations}\label{sec:benchmarkFullVMFP}

\subsection{Collisional Landau Damping}\label{sec:collisionLangmuirWave}

We return now to the Landau damping of Langmuir waves discussed in Section~\ref{sec:CollisionlessLandauDamping}, but now including the effects of collisions with our discretization of the Fokker--Planck equation.
Collisions can significantly change the damping rate, and in the limit of high collisionality, the damping can be ``shut off.'' 
This shut off happens when the mean free path becomes shorter than the wavelength, preventing the particles from resonating with the wave and gaining energy before being scattered via collisions.
We are interested in demonstrating that the discrete VM-FP system of equations in \gke~can smoothly transition from the collisionless to collisional regimes, similar to our benchmarks in Section~\ref{sec:KineticSodShock}, but now including the self-consistent plasma-electromagnetic field feedback.

We again initialize Maxwellian, \eqr{\ref{eq:ICMaxwellian}}, proton and electron distribution functions with the initial density and electric field again given by Eqns.\thinspace(\ref{eq:electronLangmuirInit}--\ref{eq:ExLangmuirInit}).
We choose a fixed $k$ for this study scanning collisionality, $k \lambda_D = 0.5$, and still set $L_x = 2\pi/k$ so exactly one wavelength fits in the domain.
The proton and electron velocity space limits are again set to $\pm 6 v_{th_s}$, with periodic boundary conditions in configuration space and zero flux boundary conditions in velocity space.

Figure~\ref{fig:langmuirDamp} shows the electric field energy as a function of time for $\nu = 0.0\omega_{pe}, 0.25 \omega_{pe},$ and $1.0 \omega_{pe}$.
\begin{figure}[!htb]
    \centering
    \includegraphics[width=0.9\textwidth]{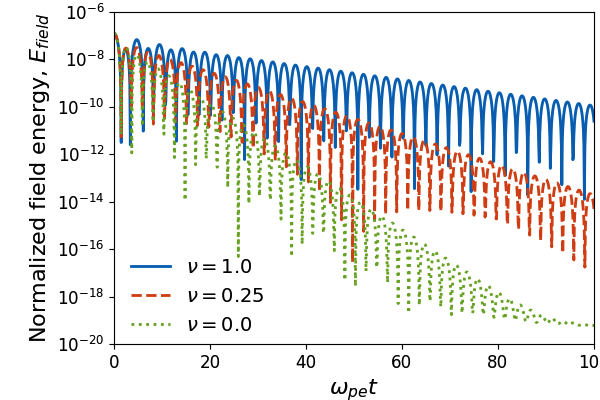}
    \caption{Field energy as a function of time for the linear collisional Landau damping problem with varying collisionality. Similar to Figure~\ref{fig:LangmuirWaveEnergy}, we compute the damping rate of each simulation by fitting to the peaks of the field energy. The collision frequency $\nu$ is normalized to the electron plasma frequency.}
    \label{fig:langmuirDamp}
\end{figure}
As the collision frequency increases, we find a rapidly decreasing damping rate in the moderate collisionality regime, as seen in Figure~\ref{fig:dampVsNu}.
\begin{figure}[!htb]
    \centering
    \includegraphics[width=0.9\textwidth]{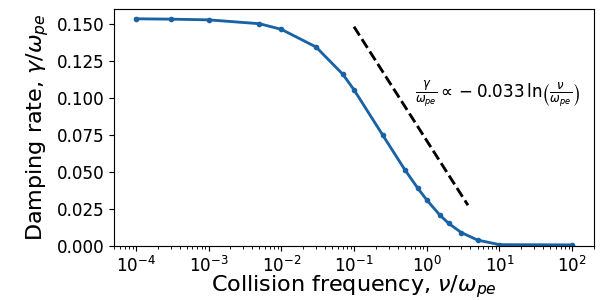}
    \caption{Damping rate versus collisionality computed from simulations such as those shown in Figure~\ref{fig:langmuirDamp}. As expected, the damping rate shuts off with increasing collisionality due to the particles being scattered by collisions before they can resonate with the wave. The black dashed line shows an analytical estimate of the damping rate computed from expressions found in \citet{Anderson:2007b} and agrees well with the results computed here.}
    \label{fig:dampVsNu}
\end{figure}
We compare this damping rate cut off with the results of \citet{Anderson:2007b}, who employ a similar simplified collision operator, though they consider the 1X3V case, and here we are examining the 1X1V case, so the results are not expected to match exactly. 
Nevertheless, a fit (black dashed line in Figure~\ref{fig:dampVsNu}) to the slope in the intermediate collisionality transition regime from the theory in \citet{Anderson:2007b} shows reasonable agreement with the numerical results.

\subsection{Heating via Magnetic Pumping}

Our final benchmark provides an opportunity to perform our most exacting test yet of the VM-FP system of equations.
We will examine heating via magnetic pumping, a process by which oscillations of the magnetic field are converted to particle energy. 
Magnetic pumping relies on the approximate conservation of the magnetic moment, $\mu = m v_\perp^2/2B$, in a magnetized plasma. 
As the magnetic field increases, to maintain magnetic moment conservation, $v_\perp^2$ should also increase.
In a collisionless system, if the magnetic field is oscillating slowly compared to the gyro period, then $v_\perp^2$ oscillates up and down in a reversible way, and there is no net heating of the plasma.  
However, collisions can provide a route to pitch angle scatter the energy into the parallel direction, leading to an overall irreversible heating of the plasma.

This mechanism was originally proposed as a heating mechanism in the early days of fusion research and investigated extensively~\citep{Berger:1958,Laroussi:1989}.
Recently, this same mechanism has been studied as a potential source of particle heating in the solar wind~\citep{Lichko:2017}.
We use a similar setup as~\citet{Lichko:2017}, with a small modification to the parameters and a different collision operator\footnote{Both our collision operator and the collision operator employed by \citet{Lichko:2017} are Fokker--Planck collision operators, but \citet{Lichko:2017} discretizes the full, unsimplified Fokker--Planck equation written in Landau form,
\begin{align}
\pfrac{f^c_s}{t}  = \sum_{s'} \nu_{s,s'} \gv \cdot \int d\mvec{v}' \overleftrightarrow{\mvec{U}}(\mvec{v}, \mvec{v}') \cdot \left ( f_{s'}(\mvec{v}') \gv f_s(\mvec{v}) - \frac{m_s}{m_{s'}} f_s(\mvec{v}) \nabla_{\mvec{v}'} f_{s'}(\mvec{v}') \right ), \label{eq:LandauCollision}
\end{align}
where
\begin{align}
\nu_{s,s'} = \frac{q^2_s q^2_{s'} \ln(\Lambda)}{8 \pi m_s \epsilon_0}
\end{align}
is the collision frequency of species $s$ colliding with species $s'$, and $\overleftrightarrow{\mvec{U}}(\mvec{v}, \mvec{v}')$ is the Landau tensor,
\begin{align}
\overleftrightarrow{\mvec{U}}(\mvec{v}, \mvec{v}') = \frac{1}{|\mvec{v} - \mvec{v}'|} \left (\overleftrightarrow{\mvec{I}} - \frac{(\mvec{v} - \mvec{v}')(\mvec{v} - \mvec{v}')}{|\mvec{v} - \mvec{v}'|^2} \right ).
\end{align}
}.
Note that the collision operator employed by \citet{Lichko:2017} retains the velocity dependence of the collision frequency, and is thus a more accurate description of collisions in a plasma.
Nevertheless, our simplified Fokker--Planck operator contains pitch-angle scattering due to the isotropic diffusion term, and thus can be used to test whether our discretization of the VM-FP system of equations contains an accurate representation of magnetic pumping.

We set up a 1X3V domain which has extents $[0,200\pi\rho_e]\times[-8v_{th,s}, 8v_{th,s}]^3$ on a $256\times24^3$ grid.
Here, $\rho_s = v_{th,s}/\Omega_{cs}$ is the gyroradius of species $s$. A perturbation is driven on a  background magnetic field $\mvec{B}=B_0 \hat{\mvec{z}} $ using an antenna that drives currents given by
\begin{align} 
\mvec{J} {=} 
\hat{\mvec{y}}J_0
\sin^2\left [\frac{\pi}{2}\thinspace\min(1, \omega_{\textrm{ramp}}t) \right ]
& \sin(\omega_{\textrm{pump}} t) \notag \\
\times & \left[
\exp\left(-\frac{(x-x_1)^2}{2\sigma_J^2}\right){-}
\exp\left(-\frac{(x-x_2)^2}{2\sigma_J^2}\right)
\right]. 
\label{eq:jPump}
\end{align}
The current is turned on slowly over one pumping period using $\omega_{\text{ramp}}=\omega_{\textrm{pump}}$. 
This ramping phase ensures that the antenna is ``turned on'' slowly and hence does not excite unwanted waves in the plasma. 
Further, we need to ensure that the plasma density is low enough that the electromagnetic waves are not ``trapped'' in the density holes that are created around the antenna.

The tests shown here use $\omega_{\textrm{pump}}=0.1 \Omega_{ce}$, $x_1=50\pi\rho_e$, $x_2=150\pi\rho_e$, $\sigma_J=200\pi\rho_e/256$, and $ \Omega_{ce}=2.5\omega_{pe}$. 
We employ a proton mass ratio $m_p/m_e=1836$ and initialize electron and proton species as Maxwellians with zero mean flow, number density $n\thinspace \rho_e^3=2.99\times10^{5}$, and thermal speed $v_{th_e}^2/c^2=\beta \Omega_{ce}^2/[2\omega_{pe}^2(1+\tau)]$.
The temperature ratio is $\tau=T_p/T_e=1$, and the ratio between plasma and magnetic pressures is $\beta=2\times10^{-4}$. 
With these quantities, the normalized background magnetic field amplitude is $\epsilon_0\omega_{pe}B_0/(en)=\Omega_{ce}/\omega_{pe}$, and we use the normalized driving current density amplitude $J_0/(enc)=\Omega_{ce}/(2\omega_{pe})$. 

Figure\thinspace\ref{fig:magPump} shows the evolution of the magnetic field and thermal energy in the middle of the domain, $x = 100 \pi \rho_e$.
\begin{figure}[!htb]
    \centering
    \includegraphics[width=\textwidth]{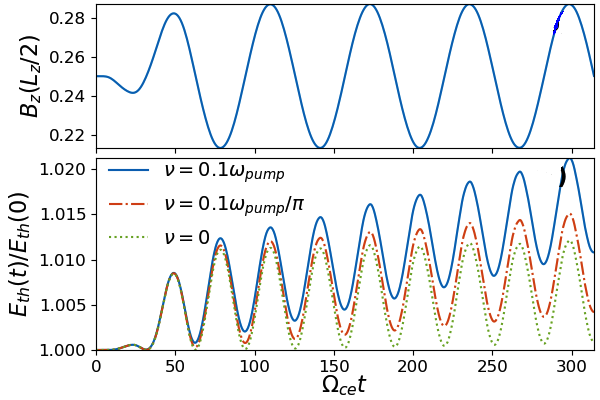}
    \caption{Time evolution of the magnetic field (top) in the middle of the domain from the magnetic pumping problem. As the antenna currents ramp up, an oscillating field is created that then transfers energy, via pitch-angle scattering, to the plasma, leading to an increase in the thermal energy (bottom). With zero collisionality (bottom, green), the energy exchange is completely reversible, and no net heating is observed, but as the collision frequency is made finite, magnetic pumping begins to heat the plasma.}
    \label{fig:magPump}
\end{figure}
As the antenna current ramps up, an oscillating magnetic field structure is created. 
The amplitude of oscillations are about $15\%$ of the background. 
This oscillating energy is then transferred to parallel heating via pitch angle scattering. 
This heating is shown in the bottom panel of the figure, which shows that as the collision frequency becomes finite, the plasma gains thermal energy through the simulation.
Importantly, these simulations show how taxing this test problem is, as it relies on every part of the discretization of the VM-FP system of equations, and that the scheme must be able to preserve the adiabatic invariants.
Were the magnetic moment, $\mu$, not conserved in the zero collisionality case, and the overall scheme not conservative, we would not be able to confidently argue the heating demonstrated is a consequence of the physics contained in the collision operator.

As a comprehensive test of the algorithm's ability to model heating via magnetic pumping, we next turn to the heating rate versus the ratio of the collisionality to the pump frequency, $\nu/\omega_{\textrm{pump}}$. In Figure~\ref{fig:heatRates}, we plot the heating rate computed from the code, 
\begin{align}
    \gamma_H = \frac{1}{\mathcal{E}} \pfrac{\mathcal{E}}{t},
\end{align}
where $\mathcal{E}$ is the second velocity moment, \eqr{\ref{eq:energyDefinition}}, the particle energy.
This quantity is computed in the middle of the domain, $x = 100  \pi \rho_e$.
\begin{figure}[!htb]
    \centering
    \includegraphics[width=\textwidth]{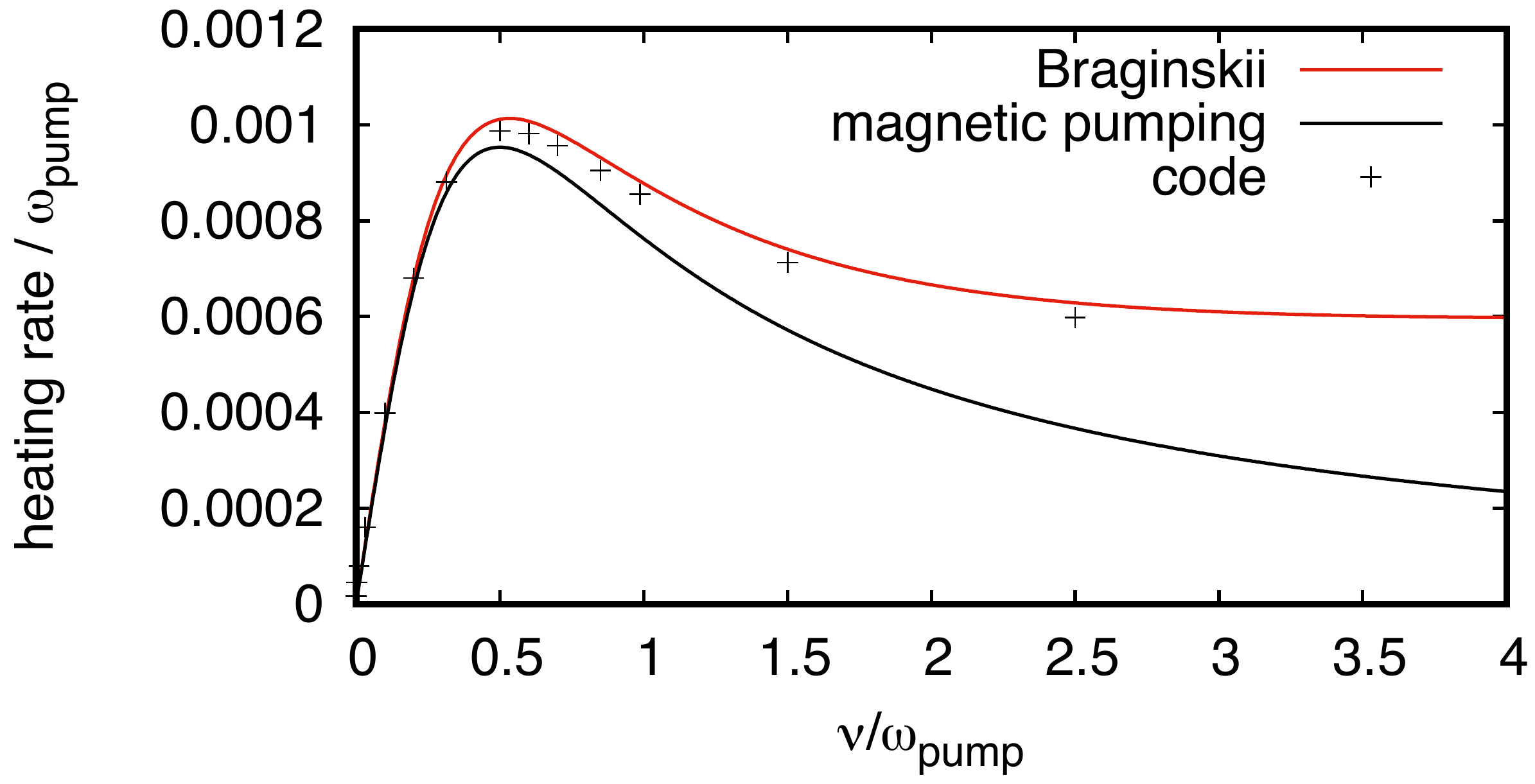}
    \caption{Heating rate via magnetic pumping, plus an additional viscous heating mechanism, as a function of normalized collision frequency. The code agrees well with the theoretical prediction (black line) magnetic pumping at lower collision frequency, but shows an additional heating mechanism at higher collisionalty due to the viscous damping of out-of-plane flows, which are included in the Braginskii-based theory (red line).}
    \label{fig:heatRates}
\end{figure}
We compare the results of \gke~simulations with our DG VM-FP solver to the heating rate predicted by the theory of magnetic pumping \citep{Lichko:2017},
\begin{align}
    \gamma_{\textrm{mp}} = \omega_{pump}^2 \frac{2}{9} \left (\frac{\delta n}{n_0} \right )^2 \frac{\nu}{(\omega_{\textrm{pump}}^2 + 4 \nu^2)},
\end{align}
where $\delta n$ and $n_0$ are computed from the central density $n(t) = n_0 + \delta n \sin(\omega_{\rm pump} t)$ after the initial transients.
Note that this heating rate is derived in terms of the magnetic fluctuations, $\delta B/B_0$, but if the plasma is frozen-in to the magnetic field, the ratios $\delta n/n_0$ and $\delta B/B_0$ are equal.
To correctly match the heating rates computed from the time evolution of the temperature, the density compression is also measured in the middle of the domain, $x = 100  \pi \rho_e$.

Our discretization of the VM-FP system of equations agrees with magnetic pumping theory for small $\nu/\omega_{\textrm{pump}} \lesssim 1$, but indicates an additional heating mechanism for larger collisionality.
This trend was also observed, but in a different parameter regime and using the Landau form of the collision operator, \eqr{\ref{eq:LandauCollision}}, in~\citet{Lichko:2017}.
Because the pump frequency is larger than the proton cyclotron frequency, $\omega_{\textrm{pump}} > \Omega_{cp}$, the protons are unmagnetized and unable to respond to the compression of the magnetic field.
Thus, when the electrons undergo compression, the protons are effectively stationary, leading to an electric field to maintain charge neutrality, but this electric field drives an $\mvec{E} \times \mvec{B}$ flow.
In our chosen geometry, the electric field to maintain quasi-neutrality develops in the $x$ direction, so the $z$ magnetic field drives a flow in the $y$ direction.
This flow is then viscously damped, leading to additional heating.

This additional heating can be derived by considering a Braginskii calculation \citep{Braginskii:1965} in which flows are viscously damped in the limit $\nu \gg \omega_{\textrm{pump}}$.
We can use the Braginskii stress tensor\footnote{Note that the Braginskii calculation is performed as an asymptotic expansion of the full Fokker--Planck collision operator, written in the Landau form, \eqr{\ref{eq:LandauCollision}}} to compute the heating rate for the viscous damping of the electron flows,
\begin{align}
    \gamma_{B} = \frac{2}{3 n T}\left[ \left( \frac{\eta_0}{3} + \eta_1 \right)\overline{\left( \frac{\partial u_x}{\partial x} \right)^2} + \eta_1 
    \overline{\left( \frac{\partial u_y}{\partial x} \right)^2}
    \right], 
\end{align}
where $\eta_0 = 0.96 n T \tau_c$ and $\eta_1 = 0.3 n T / (\tau_c  \Omega_{c}^2)$ are two of Braginskii's viscosity coefficients and $\tau_c$ is the collision time for the species.
These expressions are for $\omega_{\rm pump} \ll \nu \ll \Omega_c$, but are generalized for arbitrary $\nu/\Omega_c$ in \citet{Braginskii:1965}.
The $\eta_0$ term gives rise to magnetic pumping in the collisional limit~\citep{Kulsrud:2005,Schekochihin:2005}, and asymptotic matching can be done to extend the definition of $\eta_0$ into the low collisionality regime. 
We can then relate Braginskii's collision time, $\tau_c$, to the collision rate for our simplified Fokker--Planck collision operator by $\tau_c = 0.52 / \nu$. 
The $\eta_1$ term represents additional viscous heating due to classical cross-field momentum transport.

One can calculate the time-averaged squared shearing rate given by
\begin{align}
   \overline{\left( \frac{\partial u_x}{\partial x} \right)^2} = (1/2) \omega_{\textrm{pump}}^2 \left ( \frac{\delta n}{n_0} \right )^2,
\end{align}
and
\begin{align}
  \overline{\left( \frac{\partial u_y}{\partial x} \right)^2} = \frac{1}{2} \frac{\omega_{pe}^4}{\Omega_{ce}^2} \left ( \frac{\delta n}{n_0} \right )^2,
\end{align}
to find that the out-of-plane flows are actually larger, with $\overline{u_y^2} \approx 2.56 \, \overline{u_x^2}$ for our parameters.  
Viscous heating from damping these flows dominates at high collisionality for these parameters.
Note that because we are using $\delta n/n_0$ in the formulas, we obtain a slightly smaller heating rate since $\delta n/n_0 = 0.131$, but $\delta B/B_0 = 0.148$ in our simulations, because the plasma is not completely frozen-in.
In spite of these subtleties, the plasma not being completely frozen-in and the use of a different collision operator, we find good agreement between the theoretical heating rates in these two different parameter regimes, and our simulations add further credibility to our implementation of the DG discretization VM-FP system of equations in \gke.

Although this benchmarking section has been by no means exhaustive, we have covered a wide spectrum of functionality within our algorithm for the VM-FP systems of equations.
We have demonstrated numerically the conservation properties proved analytically in Chapter~\ref{ch:DGFEM}, and further shown numerically that our scheme satisfies discrete analogs of the Second Law of Thermodynamics and an H-theorem.
We have shown the code obtains theoretical estimates for damping rates, growth rates, and heating rates in a variety of non-trivial test cases of both the collisionless Vlasov--Maxwell implementation and the full Vlasov--Maxwell--Fokker--Planck numerical method. 
Further, we have shown that a continuum VM-FP solver provides a high fidelity representation of the particle distribution function which can be leveraged to clearly identify everything from particle trapping to resonant wave-particle interactions.
We turn now to the question of critical importance: what science can be done with this novel, well-tested tool that provides such high quality particle distribution function data?
%Chapter 5

\renewcommand{\thechapter}{5}
\epigraph{Some of the material in this chapter has been adapted from \citet{Juno:2020} and \citet*{Skoutnev:2019}}{}
\chapter{Leveraging the Uncontaminated Phase Space}\label{ch:Leverage}

We turn now to a question of the utmost importance after the meticulous work to derive, implement, and test a novel numerical method for the VM-FP system of equations: what new science can be done with this tool?
As we discussed in Chapter~\ref{ch:Benchmarks}, the continuum representation of the particle distribution function, free of the counting noise which normally pollutes a particle-based discretization, allows for the clear identification of plasma processes in phase space.
We would like now to leverage this high fidelity representation for the particle distribution function in a variety of numerical experiments to provide new perspective on energization processes and nonlinear saturation mechanisms in a number of plasma environments.

This chapter will not be an exhaustive discussion of every ongoing project with the VM-FP solver in \gke. 
It is merely our goal to demonstrate the versatility of this approach of a continuum discretization and to justify our effort in the previous chapters deriving and implementing the DG algorithm for the VM-FP system of equations.
We refer the reader to a number of publications for the breadth of applicability of the VM-FP solver, including bounded plasma and plasma sheath studies \citep{Cagas:2017a, CagasThesis:2018, Cagas:2020}, electrostatic shocks \citep{Pusztai:2018, Sundstrom:2019}, instability calculations \citep{Cagas:2017b, Ng:2019}, and simulations of the plasma dynamo \citep{Pusztai:2020}.

We will focus on the ability to directly diagnose the energy transfer between the electromagnetic fields and the plasma in phase space, and the nonlinear saturation of instabilities driven by counter-streaming beams of plasma.
Using the clean, uncontaminated phase space, we will be able to identify phase space energization signatures as a complement to other methods of determining the mechanisms of energy exchange within a plasma.
Likewise, we will leverage the high fidelity representation of the distribution function to completely characterize the nonlinear dynamics of the beam-driven instabilities discussed in Section~\ref{sec:hybridTSW}, and in doing so, showcase a situation where the particle noise inherent to particle-based methods can lead to deceptive dynamics.

\section{Directly Diagnosing the Energy Transfer in Phase Space}

\subsection{The Field-Particle Correlation}
Before we dive into the distribution function, we require a means of interpreting the structure in the distribution function and how this structure can be translated to study the energy transfer between the electromagnetic fields and the plasma.
To probe the energy exchange between electromagnetic fields and the plasma in phase space, we will utilize a technique called the field-particle correlation \citep{Klein:2016, Klein:2017a, Klein:2017b, Klein:2020, Howes:2017, Howes:2018, Li:2019}. 
The essential idea behind the field-particle correlation is to determine where in phase space the plasma is gaining or losing energy, and thereby ascertain the specifics of the energization process, or processes, that may be occurring.

To derive the field-particle correlation diagnostic, we examine the collisionless Vlasov equation weighted by $1/2\thinspace m_s |\mvec{v}|^2$,
\begin{align}
    \pfrac{w_s}{t} = -\mvec{v} \cdot \nabla_{\mvec{x}} w_s - \frac{q_s}{2}|\mvec{v}|^2\mvec{E} \cdot \nabla_{\mvec{v}} f_s - \frac{q_s}{2}|\mvec{v}|^2(\mvec{v} \times \mvec{B}) \cdot \nabla_{\mvec{v}} f_s, \label{eq:dws}
\end{align}
where we have separated out each component of the phase space flux: the configuration space streaming term, the electric field, and the magnetic field.
Here, $w_s(\mvec{x}, \mvec{v}, t) = m_s|\mvec{v}|^2 f_s(\mvec{x}, \mvec{v}, t)/2$ is the phase space energy density and is a function of the full 6D phase space, because we have not performed any integrations over phase space.

However, we can gain intuition for how $w_s$ evolves by integrating over phase space,
\begin{align}
    \pfrac{W_s}{t} =  
    -\int \int q_s \frac{|\mvec{v}|^2}{2} \mvec{E} \cdot \gv f_s \dx \dv
    & = -\int \left( \int q_s \mvec{v}  f_{s} \dv \right) \cdot \mvec{E} \dx \notag \\
    & = -\int \mvec{J_s} \cdot \mvec{E} \dx, \label{eq:dwsdt}
\end{align}
where we have split the integral over phase space into an integral over configuration space and velocity space.
Here,
\begin{align}
    W_s = \int \mathcal{E}_s \dx,
\end{align}
the integral of the particle energy over all of configuration space. Note that we have performed similar operations to the proof of Proposition~\ref{prop:collisionlessEnergyConservation} in Appendix~\ref{app:proofsContinuous}, i.e., we have integrated the velocity gradient by parts which eliminates the contribution from the magnetic field by properties of the cross product, and we have used a suitable boundary condition, such as periodic boundary conditions in configuration space and the distribution function vanishing at the edge of velocity space, to eliminate the boundary terms.
In other words, the exchange of energy between the plasma and the electromagnetic fields is governed entirely by the electric field since only the electric field can do work on the plasma, and vice versa.
Both the magnetic field and streaming term can move energy around in phase space, but neither component of the Vlasov equation corresponds to a net energization or de-energization of the plasma.

We could stop here and only use $\mvec{J}_s \cdot \mvec{E}$ as a proxy for the bulk energization of the plasma, but this would be restrictive, as $\mvec{J}_s \cdot \mvec{E}$ gives us no information about what is happening to the particles as a function of their particular velocities.
From this formulation of the energy exchange, we would be unable to distinguish between energization processes such as resonant wave-particle interactions and direct acceleration via electric fields.
In this vein, we have no way to distinguish between a transfer of energy which is oscillatory, such as a wave propagating through the plasma, and a transfer of energy which is secular, such as the wave damping on the plasma via a resonant process like Landau damping.

So, we step back from performing the integration over phase space and focus on \eqr{\ref{eq:dwsdt}}.
Since we expect the electric field to be the only participant in the direct energization and de-energization of the plasma, we will define the field-particle correlation
\begin{align}
    C(\mvec{x},\mvec{v},t,\tau) = -\frac{q_s}{2} \frac{1}{\tau} \int_{t}^{t+\tau} |\mvec{v}|^2 \mvec{E}(\mvec{x},t') \cdot \gv f_s(\mvec{x},\mvec{v},t') \thinspace dt'. \label{eq:FPCTimeAvg}
\end{align}
Here, $\tau$ defines a correlation time over which to average so we can address our previous concern about distinguishing between oscillatory and secular energy transfer by averaging over the oscillatory energy exchange.
In the limit of $\tau \rightarrow 0$, we obtain the instantaneous energy exchange,
\begin{align}
    C(\mvec{x},\mvec{v},t,0) = \pfrac{w_s}{t} = -\frac{q_s}{2} |\mvec{v}|^2 \mvec{E}(\mvec{x},t) \cdot \gv f_s(\mvec{x},\mvec{v},t). \label{eq:FPCInstantaneous}
\end{align}

Importantly, because this diagnostic does not require integrations over configuration space, it can be used as a single-point diagnostic.
This feature has already been leveraged to discover the presence of electron Landau damping in observations of the Earth's turbulent magnetosheath using spacecraft measurements \citep{Chen:2019}.
The result in \citet{Chen:2019} provides sizable motivation to apply the field-particle correlation to other plasma systems beyond the Alfv\'enic turbulence studied with the field-particle correlation in, e.g., \citet{Klein:2017b}, that gave a frame of reference for the signature of Landau damping observed in \citet{Chen:2019}.
By applying the field-particle correlation to other plasma systems, we can build a Rosetta stone that can be used to translate the signatures observed in other spacecraft observations.
We undertake such a study in the next section.

\subsection{Perpendicular Collisionless Shock}

We now examine in greater detail the results of the simulation shown in Figure~\ref{fig:proton-dist-proof-of-concept} in Section~\ref{sec:introObjectives}.
The particular simulation is a perpendicular collisionless shock.
Here, a collisionless shock refers to a shock-wave, a disturbance propagating faster than the local (magneto)sonic speed, which inevitably dissipates its bulk kinetic energy as other forms of energy, e.g., thermal energy, by means other than particle collisions, because the shock wave forms on scales smaller than the inter-particle mean-free path.
For a survey of studies of collisionless shocks relevant for the heliosphere and Earth's bow shock, we refer the reader to \citet{Wilson:2010, Wilson:2012, Wilson:2014a, Wilson:2014b} and references therein.

Since these shock-waves are collisionless, we know that the energy transfer from the kinetic energy of the incoming supersonic flow into thermal and electromagnetic energy occur due to kinetic processes such as wave-particle interactions and small-scale instabilities.
And, since this energy conversion is collisionless, it can be diagnosed directly in phase space with the aforementioned field-particle correlation technique, \eqr{\ref{eq:FPCTimeAvg}}.
We will use a perpendicular collisionless shock set-up in 1X2V to determine how the upstream kinetic energy from the supersonic plasma flows is converted to other forms of energy.
Here, perpendicular refers to the orientation of the magnetic field with respect to the shock normal, the direction of the incoming supersonic flow.
We now describe in detail the simulation parameters.

The particular geometry we choose is the one spatial coordinate is in the $x$ direction, with the initial magnetic field in the $z$ direction, $\mvec{B} (t = 0) = B_0 \mvec{\hat{z}}$.
In this geometry, we can see why we only require the two velocity dimensions perpendicular to the magnetic field to describe the dynamics because of how Maxwell's equations simplify,
\begin{align}
    \pfrac{B_z}{t} & = -\pfrac{E_y}{x}, \\
    \pfrac{E_y}{t} & = -c^2 \pfrac{B_z}{x} - \frac{J_y}{\epsilon_0}, \\
    \pfrac{E_x}{t} & = -\frac{J_x}{\epsilon_0}.
\end{align}
The electrons and protons are initialized with the same supersonic flow into a reflecting wall, which leads to a shock wave that propagates from left to right in our simulation.
Note that the particles reflect from the wall, but the ``reflecting wall'' boundary condition for the electromagnetic fields is a conducting wall boundary condition in the traditional sense, with zero normal magnetic field and zero tangential electric field.
This method of initialization is often called the ``injection'' setup, and this setup has been previously employed in numerous particle-in-cell studies of collisionless shocks \citep[e.g.,][and references therein]{Caprioli:2014a,Caprioli:2014b,Caprioli:2014c}.

Detailed parameters are as follows: the reflecting wall for the particles and conducting wall for the electromagnetic fields are at $x = 0$, and plasma is injected with a copy boundary condition\footnote{We previously employed this boundary condition in Sections~\ref{sec:KineticSodShock} and \ref{sec:threeSpeciesShock}, but we repeat the definition of this boundary condition here for completeness. A copy boundary condition means that the value in the ghost layer at the rightmost grid cell is exactly equal to the value in the rightmost grid cell, for all the quantities being evolved, including the distribution functions for the electrons and protons, and the electromagnetic fields. Because the plasma is initialized with a flow propagating from right to left, this boundary condition leads to a continuous injection of plasma from the right wall, with the corresponding electric field and magnetic field to support the $\mvec{E} \times \mvec{B}$ flow.
} 
at $x = 25 d_p$, where $d_p$ is the proton collisionless skin depth, $d_p = c/\omega_{pp}$.
Here, $c$ is the speed of light, and $\omega_{pp}$ is proton plasma frequency, $\omega_{pp} = \sqrt{e^2 n_0/\epsilon_0 m_p}$.
We use a reduced mass ratio between the protons and electrons, $m_p/m_e = 100$.
The total plasma beta, $\beta = 2 \mu_0 n_0 (T_e + T_p)/B^2 = 2$, with the proton beta, $\beta_p = 1.3$, and electron beta, $\beta_e = 0.7$.

Both the protons and electrons are non-relativistic, with $v_{th_e}/c = 1/(16\sqrt{2})$, with the previous definitions of the thermal velocity, $v_{th_s} = \sqrt{T_s/m_s}$.
The in-flow velocity to initialize the perpendicular, electromagnetic shock is $U_x = -3 v_A$ ($U_x < 0$ because the in-flow is from right to left), where $v_A$ is the proton Alfv\'en speed, $v_A = B_0/\sqrt{\mu_0 n_0 m_p}$.
Since the plasma is initialized with a flow transverse to a background magnetic field, we initialize the corresponding electric field necessary to support this flow, $\mvec{E} = -\mvec{u} \times \mvec{B} = U_x B_0 \mvec{\hat{y}}$.
With these specified parameters and initial flow, we can initialize Maxwellian velocity distribution, \eqr{\ref{eq:ICMaxwellian}}, functions for the protons and electrons.

For the grid in configuration space, we use $N_x = 1536$, $\Delta x \sim d_e/6$, with piecewise quadratic Serendipity elements for the discontinuous Galerkin basis expansion.
In velocity space, the electron extents are $\pm 8 v_{th_e}$, and the proton extents are $\pm 16 v_{th_p}$, with zero-flux boundary conditions at the edges of velocity space, and $N_{v_x} = N_{v_y} = 64$ for both species, corresponding to $\Delta v = v_{th_e}/4$ for the electrons and $\Delta v = v_{th_p}/2$ for the protons.

We solve the full VM-FP system of equations and run the simulation with a small amount of collisions to regularize velocity space. 
We find the additional boundary condition from the collision operator, \eqr{\ref{eq:additionalFPBoundaryTerm}}, also assists in stability by providing a small amount of regularization at the edge of velocity space.
In this case, we choose an electron-electron collision frequency, $\nu_{ee} = 1.0e-4 \Omega_{ce} = 0.01 \Omega_{cp}$, much less than the proton cyclotron frequency, $\Omega_{cp} = e B_0/m_p$, with the proton-proton collision frequency correspondingly smaller based on the square root of the mass ratio, $\nu_{pp} = 0.001 \Omega_{cp}$.

We will begin with a discussion of the overall structure of the collisionless shock.
In Figure~\ref{fig:perpShockFields}, we show the electromagnetic fields and reduced particle distribution functions in $x-v_x$ phase space, integrated over $v_y$, for the electrons and protons, after the perpendicular shock has formed and propagated through the simulation domain, $t_{end} = 11 \Omega_{cp}^{-1}$.
Although the downstream region after the shock has passed through the plasma is fairly oscillatory, because the energy injected into the plasma by the shock sloshes back and forth between the electromagnetic fields and particles, we can estimate the compression ratio of this low Mach number shock based on the magnetic field to be roughly, $r \sim 2.5$. 
This estimate is based on the mean value of the magnetic field, $B_z$, in the downstream region (solid black line in Figure~\ref{fig:perpShockFields}).
With this estimate for the compression ratio, we calculate the shock velocity to be $U_{shock} = U_x/(r-1) = 2 v_A$.

We have marked an approximate transition from the upstream of the shock to the shock ramp (dashed-dotted lines) and likewise an approximate transition from the shock to the downstream region (dashed lines) in Figure~\ref{fig:perpShockFields}. 
The full extent of the shock includes the foot, where the initial field variation begins, the ramp, where most of the reflected proton population can be found, and the overshoot.
It is worth emphasizing a striking feature of the electromagnetic fields through the shock: we expect the y-electric field to be the dominant component of the energization of the protons and electrons through the shock, because the x-electric field is roughly bimodal through the shock and oscillates about 0 in the downstream.
This feature is perhaps intuitive, as in this reduced dimensionality, the x-electric field is the electrostatic component of the dynamics, and so we might naively expect that the dominant energy exchange will happen through the electromagnetic component of the fields, i.e., the component of the electric field which supports the compression of the magnetic field.
Still, these features fittingly foreshadow our ultimate analysis of the phase space signature of the energization mechanism.

\begin{figure}
    \centering
    \vspace{-0.5in}
    \includegraphics[width=\textwidth]{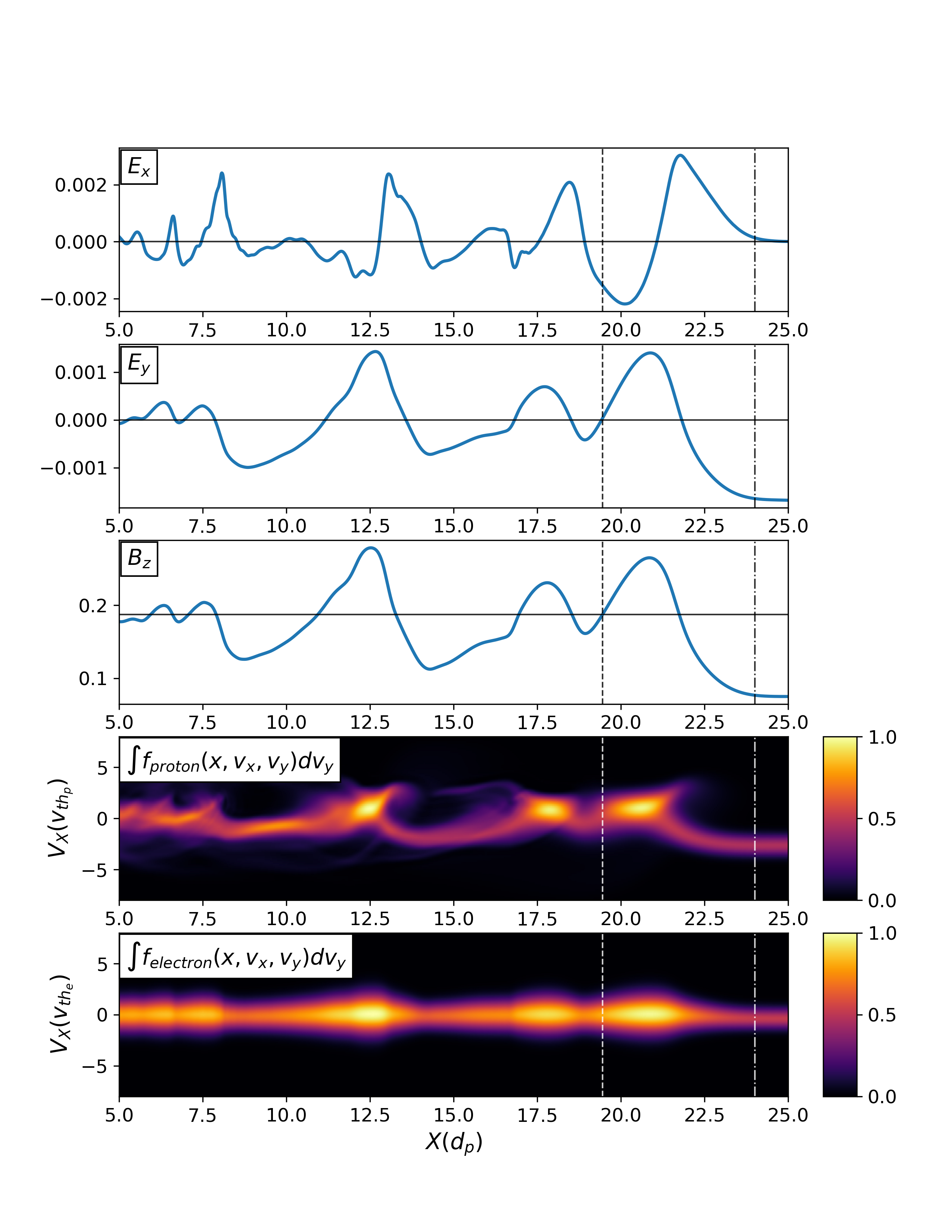}
    \vspace{-0.8in}
    \caption{The x-electric field (top), y-electric field (second from top), z-magnetic field (middle), reduced proton distribution function (second from bottom), and reduced electron distribution function (bottom), both integrated in $v_y$, after the perpendicular shock has formed and propagated through the simulation domain. We have marked an approximate transition from upstream of the shock to the shocked plasma (dashed-dotted lines), and likewise an approximate transition from the shock to the downstream region (dashed lines). To mark the mean values of the oscillating downstream electromagnetic fields, we have used a solid black line to mark the approximate compression of the magnetic field, along with $\mvec{E} = 0$.}
    \label{fig:perpShockFields}
\end{figure}

The particle distribution functions in $x-v_x$ phase space in Figure~\ref{fig:perpShockFields} are illustrative of the dynamics through the shock, showing a clear compression of the electrons and a reflected population of protons.
We can gain further insights into the dynamics of this shock by looking at the distribution function in $v_x-v_y$ at fixed points in configuration space through the shock. 
\begin{figure}[!htb]
    \centering
    \includegraphics[width=\textwidth]{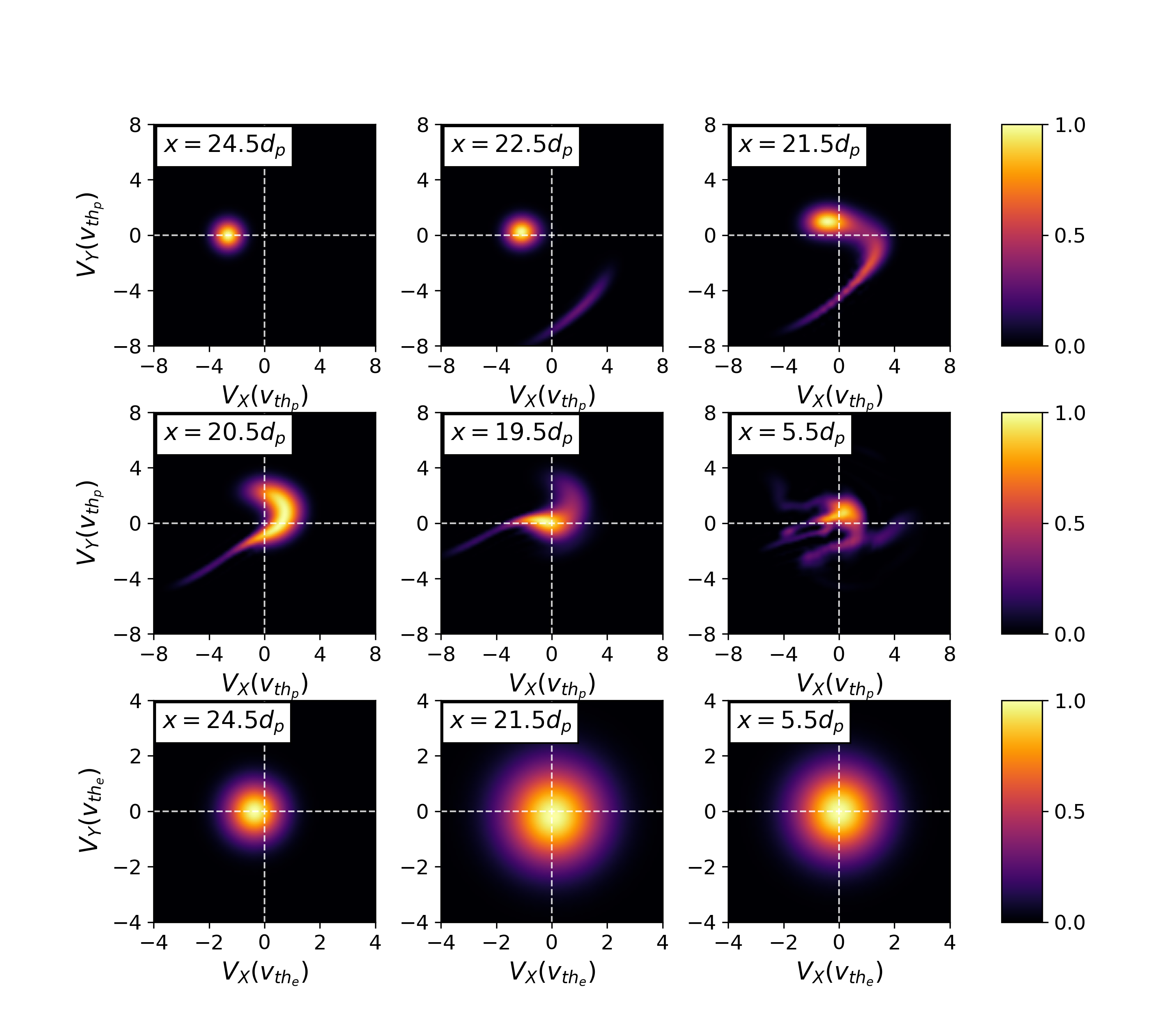}
    \caption{The proton (top two rows) and electron (bottom row) distribution functions plotted through the shock at $t = 11 \Omega_{cp}^{-1}$. As we move from upstream, $x = 24.5 d_p$, through the shock ramp centered at $x = 21.5 d_p$, we can identify the reflected proton population as well as a broadening of the electron distribution function.}
    \label{fig:distThroughShock}
\end{figure}
In Figure~\ref{fig:distThroughShock}, we plot the proton and electron distribution functions in velocity space through the shock, from upstream through the ramp to downstream.
We draw special attention to the proton distribution function in the shock ramp, where we can identify a higher energy tail in $v_x-v_y$.

\subsubsection{Proton Energization in a Perpendicular Shock}

We would like to identify the energization mechanism for this high energy tail of protons, along with the cause of the broadening of the electron distribution.
We thus turn to \eqr{\ref{eq:FPCTimeAvg}}, but instead of performing a time average, we use the instantaneous limit, \eqr{\ref{eq:FPCInstantaneous}}, since we expect the energization through this shock to be impulsive and not require any averaging over an oscillatory component of the energy exchange.
Further, we separate the field-particle correlation into the energization in each of the two velocity directions and transform the fields and velocities to the shock rest-frame,
\begin{align}
    C_{v_x} (x, v_x', v_y', t) & = -q_s \frac{(v_x'-U_{shock})^2}{2} E_x(x,t) \pfrac{f_s(x, v_x' - U_{shock}, v_y', t)}{v_x'}, \label{eq:vxCorrelation} \\
    C_{v_y} (x, v_x', v_y', t) & = -q_s \frac{v_y'^2}{2} [E_y(x,t) - U_{shock} B_z(x, t)] \pfrac{f_s(x, v_x' - U_{shock}, v_y', t)}{v_y'},  \label{eq:vyCorrelation}
\end{align}
where we have performed a Lorentz transformation of the the $y$ electric field,
\begin{align}
    \mvec{E}' = \mvec{E} - \mvec{u} \times \mvec{B}.
\end{align}
Here, primed coordinates denote the simulation frame and unprimed coordinates denote the shock rest-frame, so that, for example, the velocity in the shock rest-frame is
\begin{align}
    v_x = v_x' - U_{shock}.
\end{align}
Note that we are multiplying by the velocity squared in the particular direction of interest, as we expect the orthogonal velocity coordinates, e.g., $v_y$ for the $E_x$ correlation, will integrate to zero as the $x$ electric field can only provide net energization in the $v_x$ direction.

We first investigate the proton energization in the shock foot through the downstream transition, $x = 22.5 d_p \rightarrow 19.5 d_p$ in Figure~\ref{fig:distThroughShock}.
We plot in Figures~\ref{fig:protonFPCFootRamp} and \ref{fig:protonFPCPeakTranision} the field-particle correlation separated into the $v_x$ and $v_y$ components, Eqns. (\ref{eq:vxCorrelation}) and (\ref{eq:vyCorrelation}), as well as the corresponding proton distribution function, through the shock.
We focus in Figure~\ref{fig:protonFPCFootRamp} on the shock foot and ramp, around $x = 22.5 d_p$ and $x = 21.5 d_p$ respectively, at the specified time of $t = 11 \Omega_{cp}^{-1}$.
The blue-red signature identifies the region in phase space in which particles are being accelerated to higher velocities.
Blue regions correspond to a loss of phase space energy density, while red regions correspond to an increase, so a blue-red region means phase space energy density is being transported from the blue to the red region.
\begin{figure}
    \centering
    \vspace{-0.8in}
    \includegraphics[width=\textwidth]{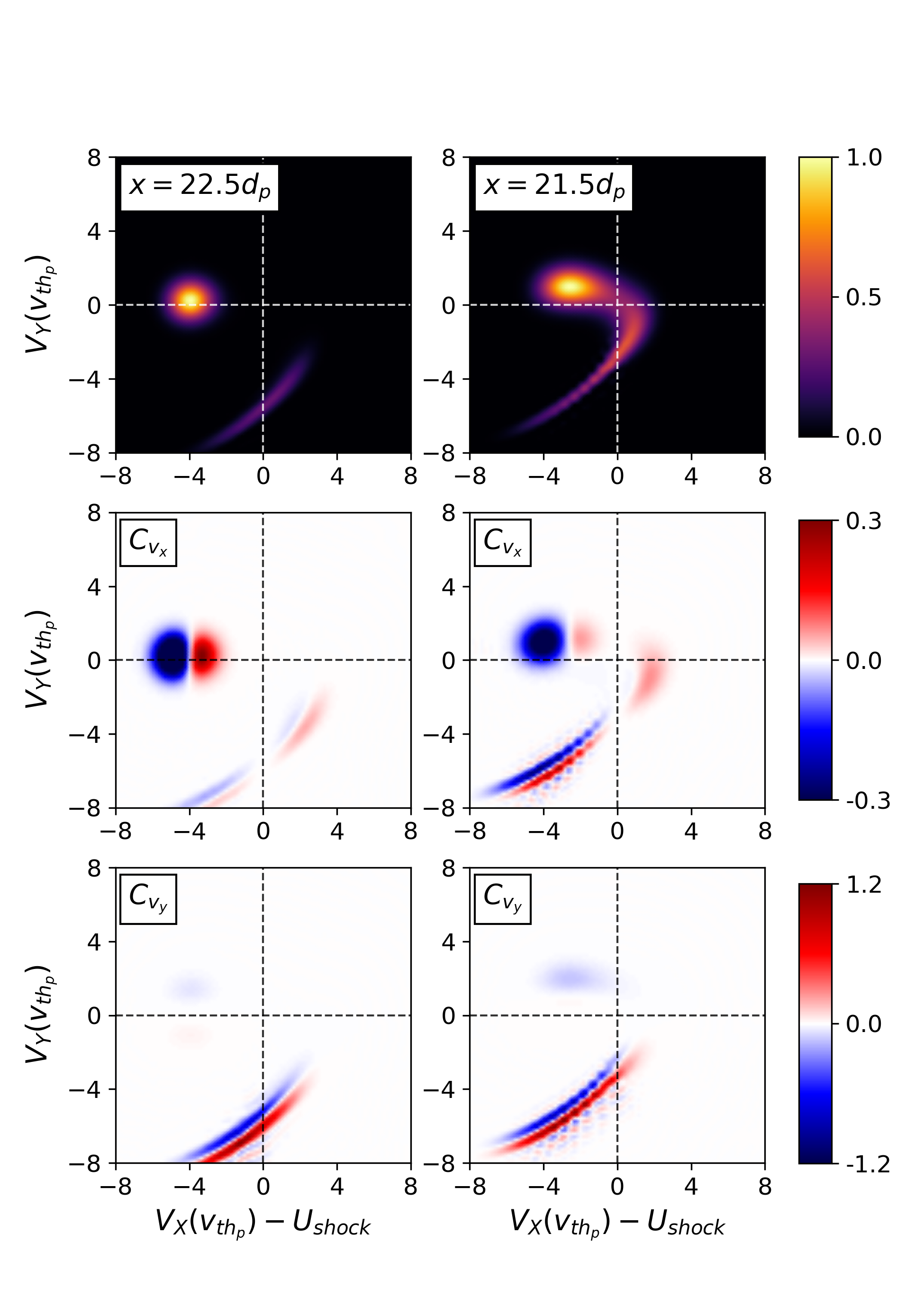}
    \vspace{-0.7in}
    \caption{Proton distribution functions (top row), $C_{v_x}$ field-particle correlations (middle row), and $C_{v_y}$ field-particle correlations (bottom row) in the shock foot and ramp region, where the shock has begun energizing the plasma. We see clear evidence in the proton distribution function of a high energy tail in $v_x - v_y$. Further, we note that the energization of the plasma is localized to this high energy tail. This energization is due to the component of the proton distribution function which returns upstream via its gyromotion, and is thus able to gain energy along the motional electric field, $E_y$, which supports the $\mvec{E} \times \mvec{B}$ drift.}
    \label{fig:protonFPCFootRamp}
\end{figure}
We note that in both the shock foot and ramp, the energization is dominantly in $v_y$ and concentrated in the vicinity of the high energy tail.

In Figure~\ref{fig:protonFPCPeakTranision}, we examine the overshoot and transition to the downstream region of the shock, where all of the secular energization is complete and the remaining energy exchange is governed by a sloshing back and forth between the electromagnetic fields and plasma.
\begin{figure}
    \centering
    \vspace{-0.8in}
    \includegraphics[width=\textwidth]{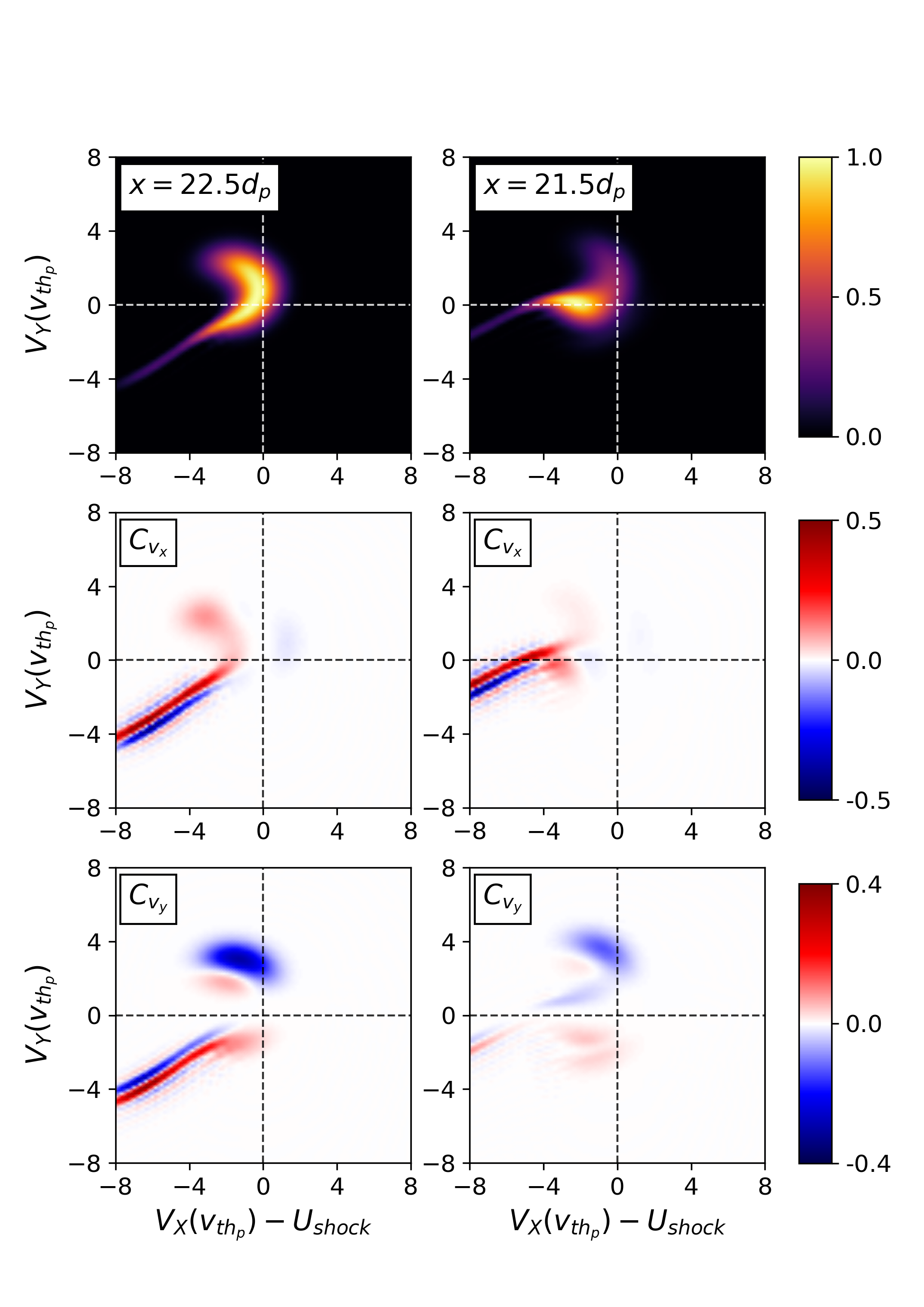}
    \vspace{-0.7in}
    \caption{Proton distribution functions (top row), $C_{v_x}$ field-particle correlations (middle row), and $C_{v_y}$ field-particle correlations (bottom row) in the overshoot and transition regions of the shock, after much of the secular energization has been completed by the shock. We see that the magnitude of the field-particle correlation has decreased in comparison to Figure~\ref{fig:protonFPCFootRamp}, and that the correlation has become more unstructured. By this point in the shock, protons in the plasma are almost downstream, and thus no long experience the gradient in the magnetic field off which the protons reflected, preventing them from gaining further energy along the motional electric field. What remains is oscillatory energy exchange between the plasma and the electromagnetic fields.}
    \label{fig:protonFPCPeakTranision}
\end{figure}
In the overshoot and transition region, we note that the energization has decreased in magnitude in the units of the field-particle correlation and become much more unstructured.
The progression from the region of direct energization to the downstream region where no further secular energization occurs and energy is merely exchanged back and forth between the fields and the particles is then nearly complete.
By this point, the shock is ``done'' in the sense of converting the incoming bulk kinetic energy of the supersonic flows to other forms of energy, though it remains for the downstream region to further partition the energy between the thermal energy of the plasma and electromagnetic energy via other collisionless processes. 

Although we can make some sense of the energy exchange occurring by the relative magnitudes of the field-particle correlation and the overall structure, we would like to understand what particular processes are present in the energy exchange.
We wish to further scrutinize the high energy tail in the shock ramp in the proton distribution function which is prominent in Figure~\ref{fig:protonFPCFootRamp} and a ``hot spot'' for the energization of the protons.
This higher energy tail in the proton distribution function arises from the component of the proton distribution function which returns upstream via its gyromotion, and is thus able to gain energy along the motional electric field, $E_y$, which supports the $\mvec{E} \times \mvec{B}$ drift.

To understand this process of protons returning upstream and gaining energy along the motional electric field, we consider a single-particle picture.
In this single-particle picture, we approximate the shock as a discontinuity in the magnetic field, since the proton gyro-orbit, or Larmor orbit, is as large or larger than the shock scale length, $\rho_p \gtrsim L_{shock} \sim d_p$.
In Figure~\ref{fig:trans_ideal}, we plot (a) the trajectory of a proton in the $(x,y)$ plane and (b) its corresponding trajectory in $(v_x,v_y)$ velocity space in the shock frame,  where the colors indicate the corresponding segments of the trajectory.
The proton velocity is normalized to the proton thermal velocity, $v_{th_p}$.
\begin{figure}
    \centering
    \includegraphics[width=0.7\textwidth]{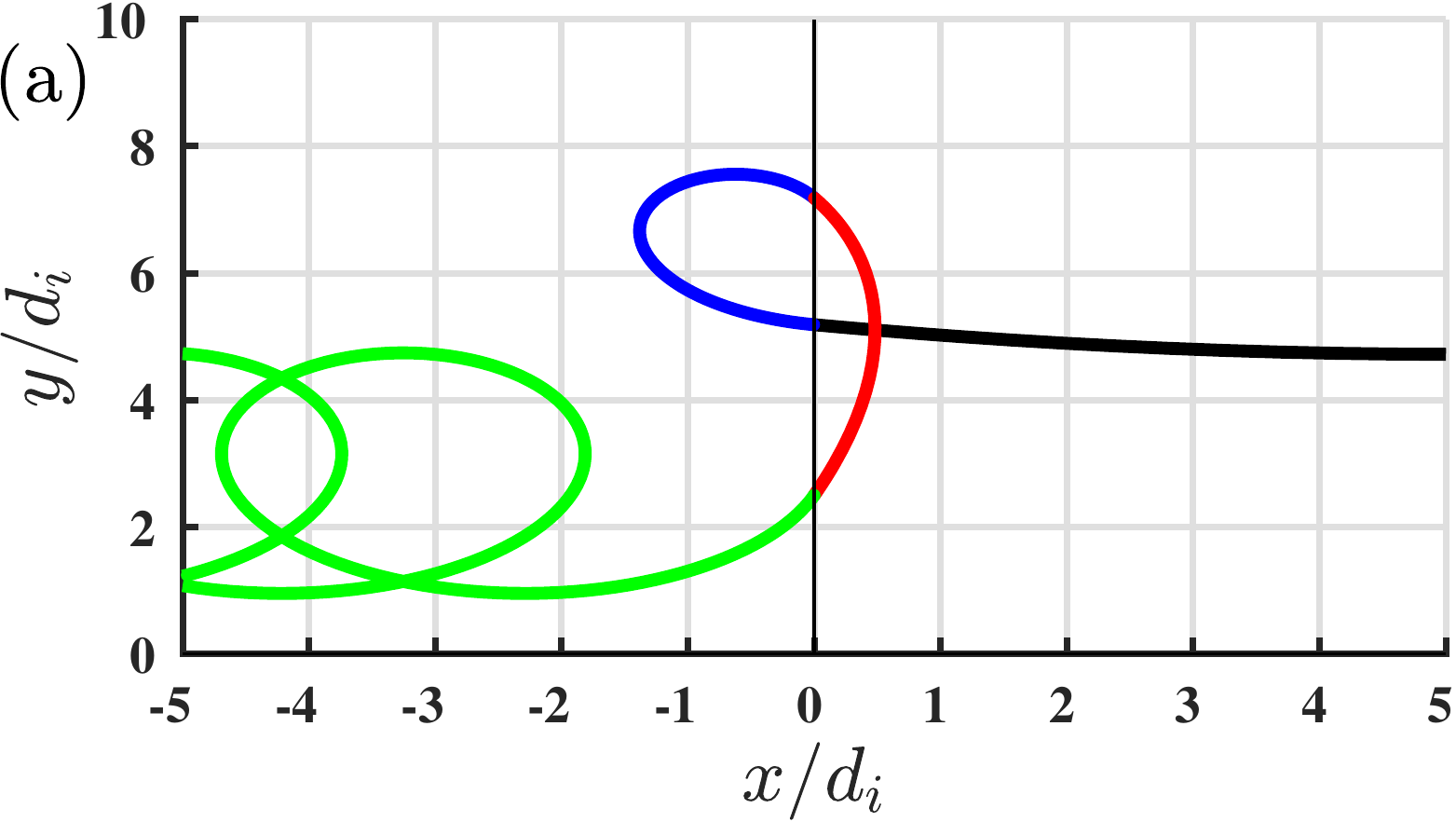}
    \includegraphics[width=0.7\textwidth]{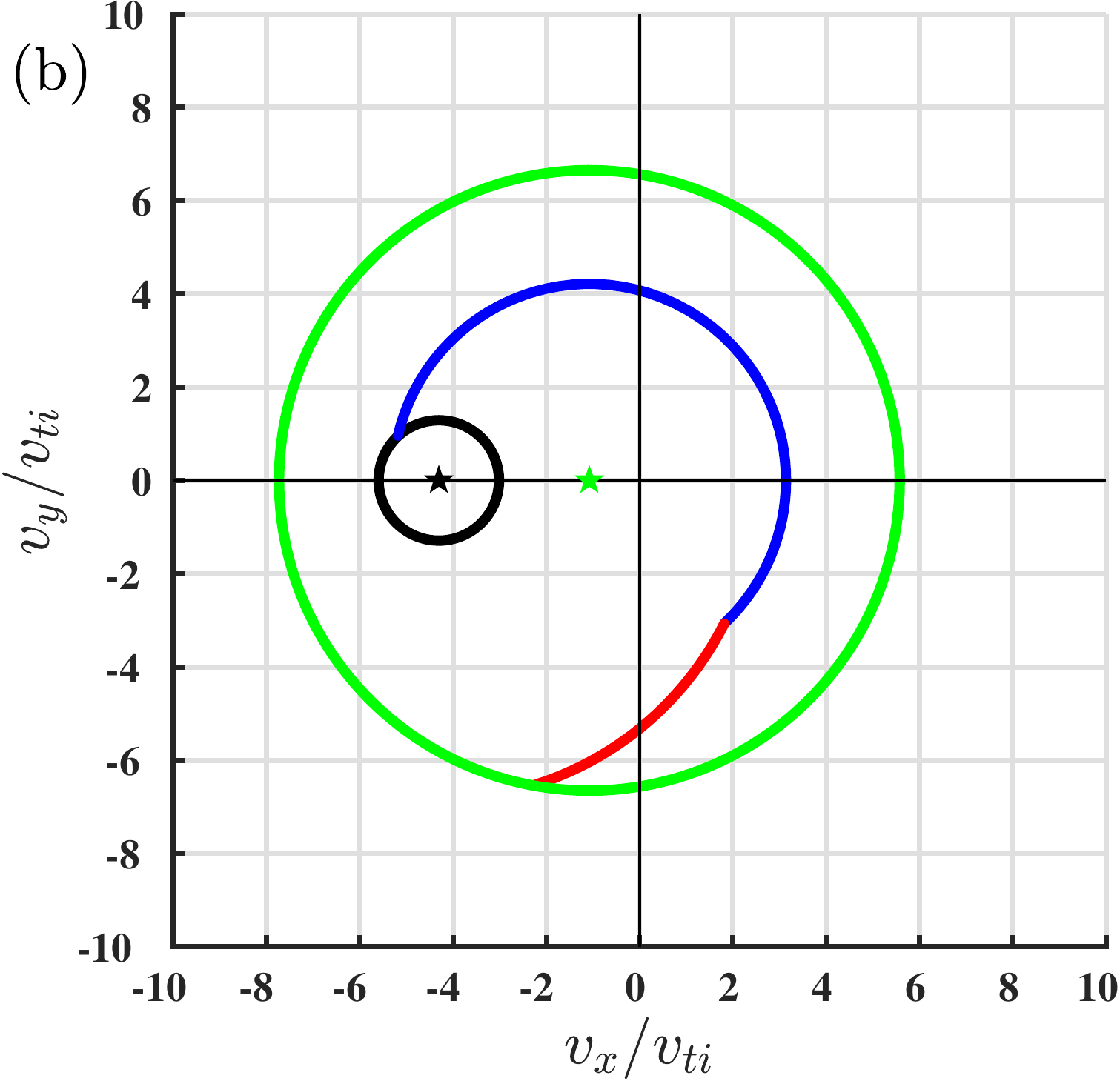}
    \caption{(a) Real space trajectory of a proton as it traverses the shock front and (b) the corresponding velocity space trajectory. Note that the magnetic gradient is assumed to be a discontinuity in this simple picture of the perpendicular shock. The colors of the particle trajectories in real space (a) correspond to the particle's location in phase space (b). Black is upstream, blue corresponds to a proton crossing the magnetic discontinuity before returning upstream, gaining energy along the red trajectory, and then returning downstream and following the green trajectory.}
    \label{fig:trans_ideal}
\end{figure}
In the upstream region, $x>0$ (black), the black circle centered about the upstream $\mvec{E} \times \mvec{B}$ velocity (black star) corresponds to the gyro-orbit of the proton about the upstream inflow velocity in the $(v_x,v_y)$ plane. 

Upon first crossing the magnetic discontinuity to $x<0$, the particle changes to a Larmor gyration in the $(v_x,v_y)$ plane (blue) about the downstream $\mvec{E} \times \mvec{B}$ velocity (green star).
In the larger amplitude downstream perpendicular magnetic field, the radius of the Larmor motion in the $(x,y)$ plane is reduced (blue), and under appropriate conditions, it can lead to the particle crossing back upstream to $x>0$ (red). 

When the proton passes back upstream to $x>0$, it will once again undergo a Larmor orbit in the $(v_x,v_y)$ plane (red) about the upstream $\mvec{E} \times \mvec{B}$ velocity (black star).  
In this segment of the trajectory (red), the proton gains perpendicular energy in the shock frame, given by the distance in velocity space of the proton from the origin of the $(v_x,v_y)$ plane.
This picture is exactly what we observe in phase space in Figure~\ref{fig:protonFPCFootRamp}, and it is no coincidence that the segment of the trajectory in red roughly corresponds to the location in phase space of the high energy tail which is gaining energy in our self-consistent perpendicular shock simulation.

Finally, the particle will eventually cross back into the downstream region to $x<0$ (green), resuming its Larmor orbit in the $(v_x,v_y)$ plane (green) about the downstream $\mvec{E} \times \mvec{B}$ velocity (green star).
Without any additional crossings of the magnetic discontinuity, the proton will simply $\mvec{E} \times \mvec{B}$ drift downstream, periodically gaining and losing energy, in the shock frame, due to work on the proton by the motional electric field $E_y<0$, but the proton will experience no net energization over a complete Larmor orbit.
This energy exchange, without any overall gain in energy, is present in Figure~\ref{fig:protonFPCPeakTranision}, wherein the field-particle correlation becomes more structured and lower amplitude.
In the transition to the downstream region, we only observe the oscillatory exchange of energy between the electromagnetic fields and plasma because the protons are drifting past the magnetic gradient.
Once the protons have drifted past the magnetic gradient, they no longer have the means to return upstream and gain energy off the motional electric field.

Whether a given proton will be ``reflected''\footnote{Note that, unlike many early simple models of collisionless shocks \citep[eg.,][]{Sckopke:1983}, this is not a specular reflection at the magnetic discontinuity at $x=0$, but rather the result of the Lorentz force leading to a return of the proton upstream to $x>0$ due to the increased magnetic field at the shock ramp.
} 
by the increased magnetic field magnitude beyond the discontinuity and return to the upstream region ($x>0$) from downstream ($x<0$) depends on three conditions in this idealized shock model:
(i) the jump in the magnetic field magnitude $B_{d}/B_{u}$; 
(ii) the perpendicular velocity in the frame of the upstream $\mvec{E} \times \mvec{B}$ velocity relative to that inflow velocity, $v_{\perp, u}/U_u$; 
and (iii) the gyrophase $\theta$ of the 
proton's gyro-orbit when it first reaches the magnetic discontinuity at $x=0$. 
For given values of $B_{d}/B_{u}$ and $v_{\perp, u}$ reflection may occur over a range of 
values of gyrophase $\theta$.
For the self-consistent perpendicular shock studied here, a portion of the the distribution of protons have the required gyrophase to reflect off the magnetic gradient and gain energy in Figure~\ref{fig:protonFPCFootRamp}.

The energization mechanism we have identified in Figure~\ref{fig:protonFPCFootRamp} is called \emph{shock-drift acceleration} and has been studied previously in the literature \citep{Paschmann:1982,Sckopke:1983, Anagnostopoulos:1994, Anagnostopoulos:2009, Ball:2001}.
We have identified, for the first time, the phase space signature of this energization process using the field-particle correlation and a continuum method for the solution of the VM-FP system of equations.
Phase space energization signatures, such as those shown in Figure~\ref{fig:protonFPCFootRamp} for shock-drift acceleration, are useful not just for the study of direct numerical simulations, but also as a means of interpreting observational results from in situ spacecraft---see \citet{Chen:2019} and the motivating theoretical studies by \citet{Howes:2017} and \citet{Klein:2017b}.

\subsubsection{Electron Energization in a Perpendicular Shock}

Having identified the proton energization mechanism, we turn now to the electron dynamics in the shock.
We again examine the field-particle correlation in $v_x$ and $v_y$ through the shock foot to the transition to the downstream in Figures~\ref{fig:electronFPCFootRamp} and \ref{fig:electronFPCPeakTranision}.
\begin{figure}
    \centering
    \vspace{-1.0in}
    \includegraphics[width=\textwidth]{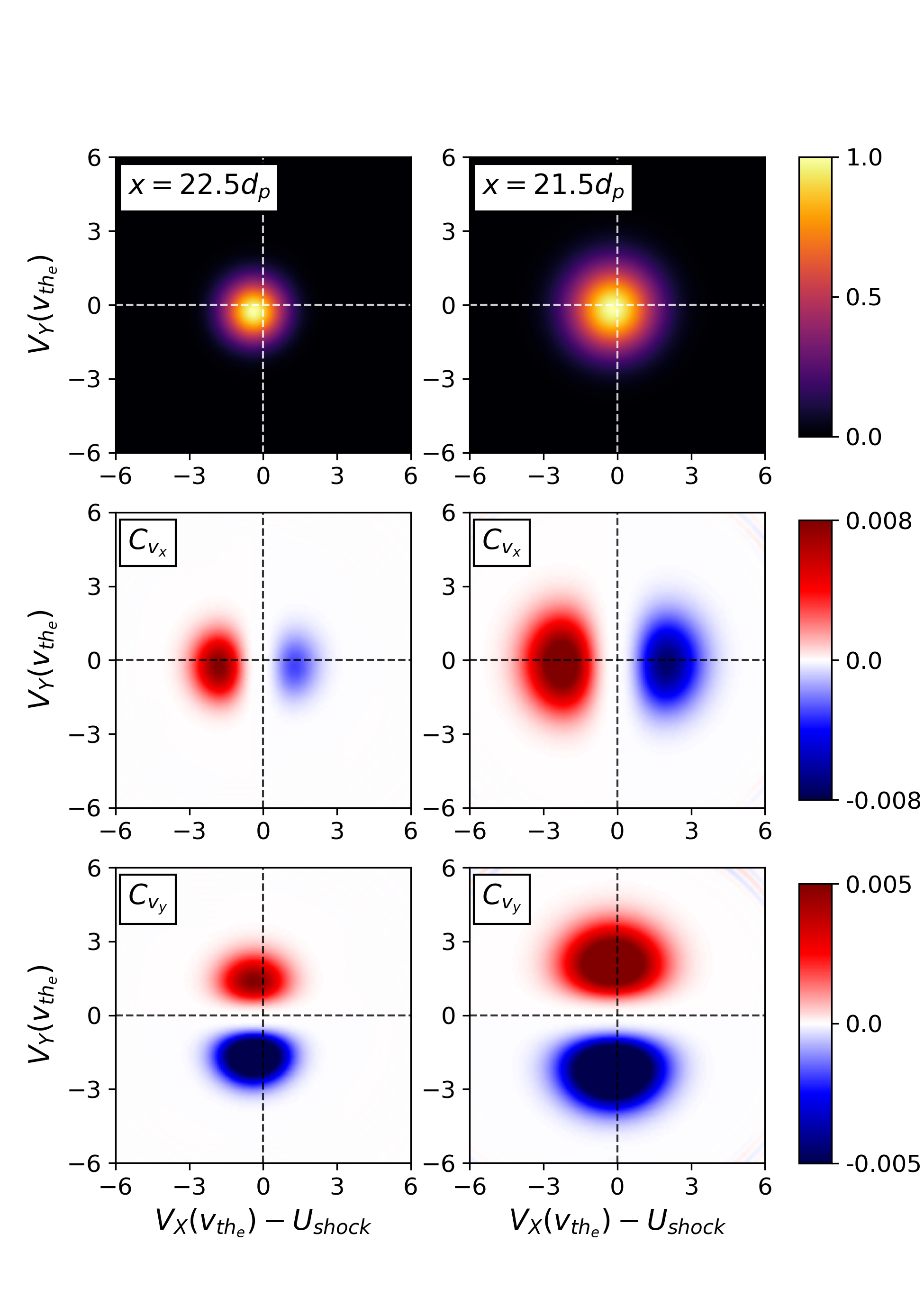}
    \vspace{-0.7in}
    \caption{Electron distribution functions (top row), $C_{v_x}$ field-particle correlations (middle row), and $C_{v_y}$ field-particle correlations (bottom row) in the shock foot and ramp region. The field-particle correlation has a slight asymmetry that corresponds to an energy gain to the $x$ field-particle correlation and an energy loss due to the $y$ field-particle correlation. The gain in energy due to $E_x$ exceeds the loss in energy due to $E_y$, corresponding to a net energization of the electrons.}
    \label{fig:electronFPCFootRamp}
\end{figure}
At first glance, the phase space signature appears to roughly cancel on each side of $v_{x,y} = 0$ for all of the correlations, each correlation has a slight asymmetry which leads to either net energization or net de-energization.
\begin{figure}
    \centering
    \vspace{-1.0in}
    \includegraphics[width=\textwidth]{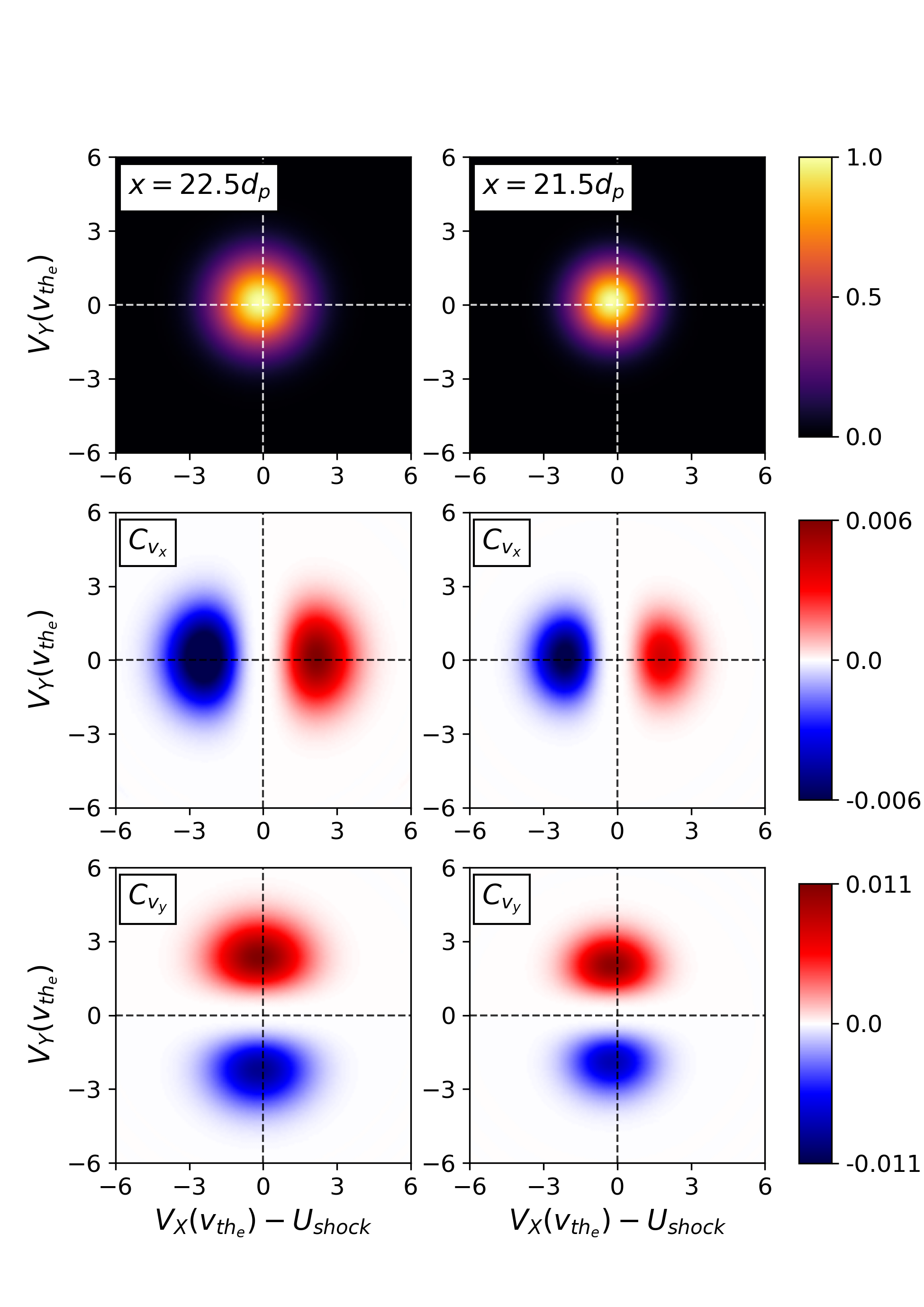}
    \vspace{-0.7in}
    \caption{Electron distribution functions (top row), $C_{v_x}$ field-particle correlations (middle row), and $C_{v_y}$ field-particle correlations (bottom row) in the overshoot and transition regions of the shock. Here, we observe the opposite behavior to Figure~\ref{fig:electronFPCFootRamp}, where now the asymmetry in the field particle correlation is such that the particles gain energy due to $E_y$ and lose energy due to $E_x$. The gain in energy due to $E_y$ still exceeds the loss in energy due to $E_x$, so the electrons continue to gain energy in this region of the shock. This particular energization signature in the $y$ field particle correlation arises from alignment of the $\gx B$ drift and the motional electric field, $E_y$, and relies on conservation of the electron's magnetic moment, the first adiabatic invariant. Because of the relationship between this energization mechanism and the electron's first adiabatic invariant, we call this adiabatic heating.}
    \label{fig:electronFPCPeakTranision}
\end{figure}
In the shock foot and ramp, Figure~\ref{fig:electronFPCFootRamp}, these slight asymmetries correspond to a gain of energy due to $E_x$, and a loss of energy due to $E_y$, and we note by their magnitudes that more energy is gained due to $E_x$ than lost due to $E_y$.
Thus, the electrons overall gain energy.
We see the opposite trend in the overshoot and transition to the downstream, Figure~\ref{fig:electronFPCPeakTranision}, wherein the electrons gain energy due to $E_y$ and lose energy due to $E_x$.
Again, the gain in energy due to $E_y$ is larger than the loss of energy due to $E_x$, so the electrons overall continue to gain energy.

The energy gain and loss due to $E_x$ can be thought of simply as electrons responding to an electrostatic potential, $E_x = -\partial \phi/ \partial x$, as $E_x$ is the electrostatic component of the electromagnetic fields.
We are especially interested, though, in the energy gain (and loss) due to $E_y$, the electromagnetic component of the electric field, since this component of the field supports the compression of the magnetic field.
To understand the energy exchange between the electrons and $E_y$, we again turn to a single-particle picture for intuition.

Because the electron gyro-orbit is much smaller than the length scale of the collsionless shock, $\rho_e \ll L_{shock} \sim d_p$, we approximate the shock in an idealized model as a linear ramp in the magnetic field.
In Figure~\ref{fig:elec_profile}, we plot in the top panel the profile of the perpendicular magnetic field $B_z(x)$ (blue) and the motional electric field $E_y(x)$ (red) along the shock normal direction, and in the middle panel the trajectory of an electron in the $(x,y)$ plane as it flows through the shock ramp, $0 \le x/d_p \le 2$.
\begin{figure}
    \centering
    \includegraphics[width=\textwidth]{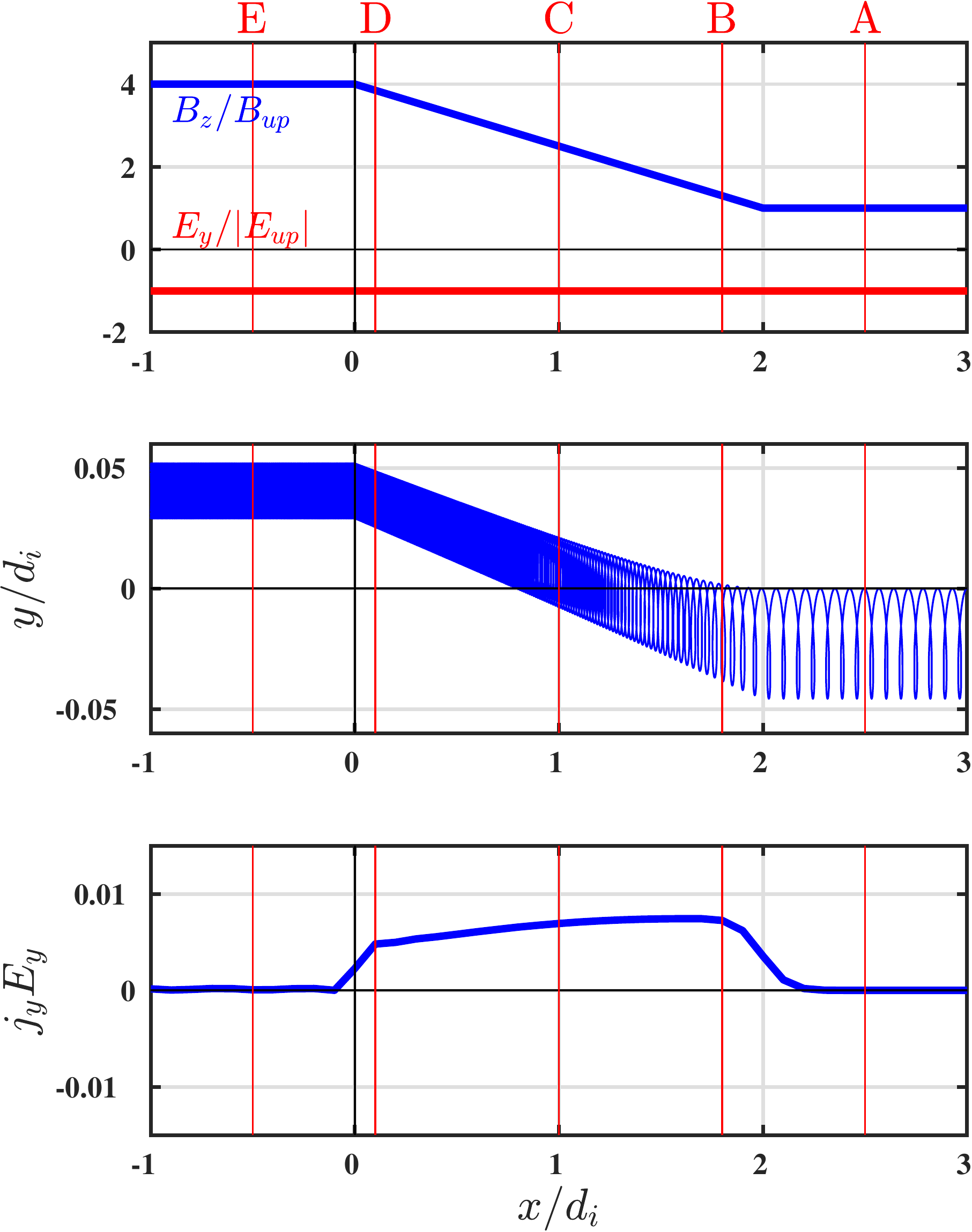}
    \caption{(Top panel) Profiles along the shock normal direction of the perpendicular magnetic field $B_z$ (blue) and the motional electric field $E_y$ (red), (Middle panel) trajectory of an electron in the $(x,y)$ plane, and 
    (Bottom panel) the rate of work done by the electric field on the electron $j_y E_y$.}
    \label{fig:elec_profile}
\end{figure}
The trajectory plot shows clearly the $\gx B$ drift in the $+y$ direction. 
A salient difference between the idealized single particle motion for electrons and protons is that the electron thermal velocity is larger than the inflow velocity, so electrons can move in the $+x$ direction, even upstream of the shock.
This condition is also satisfied for the shock parameters in our self-consistent perpendicular shock simulation, $U_{shock} \sim 2 v_A \ll v_{th_e}$.

Although the electron constantly gains and loses energy as part of its $\mvec{E} \times \mvec{B}$ drift due to the motional electric field $E_y$, the net effect on the particle energy over a Larmor orbit is zero, because the drift in the $-x$ direction is perpendicular to the electric field, $U_{\mvec{E} \times \mvec{B}} \cdot E_y = 0$.
But, in the region where the perpendicular magnetic field changes magnitude, $0 \le x/d_p \le 2$, a $\gx B$ drift arises in the $+y$ direction, which leads to a net energization of the electrons by $E_y$.
This alignment of the motional electric field, $E_y$, with a drift, in this case the $\gx B$ drift, allows the electrons to gain energy, as shown in the bottom panel of Figure~\ref{fig:elec_profile}.

As an aside, the rate of energization of the electrons by the $\gx B$ drift in the motional electric field is precisely the rate required to satisfy the conservation of the first adiabatic invariant of the electron, the electron magnetic moment, $\mu = m_e v_\perp^2/2 B_z$. 
This connection can be shown by calculating the net rate of work done by $E_y$ due to the $\gx B$ drift, which contributes to the perpendicular kinetic energy of the electrons,
\begin{align}
    \frac{d m_e v_\perp^2/2}{dt}=q_e u_{\gx B} E_y,
\label{eq:elecmu}
\end{align}
where the magnitude of the $\gx B$ drift in the $+y$ direction is given by 
\begin{align}
u_{\nabla B}= \frac{m_e v_\perp^2}{2 q_e B_z} 
\left(\frac{1}{B_z} \frac{\partial B_z}{\partial x} \right).
\end{align}
For the static fields in this idealized model, the total time derivative is determined by the $\mvec{E} \times \mvec{B}$ velocity, 
\begin{align}
    \frac{d}{dt} = \pfrac{}{t} + u_x \pfrac{}{x} = u_{\mvec{E} \times \mvec{B}} \pfrac{}{x}.
\end{align}
Substituting $u_{E \times B}=E_y/B_z$, we can manipulate \eqref{eq:elecmu} to obtain
\begin{align}
\frac{\partial }{\partial x} \frac{m_e v_\perp^2}{2B_z}= \frac{\partial \mu }{\partial x}=0,
\label{eq:mucons}
\end{align}
proving that the electron's first adiabatic invariant $\mu$ is conserved.
Because this energization process relies on the electron's first adiabatic invariant being conserved, we call this energization \emph{adiabatic heating}.

This simple model for the electron energization presumes that the only electric field participating is $E_y$, but we can see from Figure~\ref{fig:perpShockFields} that this is not the case.
Even if the electrostatic field is roughly bi-modal across the shock so that much of the energy exchange between the electrostatic field and the electrons is reversed when the electrons cross downstream, the presence of this electrostatic field still complicates the picture.
The electrostatic electric field leads to an $\mvec{E} \times \mvec{B}$ flow in the $-y$ direction which counters the $\gx B$ drift in the $+y$ direction.
Still, for at least a component of the energization through the shock, especially in the overshoot and transition region in Figure~\ref{fig:electronFPCPeakTranision}, we see a signature in the field-particle correlation of energy gain in $y$, which is characteristic of this alignment between the $\gx B$ drift and the motional electic field, $E_y$.
Because of the finite $\gx B$ drift, there are more electrons with velocities aligned with the motional electric field, $E_y$, leading to the asymmetry in the field-particle correlation in Figure~\ref{fig:electronFPCPeakTranision}, and thus a net gain of energy for the electrons.

We conclude this study of a self-consistent perpendicular shock with our DG VM-FP solver noting that, with the combination of diagnostics such as the field-particle correlation and our continuum representation of the particle distribution function, we can directly diagnose the energy exchange of kinetic plasma processes in phase space.
We have shown, for the first time, the phase space signature of shock-drift acceleration of the protons and adiabatic heating of the electrons in a collisionless shock.
Although these energization mechanisms have been studied previously, especially using the same single particle, and more generally Lagrangian, picture we used to model the particulars of the energization processes, the Eulerian phase space picture presented here is of considerable value.
Especially when interpreting spacecraft observations of particle distribution functions, which must usually be done in the Eulerian frame to obtain good enough sampling statistics, having a means of interpreting the specific energization mechanisms opens new possibilities for diagnosing the details of the phase space dynamics.

There is more that can be learned from this perpendicular shock simulation.
For example, we have only noted and not examined the competition between the electrostatic and electromagnetic electric fields in energizing electrons.
Given the requirements for adiabatic heating, $\rho_e \ll L_{shock}$, we might expect more realistic mass ratios to yield different results for this competition as well.

Finally, the distribution function structure we resolve in the downstream region, where the plasma and electromagnetic fields continually exchange energy, is a rich problem for understanding the ultimate ``mixing'' of the plasma.
Collisionless shocks are often discussed interchangeably with irreversible heating and entropy increase, though we note that the energy exchange happens on length scales much smaller than the collisional mean-free path.
Thus, despite the total energy exchange being ``done'' once the shock has passed through the plasma, we expect additional kinetic mechanisms are at play which transfer energy to smaller velocity space scales, where collisions ultimately dissipate this energy.
Given the structure we can represent in phase space with the continuum VM-FP solver presented in this thesis, we expect the ultimate diagnosis of this collisionless mixing is ideally studied by the approach taken here, as the details of the collisionless mixing may be obscured in particle-based method with the artificial collisionality introduced by finite sized particles \citep{birdsallbook}.%Add reference to particle noise appendix

The focus of this section has been on how we can use the high fidelity representation of the distribution function to more carefully analyze plasma processes in phase space.
Because diagnostics such as the field-particle correlation, \eqr{\ref{eq:FPCTimeAvg}} and \eqr{\ref{eq:FPCInstantaneous}}, involve gradients of the velocity distribution function, traditional particle-based methods may have difficulty leveraging these diagnostic to examine the precise processes present.
Counting noise can add sizable errors to the computation of these velocity space gradients, and significant spatial averaging to reduce the noise in post-processing may mix energization processes occurring in different regions of configuration space, thus making it more challenging to determine the specifics of the energy exchange between the plasma and the electromagnetic fields.
We now turn to another application which reveals a different utility of the continuum kinetic discretization: the phase space dynamics themselves being sensitive to phase space resolution.

\section{The Phase Space Dynamics of Filamentation-Type Instabilities}

We consider here an extension of the benchmark studied in Section~\ref{sec:hybridTSW}, the phase space dynamics of filamentation-type instabilities.
Recall in Figure~\ref{fig:linearTheoryTSW} for the parameters chosen for the benchmark that the oblique, $45^{\circ}$, mode had a growth rate within 20-30 percent of the two-stream.
This may not be similar enough to affect the dynamics under more general perturbations of all modes in the system for this parameter regime, $v_{th_e}/u_y = 1/3$, $v_{th_e} = 0.1c$.
But the evolution the competition of all the modes present, as would occur in the astrophysical systems where these modes are present, is likely to have an effect on the dynamics.
For example, we can ask whether the full spectrum of modes vying for dominance under more general conditions affects the efficiency of magnetic field growth from the unstable beams of plasma, a question of vital importance for the origins of the cosmological magnetic field \citep{Schlickeiser:2003, Lazar:2009}.

If we survey the parameter space more extensively, we find that these oblique modes can have comparable growth rates to the two-stream instability as the ratio of the thermal velocity to the drift speed is reduced and the beams are made colder---see Figure~\ref{fig:fullLinearTSW}.
\begin{figure}
    \centering
    \vspace{-0.2in} 
    \includegraphics[width=\textwidth]{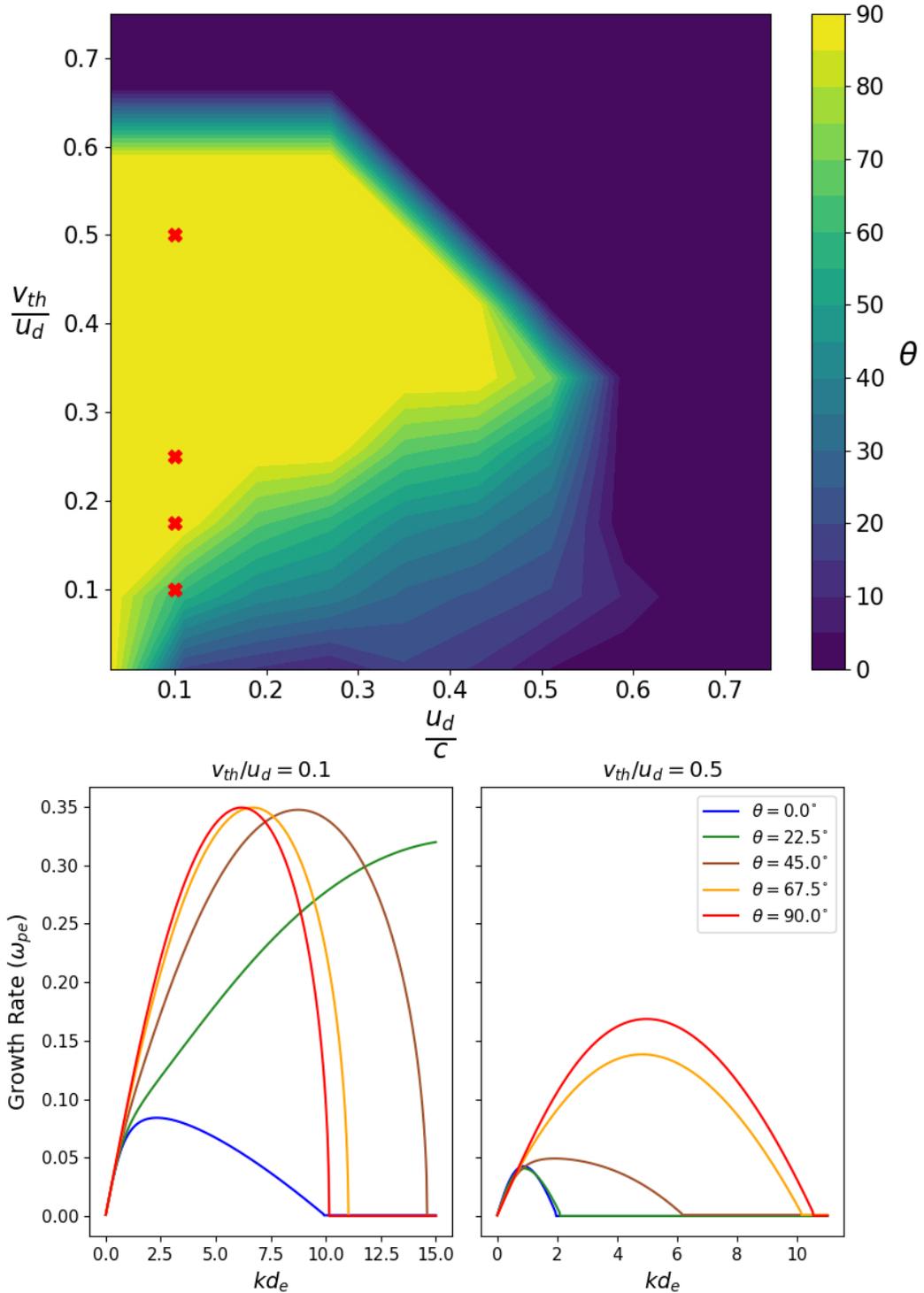}
    \vspace{-0.4in}    
    \caption{Contour plot of the angle of the fastest-growing mode in the parameter space of $v_{th_e}/u_d$ and $u_d/c$ (top panel). $\theta=90^{\circ}$ corresponds to a pure two-stream mode, and $\theta=0^{\circ}$ corresponds to a pure filamentation mode. Red crosses correspond to the four simulations presented. Growth rates versus wavenumber (bottom panels) of different modes for the hot (right panel) and cold (left panel) cases for $u_d=0.1c$. We can see in the hot case, $v_{th_e}/u_d = 0.5$, that the two-stream instability is the fastest growing mode, while when we make the beams colder, $v_{th_e}/u_d = 0.1$, the oblique modes for a variety of angles have comparable growth rates to the pure two-stream instability.}
    \label{fig:fullLinearTSW}
\end{figure}
Although some parameters, e.g., $v_{th_e}/u_d = 0.5$, clearly show that the two-stream instability is the fastest growing mode and there is not much competition for the fastest growing mode in the system, we can expect that the competition could be quite significant as the beams become colder and multiple modes spanning a wide range of angles saturate at similar times.

To study the competition between all of these modes, two-stream, oblique, and filamentation, we set-up a similar phase space domain to Section~\ref{sec:hybridTSW}, two configuration space and two velocity space dimensions (2X2V) with a drifting electron-proton plasma.
The protons are taken to be a stationary, charge-neutralizing background as before, and the electrons are initialized as two drifting Maxwellians, \eqr{\ref{eq:TSWElcInit}}.
We repeat this initial electron distribution here for clarity,
\begin{align*}
    f_e (x, y, v_x, v_y) = \frac{m_e n_0 }{2 \pi T_e} \exp & \left (- m_e \frac{(v_x)^2 + (v_y - u_d)^2}{2 T_e} \right ) \notag \\
    & + \frac{m_e n_0 }{2 \pi T_e} \exp \left (- m_e \frac{(v_x)^2 + (v_y + u_d)^2}{2 T_e} \right ).    
\end{align*}
The electromagnetic fields are initialized as a bath of fluctuations in the electric and magnetic fields in the two configuration space dimensions, i.e.,
\begin{align}
    B_z(t=0)=\sum_{n_x,n_y=0}^{16,16}\tilde B_{n_x,n_y}\sin \left (\frac{2\pi n_x x}{L_x}+\frac{2\pi n_y y}{L_y}+\tilde \phi_{n_x,n_y} \right ), \label{eq:BTSWInit}
\end{align}
where $\tilde B_{n_x, n_y}$ and $\tilde \phi_{n_x,n_y}$ are random amplitudes and phases respectively.
The electric fields, $E_x(t=0)$ and $E_y(t=0)$, are initialized similarly to \eqr{\ref{eq:BTSWInit}}, and all three dynamically important electromagnetic fields in this two dimensional geometry are given equal average energy densities, $\langle  \epsilon_0 E_x^2/2\rangle=\langle \epsilon_0 E_y^2/2\rangle=\langle B_z^2/2\mu_0\rangle\approx 10^{-7}E_K$, where $E_K$ is the initial total electron energy.

We focus on four particular simulations, whose parameters are indicated by red crosses in Figure~\ref{fig:fullLinearTSW}.
The drift velocity is fixed at $u_d=0.1c$, but we vary the temperature of the beams by choosing $v_{th_e}/u_d\in \{0.1,0.175,0.25,0.5\}$. 
The box sizes, respectively, are $L_x/d_e\in \{2.7,3.8,4.4,7.7\}$ and $L_y/d_e\in \{3.1,4.0,4.8,6.3\}$, where $d_e$ is the electron inertial length, $d_e = c/\omega_{pe}$.
Box sizes $L_x=2\pi/k_{0^{\circ}}^{max}$ and $L_y=2\pi m/k_{90^{\circ}}^{max}$ are chosen to be roughly equal, $L_x\approx L_y$, while fitting a single fastest-growing wavelength of the filamentation instability and an integer number, $m\approx k_{90^{\circ}}^{max}/k_{0^{\circ}}^{max}$, of two-stream modes.
The configuration space boundary conditions are periodic, and the velocity space boundary conditions are zero-flux.
The velocity space extents are varied for each simulation to contain the phase space evolution of the instabilities in the nonlinear regime, $[-3 u_d, u_d]^2$ to $[-5 u_d, 5 u_d]^2$.
Likewise, we vary the resolution in configuration and velocity space to obtain convergence, from $32^2 \times 32^2$ to $64^2 \times 96^2$.
All simulations use piecewise quadratic Serendipity polynomials.

We plot in Figure~\ref{fig:growth_sat} the evolution of the magnetic field energy, $\epsilon_B$, and electric field energy, $\epsilon_E$, normalized to the initial total energy of the electrons.
\begin{figure}
    \centering
    \includegraphics[width=\textwidth]{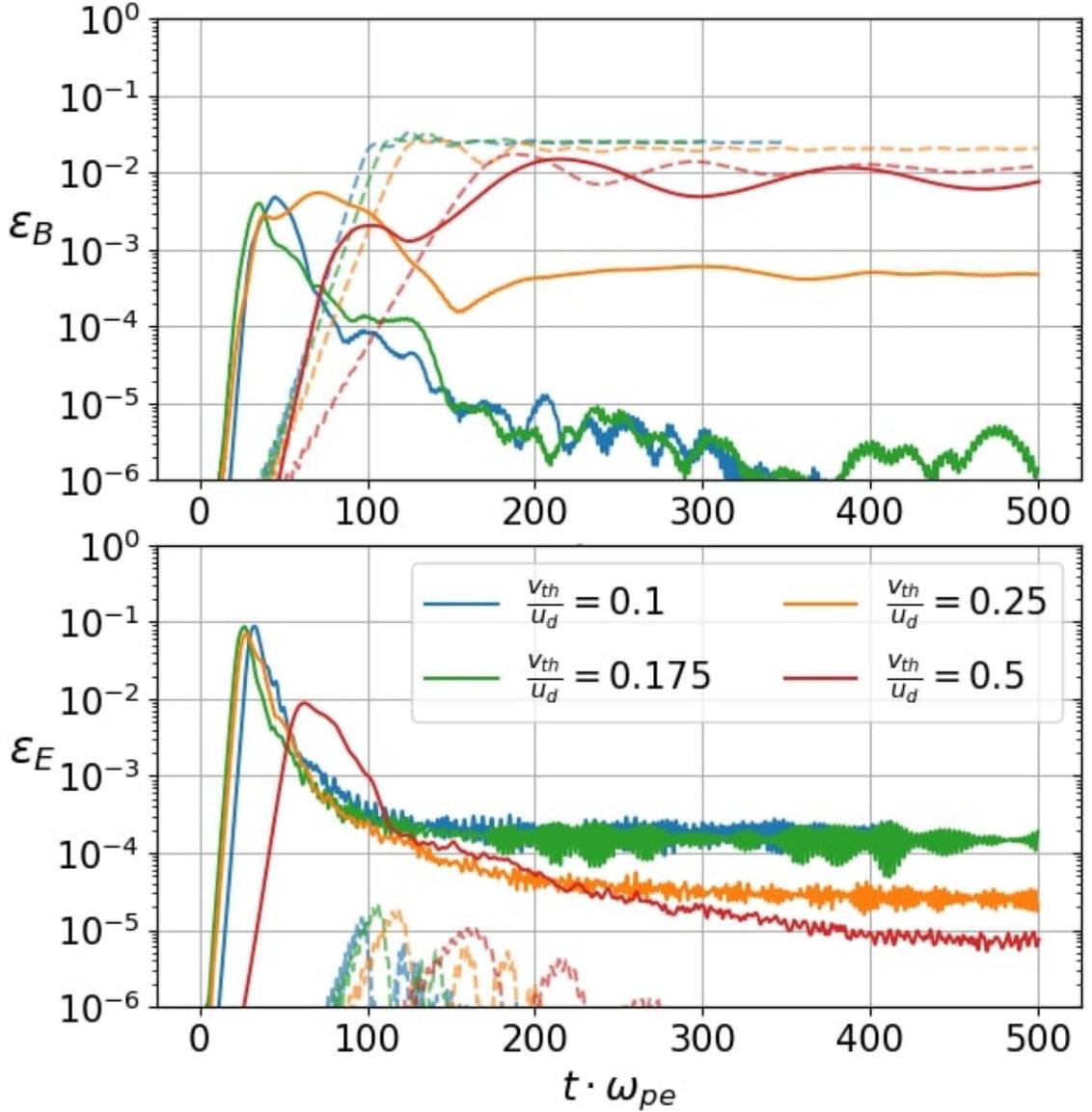}
    \vspace{-1.0in}  
    \caption{Growth and saturation of magnetic field (top panel) and electric field (bottom panel) energies normalized by the initial total electron energy for beams with drift velocity $u_d=0.1c$ at different temperatures. Solid lines correspond to 2X2V simulations with initial random modes which drive two-stream, oblique and filamentation modes, while dashed lines correspond to 1X2V simulations which only support pure filamentation modes. We can see clearly the effect of the higher dimensionality and competition between the different modes, since for all 1X2V simulations, regardless of the ratio of $v_{th_e}/u_d$, a magnetic field grows and saturates, whereas the growth of a magnetic field is sensitive to this ratio of $v_{th_e}/u_d$ when the two-stream, oblique, and filamentation modes are allowed to compete with each other in two configuration space dimensions.}
    \label{fig:growth_sat}
\end{figure}
We compare in Figure~\ref{fig:growth_sat} the results of the four simulations in 2X2V (solid lines), where two-stream, oblique, and filamentation modes are allowed to grow and compete with each other, with the results of similar 1X2V simulations (dashed lines) varying $v_{th_e}/u_d$, but which only support the filamentation instability.
We see that, while the 1X2V simulations robustly grow a magnetic field from the free energy of the unstable beams of plasma and the formation of current filaments from this free energy, irrespective of this ratio of $v_{th_e}/u_d$ and the temperature of the beams, the situation is quite different in two configuration space dimensions, wherein the various modes are permitted to compete with each other.

In 2X2V, the initial growth phase is quite different from the corresponding 1X2V simulations.
In 2X2V, we see the growth of both magnetic and electric fluctuations due to the combination of unstable oblique and two-stream modes.
The oblique modes in particular are what lead to the growth of both electric and magnetic field fluctuations, as the two-stream instability would only grow an electric field, and the filamentation instability is much more slowly growing than the other instabilities.
Following saturation, potential wells formed by the saturation of two-stream and oblique modes, the tilted current filaments of oblique modes, and the vertical, i.e., uniform in $y$, current filaments associated with the potentially still-growing filamentation instability all nonlinearly interact and vie for dominance.

To understand this interplay between the formation of current filaments and potential wells by the various instabilities, we examine the electromagnetic fields and particle distribution functions of the two limiting cases, $v_{th_e}/u_d = 0.5$, the hot case, and $v_{th_e}/u_d = 0.1$, the cold case.
\begin{figure}
    \centering
    \vspace{-0.5in}  
    \includegraphics[width=0.8\textwidth]{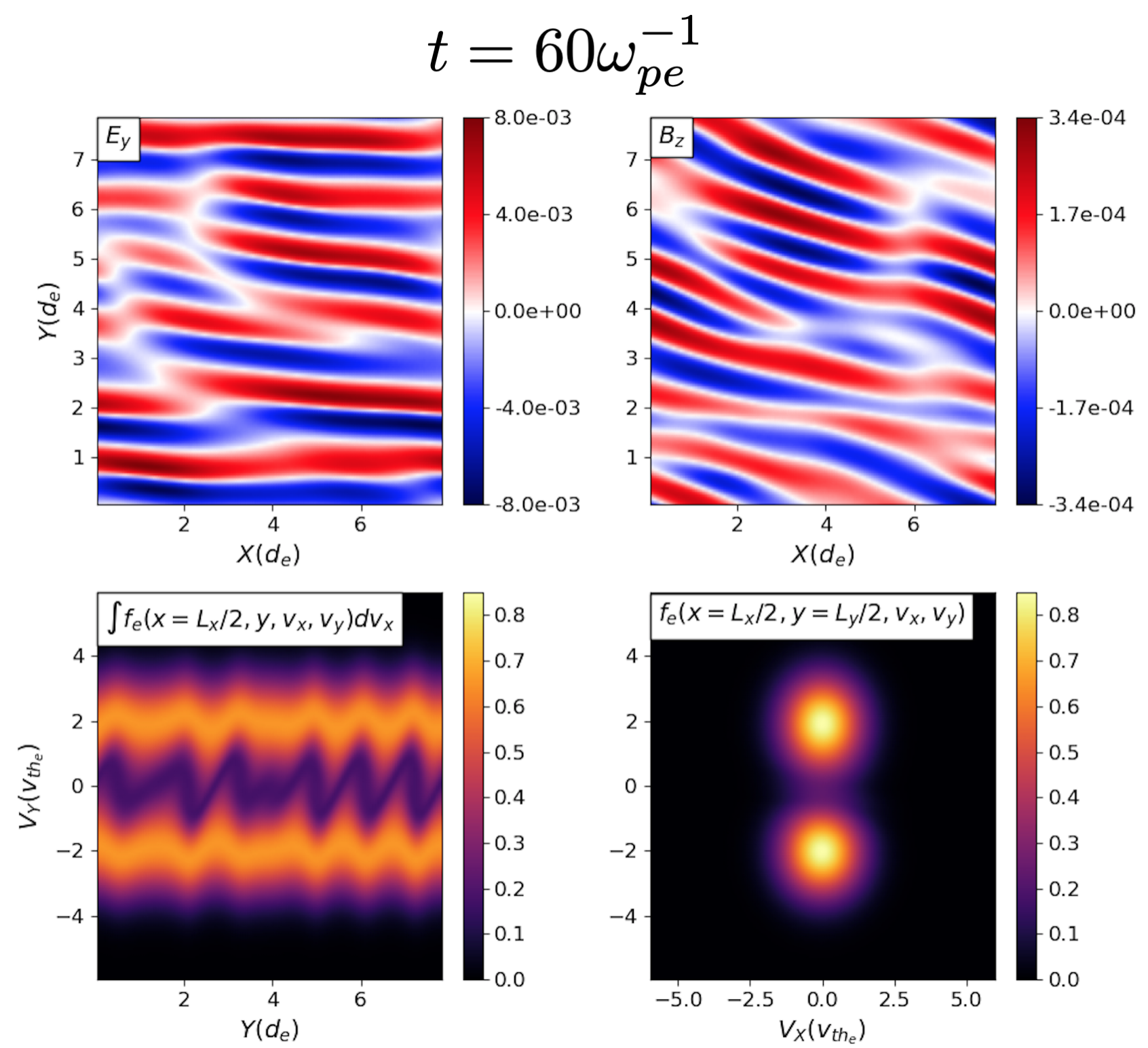}
    \includegraphics[width=0.8\textwidth]{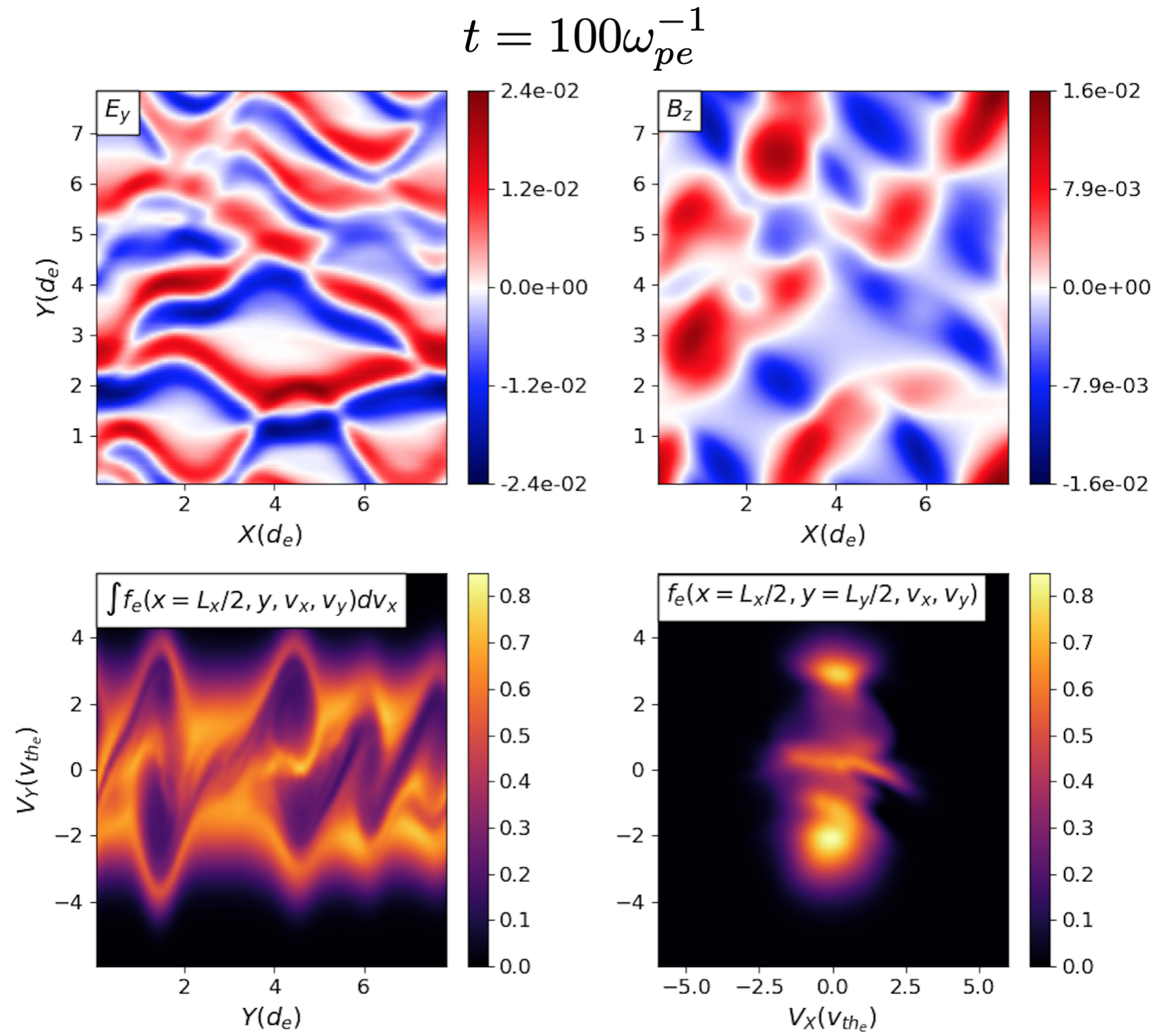}    
    \caption{$t = 60 \omega_{pe}^{-1}$ and $t = 100 \omega_{pe}^{-1}$ snapshots of the evolution of the hot case. We see the initial development of the two-stream instability and roll-up of the distribution function, before the electron tubes formed by the two-stream instability are destroyed by the more slowly growing filamentation instability.}
    \label{fig:hotTSWEarly}
\end{figure}
\begin{figure}
    \centering
    \vspace{-0.5in}  
    \includegraphics[width=0.8\textwidth]{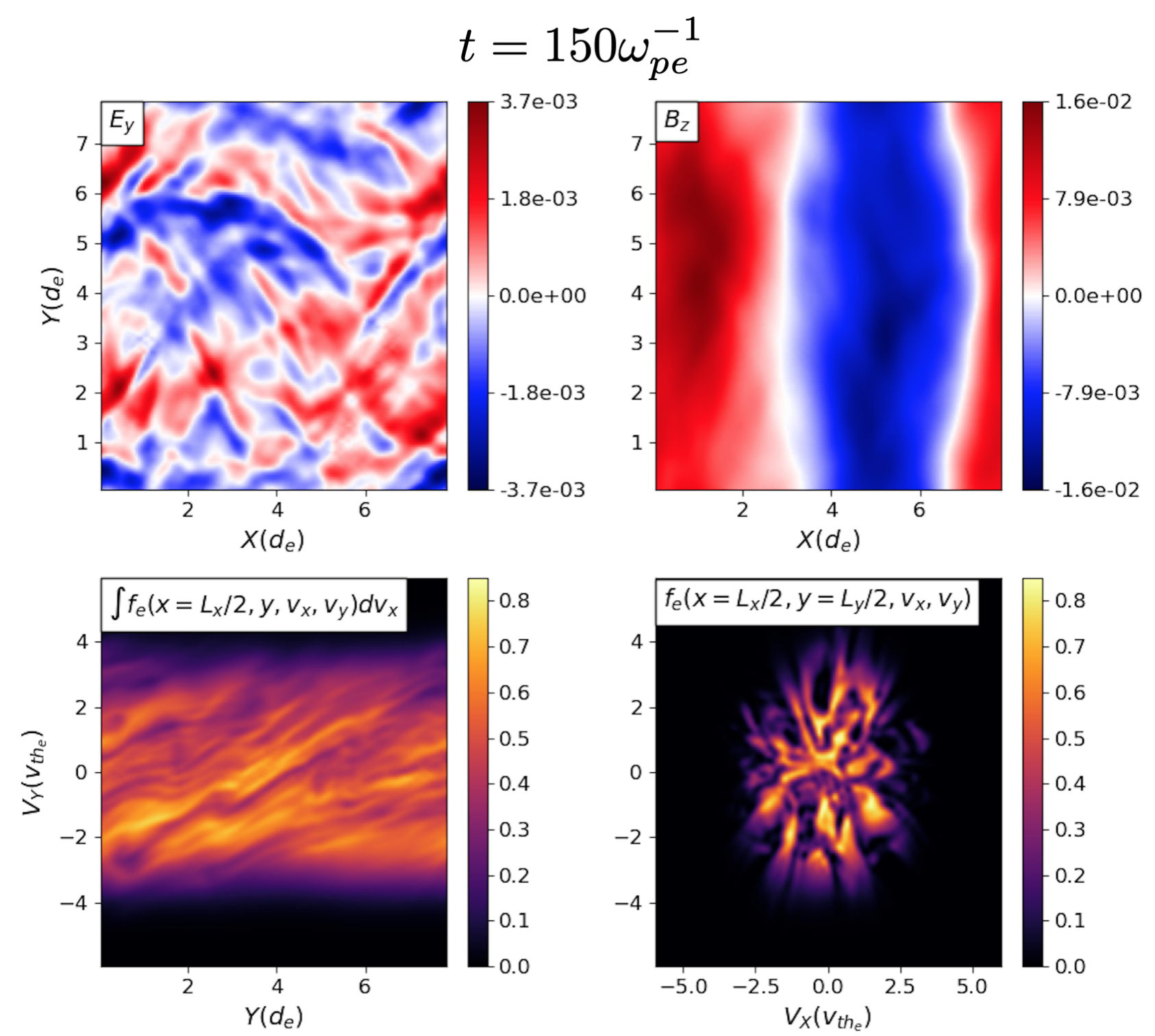}
    \includegraphics[width=0.8\textwidth]{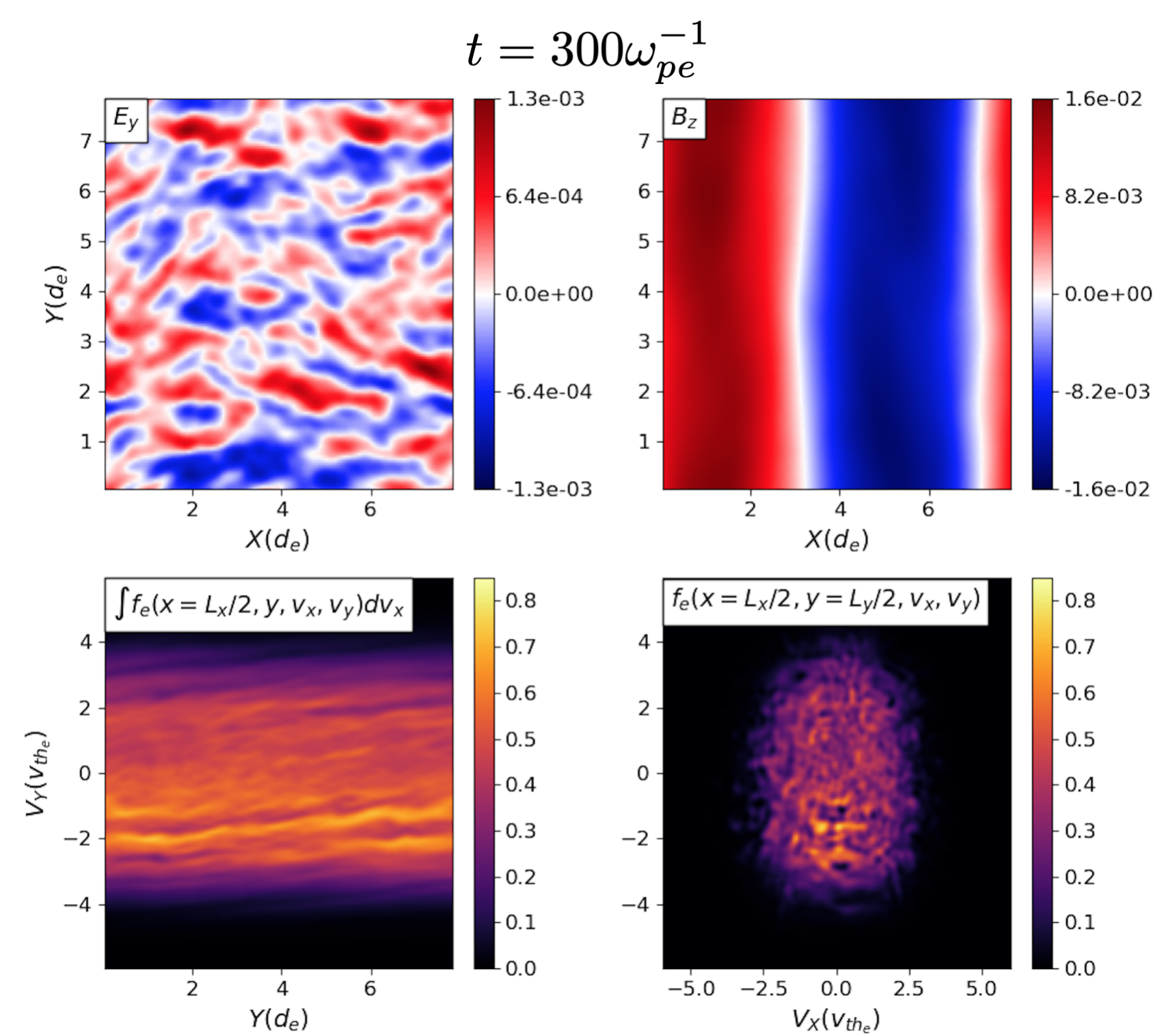}    
    \caption{$t = 150 \omega_{pe}^{-1}$ and $t = 300 \omega_{pe}^{-1}$ snapshots of the evolution of the hot case. In the deep nonlinear phase we observe the development of a temperature anisotropy in the distribution function, which provides a secondary free energy source for the secular Weibel instability. The growth of the secular Weibel instability from the temperature anisotropy ultimately supports a saturated magnetic field.}
    \label{fig:hotTSWLate}
\end{figure}
We plot in Figures~\ref{fig:hotTSWEarly} and \ref{fig:hotTSWLate} the evolution of the hot case in the early and late nonlinear stages of the plasma.
\begin{figure}
    \centering
    \vspace{-0.5in}  
    \includegraphics[width=0.8\textwidth]{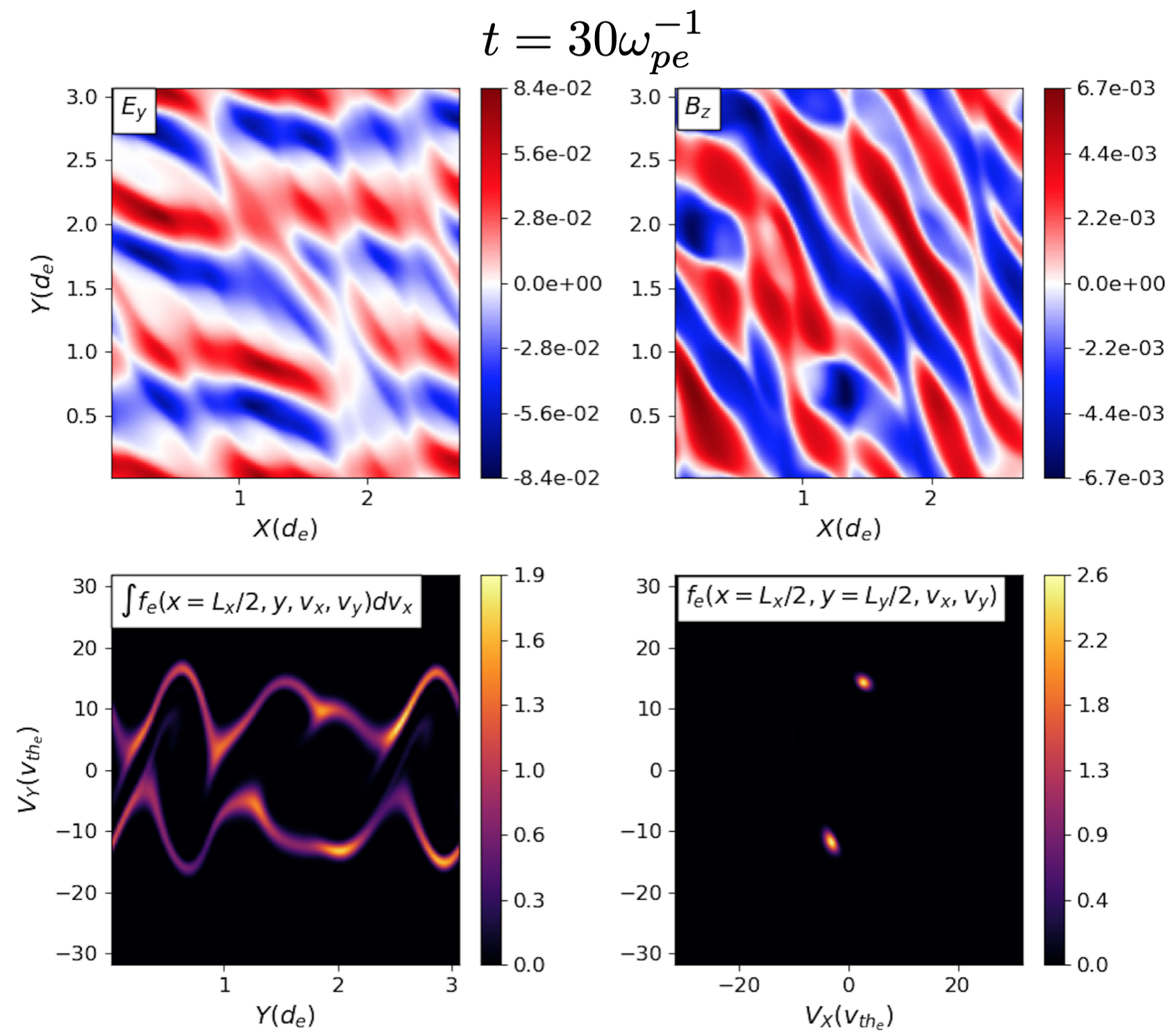}
    \includegraphics[width=0.8\textwidth]{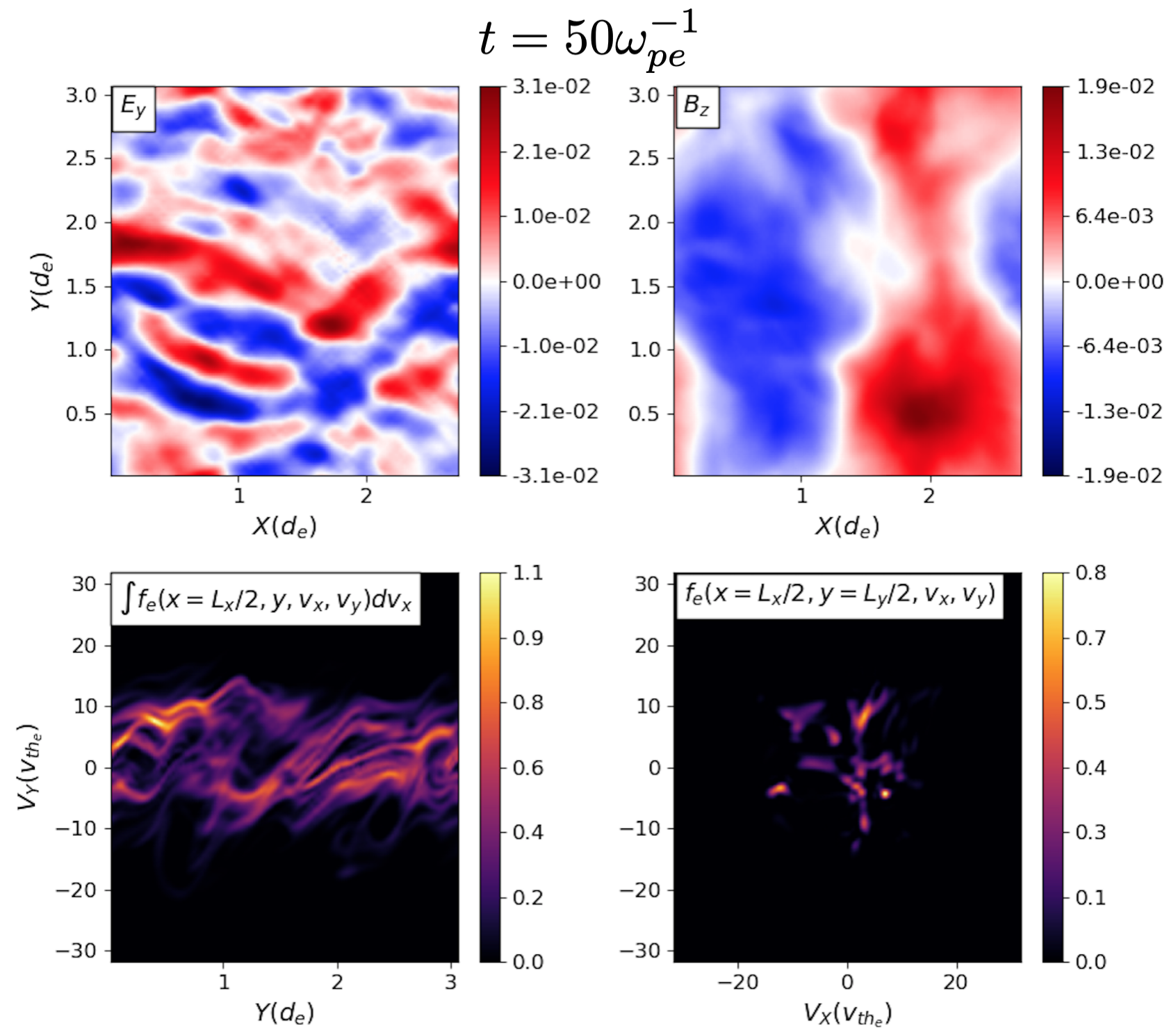}    
    \caption{$t = 30 \omega_{pe}^{-1}$ and $t = 50 \omega_{pe}^{-1}$ snapshots of the evolution of the cold case. We observe significantly more structure in the electromagnetic fields compared to the hot case in Figure~\ref{fig:hotTSWEarly}, as a variety of oblique modes all growth in tandem with the two-stream instability. These additional modes also lead to additional phase space structure, in contrast to the simple plateaus in $v_y$ which formed in the hot case.}
    \label{fig:coldTSWEarly}
\end{figure}
\begin{figure}
    \centering
    \vspace{-0.5in}  
    \includegraphics[width=0.8\textwidth]{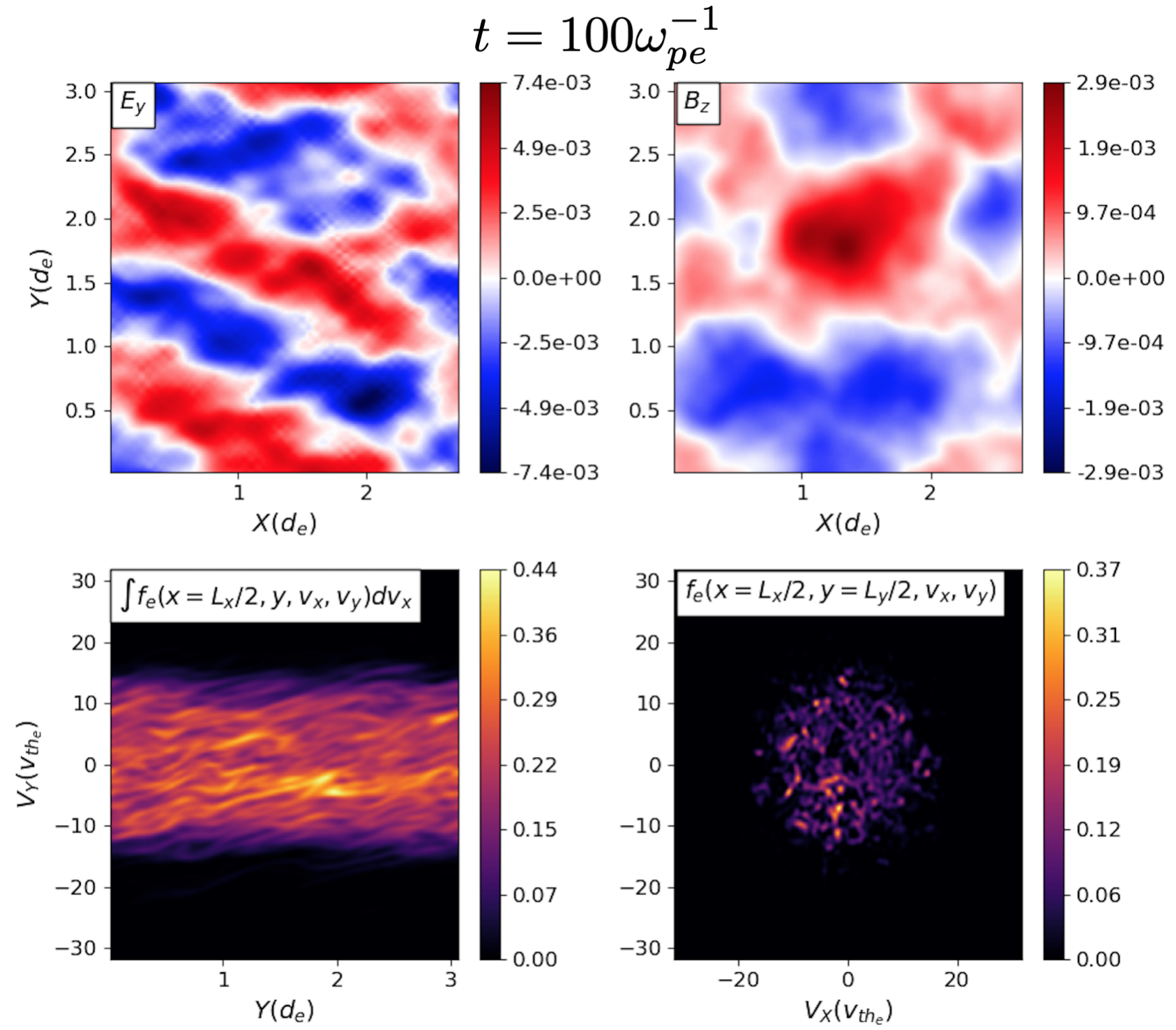}
    \includegraphics[width=0.8\textwidth]{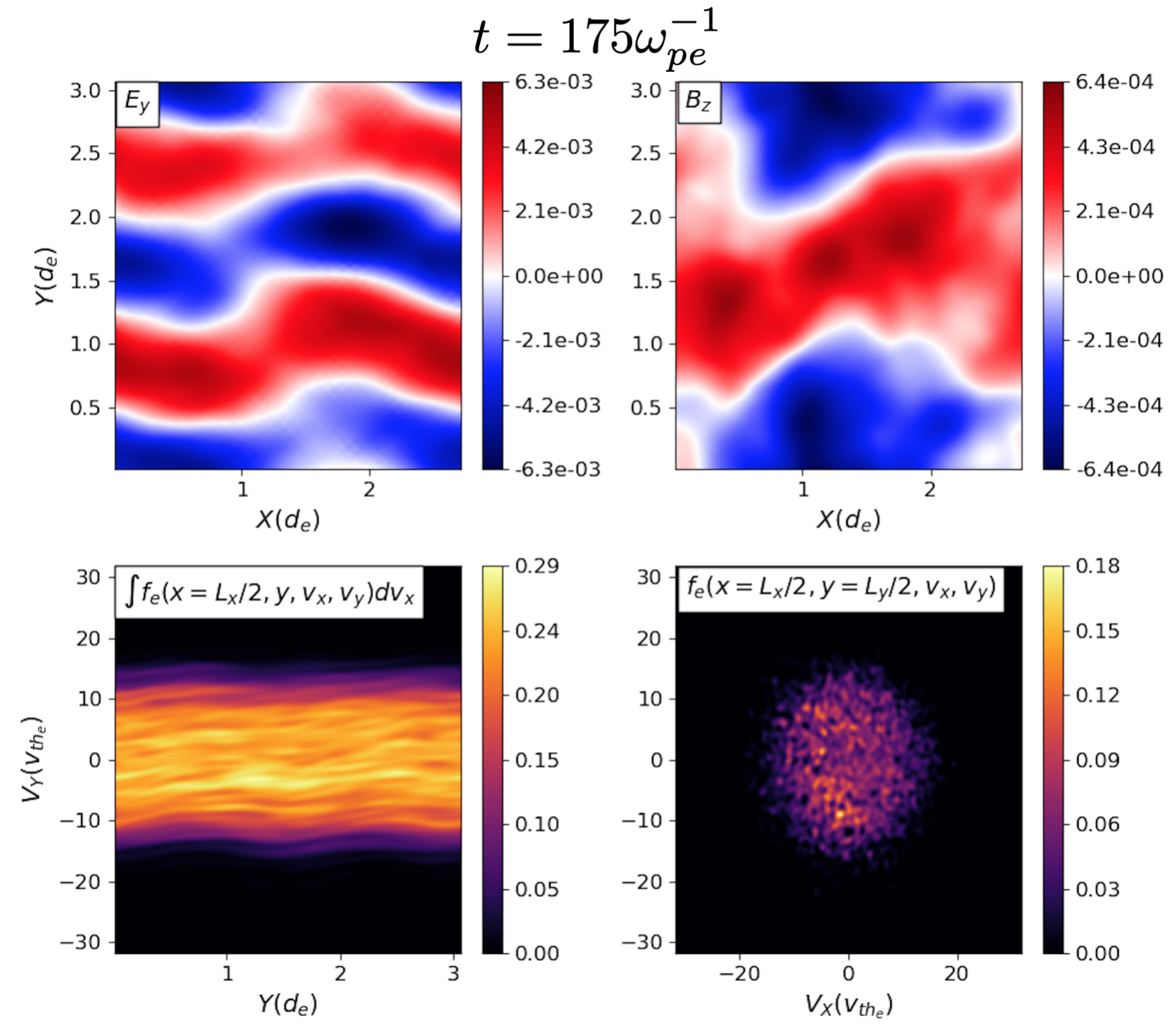}    
    \caption{$t = 100 \omega_{pe}^{-1}$ and $t = 175 \omega_{pe}^{-1}$ snapshots of the evolution of the cold case. The saturated oblique modes have now given their energy back to the electrons in a much more isotropic fashion than a pure two-stream mode, leading to almost zero temperature anisotropy. Without a temperature anisotropy to provide free energy to the Weibel instability, the magnetic field collapses, and we observe no saturated magnetic field structure.}
    \label{fig:coldTSWLate}
\end{figure}
Likewise, he cold case is presented in Figures~\ref{fig:coldTSWEarly} and \ref{fig:coldTSWLate}.

In the hot case, in the early nonlinear stage, we see the formation of the two-stream modes with their quasi-one dimensional structure in $E_y$, uniform in $x$ and multiple wavelengths of the fastest growing mode in $y$.
While there is some initial magnetic field present due to the growing oblique modes, the dynamics are dominated at this stage by the electrostatic two-stream instability.
As the two-stream modes saturate, we see the roll-up in phase space in the $y-v_y$ reduced distribution functions shown.
Importantly, in the early nonlinear stage, the more slowly growing filamentation instability arises and fractures the saturated two-stream modes.
We thus have the beginnings of magnetic field growth due to the presence of the filamentation instability.

However, the sustained growth of the magnetic field arises due to the presence of a secondary instability in the hot case.
The fast saturation of the two-stream instability, along with the disruption and release of the stored electrostatic energy from the saturated two-stream modes by the filamentation instability, heats the electrons primarily in one direction in velocity space, $v_y$, because the electrostatic two-stream instability is fundamentally one-dimensional.
But this leads to a temperature anisotropy in the electron distribution, as can be seen forming in the late nonlinear evolution of the hot case in Figure~\ref{fig:hotTSWLate}.
This temperature anisotropy provides a source of free energy for the secular Weibel instability \citep{Weibel:1959}, and a saturated magnetic field.
We can clearly see this temperature anisotropy by inspection of the electron distribution function in $v_x-v_y$ in the late nonlinear time, as the distribution function is visibly broadened in $v_y$.
Note that the magnetic energy saturates at $\epsilon_B\sim 10^{-2}$, near the Alfv\'en-limited regime, $\rho_e\sim m_e u_d/(e B_z)\sim 7d_e\sim L_x$, and enters a steady-state oscillation at the magnetic bounce frequency, agreeing closely with previous particle-in-cell studies \citep{Fonseca:2003, Silva:2003, Nishikawa:2003, Nishikawa:2005, Kato:2008, Kumar:2015, Takamoto:2018} and 1X2V simulations \citep{Califano:1998, Cagas:2017b}.

The cold case is strikingly different, as we see that the two-stream mode is now competing with a spectrum of oblique modes in the early nonlinear stage in Figure~\ref{fig:coldTSWEarly}.
The electric and magnetic fields are much more structured, and while a single oblique mode is relatively dominant, we see that the distribution function structure from the initial saturation of the instabilities is not as simple as the roll-up and formation of electron tubes observed in the hot case.
Critically, the saturation of a spectrum of oblique modes at similar times leads to a heating of the electrons in a roughly isotropic fashion, as can be seen in the $v_x-v_y$ cuts in Figure~\ref{fig:coldTSWLate}.
This isotropic energization means that there is no temperature anisotropy to provide free energy for the Weibel instability to sustain the growing magnetic field.
The magnetic field that has grown ultimately collapses as the oblique modes damp on the electrons, giving their energy back to the electrons.

We can explicitly quantify this difference in the anisotropy after these instabilities have gone nonlinear.
In Figure~\ref{fig:aniso}, we find that the spatially averaged temperature anisotropy, $\bar{A}$, drops from a large initial value in both the hot, $\bar{A}=5$, and cold, $\bar{A}=101$, cases to some residual value as the instabilities present grow off this effective temperature anisotropy, where the spatially averaged temperature anisotropy is defined as
\begin{align}
    \bar{A}=\int_0^{L_y} \int_0^{L_x} \frac{\int (v_y-u_y)^2f(x,y,\mvec{v})\dv}{\int (v_x-u_x)^2f(x,y,\mvec{v})\dv}\dx,    
\end{align}
where $u_{x,y}$ are the flows in the $x$ and $y$ dimensions respectively.
\begin{figure}[!htb]
    \centering
    \includegraphics[width=\textwidth]{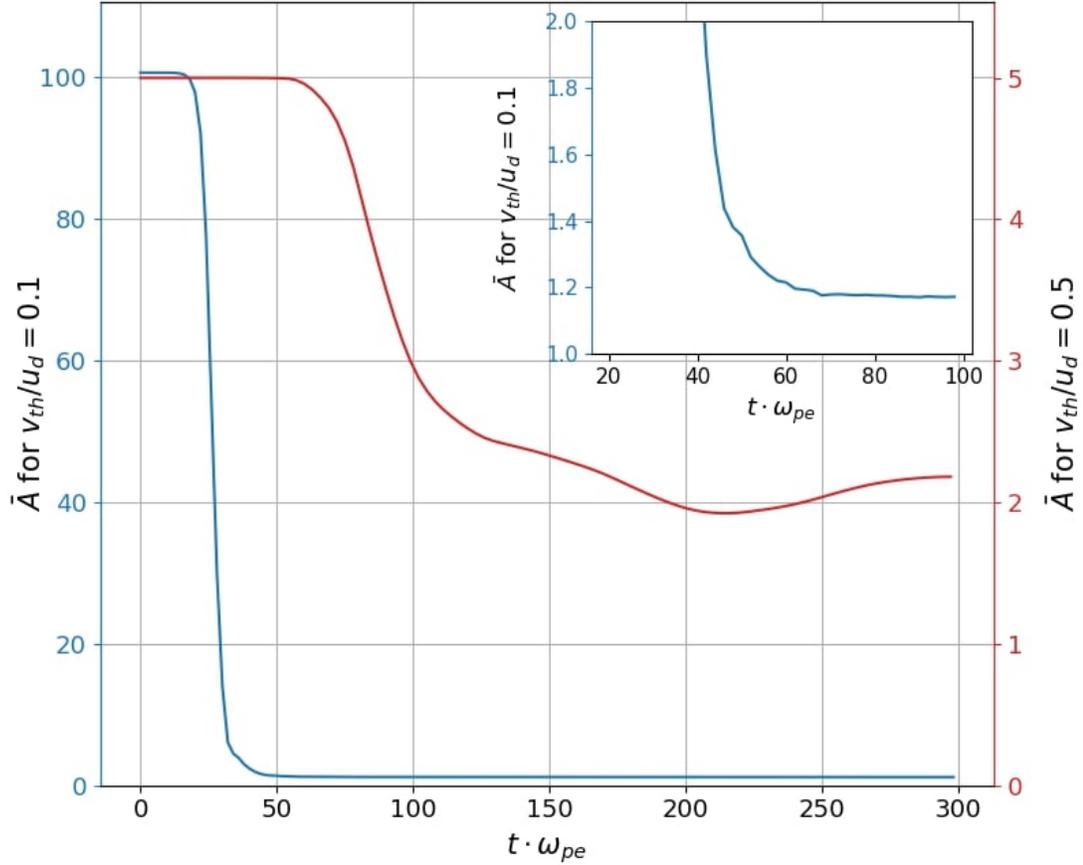}
    \caption{Effective temperature anisotropy of the hot case (red) and cold case (blue) over time. The effective temperature anisotropy starts at a finite value because of the initial beams in $v_y$ and then decreases as the beam-driven instabilities are excited. For the hot case, the temperature anisotropy reduces to a finite value, off which the secular Weibel instability can ultimately feed. In the cold case, the effective temperature anisotropy decreases to a value close to one, i.e., close to isotropy, and thus there is no free energy source for the secular Weibel instability to grow and support a saturated magnetic field.}
    \label{fig:aniso}
\end{figure}
We note the evolution of the temperature anisotropy in the hot case, where we observe a decrease in the anisotropy from $\bar{A}=5$ to a finite value, $\bar{A}\approx 2$, that explains the source of free energy for the secular Weibel instability that ultimate supports the saturated magnetic field.
The cold case on the other hand, has functionally no temperature anisotropy after nonlinear saturation, having collapsed from the large effective temperature anisotropy of two cold beams, $\bar{A}=101$, to $\bar{A}\approx 1.2$.

This collapse of the magnetic field and inefficient conversion of the initial kinetic energy of the cold beams to any appreciable amount of magnetic energy has not been previously observed in the literature, and in fact contradicts previous particle-in-cell studies in similar parameter regimes \citep{Kato:2008}.
While there are many differences between the study performed here and the study performed in a similar parameter regime in \citet{Kato:2008}, e.g., \citet{Kato:2008} includes the effect of the protons on the dynamics and self-consistently drives the system by studying a collisionless shock which excites these instabilities, we consider here the effect particle noise can have on simulations in this parameter regime.
Since the magnetic field collapses by orders of magnitude as a result of these oblique modes isotropically heating the electrons as these instabilities nonlinearly saturate, we are interested in determining the effective phase space resolution required to adequately resolve this process.

We plot in Figure~\ref{fig:PICGkeyll} a suite of simulations using the particle-in-cell code \texttt{p3d} \citep{Zeiler:2002}.
\begin{figure}[!htb]
    \centering
    \includegraphics[width=\textwidth]{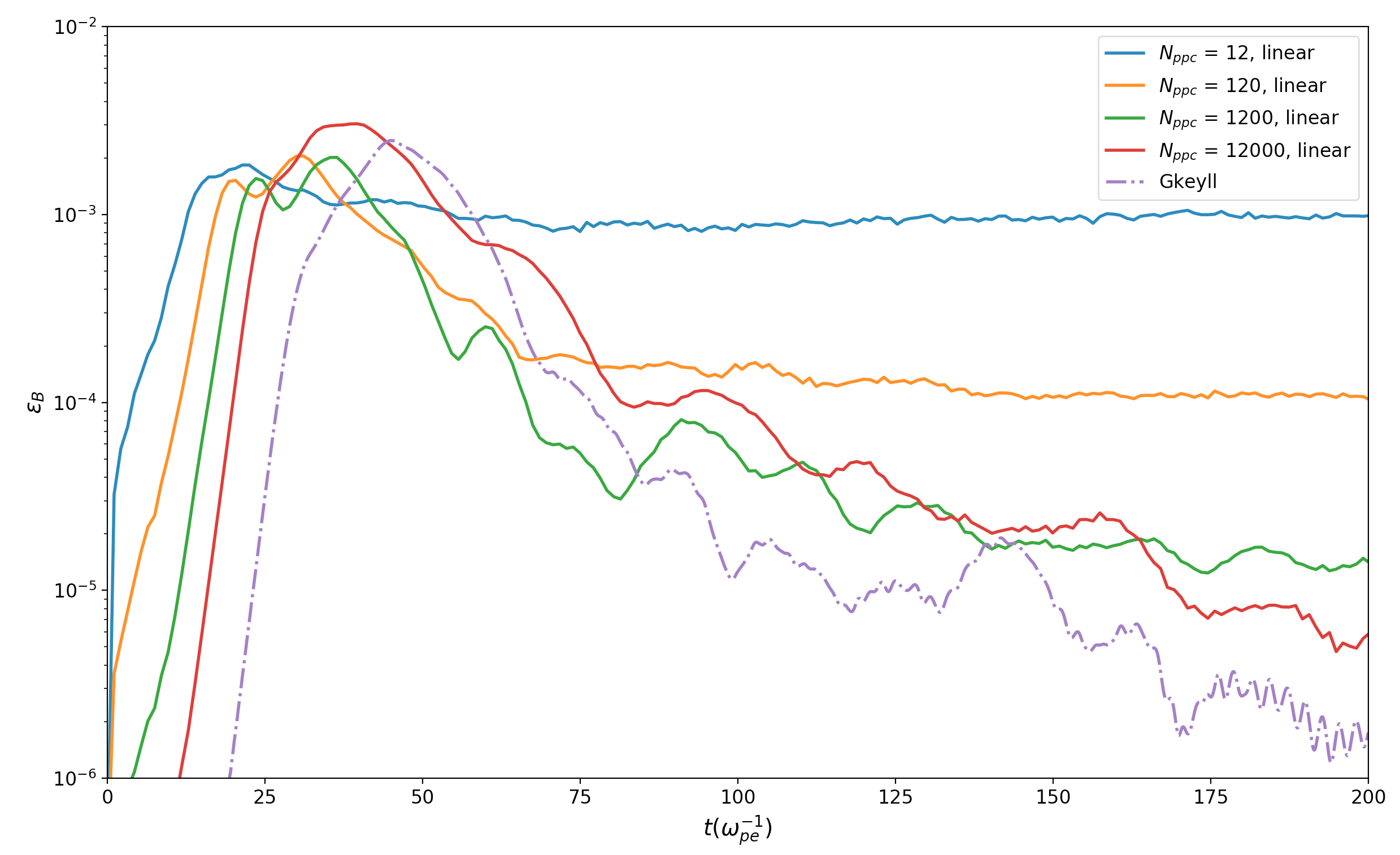}
    \caption{Comparison of the integrated magnetic field energy between a number of particle-in-cell simulations, varying the particles-per-cell, and the \gke VM-FP simulation of the cold case. In the limit of large particle-per-cell counts, the particle-in-cell simulations agree with the continuum kinetic result, but as the number of particles-per-cell is decreased, a saturated magnetic field appears.}
    \label{fig:PICGkeyll}
\end{figure}
We initialize the simulations in exactly the same way as the continuum VM-FP simulations using \gke, we specify two drifting Maxwellians for the electrons, and a bath of fluctuations in the electromagnetic fields given by \eqr{\ref{eq:BTSWInit}}.
The particle-in-cell simulations are performed using linear particle interpolants (triangle shaped particles), and the number of particles per cell is varied to determine the effect that particle noise has on the solution.

We can see that indeed, particle noise does appear to lead to a saturated magnetic field state.
Further, the convergence to the continuum, grid-based method is slow, as it requires a considerable number of particles to recreate the behavior of the collapsing magnetic field.
The saturated magnetic field in the low particle count simulations is a result of ``quasi-thermal'' noise in the sampling of the current to produce the magnetic field.
Essentially, in the same way that particle noise can manifest as fluctuations in the electric field due to errors in the sampling of the density of the particle distribution function \citep{Langdon:1979}, so too can these errors manifest in the current, giving rise to and supporting a magnetic field.

Given the fact that the low particle count simulations saturate at what appears to be the noise floor of the simulations, we can potentially improve the comparison by filtering the particle-in-cell data using a simple low pass filter at the largest wavenumber fluctuations in the box.
We plot the same comparison between our continuum VM-FP simulation of the cold case, and the two extreme particle counts, with and without filtering, in Figure~\ref{fig:PICGkeyllFilter}.
\begin{figure}[!htb]
    \centering
    \includegraphics[width=\textwidth]{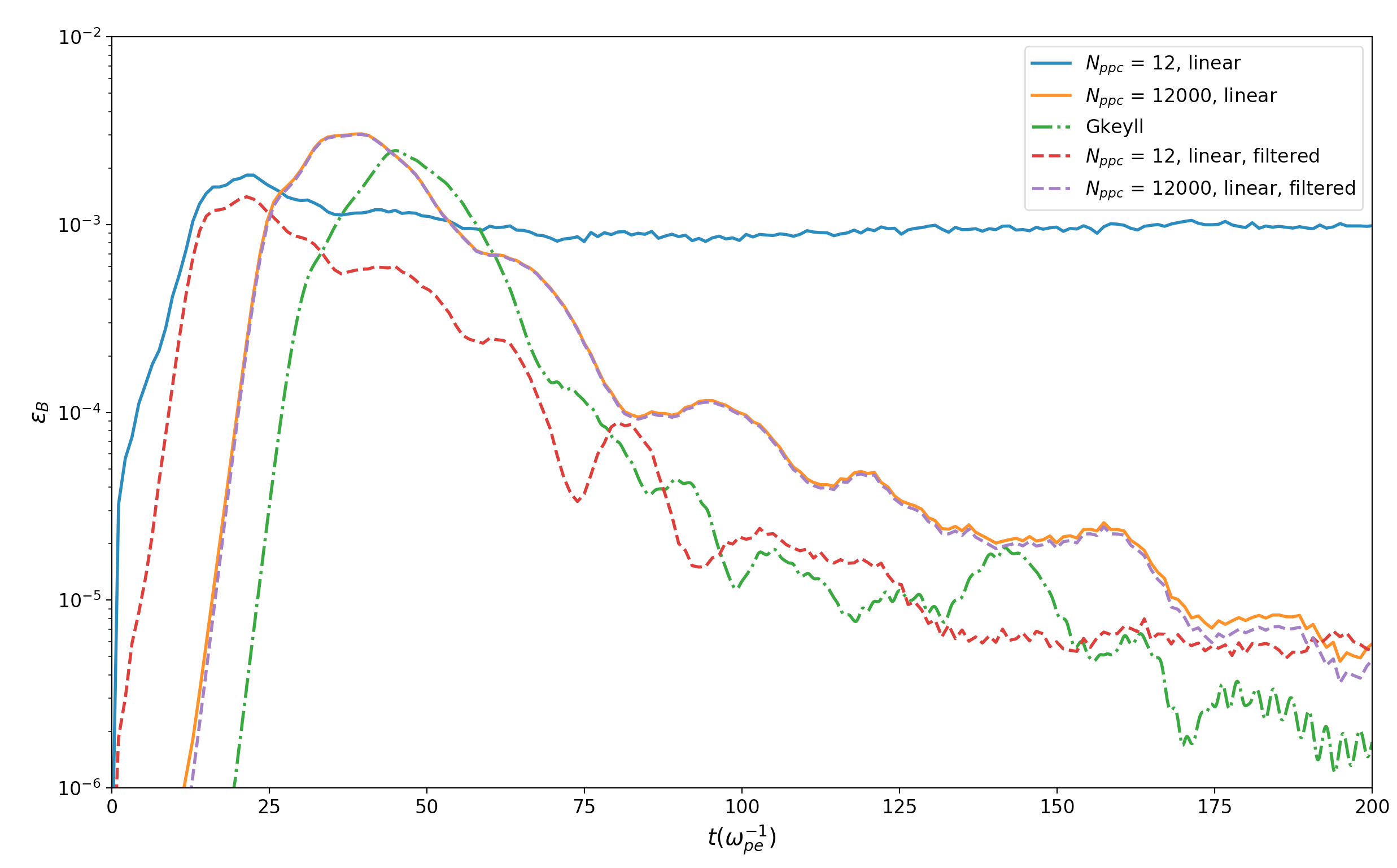}
    \caption{Comparison of the integrated magnetic field energy between the largest and smallest particle-per-cell counts, with and without a low pass filter, and the \gke VM-FP simulation of the cold case. We can see that the filter does allow for the recovery of the collapsing magnetic field in the low particle-per-cell count, adding credibility to the interpretation that the saturated magnetic field is due to noise.}
    \label{fig:PICGkeyllFilter}
\end{figure}
The improvement from a low-pass filter adds further credibility to the interpretation that the saturated magnetic field in the low particle count particle-in-cell calculations arises due to counting noise.

Importantly, while these isolated simulations can be improved with filtering, it does not eliminate the possibility that particle noise is at least partially responsible for the lack of agreement between the \gke results presented here and other particle-in-cell studies \citep{Kato:2008}.
While filtering as a post-processing step works robustly for this problem set-up, where the plasma is perturbed and allowed to evolve, a driven simulation in which the plasma instabilities are constantly being excited may be polluted by this same noise that we can see in the non-filtered case.
It is much more difficult to filter the noise in-situ, and thus the dynamics may be affected by the magnetic field attempting to collapse due to the electron instabilities, but being unable to, under the stress of a constant injection of noise-polluted, unstable fluctuations.
Given the sensitivity of the overall dynamics and magnetic field growth to parameter regimes of relevance in astrophysical plasmas, it is vital that care be taken when resolving the phase space evolution of these instabilities.
Further details of this comparison can be found in \citet{Juno:2020}.

We conclude this section having presented a series of simulations of unstable plasmas, in which novel behavior in the competition between beam-driven instabilities was found in the limit of the beam temperature and the ratio $v_{th_e}/u_d$ decreasing. 
While the continuum DG VM-FP solver described in this thesis recovers the results of previous kinetic studies when the electron beams are hot, we find that the secular Weibel instability can feed off the residual temperature anisotropy remaining from saturated two-stream modes, the picture changes dramatically as the beams grow colder.
The oblique modes that exist between the filamentation instability and two-stream instability become as fast growing as, or faster than, the two-stream instability, significantly complicating the initial nonlinear phase.
Without the dominant two-stream mode in the early nonlinear saturation, the electrons are ultimately energized quasi-isotropically, leading to a collapse of the temperature anisotropy and lack of a free energy source to support a saturated magnetic field.

We attempted to replicate this result in analogous particle-in-cell simulations and found that the result is sensitive to the particle noise arising from the number of particles per cell employed in the simulation.
Simulations with very few particles per cell attain saturated magnetic field states arising from the presence of quasi-thermal noise in the magnetic field, i.e., sampling error in the computation of the current from the particles discretizing the distribution function. 
While these errors can be mitigated with filtering in this isolated system, we emphasize that recovering the behavior of these instabilities in a driven context, such as a collisionless shock, may be more challenging.
We thus argue for the utility of the continuum approach presented in this thesis as a means of obtaining an accurate solution for plasma dynamics that are sensitive to phase space resolution.
%Chapter 6

\renewcommand{\thechapter}{6}

\chapter{Summary and Future Work}\label{ch:Summary}

We have presented in this thesis the derivation, implementation, and application of a discretization of the Vlasov--Maxwell--Fokker--Planck (VM-FP) system of equations which uses the discontinuous Galerkin (DG) finite element method to numerically integrate the VM-FP equation system on a phase space grid.
In contrast to traditional particle-based approaches to the numerical integration of the kinetic equation, this approach provides a high fidelity representation of the particle distribution function, free of the counting noise inherent to Monte Carlo methods.
This unpolluted discrete representation of the particle dynamics in the full phase space affords new opportunities for analysis of the plasma processes present directly in phase space, and makes new problems accessible by increasing the signal-to-noise ratio.

We identified and solved a number of analytic and numerical challenges throughout this thesis.
We showed what is required in the mathematical formulation of the DG algorithm for the discrete VM-FP system of equations to retain important properties of the continuous system, such as conservation of mass and energy.
In the implementation stage, we noted that the direct discretization of the VM-FP system of equations was rife with difficulties owing to the high dimensional nature of the equation system.
Importantly, we noted that standard means of lowering the cost of DG schemes would be catastrophic for the discretization of the VM-FP system of equations, as numerical integration errors that could reduce the computational complexity of the algorithm would inevitably destroy the implicit properties of the VM-FP system of equations, such as conservation of energy.
We designed a scheme free of aliasing errors in the integration, and further formulated a basis set for our DG scheme using orthonormal, modal polynomials that sparsified the resulting tensor-tensor convolutions.

We benchmarked the implementation of the DG VM-FP solver against a large suite of tests, and numerically demonstrated the analytically proved properties of the scheme.
The DG VM-FP solver was then deployed to study the energization of plasmas in fundamental plasma processes such as collisionless shocks as well as the details of the nonlinear saturation of beam-driven instabilities.
Using the increased phase space resolution afforded to us by a continuum discretization of the VM-FP system of equations, we were able to directly diagnose the energization processes such as shock-drift acceleration in phase space.
Likewise, we were able to identify a new parameter regime as the unstable beams became colder for the saturation of filamentation-type instabilities.
In this cold parameter regime, we observed no saturated magnetic field as a result of the competition between additional unstable modes that could grow more quickly in the cold beam parameter regime.
We drew special attention to this result, as analogous particle-in-cell simulations of this system found saturated magnetic fields when using low numbers of particles per cell due to particle noise.

There are a number of avenues of future research to build off the algorithmic and physics work presented in this thesis.
The methods in this thesis can be extended to other kinetic systems, for example the relativistic Vlasov-Maxwell system of equations.
In addition, it is worth exploring whether the recovery procedure described in Chapters~\ref{ch:DGFEM} and \ref{ch:ImplementationDGFEM} for the diffusion operator in the Fokker--Planck equation can also be applied to other components of the update, e.g., the discretization of Maxwell's equations.
Given some of the challenges in discretizing Maxwell's equations, especially in the choice of numerical flux function, an alternative approach that reconstructs continuous functions at the interface could be particularly powerful.

On the physics side, we have demonstrated that the field-particle correlation, combined with our continuum discretization of the VM-FP system of equations, provides a particularly useful way to characterize the energization processes present in phase space, but we have only scratched the surface of what can be done.
Even amongst the benchmarks presented, for example the lower hybrid drift instability and magnetic pumping, identifying the phase space energization signatures of these processes would further build a Rosetta stone for assistance in interpreting future numerical and observational solutions.
We can also extend the study of filamentation-type instabilities to include the proton dynamics as well as inhomogeneities in the beams, e.g., if the two beams have different densities.

But we conclude noting the power and utility of our continuum VM-FP solver in the \gke~ simulation framework, and emphasize that there is an enormous array of problems that can be tackled with this solver, especially if one requires high phase space resolution.
\titleformat{\chapter}
{\normalfont\large}{Appendix \thechapter:}{1em}{}
%\include{AppendixA}
%\include{AppendixB}
%Appendix --
\appendix
\renewcommand{\thechapter}{A}
\renewcommand{\chaptername}{Appendix}

\chapter{Proofs of the Properties of the Continuous \\ Vlasov--Maxwell--Fokker--Planck System of Equations}\label{app:proofsContinuous}
\textbf{Proof of Proposition~\ref{prop:collisionlessMassConservation} (The collisionless Vlasov--Maxwell system of equations conserves mass.)}
\begin{proof}
If we multiply the conservation equation form of the collisionless Vlasov equation, \eqr{\ref{eq:vlasovConservationEq}}, by the mass of the particle, integrate over the phase space domain $K$, and apply the divergence theorem, we obtain,
\begin{align}
    \frac{d}{dt} \left ( m_s \int_K f_s \dz \right ) = -m_s \oint_{\partial K} \gvec{\alpha}_s f_s \thinspace dS = 0,
\end{align}
by our assumed boundary conditions.
Note that this proposition holds individually for each species $s$ in the plasma as we are not including the effects of source terms such as ionization or recombination in our system.
\end{proof}
\textbf{Proof of Proposition~\ref{prop:collisionlessL2} (The collisionless Vlasov--Maxwell system of equations conserves the $L^2$ norm of the particle distribution function.)}
\begin{proof}
We first multiply the conservation equation form of the collisionless Vlasov equation, \eqr{\ref{eq:vlasovConservationEq}}, by the distribution function $f_s$ and integrate over the full phase space to obtain
\begin{align}
     \frac{d}{dt} \left ( \frac{1}{2} \int_K f_s^2 \dz \right ) = -\oint_{\partial K} \gvec{\alpha}_s f_s^2 \thinspace dS + \int_K \gz f_s \cdot \gvec{\alpha}_s f_s \dz,
\end{align}
where we have used the chain rule,
\begin{align}
    f_s \frac{d}{dt} f_s = \frac{1}{2} \frac{d}{dt} \left (f_s^2 \right),
\end{align}
to simplify the left hand side and integration by parts to rewrite the right hand side.
We can again use our assumed boundary conditions to eliminate the surface integral, and the product rule to rewrite the volume integral,
\begin{align}
     \gz f_s \cdot \gvec{\alpha}_s f_s = \gz \cdot \left (\frac{1}{2} \gvec{\alpha}_s f_s^2 \right ) - \frac{1}{2} \left (\gz \cdot \gvec{\alpha}_s \right ) f_s^2 = \gz \cdot \left (\frac{1}{2} \gvec{\alpha}_s f_s^2 \right ),
\end{align}
since phase space is incompressible,
\begin{align}
    \gz \cdot \gvec{\alpha}_s = \left (\gx \cdot \mvec{v}, \frac{q_s}{m_s} \gv \cdot [\mvec{E} + \mvec{v} \times \mvec{B}] \right ) = 0. \label{eq:phaseSpaceIncomp}
\end{align}
But, since we can rewrite the volume term as a total derivative, we can again apply the divergence theorem and use boundary conditions to eliminate the remainder of the right hand side,
\begin{align}
    \frac{d}{dt} \left ( \frac{1}{2} \int_K f_s^2 \dz \right ) = 0. \notag
\end{align}
This completes the proof.
As with the conservation of particles, the conservation of the $L^2$ norm by the collisionless Vlasov--Maxwell system holds individually for each species $s$ in the system. 
\end{proof}
\textbf{Proof of Proposition~\ref{prop:collisionlessEntropy} (The collisionless Vlasov--Maxwell system of equations conserves the entropy density $S = -f \ln(f)$ of the system.)}
\begin{proof}
Again using the conservation equation form of the collisionless Vlasov equation, \eqr{\ref{eq:vlasovConservationEq}}, multiplying by $-\ln f_s$, and integrating over phase space we obtain
\begin{align}
    \frac{d}{dt} \left [\int_K -f_s \ln(f_s) \dz \right ] = \oint_{\partial K} \ln (f_s) \left ( \gvec{\alpha}_s f_s \right ) \thinspace dS - \int_{K} \pfrac{f_s}{t} + \gz \ln (f_s) \cdot \gvec{\alpha}_s f_s,  
\end{align}
where we have again used the chain rule to rewrite the time derivative,
\begin{align}
    -\ln(f_s) \pfrac{}{t} f_s = \pfrac{f_s}{t} - \pfrac{\ln(f_s) f_s}{t},
\end{align}
since
\begin{align}
    \pfrac{\ln(f_s)}{t} = \pfrac{f_s}{t} \frac{1}{f_s}.
\end{align}
We have also again used integration by parts on the right hand side of \eqr{\ref{eq:vlasovConservationEq}} and can eliminate the surface integral with our boundary conditions in phase space.
Using the chain rule and the incompressibility of phase space, \eqr{\ref{eq:phaseSpaceIncomp}}, we find
\begin{align}
     \gz \ln (f_s) \cdot \gvec{\alpha}_s f_s = \gvec{\alpha}_s \cdot \gz f_s = \gz \cdot \left ( \gvec{\alpha}_s f_s \right ),
\end{align}
but this expression means the the right hand side is simply the collisionless Vlasov equation, \eqr{\ref{eq:vlasovConservationEq}}, which is equal to zero, completing the proof,
\begin{align}
    \frac{d}{dt} \left (\int_K -f_s \ln(f_s) \dz \right ) = 0. \notag
\end{align}
As before with mass conservation and conservation of the $L^2$ norm, conservation of entropy holds independently for each species in the collisionless Vlasov--Maxwell system.
\end{proof}
\textbf{Proof of Proposition~\ref{prop:collisionlessMomentum} (The collisionless Vlasov--Maxwell system of equations conserves the total, particles plus fields, momentum.)}
\begin{proof}
We begin by multiplying \eqr{\ref{eq:collisionlessComponent}} by $m_s \mvec{v}$, summing over species, and integrating over phase space to obtain
\begin{align}
    \underbrace{\int_K \sum_s m_s \mvec{v} \pfrac{f_s}{t} \dz}_{\int_\Omega \sum_s \pfrac{\gvec{\mathcal{M}}_s}{t} \dx } = - \int_K \sum_s & m_s \mvec{v} \gx \cdot (\mvec{v} f_s) \dz \notag \\
    & - \int_K \sum_s m_s \mvec{v} \gv \cdot \left [ \frac{q_s}{m_s} \left ( \mvec{E} + \mvec{v} \times \mvec{B} \right ) f_s \right ] \dz.
\end{align}
Since the velocity coordinate does not depend on configuration space, we can bring $m_s \mvec{v}$ inside the divergence in the first term on the right side, apply the divergence theorem, and eliminate this term by our configuration space boundary conditions.
For the second term on the right hand side, we can use integration by parts to move the velocity divergence onto $m_s \mvec{v}$, eliminating the surface term using our boundary condition in velocity space,
\begin{align}
     -\int_K \sum_s m_s \mvec{v} \gv \cdot \left [ \frac{q_s}{m_s} \left ( \mvec{E} + \mvec{v} \times \mvec{B} \right ) f_s \right ] \dz & = \int_K \sum_s q_s \gv \mvec{v} \cdot \left ( \mvec{E} + \mvec{v} \times \mvec{B} \right ) f_s \dz, \notag \\
     & = \int_\Omega \rho_c \mvec{E} + \mvec{J} \times \mvec{B} \dx, \label{eq:mom-proof-1}
\end{align}
where we have used the fact that $\gv \mvec{v} = \overleftrightarrow{\mvec{I}}$ and the definitions of the charge density and current density, Eqns.\thinspace\ref{eq:chargeDensity}--\ref{eq:currentDensity}, to perform the integral over velocity space.
To make further progress, we consider Maxwell's equations. Taking the cross-product of \eqr{\ref{eq:dbdt}} with $\epsilon_0\mvec{E}$, the cross-product of \eqr{\ref{eq:dedt}} with $\mvec{B}/\mu_0$, and subtracting the resulting equations we obtain
  \begin{align}
    \epsilon_0\pfraca{t} (\mvec{E}\times\mvec{B}) + \epsilon_0 \underbrace{\mvec{E}\times(\gx\times\mvec{E})}_{(\gx\mvec{E})\cdot\mvec{E} - (\mvec{E}\cdot\gx)\mvec{E}}
    + \frac{1}{\mu_0} \underbrace{\mvec{B}\times(\gx\times\mvec{B})}_{(\gx\mvec{B})\cdot\mvec{B} - (\mvec{B}\cdot\gx)\mvec{B}}
    =
    -\mvec{J}\times\mvec{B},
  \end{align}
for the evolution of the electromagnetic momentum density,\footnote{The electromagnetic momentum density is also commonly written as
\begin{align}
\mvec{p}_{EM} = \frac{\mvec{S}}{c^2},
\end{align}
where $\mvec{S}$ is the Poynting flux,
\begin{align}
    \mvec{S} = \frac{1}{\mu_0} \mvec{E} \times \mvec{B},
\end{align}
and $c$ is the speed of light, $c = 1/\sqrt{\epsilon_0 \mu_0}$.
}
$\epsilon_0 \mvec{E} \times \mvec{B}$. Now, for any vector field $\mvec{A}$ we have
\begin{align}
    (\nabla\mvec{A})\cdot\mvec{A} & = \nabla |\mvec{A}|^2/2, \\
    (\mvec{A}\cdot\nabla)\mvec{A} & = \nabla\cdot(\mvec{A}\mvec{A})-\mvec{A}\nabla\cdot\mvec{A}.
\end{align}
Using these vector identities and the divergence Eqns.\thinspace(\ref{eq:divE}) and (\ref{eq:divB}) to replace $\gx \cdot \mvec{E} = \rho_c/\epsilon_0$ and $\gx \cdot \mvec{B} = 0$ gives
  \begin{align}
    \epsilon_0\pfraca{t} (\mvec{E}\times\mvec{B}) + \gx\left( \frac{\epsilon_0}{2}|\mvec{E}|^2 + \frac{1}{2\mu_0}|\mvec{B}|^2 \right)
    - \gx\cdot\left( \epsilon_0\mvec{E}\mvec{E} + \frac{1}{\mu_0}\mvec{B}\mvec{B} \right)
    + \varrho_c\mvec{E} = -\mvec{J}\times\mvec{B}. \label{eq:mom-proof-2}
  \end{align}
We recognize the spatial gradients and divergences in \eqr{\ref{eq:mom-proof-2}} to be acting on the Maxwell stress tensor,
\begin{align}
    \overleftrightarrow{\gvec{\sigma}} = \epsilon0 \left (\mvec{E}\mvec{E} - \frac{1}{2} |\mvec{E}|^2 \overleftrightarrow{\mvec{I}} \right ) + \frac{1}{\mu_0} \left (\mvec{B}\mvec{B} - \frac{1}{2} |\mvec{B}|^2 \overleftrightarrow{\mvec{I}} \right ).
\end{align}
So, inserting \eqr{\ref{eq:mom-proof-2}} into \eqr{\ref{eq:mom-proof-1}} and using configuration space boundary conditions to eliminate the total derivatives of the Maxwell stress tensor gives our desired conservation relation,
\begin{align}
    \frac{d}{dt} \left (\int_\Omega \sum_s \gvec{\mathcal{M}}_s  +  \epsilon_0 \mvec{E}\times\mvec{B} \dx \right )= 0. \notag
\end{align}
We emphasize that the linear momentum is a conserved vector quantity.
In other words, only the corresponding components of the particle and electromagnetic momentum can be exchanged, e.g., the $x$ particle momentum can be exchanged with the $x$ component of the electromagnetic momentum.
Of course the stress tensor for the particles,
\begin{align}
    \overleftrightarrow{\mvec{S}}_s = \int \gx \cdot (\mvec{v} \mvec{v} f_s) \dv,
\end{align} 
can move momentum between the various components of the particle momentum density, and likewise the Maxwell stress tensor can move momentum between the various components of the electromagnetic momentum density.
But when the particles and electromagnetic fields exchange momentum, they do so component by component.
\end{proof}
\textbf{Proof of Proposition~\ref{prop:collisionlessEnergyConservation} (The collisionless Vlasov--Maxwell system of equations conserves the total, particles plus fields, energy.)}
\begin{proof}
We proceed in a similar fashion to our proof of momentum conservation, but we now multiply \eqr{\ref{eq:collisionlessComponent}} by $1/2 \thinspace m_s |\mvec{v}|^2$, sum over species, and integrate over phase space to obtain
\begin{align}
    \underbrace{\int_K \sum_s \frac{1}{2} m_s |\mvec{v}|^2 \pfrac{f_s}{t} \dz}_{\int_\Omega \sum_s \pfrac{\mathcal{E}_s}{t} \dx } = - \int_K \sum_s &  \frac{1}{2} m_s |\mvec{v}|^2 \gx \cdot (\mvec{v} f_s) \dz \notag \\
    & - \int_K \sum_s  \frac{1}{2} m_s |\mvec{v}|^2 \gv \cdot \left [ \frac{q_s}{m_s} \left ( \mvec{E} + \mvec{v} \times \mvec{B} \right ) f_s \right ] \dz.
\end{align}
Since the velocity coordinate does not depend on configuration space, we can move $1/2 \thinspace m_s |\mvec{v}|^2$ inside the configuration space divergence, forming a total derivative and allowing us to use the divergence theorem and boundary conditions to eliminate this term.
As before with momentum conservation, we use integration by parts and velocity space boundary conditions on the second term on the right hand side,
\begin{align}
    - \int_K \sum_s  \frac{1}{2} m_s |\mvec{v}|^2 \gv \cdot \left [ \frac{q_s}{m_s} \left ( \mvec{E} + \mvec{v} \times \mvec{B} \right ) f_s \right ] \dz & = \int_K \frac{q_s}{2} \gv |\mvec{v}|^2 \cdot \left ( \mvec{E} + \mvec{v} \times \mvec{B} \right ) f_s \dz \notag \\
    & = \int_\Omega \mvec{J} \cdot \mvec{E} \dz,  \label{eq:fenergy}
\end{align}
where we have used the fact that $\mvec{v} \cdot (\mvec{v} \times \mvec{B}) = 0$ by properties of the cross product to eliminate the magnetic field term.
To make further progress, we again examine Maxwell's equations. 
Taking the dot product of \eqr{\ref{eq:dedt}} with $\mvec{E}/\mu_0$, the dot product of \eqr{\ref{eq:dbdt}} with $\mvec{B}/\mu_0$, and adding the resulting equations gives us
  \begin{align}
    \pfraca{t}\left( \frac{\epsilon_0}{2}|\mvec{E}|^2 + \frac{1}{2\mu_0}|\mvec{B}|^2 \right)
    +
    \frac{1}{\mu_0}\underbrace{ \left [ \mvec{B}\cdot(\gx\times\mvec{E}) - \mvec{E}\cdot(\gx\times\mvec{B}) \right ] }_{= \gx\cdot(\mvec{E}\times\mvec{B})}
    =
    -\mvec{J}\cdot\mvec{E}.
  \end{align}
Using this result in \eqr{\ref{eq:fenergy}}, along with configuration space boundary conditions to eliminate the divergence of the Poynting flux, gives the total energy conservation law,
\begin{align}
    \frac{d}{dt} \left (\int_\Omega \sum_s \mathcal{E}_s+ \frac{\epsilon_0}{2} |\mvec{E}|^2 + \frac{1}{2 \mu_0}|\mvec{B}|^2 \dx \right )= 0. \notag
\end{align}
\end{proof}
\textbf{Proof of Proposition~\ref{prop:collisionMassConservation} (The Fokker--Planck equation conserves mass.)}
\begin{proof}
If we multiply \eqr{\ref{eq:collisionalComponent}} by the mass of the particle and integrate over phase space, just as with Proposition~\ref{prop:collisionlessMassConservation}, we can use the boundary conditions in velocity space to obtain
\begin{align}
    \frac{d}{dt}\left ( m_s \int_K f^c_s \dz \right ) = \oint_{\partial K} \nu_s \left [ (\mvec{v} - \mvec{u}_s) f_s + \frac{T_s}{m_s} \gv f_s \right ] \thinspace dS = 0.
\end{align}
Because we are not including particle sources such as ionization and recombination, this conservation relation holds for each plasma species.
Importantly, because the Fokker-Planck operator only involves derivatives in velocity space, this conservation is \emph{local},
\begin{align}
\int_{K\setminus\Omega} m_s \pfrac{f^c_s}{t} \dv = \pfrac{\rho_s}{t} = 0,
\end{align}
i.e., the Fokker-Planck collision operator does not change the local mass (or number) density in configuration space.
\end{proof}
\textbf{Proof of Proposition~\ref{prop:collisionMomentumConservation} (The Fokker--Planck equation conserves the particle momentum.)}
\begin{proof}
If we first multiply \eqr{\ref{eq:collisionalComponent}} by $m_s \mvec{v}$ and integrate over phase space, we can integrate the collision operator by parts once to obtain
\begin{align}
    \frac{d}{dt}\left ( \int_K m_s \mvec{v} f^c_s \dz \right ) = \oint_{\partial K} & m_s \mvec{v} \thinspace \nu_s \left [ (\mvec{v} - \mvec{u}_s) f + \frac{T_s}{m_s} \gv f_s \right ] \thinspace dS \notag \\
    & - \int_K m_s \nu_s \gv \mvec{v} \cdot \left [ (\mvec{v} - \mvec{u}_s) f_s + \frac{T_s}{m_s} \gv f_s \right ] \dz.
\end{align}
We can eliminate the surface integral with our boundary conditions in velocity space. Recall that $\gv \mvec{v} = \overleftrightarrow{\mvec{I}}$, so the volume integral simplifies to
\begin{align}
   \int_K m_s \gv \mvec{v} \cdot \left [ (\mvec{v} - \mvec{u}_s) f_s + \frac{T_s}{m_s} \gv f_s \right ] & \dz = \int_K m_s (\mvec{v} - \mvec{u}) f_s \dz, \notag \\
   & = \int_\Omega \left ( \gvec{\mathcal{M}}_s - m_s n_s \mvec{u}_s \right ) \dx = 0, \label{eq:collisions-mom-1}
\end{align}
where we have dropped the velocity independent collision frequency for notational convenience.
In our simplification to \eqr{\ref{eq:collisions-mom-1}} we have used the fact that the diffusion coefficient, $T_s/m_s$, does not depend on velocity space to write what remains of the diffusion term as a total derivative, which upon integrating the total derivative and using the boundary conditions in velocity space, eliminates the diffusion term.
\eqr{\ref{eq:collisions-mom-1}} completes the proof.
As with conservation of mass in Proposition~\ref{prop:collisionMassConservation}, since the Fokker--Planck collision operator only includes derivatives in velocity space, we can construct a \emph{local} conservation law,
\begin{align}
\int_{K\setminus\Omega} m_s \mvec{v} \pfrac{f^c_s}{t} \dv = \pfrac{\gvec{\mathcal{M}}_s}{t} = 0,
\end{align}
i.e., the Fokker--Planck collision operator does not change the local momentum density in configuration space.
\end{proof}
\textbf{Proof of Proposition~\ref{prop:collisionEnergyConservation} (The Fokker--Planck equation conserves the particle energy.)}
\begin{proof}
In analogy with Proposition~\ref{prop:collisionMomentumConservation}, we multiply \eqr{\ref{eq:collisionalComponent}} by $1/2 \thinspace m_s |\mvec{v}|^2$, integrate over phase space, and use integration by parts to obtain
\begin{align}
    \frac{d}{dt}\left ( \int_K \frac{1}{2} m_s |\mvec{v}|^2 f^c_s \dz \right )& = \oint_{\partial K} \frac{1}{2} m_s |\mvec{v}|^2 \thinspace \nu_s \left [ (\mvec{v} - \mvec{u}_s) f + \frac{T_s}{m_s} \gv f_s \right ] \thinspace dS \notag \\
    & - \int_K \frac{1}{2} m_s \nu_s \gv |\mvec{v}|^2 \cdot \left [ (\mvec{v} - \mvec{u}_s) f_s + \frac{T_s}{m_s} \gv f_s \right ] \dz.
\end{align}
As before, we can eliminate the surface integral with our boundary conditions in velocity space. Using the fact that $\gv |\mvec{v}|^2 = 2\mvec{v}$, the volume integral can be rewritten as
\begin{align}
    \int_K  m_s \mvec{v} \cdot \left [ (\mvec{v} - \mvec{u}_s) f_s + \frac{T_s}{m_s} \gv f_s \right ] & \dz = \int_K \left [ m_s \left ( |\mvec{v}|^2 - \mvec{v} \cdot \mvec{u}_s \right ) f_s + T_s \mvec{v} \cdot \gv f_s \right ] \dz, \notag \\
    & = \int_\Omega 2 \mathcal{E}_s - m_s n_s |\mvec{u}_s|^2 - 3 n_s T_s \dx = 0, \label{eq:collisions-en-1}
\end{align}
where we have dropped the velocity independent collision frequency for notational convenience and used integration by parts and the velocity space boundary conditions to simplify
\begin{align}
    \int_K T_s \mvec{v} \cdot \gv f_s \dz = \int_K T_s f_s (\gv \cdot \mvec{v})  \dz = \int_\Omega 3 n_s T_s \dx.
\end{align}
\eqr{\ref{eq:collisions-en-1}} completes the proof.
We note that as with conservation of mass in Proposition~\ref{prop:collisionMassConservation} and conservation of momentum in Proposition~\ref{prop:collisionMomentumConservation}, since the Fokker--Planck collision operator only includes derivatives in velocity space, we can construct a \emph{local} conservation law,
\begin{align}
\int_{K\setminus\Omega} \frac{1}{2} m_s |\mvec{v}|^2 \pfrac{f^c_s}{t} \dv = \pfrac{\mathcal{E}_s}{t} = 0,
\end{align}
i.e., the Fokker--Planck collision operator does not change the local energy density in configuration space.
\end{proof}
\textbf{Proof of Proposition~\ref{prop:collisionEntropy} (The Fokker--Planck equation satisfies the Second Law of Thermodynamics and leads to a non-decreasing entropy density $S = -f \ln(f)$.)}
\begin{proof}
Defining the total entropy as
\begin{align}
\mathcal{S}_s = -\int_K f_s \ln{f_s} \dz, \label{eq:totalEntropyDef}
\end{align}
and taking the time derivative of the total entropy, we have
\begin{align}
\pfrac{\mathcal{S}_s}{t} = -\int_K \pfrac{f_s}{t}[\ln (f_s) + 1] \dz.
\end{align}
We can rewrite the Fokker--Planck operator as a flux in velocity space,
\begin{align}
\pfrac{f_s^c}{t} = \gv \cdot \mvec{F},
\end{align}
where
\begin{align}
\mvec{F} = (\mvec{v}-\mvec{u}_s)f_s + \frac{T_s}{m_s} \gv f_s,
\end{align}
and we have dropped the velocity independent collision frequency $\nu_s$ for notational convenience without loss of generality.
Because we have already proved the collisionless component of the VM-FP system of equations does not change the entropy of the system, Proposition~\ref{prop:collisionlessEntropy}, we need only consider the contribution of the Fokker--Planck equation to the evolution of the total entropy,
\begin{align}
\pfrac{\mathcal{S}_s}{t} = -\int_K \gv \cdot \mvec{F} [\ln (f_s) + 1] \dz.
\end{align}
We can integrate the flux by parts and use our boundary conditions in velocity space to obtain
\begin{align}
    \pfrac{\mathcal{S}_s}{t} = \int_K \frac{1}{f_s} \gv f_s \cdot \mvec{F} \dz. \label{eq:collision-entropy-1}
\end{align}
We now substitute
\begin{align}
\gv f_s = \frac{m_s}{T_s} [\mvec{F}-(\mvec{v}-\mvec{u}_s)f_s],
\end{align}
into \eqr{\ref{eq:collision-entropy-1}} to obtain,
\begin{align}
\pfrac{\mathcal{S}_s}{t} = \int_K \frac{m_s}{T_s} \left [ \frac{|\mvec{F}|^2}{f_s} - (\mvec{v} - \mvec{u}_s) \cdot \mvec{F} \right ] \dz.
\end{align}
Using the definition of $\mvec{F}$, the second term in this equation becomes
\begin{align}
\int_K (\mvec{v} - \mvec{u}_s) \cdot \mvec{F} \dz & = \int_K (|\mvec{v}|^2-2\mvec{u}_s \cdot \mvec{v} + |\mvec{u}_s|^2) f_s + \frac{T_s}{m_s} (\mvec{v}-\mvec{u}_s) \cdot \gv f_s \dz \notag \\
 & = \int_\Omega \frac{2}{m_s} \mathcal{E}_s - 2 n_s |\mvec{u}_s|^2 + n_s |\mvec{u}_s|^2 - 3 n_s \frac{T_s}{m_s} \dx = 0,
\end{align}
where we have used integration by parts on the $\gv f_s$ term. Hence,
\begin{align} \label{eq:entropy2}
\pfrac{\mathcal{S}_s}{t} = \int_K \frac{m_s}{T_s} \frac{1}{f_s}|\mvec{F}|^2 \dz
\geq 0,
\end{align}
as long as $f_s \geq 0$\footnote{And $\nu_s > 0$ of course. If the collision frequency was not positive definite, that would be a real problem!}.
Given the preceding discussion, we can also define a velocity integrated entropy density,
\begin{align}
    s_s(\mvec{x}, t) = -\int_{K\setminus\Omega} f_s(\mvec{x}, \mvec{v}, t) \ln(f_s(\mvec{x}, \mvec{v}, t)) \dv,
\end{align}
which is a monotonically increasing function,
\begin{align}
    \pfrac{s_s(\mvec{x}, t)}{t} = \int_{K\setminus\Omega} \frac{m_s}{T_s(\mvec{x}, t)} \frac{1}{f_s(\mvec{x}, \mvec{v}, t)}|\mvec{F}(\mvec{x}, \mvec{v}, t)|^2 \dv \geq 0,
\end{align}
since the Fokker--Planck operator only involves derivatives in velocity space.
In other words, the collision operator leads to non-decreasing entropy at each point in configuration space, and further mixing in configuration space is required to attain a global maximum entropy state.
We might be unsurprised by this statement, as the entropy increase in velocity space corresponds to the second of Bogoliubov's timescales, while the entropy increase in all of phase space corresponds to the third of Bogoliubov's timescales.
\end{proof}
\textbf{Proof of Corollary~\ref{coro:HTheorem} (The maximum entropy solution to the Fokker--Planck collision operator is the Maxwellian velocity distribution.)}
\begin{proof}
By Proposition~\ref{prop:collisionEntropy}, we know that the entropy is a monotonically increasing function. But, if the entropy is a monotonically increasing function in time, and the entropy is a well-defined quantity, i.e., \eqr{\ref{eq:totalEntropyDef}} is not a divergent integral, then the extremum of the entropy must necessarily maximize the entropy.
Thus, we need only find when
\begin{align}
    \pfrac{\mathcal{S}_s}{t} = 0. \label{eq:entropyExtremum}
\end{align}
The time evolution of the entropy vanishes when
\begin{align}
    \mvec{F} = 0,
\end{align}
i.e.,
\begin{align}
    \gv f_s = -\frac{m_s}{T_s}(\mvec{v}-\mvec{u}_s)f_s.
\end{align}
Solving for the distribution function $f_s$, we find
\begin{align}
    f_s = A \exp \left (-m_s \frac{|\mvec{v} - \mvec{u}_s|^2}{2 T_s} \right ),
\end{align}
where $A$ is some constant of integration.
To find the constant of integration, we exploit the requirement that the integral over velocity space of the distribution function must by definition give the density,
\begin{align}
    n_s = \int_{K\setminus\Omega} A \exp \left (-m_s \frac{|\mvec{v} - \mvec{u}_s|^2}{2 T_s} \right ) \dv,
\end{align}
which means
\begin{align}
    A = n_s \left (\frac{m_s}{2 \pi T_s} \right )^{\frac{3}{2}},
\end{align}
where the integral over each velocity direction naturally gives a factor of $\sqrt{2 \pi T_s/m_s}$.
\end{proof}
\textbf{Further discussions of the Maxwellian velocity distribution and its connection to thermodynamic equilibrium.}

\eqr{\ref{eq:entropyExtremum}} is often referred to as the principle of detailed balance.
To give ourselves physical intuition for what it means for the time evolution of the entropy to vanish, we must consider what we mean by the plasma being in thermodynamic equilibrium.
A useful way to define equilibrium is that every process ongoing in the plasma is exactly compensated by its reverse, e.g., every Coulomb collision a particle in the plasma experiences is exactly balanced by an equal and opposite Coulomb collision.
The contribution of Coulomb collisions to the plasma's dynamics would then vanish.
But the contribution of Coulomb collisions vanishing was exactly the requirement for entropy production to disappear.
Inevitably, the velocity distribution function for which Coulomb collisions are ``in balance'' defines our equilibrium state and the state of maximum entropy.

There are additional subtleties worth mentioning; for example, we have used the total entropy vanishing to derive the Maxwellian as the maximum entropy distribution, but the plasma is free to be a different Maxwellian at each point in configuration space since the density, flow, and temperature may vary in space.
In this case, the entropy density can be maximized at a given configuration space location, but the total entropy may not yet be maximized.
For example, a spatially varying Maxwellian may itself be unstable and drive the system to a still higher entropy state.

We wish to make one additional note about the interconnection between the Maxwellian velocity distribution, the Fokker--Planck equation, and the entropy.
The Maxwellian velocity distribution is actually the naturally arising weight function when considering additional properties of the Fokker--Planck operator in \eqr{\ref{eq:collisionalComponent}}.
For example, we can show that the Fokker--Planck operator is self-adjoint, i.e., for arbitrary functions $g(\mvec{x},\mvec{v}, t)$, $f(\mvec{x},\mvec{v}, t)$, 
\begin{align}
\left ( g,  \pfrac{f^c}{t} \right )_{f_M} = \left ( f,  \pfrac{g^c}{t} \right )_{f_M}, \label{eq:lbo-adjoint}
\end{align}
with the inner product defined as
\begin{align}
( f, g )_{f_M} = 
\int_{K\setminus\Omega} \frac{1}{f_M}f g \dv. \label{eq:maxwellInnerProduct}
\end{align}
Note that $(\cdot, \cdot)_{f_M}$ is a bilinear operator taking two arguments, defined by the integral equation in \eqr{\ref{eq:maxwellInnerProduct}}.
Here, we consider only the integrals over velocity space for simplicity and $f_M$ is the Maxwellian for which the collision operator vanishes.
Note that we have dropped the species subscript.
Integrating \eqr{\ref{eq:lbo-adjoint}} by parts we get
\begin{align}
    \left ( g, \pfrac{f^c}{t} \right )_{f_M} =
    -\int_{K\setminus\Omega}
    \gv\left(\frac{g}{f_M}\right)
    \cdot
    \left[
      (\mvec{v}-\mvec{u}) f + \frac{T}{m} \gv f
    \right]
    \dv.
\end{align}
We have the identity
\begin{align}
\frac{T}{m} f_M \gv\left(\frac{f}{f_M}\right) = (\mvec{v}-\mvec{u}) f + \frac{T}{m} \gv f.
\end{align}
Using this identity leads to
\begin{align}
      \left ( g, \pfrac{f^c}{t} \right )_{f_M} =
    - \frac{T}{m} \int_{K\setminus\Omega}  
    f_M \gv\left(\frac{g}{f_M}\right) \cdot \gv\left(\frac{f}{f_M}\right)
    \dv.
    \label{eq:selfad}
\end{align}
This equation is symmetric in $f$ and $g$ from which the self-adjoint property follows.

As an aside, the self-adjoint property indicates that the eigenvalues of the operator are all real and hence all solutions are damped.
In other words, the Fokker--Planck operator in the VM-FP system of equations does not support any oscillatory modes.
One can show that the eigenfunctions of the operator \eqr{\ref{eq:collisionalComponent}} are simply the multi-dimensional tensor Hermite functions~\citep{Grant:1967,Hammett:1993,Harris:2004,Anderson:2007a,Patarroyo:2019} and each mode is damped proportional to the mode number.

We can use the self-adjoint property to discuss the behavior of the distribution function squared, $f^2$, at least in this norm with the Maxwellian weight.
If we set $g=f$ in \eqr{\ref{eq:selfad}} we get
\begin{align}
\int_{K\setminus\Omega}
\frac{f}{f_M} \pfrac{f^c}{t}
\dv
& =
\frac{d}{dt} \int_{K\setminus\Omega} \frac{1}{2}\frac{\left(f^c\right)^2}{f_M} \dv \notag \\
& = -\frac{T}{m} \int_{K\setminus\Omega} f_M \gv\left(\frac{f}{f_M}\right) \cdot  \gv\left(\frac{f}{f_M}\right) \dv \leq 0,
\end{align}
which shows that the Fokker--Planck operator will decay $f^2/f_M$ integrated over velocity space.
But what about $f^2$, the $L^2$ norm, without the Maxwellian weight?

We previously discussed the $L^2$ norm of the collisionless component of the VM-FP system of equations in Proposition~\ref{prop:collisionlessL2}, showing it is a conserved quantity in the evolution of the distribution function from the collisionless part of the VM-FP system of equations.
We proceed in a similar fashion to Proposition~\ref{prop:collisionlessL2}, but now with the Fokker--Planck equation,
\begin{align}
    \frac{d}{dt}\int_{K\setminus\Omega} \frac{1}{2} f^2 \dv = 
    -\int_{K\setminus\Omega} \gv f\cdot \left [ (\mvec{v}-\mvec{u})f + \frac{T}{m} \gv f \right ] \dv,
\end{align}
where we have already integrated by parts once and used our velocity space boundary conditions to eliminate the surface term. We now write the first term as
\begin{align}
    \gv f\cdot (\mvec{v}-\mvec{u}) f = \gv\bigg( \frac{1}{2}f^2 \bigg) \cdot (\mvec{v}-\mvec{u}) = \mvec{v} \cdot \gv\bigg( \frac{1}{2}f^2 \bigg) - \gv\cdot\bigg( \mvec{u}\frac{1}{2}f^2 \bigg).
\end{align}
The second term is a total derivative and will vanish on upon the use of the divergence theorem and our velocity space boundary conditions. This procedure leaves
\begin{align}
    \frac{d}{dt}\int_{K\setminus\Omega} \frac{1}{2} f^2 \dv = 
    -\int_{K\setminus\Omega} \mvec{v} \cdot \gv\bigg( \frac{1}{2}f^2 \bigg)
    + \frac{T}{m} |\gv f|^2 \dv.
\end{align}
Performing integration by parts on the first term we obtain
\begin{align}
    \frac{d}{dt}\int_{K\setminus\Omega} \frac{1}{2} f^2 \dv = 
    \int_{K\setminus\Omega} \frac{3}{2}f^2 - \frac{T}{m} |\gv f|^2 \dv.
\end{align}
For a Maxwellian, the right-hand side vanishes,
\begin{align}
    \frac{d}{dt}\int_{K\setminus\Omega} \frac{1}{2} f_M^2 \dv & = \int_{K\setminus\Omega} \frac{3}{2}f_M^2 - \frac{T}{m} \left ( -\frac{m (\mvec{v} - \mvec{u})}{T} f_M \right )^2 \dv, \notag \\
    & = \int_{K\setminus\Omega} \frac{3}{2}f_M^2 - \frac{m}{T}|\mvec{v} - \mvec{u}|^2 f_M^2 \dv = 0,
\end{align}
but one can construct perturbations on the Maxwellian that may change the sign. To see this, perform a perturbation around a Maxwellian $f = f_M + \delta f$ to get the variation,
\begin{align}
    \delta \frac{d}{dt} \int_{K\setminus\Omega} \frac{1}{2} f^2 \dv
    & =
    \int_{K\setminus\Omega} \left( 3 f_M - 2 \frac{T}{m} \nabla^2_{\mvec{v}} f_M \right) \delta f \dv, \notag \\
    & =
    \int_{K\setminus\Omega} \left(3 - \frac{2 m |\mvec{v} - \mvec{u}|^2}{T}\right)f_M \delta f \dv.
\end{align}
% Need to be more careful with this statement about the L2 norm. 
Clearly, $\delta f$ can be of any sign. This result shows that the $L^2$ norm is not monotonic and the Maxwellian is not the extremum of the $L^2$ norm. Physically, as the drag velocity $\mvec{v}-\mvec{u}$ is compressible, the contribution from the drag term cannot be turned into a total derivative. The compressibility of the drag term is in contrast to the collisionless case, in which the phase-space velocity is incompressible and hence the phase-space integrated $f^2$ is constant.

We have focused on these additional properties of the Fokker--Planck collision operator---the operator is self-adjoint and decays $f^2/f_M$, but not $f^2$---to make the connection between the Maxwellian velocity distribution and the entropy production of the operator even more explicit.
Proposition~\ref{prop:collisionEntropy}, that the VM-FP system of equations obeys the Second Law of Thermodynamics, and Corollary~\ref{coro:HTheorem}, that the VM-FP system of equations obeys Boltzmann's H-theorem, are inseparable, and Corollary~\ref{coro:HTheorem} naturally follows from Proposition~\ref{prop:collisionEntropy}.
The fact that the Maxwellian velocity distribution is then a natural weight function for discussing additional properties of the collision operator in the VM-FP system of equations should thus be unsurprising, and we cannot avoid including this weight function when discussing the behavior of quantities such as the distribution function squared, $f^2$.

We will conclude this discussion with one final way to think about the connection between the Maxwellian velocity distribution, entropy production, and the $1/f_M$ weighting of the inner product.
$1/f_M$ naturally arises when measuring how much a distribution function deviates away from a Maxwellian in terms of entropy.
In other words, writing $f=f_M + \delta f$, then the entropy $S[f] = - \int_K f \ln(f) \thinspace \dv$ as a functional of $f$ can be written as,
\begin{align}
    S[f_M+\delta f] = S[f_M] - (1/2) \int (\delta f)^2 / f_M \dv + \ldots,
\end{align}
through second order.
This expansion is consistent with the result that any small deviation, $\delta f \ll f_M$, away from a Maxwellian is a state of lower entropy.
Note that, to derive this, we have made use of
\begin{align}
    \int \mvec{v}^p\delta f \dv = 0 \quad \textrm{ for } p=0, 1, 2,
\end{align}
because the Maxwellian $f_M$ has the same zeroth through second moments as $f$ by construction.
In other words, any finite zeroth through second moments in $\delta f$ could just be absorbed into the Maxwellian $f_M$, and $f_M$ redefined.
This norm for $\delta f$ is equivalent to a norm on the total $f$, plus a constant, since
\begin{align}
    \int f^2 / f_M \dv = \int (f_M + \delta f)^2 / f_M  \dv = n + \int (\delta f)^2 / f_M \dv,
\end{align}
where the density $n = \int  f \dv$ is conserved by the collision operator.
This result that 
\begin{align}
    S[f_M+\delta f] = {\rm constant} - (1/2) \int  f^2 / f_M  \dv  + \ldots,
\end{align} 
shows a relationship between the collision operator causing the entropy to be never decreasing and the $1/f_M$-weighted norm to be never increasing.

% Commenting out appendix on the noise in particle methods for now
%\include{AppendixD}

\renewcommand{\baselinestretch}{1}
\small\normalsize

%Bibliography
\newpage
\addcontentsline{toc}{chapter}{Bibliography}
\bibliographystyle{abbrvnat}
\bibliography{thesis_abbrev.bib,thesis-chapter1.bib,thesis-chapter2.bib,thesis-chapter3.bib,thesis-chapter4.bib,thesis-chapter5.bib}

\end{document}